\documentclass[11pt]{article}

\pdfoutput=1

\usepackage{announce_macros}

\bibliography{homological_toolbox_refs}

\begin{document}

\long\def\symbolfootnote[#1]#2{\begingroup%
\def\thefootnote{\fnsymbol{footnote}} \footnote[#1]{#2}\endgroup}

\title{Homological Tools for the Quantum Mechanic}
\author{Tom Mainiero\thanks{tom.mainiero@rutgers.edu or mainiero@physics.utexas.edu}}
\affil{New High Energy Theory Center, Rutgers University, Piscataway, NJ 08854, USA}
\date{}

\maketitle

\begin{abstract}
	This paper is an introduction to work motivated by the question ``can multipartite entanglement be detected by homological algebra?"
	We introduce cochain complexes associated to multipartite density states whose cohomology detects factorizability. 
	The $k$th cohomology components of such cochain complexes produce tuples of $(k+1)$-body operators that are non-locally correlated due to the non-factorizability of the state.
	Associated Poincar\'{e} polynomials are invariant under local invertible linear transformations (automorphisms that decompose as tensor products of automorphisms). 
	These complexes can be considered as a step toward realizing mutual information as an Euler characteristic. 
	We motivate the definition of the ``state index" associated to a multipartite state: a three-parameter function which is invariant under local invertible transformations, well-behaved under tensor products of states, and interpolates between multipartite mutual information and the integer-valued Euler characteristics of our complexes.
	We compute cohomologies and state indices of multipartite W and GHZ states.
	The approach in this paper is directed toward practitioners of finite-dimensional quantum mechanics, although the machinery developed generalizes far beyond.  
	Some results are applicable in infinite dimensions and should be generalizable to the context of quantum field theory. 
	In order to compensate for the long length, a detailed summary is provided in the introduction section.
\end{abstract}

\tableofcontents

\section{Introduction \label{sec:introduction}}
We begin by considering two famous ``factorizability" questions:
\begin{enumerate}
	\item[(C)] Let $\Omega = \Omega_{1} \times \cdots \times \Omega_{n}$ be a Cartesian product of finite sets; suppose we are handed a probability measure $\mu: \Omega \rightarrow \mathbb{R}_{\geq 0}$, is $\mu$ a product measure?
	I.e. does $\mu = \mu_{1} \times \cdots \times \mu_{n}$ for some $\mu_{i}: \Omega_{i} \rightarrow \mathbb{C}$?
	
	\item[(Q)] Let $\hilb = \hilb_{1} \otimes \cdots \otimes \hilb_{n}$ be a tensor product of Hilbert spaces; suppose we are handed a state $\psi \in \hilb$, does $\psi = \psi_{1} \otimes \cdots \otimes \psi_{n}$ for some $\psi_{i} \in \hilb_{i}$?
\end{enumerate}
The first question is one of classical probability theory: one is asking if random variables associated to $\Omega_{i}$ ($\mathbb{C}$-valued functions $\Omega_{i} \rightarrow \mathbb{C}$) are independent of random variables associated to $\Omega_{j}$ for $j \neq i$.
On the other hand, the second question is quantum mechanical in nature: states that fail to factorize in such a way are called \textit{entangled}, and play a fundamental role in quantum information/computation \cite{nc:qcqi,preskill}.
The study of entangled states has also found its way into high-energy physics---particularly after the work of Ryu-Takayanagi \cite{Ryu:2006bv}, relating entanglement entropy (a numerical measure of entanglement in a field-theoretic state) of states in a conformal field theory to the area of minimal surfaces in a gravity dual theory in one spatial dimension higher.
Some of these ideas have discrete, computationally accessible versions \cite{Pastawski:2015qua,Hayden:2016cfa,harlow16} using the techniques of tensor networks  \cite{Verstraete:2004cf, vidal:tn, Verstraete2009} and error correcting codes. 

Expressed purely algebraically, both questions are special cases of a more general question of the form:
\begin{enumerate}
	\item[(QC)] Given an \textit{algebra of random variables} $\mathcal{A}$ (an algebra over a field $k$) that factorizes as a tensor product $\mathcal{A} = \bigotimes_{i = 1}^{n} \mathcal{A}_{i}$ and an \textit{expectation value} linear functional $\mathbbm{E}: \mathcal{A} \rightarrow k$, does $\mathbbm{E}: \mathcal{A} \rightarrow k$ factorize as a tensor product of expectation values on each $\mathcal{A}_{i}$? 
	I.e.\ are there expectation values $\mathbbm{E}_{i}: \mathcal{A}_{i} \rightarrow k$ such that $\mathbbm{E}(a_{1} \otimes a_{2} \otimes \cdots \otimes a_{k}) = \mathbbm{E}_{1}(a_{1}) \mathbbm{E}_{2}(a_{2}) \cdots \mathbbm{E}_{n}(a_{n})$ for all $a_{i} \in \mathcal{A}_{i}$?
\end{enumerate}
In question (C), $\mathcal{A}$ is the commutative algebra of $\mathbb{C}$-valued functions on $\Omega$, and (Q) is a question about the non-commutative algebra of (bounded) endomorphisms of $\hilb$.
Serious practitioners of non-commutative measure theory would tell us that the algebra of random variables $\mathcal{A}$ here should be taken to be a non-commutative $W^{*}$-algebra: special types of $C^{*}$-algebras that abstract von Neumann algebras,\footnote{Lovers of two-dimensional topological field theory might take $\mathcal{A}$ to be a commutative Fr\"{o}benius algebra; most of the results in this paper are generalizable to this situation.} and the expectation value to be a normal state.
In finite-dimensions $W^{*}$-algebras and $C^{*}$-algebras algebras are isomorphic to direct sums (i.e. Cartesian products of algebras) of matrix algebras and normal states are uniquely associated to direct sums/tuples of density operators.

In practice, the tensor factors of the algebra of random variables are associated to disjoint ``local" subsystems or ``observers": e.g.\ lattice sites, individual atoms in a molecule, disjoint regions on a spatial slice of spacetime, disjoint causal diamonds, etc.\  The failure to factorize is indicative of ``non-local" or ``global" correlations among subsystems.
Those words in quotation marks---particularly ``global"---should set off some cohomological alarms, indicating that perhaps there is some cochain complex whose cohomology captures how badly an expectation value fails to factorize.
Yet the reader should not be convinced simply by a couple of vague analogies, so we provide the following diagram of boxes and squiggly arrows as further evidence.

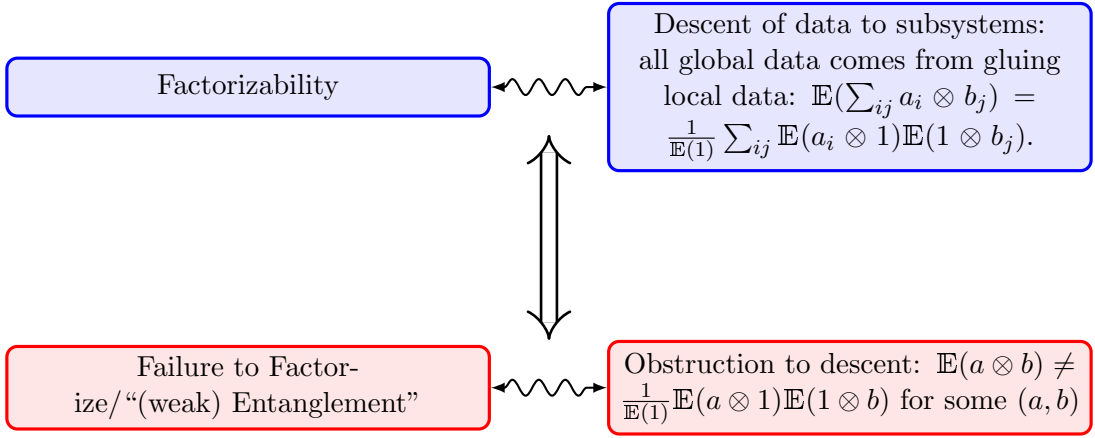
\begin{center}
	\begin{tikzpicture}
		\tikzstyle{block} = [rectangle, draw=blue, thick, fill=blue!10,
		text width=16em, text centered, rounded corners, minimum height=2em]

		\tikzstyle{blockred} = [rectangle, draw=blue, thick, fill=red!10,
		text width=16em, text centered, rounded corners, minimum height=2em]

		\node at (-4, 0)[block,draw=blue,very thick] (fact) {Factorizability};

		\node at (4, 0)[block,draw=blue,very thick] (descent) {Descent of data to subsystems: all global data comes from gluing local data: $\mathbbm{E}(\sum_{i j} a_{i} \otimes b_{j}) = \frac{1}{\mathbbm{E}(1)} \sum_{i j} \mathbbm{E}(a_{i} \otimes 1) \mathbbm{E}(1 \otimes b_{j})$.};

		\node at (-4,-4)[blockred,draw=red,very thick] (failfact) {Failure to Factorize/``(weak) Entanglement"};

		\node at (4, -4)[blockred,draw=red,very thick] (obsdesc) {Obstruction to descent: $\mathbbm{E}(a \otimes b) \neq \frac{1}{\mathbbm{E}(1)}\mathbbm{E}(a \otimes 1) \mathbbm{E}(1 \otimes b)$ for some $(a,b)$};

		\node at (0, -.5) (midpointup) {};

		\node at (0,-3.5) (midpointdown) {};

		\draw[latex-latex,decorate, thick, decoration={snake,amplitude=3pt,pre length=5pt,post length=4pt}] (fact) -- (descent);
		\draw[latex-latex,decorate, thick,  decoration={snake,amplitude=3pt,pre length=5pt,post length=4pt}] (failfact) -- (obsdesc);

		\draw[implies-implies,line width=1pt,double distance=5pt] (midpointup) -- (midpointdown);

	\end{tikzpicture}
\end{center}
		
To elaborate on this diagram: because any random variable is decomposable as a sum of tensor products of local random variables, an expectation value is factorizable if and only if it can be reconstructed via sums and products of local expectation values, i.e.\ ``glued together from local data".
Hence, the failure to be factorizable is captured by tuples of local random variables $(a_{1},\cdots,a_{k})$ such that $\mathbbm{E}(a_{1} \otimes \cdots \otimes a_{k}) \neq \mathbbm{E}_{1}(a_{1}) \mathbbm{E}_{2}(a_{2}) \cdots \mathbbm{E}_{k}(a_{k})$---where  $\mathbbm{E}_{k}$ is (up to a proportionality constant) given by the pullback of $\mathbbm{E}$ along the embedding $A_{k} \hookrightarrow \otimes_{i = 1}^{n} A_{i}$ defined by $a_{k} \mapsto 1_{A_{1}} \otimes 1_{A_{2}} \otimes \cdots \otimes 1_{A_{k-1}} \otimes a_{k} \otimes 1_{A_{k+1}} \otimes \cdots \otimes 1_{A_{n}}$; these are precisely tuples of variables with non-local correlations.
This should be a somewhat more convincing argument for cohomology, as cohomological techniques are precisely designed for extracting those globally defined quantities that fail to be completely determined by locally defined ones; in this situation, the cohomology components should be directly related to tuples of operators with non-local correlations representing the obstructions to factorization.

On the other hand, there are already plenty of \textit{numerical} quantities that measure how states fail to factorize, especially in the bipartite situation (two tensor factors); in this paper, our focus will be on a particularly famous one: the mutual information and its multipartite versions.
Suppose we are given observers\footnote{Also referred to as ``primitive subsystems" or ``tensor factors" throughout this paper.} $\mathsf{A}$ and $\mathsf{B}$, each whom is equipped with an associated local algebra of random variables $\mathcal{A}_{\mathsf{A}}$ and $\mathcal{A}_{\mathsf{B}}$ (respectively).
Then given the data of an expectation value functional $\mathbbm{E}: \mathcal{A}_{\mathsf{A}} \otimes \mathcal{A}_{\mathsf{B}} \rightarrow \mathbb{C}$, we associate the real-valued quantity
\begin{align*}
	I_{2} = S_{\mathsf{A}} + S_{\mathsf{B}} - S_{\mathsf{\sAB}}
\end{align*}
where $\mathsf{\sAB}$ is a shorthand notation for the joint system of $\mathsf{A}$ and $\mathsf{B}$, and $S_{\mathsf{T}},\,$ for $\mathsf{T} \in \{\sA, \sB, \sAB \}$ denotes the entropy associated to the system $\mathsf{T}$.
When the entropy is defined properly, $I_{2}$ is always a positive quantity, and vanishes if and only if the expectation value on $\sAB$ factorizes.
As the name suggests, $I_{2}$ is a quantitative measure of the information that is shared between the disjoint subsystems $\sA$ and $\sB$; hence, a good measure of non-local correlations.
For more than two systems, there are several possible generalizations of the mutual information; one approach takes seriously the inclusion-exclusion nature of the sum of entropies in the definition of $I_{2}$:  for an expectation value $\mathbbm{E}: \bigotimes_{p \in P} \mathcal{A}_{p} \rightarrow \mathbb{C}$, where the indexing set $P$ is ordered and of size $n$, we define
\begin{align*}
	I_{n} := \sum_{k = 0}^{n} (-1)^{k} \left[ \sum_{\{T \subseteq P: |T| = k+1 \} } S_{T} \right]
\end{align*}
where the second nested sum is over all joint subsystems of $k$ elements of $P$.
The multipartite mutual information $I_{n}$ vanishes if there is any subsystem---represented by a subset of $P$---that is independent in the sense that the expectation value functional factorizes with respect to a partition defined by that subsystem and its complement.
In this sense, multipartite mutual information can be thought of as a measure of the information that is shared between all possible subsystems due to any non-local correlations induced by the expectation value functional.\footnote{Unlike the bipartite mutual information, however, tripartite and higher mutual informations can take on negative values.}
The inclusion-exclusion sum defining mutual information has the superficial appearance of an Euler characteristic of some chain complex, with the sum of entropies of systems of size $k$ acting as the dimension of the $k$th cochain component; so one might ask if there is an associated cochain complex whose associated Euler characteristic (an alternating sum of dimensions) is related to mutual information.
Because we expect a cochain complex to be a much richer mathematical object than a numerical quantity, if one were able to produce such a complex, we would likely gain further insight: i.e.\ we might be able to find out not only a numerical measure of the amount of information that is shared between subsystems, but specifically \textit{what} information is shared.
This was already indicated above where we conjectured the cohomology of such complexes should be related to non-locally correlated operators, i.e.\ those operators representing shared information among disjoint subsystems.

Moreover, cohomology maintains information about how objects are glued together in contrast to the Euler characteristic which only depends on the number of things being glued together.
This property of this Euler characteristic is advantageous for ease of computability, but at the cost of information.
One can have two cochain complexes with different cohomologies but the same Euler characteristic, or an individual cochain complex with vanishing Euler characteristic but non-trivial cohomology.
As an extreme example of both phenomena: the Euler characteristic of any orientable compact manifold of odd-dimension is vanishing, making it useless for detecting the difference between two such manifolds, let alone detecting the information about the topology (as opposed to the calculation of homology or cohomology).
Similarly, there  may be situations where the mutual information of two states are identical, but are distinguished at the level of cohomology or the mutual information of a state is vanishing despite non-vanishing cohomology.
This latter situation cannot happen for bipartite systems---where mutual information completely determines factorizability---but can happen for tripartite systems: e.g.\ the GHZ and W-states both have vanishing tripartite mutual informations, but are clearly entangled states.

Cochain complexes motivated by the above discussion are more than just fantastical musings; in this paper we will provide an exposition of definitions and results that will be elaborated and expanded upon in forthcoming papers, likely with a more sophisticated approach.
We limit our approach primarily to the purely quantum mechanical: where algebras of random variables are given by endomorphisms on Hilbert spaces, and expectation value functionals are given by tracing with respect to density states.
In the spirit of maintaining a reasonable length, and to appeal to a wide audience of both physicists and mathematicians, we resist the urge to utilize categorical and homotopical techniques that have not yet found their way into the physics lingo (beyond a few specialists).
This certainly has disadvantages, but emphasizes that the ideas here are primarily linear algebraic in nature; our approach is sufficient for the reader to engage in concrete computations which should be especially appealing to physicist readers.

\subsection{Provided Software}
``Alpha" versions of software for calculating the cohomologies of multipartite states are provided in the arXiv source of this paper; the most up-to-date versions are available on GitHub.\footnote{See: \href{https://github.com/tmainiero}{\url{https://github.com/tmainiero/homological-tools-4QM-octave}} and \href{https://github.com/tmainiero}{\url{https://github.com/tmainiero/homological-tools-4QM-mathematica}}}
  There are two versions of this software corresponding to implementations in \textit{Mathematica} and \textit{Octave/Matlab}.
  Both versions are able to compute dimensions of cohomology components (which may be thought of as coefficients of associated Poincar\'{e} polynomials); the Mathematica version is able to output explicit generators for each cohomology component. 
  Some documentation is provided for both versions (in the form of a quick-start notebook for the Mathematica version).
  Readers are encouraged to experiment with this software, using this paper as a guideline for definitions and observations of the properties of associated cohomologies.\footnote{As a disclaimer: this software was written in a very early stage of the author's understanding; as a result, its implementation uses `co-\v{C}ech" rather than the \v{C}ech techniques suggested in this paper.
  However, it is possible to verify the resulting cochain complexes are chain isomorphic via a ``sign-correcting" chain isomorphism.
  The author intends to rectify this with updated software in future versions of this paper.}   

The reader with either no knowledge of homological techniques, or an itching impatience can use the software as a black box and employ the following famous schematic.
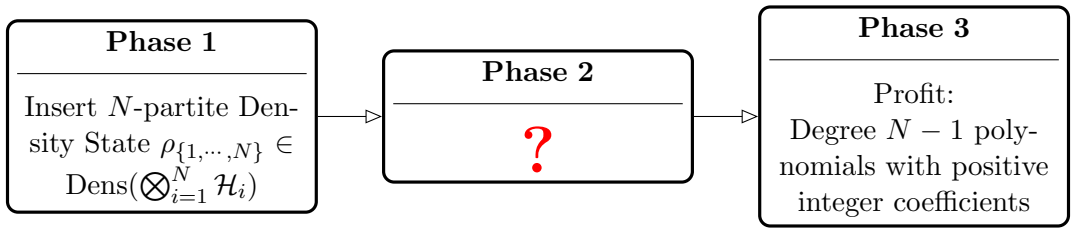
\begin{center}
	\begin{tikzpicture}
		\tikzstyle{block} = [rectangle, thick,
		text width=10em, text centered, rounded corners, minimum height=2em]

		\node at (-5,0)[block,draw=black,very thick,align=center] (collectunderpants) {\textbf{Phase 1}\\ \vspace{3mm} \hrule \vspace{3mm} Insert $N$-partite Density State $\dens{\rho}_{\{1, \cdots, N \}} \in \Dens(\bigotimes_{i = 1}^{N} \hilb_{i})$};

		\node at (0,0)[block,draw=black,very thick,align=center] (magic) {\textbf{Phase 2}\\ \vspace{3mm} \hrule \vspace{3mm} \textcolor{red}{\Huge \textbf{?}}};

		\node at (5,0)[block,draw=black,very thick,align=center] (profit) {\textbf{Phase 3}\\ \vspace{3mm} \hrule \vspace{3mm} Profit:\\ Degree $N-1$ polynomials with positive integer coefficients};

		\draw[->, >=open triangle 45] (collectunderpants) to (magic);

		\draw[->, >=open triangle 45] (magic) to (profit);
	\end{tikzpicture}
\end{center}
The strange patterns in its output should hopefully be enough inspiration to dig deeper.

\subsection{Summary and Key Results \label{sec:key_results}}
In order to compensate for the long length of this paper, we provide a detailed summary; the reader can reference particular sections for even further detail. 
This paper can be divided into two parts which can, for the most part, be read independently. 
The first part consists of sections \S\ref{sec:alg_and_states}-\ref{sec:multipartite_complexes} and is focused on the construction of cohomology and its properties.
The second part, which is contained in \S\ref{sec:categorification}, is an exposition focused on a motivation of the state index via an attempt to realize multipartite mutual information as an Euler characteristic of some cohomology theory.
The ideas of both sections are mildly mixed together in \S\ref{sec:W_vs_GHZ}, which consists of some concrete computations.

\subsubsection{Cohomology For Bipartite Density States}
In the first part of this paper we apply homological techniques to the study of bipartite density states: density states on the tensor product of two Hilbert spaces.
Formally a \textit{bipartite density state} is given by a tuple of data $(\hilb_{\sA}, \hilb_{\sB}, \dens{\rho}_{\sAB})$ where:
\begin{itemize}
	\item $\sA$ and $\sB$ label the two \textit{tensor factors} or \textit{primitive subsystems},

	\item $\hilb_{\sA}$ and $\hilb_{\sB}$ are finite-dimensional Hilbert spaces,\footnote{Most of the results in this paper generalize appropriately to infinite dimensions.  See \S\ref{sec:alg_and_states}.} and 

	\item $\dens{\rho}_{\sAB}$ is a positive-semidefinite trace 1 endomorphism of $\hilb_{\sA} \otimes \hilb_{\sB}$ i.e.\ a \textit{density state} on $\hilb_{\sA} \otimes \hilb_{\sB}$.
\end{itemize}
To lighten notation we denote such a tuple of data as a boldface version of its density state $\bdens{\rho}_{\sAB}$.
From any bipartite density state $\bdens{\rho}_{\sAB}$ we construct associated chain complexes 
\begin{align*}
	\cpx{g}(\bdens{\rho}_{\sAB}) &= \cdots 
	\longleftarrow 0 \longleftarrow
	\mathbb{C} 
	\overset{\partial^{\cpx{g}}_{0}}{\longleftarrow}
	\mathtt{g}^{0}(\bdens{\rho}_{\sAB})
	\overset{\partial^{\cpx{g}}_{1}}{\longleftarrow}
	\mathtt{g}^{1}(\bdens{\rho}_{\sAB})
	\longleftarrow 0 \longleftarrow \cdots, \\
	\cpx{e}(\bdens{\rho}_{\sAB}) &= \cdots 
	\longleftarrow 0 \longleftarrow
	\mathbb{C} 
	\overset{\partial^{\cpx{e}}_{0}}{\longleftarrow}
	\mathtt{e}^{0}(\bdens{\rho}_{\sAB})
	\overset{\partial^{\cpx{e}}_{1}}{\longleftarrow}
	\mathtt{e}^{1}(\bdens{\rho}_{\sAB})
	\longleftarrow 0 \longleftarrow \cdots,
\end{align*}
and cochain complexes
\begin{align*}
	\cpx{G}(\bdens{\rho}_{\sAB}) &= \cdots 
	\longrightarrow 0 \longrightarrow
	\mathbb{C} 
	\overset{d_{\cpx{G}}^{-1}}{\longrightarrow}
	\mathtt{G}^{0}(\bdens{\rho}_{\sAB})
	\overset{d_{\cpx{G}}^{0}}{\longrightarrow}
	\mathtt{G}^{1}(\bdens{\rho}_{\sAB})
	\longrightarrow 0 \longrightarrow \cdots, \\
	\cpx{E}(\bdens{\rho}_{\sAB}) &= \cdots 
	\longrightarrow 0 \longrightarrow
	\mathbb{C} 
	\overset{d_{\cpx{E}}^{-1}}{\longrightarrow}
	\mathtt{E}^{0}(\bdens{\rho}_{\sAB})
	\overset{d_{\cpx{E}}^{0}}{\longrightarrow}
	\mathtt{E}^{1}(\bdens{\rho}_{\sAB})
	\longrightarrow 0 \longrightarrow \cdots. 
\end{align*}
In Prop.~\ref{prop:bipartite_cocomplex_to_dualcomplex} it is shown that the cochain complexes are the canonically duals of the chain complexes:
\begin{align*}
	\cpx{G}(\bdens{\rho}_{\sAB}) &\cong \left[ \cpx{g}(\bdens{\rho}_{\sAB}) \right]^{\vee}, \\
	\cpx{E}(\bdens{\rho}_{\sAB}) &\cong \left[ \cpx{e}(\bdens{\rho}_{\sAB}) \right]^{\vee};
\end{align*}
which results isomorphisms:
\begin{align*}
	H^{k}\left[\cpx{G}(\bdens{\rho}_{\sAB})\right] &\cong H_{k}\left[\cpx{g}(\bdens{\rho}_{\sAB})\right]^{\vee},\\
	H^{k}\left[\cpx{E}(\bdens{\rho}_{\sAB})\right] &\cong H_{k}\left[\cpx{e}(\bdens{\rho}_{\sAB})\right]^{\vee},
\end{align*}
for $k \in \mathbb{Z}$.
With this in mind, we primarily focus on cochain complexes and cohomology (rather than complexes and homology).
Letting $\cpx{C}(\bdens{\rho}_{\sAB})$ denote one of the cochain complexes above, its components are of the form 
\begin{align*}
	\cpx{C}^{0}(\bdens{\rho}_{\sAB}) &= \mathtt{B}(\dens{\rho}_{\sA}) \times \mathtt{B}(\dens{\rho}_{\sB}), \\
	\cpx{C}^{1}(\bdens{\rho}_{\sAB}) &= \mathtt{B}(\dens{\rho}_{\sAB})
\end{align*}
where $\dens{\rho}_{\sX}$ is the reduced density state on $\hilb_{\sX}$, and---for any density state $\dens{\hilb}$ on a Hilbert space $\hilb$---$\mathtt{B}(\dens{\rho})$ is one of the building blocks of \S \ref{sec:building_blocks}.
In particular, $\mathtt{B}(\dens{\rho})$ is constructed as a subspace of $\algebra{\hilb}$---the space of (bounded) endomorphisms of $\hilb$---that depends on the data of the density state $\dens{\rho}$.
The coboundary maps of the cochain complexes are constructed by ``tensoring by the appropriate identity" (the trace-pairing dual of partial trace maps), followed by an projection onto the subspace of operators of interest. A pictorial example of such a composition of maps is the following: 
\begin{equation}
	\begin{tikzpicture}[baseline=(current bounding box.center), descr/.style={fill=white,inner sep=1.5pt}]
		\matrix (m) [
		matrix of math nodes,
		row sep=2em,
		column sep=4em,
		text height=1.5ex, text depth=0.25ex
		]
		{\algebra{\hilb_{\sA}}  & \algebra{\hilb_{\sAB}} \\
		\mathtt{B}(\dens{\rho}_{\sA}) & \mathtt{B}(\dens{\rho}_{\sAB})\\};

		\draw[->] (m-1-1) to node[above, font=\scriptsize]{$a \mapsto a \otimes 1_{\sB}$} (m-1-2);
		\draw[right hook->] (m-2-1) to (m-1-1);
		\draw[->] (m-1-2) to node[right, font=\scriptsize]{$\mathrm{proj}_{\sAB}$}(m-2-2);
	\end{tikzpicture}
	\label{eq:lift_and_project}
\end{equation}
Where the form of the linear map labelled ``$\mathrm{proj}_{\sAB}$" depends on the building block $\mathtt{B}(\dens{\rho}_{\sAB})$ and involves either right multiplication or compression (left and right multiplication) by the \textit{support projection} of $\dens{\rho}_{\sAB}$: the orthogonal projection onto the image of $\dens{\rho}_{\sAB}$.
Then, for instance, the map $d_{\cpx{C}}^{0}$ takes the form:
\begin{align*}
	d_{\cpx{C}}^{0}(a,b) = \mathrm{proj}_{\sAB} \left[ a \otimes 1_{\sB} - 1_{\sA} \otimes b \right].
\end{align*}
The map $d^{-1}_{\mathrm{C}}$ acts by sending $\lambda \in \mathbb{C}$ to $(\lambda \supp_{\dens{\rho}_{\sA}}, \lambda \supp_{\dens{\rho}_{\sB}})$.
The fact that the coboundary of our (co)chain complexes square to zero relies on a compatibility condition between the support projections of reduced density states.
The key lemmata here are Lem.~\ref{lem:compat_of_supports} (and its manifestation in terms of partial traces: Lem.~\ref{lem:compat_of_supports_predual}).

Because the building blocks are explicit subspaces of algebras of operators, an element of either the first or second cohomology components is identifiable with an equivalence classes of operators.
In particular, the zeroth cohomology components $\ker(d_{\cpx{C}}^{0})/\image(d_{\cpx{C}}^{0})$ of either cochain complex is identifiable with equivalence classes of pairs of operators/random variables $(a,b) \in \algebra{\hilb_{\sA}} \times \algebra{\hilb_{\sB}}$.
Following some of the ideas outlined in the previous section, our hope is that such elements act as obstructions to factorizability of the density state $\dens{\rho}_{\sAB}$ (or, equivalently, its associated expectation value functional). 
This is indeed the case in the context of the cochain complex $\cpx{G}(\bdens{\rho}_{\sAB})$ (referred to as the \textit{GNS complex}).

The building blocks $\mathtt{B}(\dens{\rho})$ used to construct $\cpx{G}(\bdens{\rho}_{\sAB})$ are denoted $\mathtt{GNS}(\dens{\rho})$; they are formed by noticing that to each density state $\dens{\rho} \in \Dens(\hilb)$ there is an associated left module $\mathtt{GNS}(\dens{\rho})$---dubbed the \textit{GNS module}---for the algebra $\algebra{\hilb}$ of (bounded) endomorphisms of the Hilbert space $\hilb$;\footnote{By ``module" we mean a module in the purely algebraic sense: i.e.\ a module for the underlying $\mathbb{C}$-algebra of $\algebra{\hilb}$.} it can be realized as a left submodule of $\algebra{\hilb}$.
In finite dimensions there is a canonical isomorphism (of left $\algebra{\hilb}$-modules):
\begin{align*}
	\mathtt{GNS}(\dens{\rho}) \cong \hilb \otimes \image(\dens{\rho})^{\vee} \hookrightarrow \algebra{\hilb}
\end{align*}
with the action of $\algebra{\hilb}$ acting on the left tensor component in the obvious manner.
On one hand, the GNS module is a precursor to the Gelfand-Neumark-Segal representation---this is what inspires its name---but in this paper we emphasize that it can be identified with representatives of \textit{right essential equivalence classes} of operators (c.f.\ Def.~\ref{def:right_essential_equivalence}): a non-commutative generalization of the notion of ``almost-everywhere" equivalence classes of functions that appear in commutative measure theory.

The zeroth cohomology component of $\cpx{G}(\bdens{\rho}_{\sAB})$ is identifiable with operators on separate tensor components that are ``maximally correlated".
Let us make this statement more precise: to the bipartite density state $\bdens{\rho}_{\sAB}$ we associate (c.f.\ Def.~\ref{def:cov_var}):
\begin{itemize}
	\item A sesquilinear form $\Cov: \algebra{\hilb_{\sA}} \times \algebra{\hilb_{\sB}} \rightarrow \mathbb{C}$ that measures the covariance between operators associated to different tensor components (a linear measure of ``global correlations");

	\item Quadratic forms $\Var_{\sX}: \algebra{\hilb_{\sX}} \rightarrow \mathbb{C}$, for $\sX \in \{\sA, \sB \}$. 
\end{itemize}
Then, as with the covariance and variances one encounters in classical probability theory, one can use Cauchy-Schwarz arguments to show that the square of the covariance is bounded above by the product of variances.
Quoting Lem.~\ref{lem:covariance_saturation}: for any pair of local\footnote{Throughout this paper the word ``local" will be used to refer to quantities associated to or operations on each tensor factor/primitive subsystem.  When partitioning up the tensor factors in a different manner, the word ``local" is in reference to the new tensor factors/primitive subsystems defined by the components of the partition.} 
operators $(a,b) \in \algebra{\hilb_{\sA}} \times \algebra{\hilb_{\sB}}$ we have
\begin{align*}
		\left |\Cov(a,b)  \right|^{2} \leq \Var_{\sA}(a) \Var_{\sB}(b).
\end{align*}
The quantities $\Cov$ and $\Var_{\sX}$ restrict to functions on the subspace $\mathtt{GNS}(\dens{\rho}_{\sX})$;\footnote{As mentioned above, $\mathtt{GNS}(\dens{\rho}_{\sX}$ is canonically isomorphic to the module of right essential equivalence classes: which is formed as a quotient of $\mathcal{A}(\hilb)$ by a left-ideal.
Luckily $\Cov$ and $\Var_{\sX}$ descend naturally to this quotient.}
this allows us to make sense of the following identification;
\begin{corollary*}[]{(C.f.\ Cor.~\ref{cor:GNS_cohom_saturated_cov})}{}
	\begin{align*}
		H^{0}[\cpx{G}(\bdens{\rho}_{\sAB})] &= \left \{
			(a,b) \in \mathtt{GNS}(\dens{\rho}_{\sA}) \times \mathtt{GNS}(\dens{\rho}_{\sB}) : 
			\text{\parbox{13em}{
					\centering \small $\Cov(a,b) = \Var_{\sA}(a) = \Var_{\sB}(b)$ \\
					and \\ 
				$\Tr[\dens{\rho}_{\sA}a] = \Tr[\dens{\rho}_{\sB}b]$}
			}
		\right \} 
		\Bigg/ \{\Cov = 0\} ,
	\end{align*}
	where $\{\Cov = 0 \}$ is shorthand for $\{(a,b) \in \mathtt{GNS}(\dens{\rho}_{\sA}) \times \mathtt{GNS}(\dens{\rho}_{\sB}): \Cov(a,b) = 0 \}$. 
\end{corollary*}

This is in harmony with what was said in \S\ref{sec:introduction}: there should be some cohomology theory that encodes the failure to recover global expectation values from local ones---more precisely, this cohomology theory should encode pairs of operators $(a,b)$ such that the expectation value of $a \otimes b$ cannot be constructed as the product of expectation values.
A quantitative measure of the failure of such a factorization of expectation values is given by the magnitude of the covariance; in this sense, non-trivial elements of $H^{0}[\cpx{G}(\bdens{\rho}_{\sAB})]$ represent the ``worst offenders". 

When the density state $\bdens{\rho}_{\sAB}$ is a \textit{pure} bipartite state, then one can use the covariance saturation condition to compute the zeroth cohomology component of the GNS complex explicitly in terms of a Schmidt decomposition.
\begin{theorem*}{(C.f.\ Thm.~\ref{thm:pure_bipartite_cohomology})}{}
	Let $\bdens{\rho}_{\sAB}$ be a pure bipartite density state with $\dens{\rho}_{\sAB} = \psi \otimes \psi^{\vee}$ for some $\psi \in \hilb_{\sA} \otimes \hilb_{\sB}$. Decompose $\psi$ as:
	\begin{align*}
		\psi = \sum_{i = 1}^{S} \sqrt{p_{i}} \xi^{\sA}_{i} \otimes \xi^{\sB}_{i}
	\end{align*} 
	for positive coefficients $\{p_{i} \}_{i = 1}^{S} \subseteq \mathbb{R}_{>0}$ and orthonormal vectors $\{\xi^{\sX}_{i}\}_{i = 1}^{S} \subset \hilb_{\sX}$---i.e.\ a Schmidt decomposition of $\psi$---then 
	\begin{align*}
		H^{0}[\cpx{G}(\bdens{\rho}_{\sAB})] = \operatorname{span}_{\mathbb{C}} \left[(\mathbbm{e}^{\sA}_{ij},\mathbbm{e}^{\sB}_{ij}): i,j = 1,\cdots, S \right]/\operatorname{span}_{\mathbb{C}} \{(\supp_{\sA}, \supp_{\sB}) \}
	\end{align*}
	where $\supp_{\sX}$ denotes the support projection of $\dens{\rho}_{\sX}$ (the orthogonal projection onto the image of $\dens{\rho}_{\sX}$), and
	\begin{align*}
		\mathbbm{e}^{\sX}_{ij} := \xi^{\sX}_{i} \otimes \left(\xi^{\sX}_{j} \right)^{\vee} \in \mathtt{GNS}(\dens{\rho}_{\sX})
	\end{align*}
	for $\sX \in \{\sA, \sB \}$.
	In particular, $\dim H^{0}[\cpx{G}(\bdens{\rho}_{\sAB})] = S^2 - 1$, where $S$ is the Schmidt rank of $\psi$ (equivalently, $\dim H^{0}[\cpx{G}(\bdens{\rho}_{\sAB})] = \mathrm{rank}(\rho_{\sA})^2 -1 = \mathrm{rank}(\rho_{\sB})^2 -1$).
\end{theorem*}
Because the cohomology of bipartite complexes are concentrated in degrees 0 and 1, one can determine the dimension of one cohomology component from the other by an Euler characteristic computation.  It follows that for a pure bipartite density state as in the theorem: 
\begin{align}
	\dim H^{1}[\cpx{G}(\bdens{\rho}_{\sAB})] = (d_{\sA} - S) (d_{\sB} - S) 
	\label{eq:H1_pure_bipartite_intro}
\end{align}
which provides a very coarse measure of how far $\psi$ is from being ``maximally entangled".

The building blocks for the components of the cochain complex $\cpx{E}(\bdens{\rho}_{\sAB})$ are identifiable with spaces of endomorphisms of $\mathtt{GNS}(\dens{\rho})$ that are equivariant with respect to the left $\algebra{\hilb}$-action---a vector space we call $\mathtt{Com}(\dens{\rho})$: it is canonically identifiable with (bounded) endomorphisms on the image of $\dens{\rho}$
\begin{align*}
	\mathtt{Com}(\dens{\rho}) \cong \End[\image(\dens{\rho})] \hookrightarrow \algebra{\hilb}.
\end{align*}
$\mathtt{Com}(\dens{\rho})$ is a purely algebraic version of what is referred to as the ``commutant" in $C^{*}$-algebra theory, so we call $\cpx{E}(\bdens{\rho}_{\sAB})$ the \textit{commutant complex}.\footnote{Although we are working with purely algebraic structures---in particular we forget the inner product on $\mathtt{GNS}(\dens{\rho})$---it happens that the space of equivariant endomorphisms of the underlying algebraic module agrees with the commutant in the $C^{*}$-algebraic sense; so this terminology is not misleading.}  The commutant complex and its cohomology might be considered as a subtle way of recovering some information about the module structure of the GNS representation: information that is lost when forgetting down from modules to vector spaces when computing cochain complexes and cohomology.

The cohomology components of the cochain complex $\cpx{E}(\bdens{\rho}_{\sAB})$ can also be computed for a pure bipartite density state $\bdens{\rho}_{\sAB}$ with $\dens{\rho}_{\sAB} = \psi \otimes \psi^{\vee}$ as:
\begin{align*}
	H^{1}[\cpx{E}(\bdens{\rho}_{\sAB})] &= 0,\\
	H^{0}[\cpx{E}(\bdens{\rho}_{\sAB})] &= \{(a,b) \in \mathtt{Com}(\dens{\rho}_{\sA}) \times \mathtt{Com}(\dens{\rho}_{\sB}): \langle \psi, a \psi \rangle = \langle \psi, b \psi \rangle = 0 \}/\mathrm{span}_{\mathbb{C}}\{(\supp_{\sA}, \supp_{\sB}) \}.
\end{align*}
Here the zeroth cohomology component does not have an interpretation as correlated operators,\footnote{For instance, if $\bdens{\rho}$ is pure as Thm.~\ref{thm:pure_bipartite_cohomology}, then $(0, \mathbbm{e}^{\sB}_{ij})$ is the representative of a non-trivial cohomology class for all $1 \leq i,j \leq S$; but this pair of operators has zero covariance.} however, the Schmidt rank does make a reappearance if we compute the dimensions.
\begin{theorem}{(C.f. Thm.~\ref{thm:pure_commutant_cohomology})}{}
	Let $\bdens{\rho}_{\sAB}$ be a bipartite density state with $\dens{\rho}_{\sAB} = \psi \otimes \psi^{\vee}$ for some $\psi \in \hilb_{\sA} \otimes \hilb_{\sB}$, then
	\begin{equation*}
		\dim H^{k}[\cpx{E}(\bdens{\rho}_{\sAB})] =
		\left \{
			\begin{array}{ll} 
				2(S^2 -1) &, \text{if $k=0$} \\
				0 &, \text{otherwise}
			\end{array}
		\right.
	\end{equation*}
	where $S$ is the Schmidt rank of $\psi$ (equivalently given as $\rank \left( \Tr_{\sX}[\psi \otimes \psi^{\vee}] \right)$ for $\sX \in \{\sA, \sB \})$. 
	In particular $\psi$ is factorizable if and only if $H^{k}[\mathrm{E}(\bdens{\rho}_{\sAB})] = 0$ for all $k$.
\end{theorem}

All of these (co)chain complexes obey some desirable properties outlined in \S\ref{sec:bipartite:properties}.
In particular, they are equivariant under local invertible transformations: invertible linear maps applied to the Hilbert space on each tensor factor (c.f.\ Def.~\ref{def:local_invertible_bipartite} and Def.~\ref{def:local_invertible_multipartite}).
As a result, their associated Poincar\'{e} polynomials should be invariant under local invertible transformations.\footnote{The set of equivalence classes of multipartite pure states up to local invertible equivalence coincides with the notion of an SLOCC (Stochastic Local Operations and Classical Communication) class: classes of states that can be transformed into one another (with non-zero probability) via completely positive maps living inside the class of Local Operations and Classical Communication. 
So, using this language, Poincar\'{e} polynomials are SLOCC invariants.
The collection of SLOCC equivalence classes is naturally a projective variety; we do not explore if these polynomials are compatible with the algebraic structure of this projective variety.}

Given that the Schmidt rank classifies pure bipartite density states up to local invertible transformations, then as corollary of the results above we have that the Poincar\'{e} polynomials:
\begin{align*}
	P_{\cpx{G}}(\bdens{\rho}_{\sAB}) &:= \dim H^{0}[ \cpx{G}(\bdens{\rho}_{\sAB}) ] + y  \dim H^{1}[ \cpx{G}(\bdens{\rho}_{\sAB}) ],\\
	P_{\cpx{E}}(\bdens{\rho}_{\sAB}) &:= \dim H^{0}[ \cpx{E}(\bdens{\rho}_{\sAB}) ] + y  \dim H^{1}[ \cpx{E}(\bdens{\rho}_{\sAB}) ].
\end{align*}
provide complete invariants of pure bipartite states up to local invertible transformations. 
\begin{remark}{(C.f.\ Rmk.~\ref{rmk:bipartite_poincare})}{}
	Suppose $\bdens{\rho}_{\sAB}$ is a pure bipartite density state with Schmidt rank $S$ then, by the above discussion, its Poincar\'{e} polynomials are given by: 
	 \begin{align*}
		 P_{\cpx{G}}(\bdens{\rho}_{\sAB}) &:= P[\cpx{G}(\bdens{\rho}_{\sAB})] = (S^2 - 1) + (d_{\sA} - S)(d_{\sB} - S) y,\\
		 P_{\cpx{E}}(\bdens{\rho}_{\sAB}) &:= P[\cpx{E}(\bdens{\rho}_{\sAB})] = 2(S^2-1).
	 \end{align*}
	 Spanning over all bipartite density states, there are $\mathrm{min}(d_{\sA}, d_{\sB})$ possibilities for each Poincar\'{e} polynomial, each uniquely labelled by the Schmidt rank.
	 In particular the Poincar\'{e} polynomials completely classify pure bipartite density states up to local invertible transformations.
\end{remark}

As summarized in the above remark, both GNS and commutant cohomologies are separately good measures of factorizability/entanglement for a \textit{pure} bipartite density state.  However, the situation for mixed states is more subtle.
There are vanishing theorems for the zeroth cohomology components when the density state is \textit{support factorizable}: a weaker condition than factorizability (c.f.\ \S\ref{sec:factorizabilia}), however the converse does not hold true as there are bipartite density states with trivial cohomology but not support factorizable (c.f.\ Ex.~\ref{ex:mixed_state_counterexamples}).\footnote{Moreover, there are separable states with non-trivial cohomology.}
Details for the mixed state situation are present in \S\ref{sec:factorizability_mixed}.

\subsubsection{Cohomology for Multipartite Density States}
In \S\ref{sec:multipartite_complexes} we introduce multipartite generalizations of the bipartite (co)chain complexes introduced in \S\ref{sec:bipartite_complexes}.\footnote{Although it is not emphasized in this paper, the cochain complexes arise as \v{C}ech complexes constructed from presheaves of vector spaces over the space of tensor factors (thought of as equipped with the discrete topology) and open covers given by the complement of each tensor factor.}
Generalizing the formal definition for a bipartite density state, an $N$-partite density state, $N \geq 2$ is formally a tuple $(P, (\hilb_{p})_{p \in P}, \dens{\rho})$, where
\begin{itemize}
	\item $P$ is an ordered set of size $|P|=N$, whose elements are called \textit{tensor factors} or \textit{primitive subsystems},

	\item $(\hilb_{p})_{p \in P}$ is a tuple of (finite-dimensional) Hilbert spaces indexed by $P$, and

	\item $\dens{\rho}$ is a unit trace positive semidefinite endomorphism of $\bigotimes_{p \in P} \hilb_{p}$.
\end{itemize}
Once again, when the tuple of Hilbert spaces is understood, we denote such a tuple of data by the boldface quantity $\bdens{\rho}_{P}$.

Using the same building blocks as previously, we construct the \textit{GNS (cochain) complex}: 
\begin{align*}
	\newmath{\cpx{G}(\bdens{\rho}_{P})} :=
	\cdots \longrightarrow 0
	\longrightarrow \mathbb{C}
	\overset{d_{\cpx{G}}^{-1}}{\longrightarrow}
	\cpx{G}^{0}(\bdens{\rho}_{P})
	\overset{d_{\cpx{G}}^{0}}{\longrightarrow}
	\cpx{G}^{1}(\bdens{\rho}_{P})
	\overset{d_{\cpx{G}}^{2}}{\longrightarrow} \cdots
	\overset{d_{\cpx{G}}^{N-2}}{\longrightarrow}
	\cpx{G}^{N-1}(\bdens{\rho}_{P})
	\longrightarrow 0 \longrightarrow
	\cdots
\end{align*}
and the \textit{commutant (cochain) complex}
\begin{align*}
	\newmath{\cpx{E}(\bdens{\rho}_{P})} :=
	\cdots \longrightarrow 0
	\longrightarrow \mathbb{C}
	\overset{d_{\cpx{E}}^{-1}}{\longrightarrow}
	\cpx{E}^{0}(\bdens{\rho}_{P})
	\overset{d_{\cpx{E}}^{0}}{\longrightarrow}
	\cpx{E}^{1}(\bdens{\rho}_{P})
	\overset{d_{\cpx{E}}^{1}}{\longrightarrow} \cdots
	\overset{d_{\cpx{E}}^{N-2}}{\longrightarrow}
	\cpx{E}^{N-1}(\bdens{\rho}_{P})
	\longrightarrow 0 \longrightarrow
	\cdots 
\end{align*}
Along with chain complex versions $\cpx{g}(\bdens{\rho}_{P})$ and $\cpx{e}(\bdens{\rho}_{P})$, once again satisfying (c.f.\ Prop.~\ref{prop:trace_duality_multipartite}):
\begin{align*}
	\cpx{G}(\bdens{\rho}_{\sAB}) &\cong \left[ \cpx{g}(\bdens{\rho}_{\sAB}) \right]^{\vee}, \\
	\cpx{E}(\bdens{\rho}_{\sAB}) &\cong \left[ \cpx{e}(\bdens{\rho}_{\sAB}) \right]^{\vee}.
\end{align*}

Letting $\cpx{C}$ be one of the cochain complexes, above, the non-trivial components of degree $-1 \leq k \leq N$ are defined using building blocks:
\begin{align*}
	\cpx{C}^{k}(\bdens{\rho}_{P}) &= \prod_{\{T \subseteq P: |T| = k + 1 \}} \mathtt{B}(\dens{\rho}_{T}),
\end{align*}
where:
\begin{itemize}
	\item $\dens{\rho}_{T}$ is the reduced density state associated to the subset/subsystem $T \subseteq P$;
	\item $\mathtt{B} = \mathtt{GNS}$ for the GNS complex $\cpx{G}(\bdens{\rho}_{P})$, and $\mathtt{B} = \mathtt{Com}$ for the commutant complex $\cpx{E}(\bdens{\rho}_{P})$.
\end{itemize}
In \S\ref{sec:visualization} we give a geometric picture of these complexes.\footnote{This picture is not specific to the complexes we have developed, it is a general property of cochain complexes arising from \v{C}ech cohomology of a presheaf.}
We begin by visualizing the standard $(N-1)$-simplex whose vertices are labelled by elements of the ordered set $P$, call it $\Delta_{P}$.
Dimension $k$-faces of this simplex are in bijective correspondence with size $k+1$ subsets of $P$.
With this in mind, let $F_{T}$ denote the dimension $(|T|-1)$-simplex associated to the set $|T|$; then we have an assignment:
\begin{align*}
	\text{Faces of $\Delta_{P}$} & \longrightarrow \text{Subspaces of operators}\\
	F_{T}                          & \longmapsto \mathtt{B}(\bdens{\rho}_{T}) \leq \algebra{\hilb_{T}}.
\end{align*}
where $\hilb_{T} := \bigotimes_{t \in T} \hilb_{t}$.
Elements of the set of $k$-cochains $\cpx{C}^{k}(\bdens{\rho}_{P})$ are precisely sections of this assignment over the $k$-skeleton of $\Delta_{P}$: the union of all faces of dimension $k$. 
In particular, an element $R \in \cpx{C}^{k}(\bdens{\rho}_{P})$ can be thought of as consisting of as an assignment of an element 
\begin{align*}
	R_{T} \in \mathtt{B}(\dens{\rho}_{T}) \subseteq \algebraext{\bigotimes_{t \in T} \hilb_{t}} = \{\text{Operators/random variables associated to subsystem $T$} \}.
\end{align*}
to each face $F_{T}$ with $|T| = k+1$.
The coboundary is constructed as an alternating sum over the $(k+1)$-possible ways of lifting to the $(k+1)$-skeleton.
For example, returning to the bipartite situation: given a 1-cochain
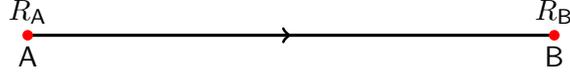
\begin{center}
	\begin{tikzpicture}[very thick,decoration={markings, mark=at position 0.5 with {\arrow{>}}}]
		\coordinate [label=below:$\sA$, label=above:$R_{\sA}$] (A) at (0,0);
		\coordinate [label=below:$\sB$, label=above:$R_{\sB}$] (B) at (7,0);

		\draw[postaction={decorate}] (A) to (B);

		\fill[red]  (A) circle [radius=2pt]; 
		\fill[red]  (B) circle [radius=2pt]; 
	\end{tikzpicture}
\end{center}
we have the coboundary:
\begin{center}
	\begin{tikzpicture}[very thick,decoration={markings, mark=at position 0.5 with {\arrow{>}}}]
		\coordinate [label=below:$\sA$] (A) at (0,0);
		\coordinate [label=below:$\sB$] (B) at (7,0);

		\draw[postaction={decorate}] (A) to node[sloped,above,midway] {$(d_{\cpx{C}}^{0}R)_{\sAB} = \mathrm{proj}_{\sAB}\left[R_{\sA} \otimes 1_{\sB} - 1_{\sA} \otimes R_{\sB} \right]$} (B);

		\fill[red]  (A) circle [radius=2pt]; 
		\fill[red]  (B) circle [radius=2pt]; 
	\end{tikzpicture}
\end{center}
Given a tripartite density state $\bdens{\rho}_{\sABC}$ over the set of tensor factors $(\sA, \sB, \sC)$, and a 0-cochain $Q \in \cpx{C}^{0}(\bdens{\rho}_{\sABC})$, we have the coboundary:
\begin{center}
	\begin{tikzpicture}[very thick,decoration={markings, mark=at position 0.5 with {\arrow{>}}}]
		\coordinate [label=below:$\sA$, label=left:$Q_{\sA}$] (AL) at (-9,0);
		\coordinate [label=below:$\sB$, label=right:$Q_{\sB}$] (BL) at (-6,0);
		\coordinate [label=above:$\sC$, label=left:$Q_{\sC}$] (CL) at (-9,3);
		\coordinate [label=below:$\sA$] (AR) at (0,0);
		\coordinate [label=below:$\sB$] (BR) at (3,0);
		\coordinate [label=above:$\sC$] (CR) at (0,3);

		\draw[fill=blue!30] (AL)--(BL)--(CL)--cycle;

		\draw[fill=blue!30] (AR)--(BR)--(CR)--cycle;

		\draw[postaction={decorate}] (AL) to node[sloped,below,midway] {} (BL);

		\draw[postaction={decorate}] (BL) to node[sloped,above,midway] {} (CL);

		\draw[postaction={decorate}] (AL) to node[sloped,above,midway] {} (CL);

		\draw[postaction={decorate}] (AR) to node[sloped,below,midway] {$(d_{\cpx{C}}^{0}Q)_{\sAB}$} (BR);

		\draw[postaction={decorate}] (BR) to node[sloped,above,midway] {$(d_{\cpx{C}}^{0}Q)_{\sBC}$} (CR);

		\draw[postaction={decorate}] (AR) to node[sloped,above,midway] {$(d_{\cpx{C}}^{0}Q)_{\sAC}$} (CR);

		\fill[red]  (AL) circle [radius=2pt]; 
		\fill[red]  (AR) circle [radius=2pt]; 
		\fill[red]  (BL) circle [radius=2pt]; 
		\fill[red]  (BR) circle [radius=2pt]; 
		\fill[red]  (CL) circle [radius=2pt]; 
		\fill[red]  (CR) circle [radius=2pt]; 

		\draw[-latex] (-5.5,1.5) to node[above,midway]{\small $d_{\cpx{C}}^{0}: \cpx{C}^{0}(\bdens{\rho}_{\sABC}) \rightarrow \cpx{C}^{1}(\bdens{\rho}_{\sABC})$} (-1.5,1.5);
	\end{tikzpicture}
\end{center}
where
\begin{align*}
	(d^{0}Q):
	\left \{
		\begin{array}{ll}
			\sAB \mapsto \mathrm{proj}_{\sAB}(1_{\sA} \otimes Q_{\sB} - Q_{\sA} \otimes 1_{\sB})\\
			\sAC \mapsto \mathrm{proj}_{\sAC}(1_{\sA} \otimes Q_{\sC} - Q_{\sA} \otimes 1_{\sC})\\
			\sBC \mapsto \mathrm{proj}_{\sBC}(Q_{\sB} \otimes 1_{\sC} - Q_{\sB} \otimes 1_{\sC})
		\end{array}
	\right..\\
\end{align*}
Given a 1-cochain $R \in \cpx{C}^{1}(\bdens{\rho}_{\sABC})$, we have the coboundary:
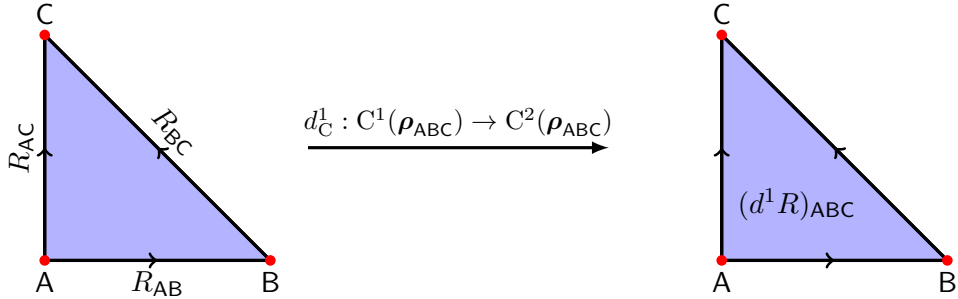
\begin{center}
	\begin{tikzpicture}[very thick,decoration={markings, mark=at position 0.5 with {\arrow{>}}}]
		\coordinate [label=below:$\sA$] (AL) at (-9,0);
		\coordinate [label=below:$\sB$] (BL) at (-6,0);
		\coordinate [label=above:$\sC$] (CL) at (-9,3);
		\coordinate [label=below:$\sA$] (AR) at (0,0);
		\coordinate [label=below:$\sB$] (BR) at (3,0);
		\coordinate [label=above:$\sC$] (CR) at (0,3);

		\draw[fill=blue!30] (AL)--(BL)--(CL)--cycle;

		\draw[fill=blue!30] (AR)--(BR)--(CR)--cycle;

		\draw[postaction={decorate}] (AL) to node[sloped,below,midway] {$R_{\sAB}$} (BL);

		\draw[postaction={decorate}] (BL) to node[sloped,above,midway] {$R_{\sBC}$} (CL);

		\draw[postaction={decorate}] (AL) to node[sloped,above,midway] {$R_{\sAC}$} (CL);

		\draw[postaction={decorate}] (AR) to node[sloped,below,midway] {} (BR);

		\draw[postaction={decorate}] (BR) to node[sloped,above,midway] {} (CR);

		\draw[postaction={decorate}] (AR) to node[sloped,above,midway] {} (CR);

		\fill[red]  (AL) circle [radius=2pt]; 
		\fill[red]  (AR) circle [radius=2pt]; 
		\fill[red]  (BL) circle [radius=2pt]; 
		\fill[red]  (BR) circle [radius=2pt]; 
		\fill[red]  (CL) circle [radius=2pt]; 
		\fill[red]  (CR) circle [radius=2pt]; 
		
		\node at  (1,0.75) {$(d^{1}R)_{\sABC}$};

		\draw[-latex] (-5.5,1.5) to node[above,midway]{\small $d_{\cpx{C}}^{1}: \cpx{C}^{1}(\bdens{\rho}_{\sABC}) \rightarrow \cpx{C}^{2}(\bdens{\rho}_{\sABC})$} (-1.5,1.5);
	\end{tikzpicture}
\end{center}
Here the coboundary $d^{1}R$ is specified by its value on the only 2-face $F_{\sABC}$:
\begin{align*}
	(d^{1}_{\cpx{C}} R)_{\sABC} = \mathrm{proj}_{\sABC} \left[1_{\sA} \otimes R_{\sBC} - \Sigma_{\sABC, \sB} \left(R_{\sAC} \otimes 1_{\sB}\right) + R_{\sAB} \otimes 1_{\sC} \right],
\end{align*}
where $\Sigma_{\sABC, \sB}$ is a reshuffling of tensor factors, and $\mathrm{proj}_{\sABC}$ consists of a projection to the subspace $\mathtt{B}(\dens{\rho}_{\sABC})$ by either right multiplication or compression by support projections. 

In \S\ref{sec:classes_and_correlations} we describe how non-zero classes in multipartite GNS or commutant cohomology consist of (equivalence classes) of tuples of operators that exhibit non-local correlations indicative of non-factorizability (a generalization of the story for bipartite GNS cohomology).
When $k \leq N-2$, elements of $\ker(d^{k}_{\cpx{C}})$---a.k.a. \textit{$k$-cocycles}---are those sections over the $k$-skeleton that satisfy linear relations of the form
\begin{align}
	\sum_{l = 0}^{k+1} (-1)^{l} \underline{R_{\partial_{l}V}} \sim_{V}  0
	\label{eq:cocycle_relation_intro}
\end{align}
for each $k$-face $F_{V}$, where:
\begin{itemize}
	\item $R_{\partial_{l}V} \in \mathtt{B}(\bdens{\rho}_{\partial_{l}V})$ denotes the value of the cochain $R$ on the $l$th boundary face $F_{\partial_{l} V}$ of $F_{V}$: Letting $V(l)$ denote the $l$th element of $V$, then $\partial_{l} V = V \backslash V(l)$;

	\item $\underline{R_{\partial_{l}V}}$ is the lift of $R_{\partial_{l}V}$ to an element of $\algebra{\hilb_{V}}$: Letting $V(l)$ denote the $l$th element of $V$ and $\partial_{l} V = V \backslash V(l)$ be the $l$th face of, $\underline{R_{\partial_{l}V}} := \Sigma_{V, l} \left(R_{\partial_{l} V} \otimes 1_{V(l)} \right)$, where  $V(l)$ is the $l$th element of $V$,\, $\partial_{l} V = V \backslash V(l)$, and $\Sigma_{(V,l)}$ is a reshuffling of tensor factors;
	
	\item $\sim_{V}$ is the equivalence relation on operators in $\algebra{\hilb_{V}}$ defined by: $a \sim_{V} b$ if $\mathrm{proj}_{V}a = \mathrm{proj}_{V} b$ for $a,\, b, \in \algebra{\hilb_{V}}$.
		(Once again, $\mathrm{proj}_{V}$ is an appropriate projection onto the subspace $\mathtt{B}(\dens{\rho}_{V})$.)
\end{itemize}
The relation \eqref{eq:cocycle_relation_intro} is indicative of possible non-trivial, non-local correlations between distinct tensor factors.
This is most evident in the bipartite situation, where the cocycle condition for $0$-chains reduces to
\begin{align}
	R_{\sA} \sim_{\sAB} R_{\sB}.
	\label{eq:bipartite_relation_intro}
\end{align}
However, there are trivial solutions to \eqref{eq:cocycle_relation_intro} that we are not interested in if we are seeking indicators of non-factorizability.
For example: if we are studying the GNS complex for some bipartite density state $\bdens{\rho}_{\sAB} = (\hilb_{\sA}, \hilb_{\sB}, \dens{\rho}_{\sAB})$, then it follows from our discussion above that the relation \eqref{eq:bipartite_relation_intro} is equivalent to a saturation of covariance condition:
\begin{align}
	\Cov(R_{\sA}, R_{\sB}) = \Var_{\sA}(R_{\sA}) = \Var_{\sB}(R_{\sB})
	\label{eq:covariance_condition_intro}
\end{align}
on pairs of operators $(R_{\sA}, R_{\sB}) \in \mathtt{GNS}(\dens{\rho}_{\sA}) \times \mathtt{GNS}(\dens{\rho}_{\sB})$.
There is a trivial solution to \eqref{eq:covariance_condition_intro} given by $R_{\sX} = \supp_{\sX}$ (the support projection of $\dens{\rho}_{\sX}$), this is the unique solution for which the covariance and variances in \eqref{eq:covariance_condition_intro} vanish.
The trivial solution is also the only solution that exists for the factorizable bipartite density state $(\hilb_{\sA}, \hilb_{\sB}, \dens{\rho}_{\sA} \otimes \dens{\rho}_{\sB})$ constructed from the reduced density states $\dens{\rho}_{\sA}$ and $\dens{\rho}_{\sB}$ of $\dens{\rho}_{\sAB}$, and it descends in a precise way to a solution of $\eqref{eq:covariance_condition_intro}$ for $\bdens{\rho}_{\sAB}$.

In general, given an arbitrary multipartite density state $\bdens{\rho}_{P}$, one can construct an associated ``fully factorizable" multipartite density state from the reduced density states $\dens{\rho}_{\{p \}},\, p \in P$ of $\bdens{\rho}_{P}$.
Solutions to \eqref{eq:cocycle_relation_intro} for this associated factorizable state descend to solutions to \eqref{eq:cocycle_relation_intro} for $\bdens{\rho}_{P}$, but such solutions should be considered trivial as the full factorizable state cannot have any associated non-trivial, non-local correlations along the boundaries of any face $F_{V}$. 
As it so happens the subspace of trivial solutions is equivalent to the subspace $\image(d^{k-1}_{\cpx{C}})$ when $k \leq N-2$.
As a result, when $k \leq N-2$ we have the identification: 
\begin{align*}
	H^{k} \left[ \cpx{C}(\bdens{\rho}_{P}) \right] &= \ker(d_{\cpx{C}}^{k})/\image(d_{\cpx{C}}^{k-1})\\[5pt]
                                                   &= \left\{\text{\parbox{15.5em}{Sections of $F_{T} \mapsto \mathtt{B}(\bdens{\rho}_{T})$ over the $k$-skeleton of the $(N-1)$-simplex that exhibit \textit{possible} non-local correlations along each face: i.e.\ solutions to \eqref{eq:cocycle_relation_trivial}.}} \right\} \Big / \left\{\text{\parbox{15em}{Trivial solutions to \eqref{eq:cocycle_relation_trivial}: i.e.\ those solutions that do not encode correlations due to non-factorizability of $\bdens{\rho}_{P}$ along the boundaries of any face.}} \right\}.\\
  \end{align*}
As described in \S\ref{sec:interpretation}, when $k \in\{0, \cdots, N-2 \}$, a representative of a non-zero classes in $H^{k}[\cpx{C}(\bdens{\rho}_{P})]$ is identifiable with a section of $F_{T} \rightarrow \mathtt{B}(\bdens{\rho}_{P})$ that exhibits non-trivial, non-local correlations along the boundary of at least one $(k+1)$-face: a property indicative of non-factorizability of $\bdens{\rho}_{P}$.
Fig.~\ref{fig:cocycles_representation_intro} provides an explicit 1-cocycle representative of a non-trivial GNS cohomology class for the tripartite GHZ state.   

From another perspective: we can think of a $k$-cochain as an assignment to each size $(k+1)$ subset of tensor factors a $(k+1)$-body operator\footnote{An operator that can only act non-trivially on $(k+1)$-tensor factors.} implementable or detectable by an observer with access to that collection of tensor factors.
The operators in such an assignment always live inside an appropriate version of the building blocks we have constructed. 
In particular, when considering GNS cohomology, the operators in such an assignment live inside of the GNS building blocks, and can be thought of as canonical representatives of right essential equivalences class of operators.
A $k$-cocycle ($0 \leq k \leq N-2$) with non-trivial cohomology class is then a $k$-cochain where there is at least one size $k+2$ subset of tensor factors $V \subseteq P$ such that an observer with access to the tensor factors of $V$ can see correlations among the $(k+2)$, $(k+1)$-body operators that arise from the assignments of the $k$-cochain to the size $k+1$ subsets that are given by removing a single element of $V$.

We do not offer as satisfying of an interpretation for the cohomology component in degree $N-1$: the highest degree component that might be non-trivial.
However, as described at the end of \S~\ref{sec:highest_multipartite_cohomology}, non-trivial representatives in this component consist of $N$-body operators that \textit{cannot} be written as a linear combination of lifts of $(N-1)$-body operators up to the equivalence relation $\sim_{V}$.
The ability to write a $N$-body operator as such a linear combination is indicative of correlations; so we should expect the dimension for $H^{N-1}\left[\cpx{C}(\bdens{\rho}_{P})\right]$ to be larger for multipartite density states that are more strongly correlated. 
This is, in fact, the case for the GNS cohomology of a pure bipartite state where the dimension of the first cohomology component is coarse measure of how far the state is from maximal entanglement (c.f.\ \eqref{eq:H1_pure_bipartite_intro} above).

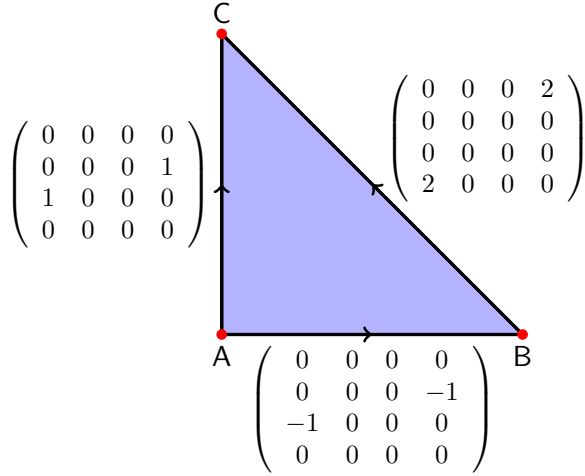
\begin{figure}
	\centering
	\begin{tikzpicture}[very thick,decoration={markings, mark=at position 0.5 with {\arrow{>}}}]
		\coordinate [label=below:$\sA$] (A) at (0,0);
		\coordinate [label=below:$\sB$] (B) at (4,0);
		\coordinate [label=above:$\sC$] (C) at (0,4);

		\draw[fill=blue!30] (A)--(B)--(C)--cycle;


		\draw[postaction={decorate}] (A) to node[sloped,below,midway] {\small $\left(
					\begin{array}{cccc}
						0 & 0 & 0 & 0 \\
						0 & 0 & 0 & -1 \\
						-1 & 0 & 0 & 0 \\
						0 & 0 & 0 & 0 \\
					\end{array}
			\right)$} (B);

		\draw[postaction={decorate}] (B) to node[above,midway,xshift=4em,yshift=-1em] {\small $\left(
					\begin{array}{cccc}
						0 & 0 & 0 & 2 \\
						0 & 0 & 0 & 0 \\
						0 & 0 & 0 & 0 \\
						2 & 0 & 0 & 0 \\
					\end{array}
			\right)$} (C);

		\draw[postaction={decorate}] (A) to node[left,midway] {\small $\left(
					\begin{array}{cccc}
						0 & 0 & 0 & 0 \\
						0 & 0 & 0 & 1 \\
						1 & 0 & 0 & 0 \\
						0 & 0 & 0 & 0 \\
					\end{array}
			\right)$} (C);

		\fill[red]  (A) circle [radius=2pt]; 
		\fill[red]  (B) circle [radius=2pt]; 
		\fill[red]  (C) circle [radius=2pt]; 
	\end{tikzpicture}
	\caption{1-cocycle representative of a non-trivial GNS cohomology class for the GHZ state $\frac{1}{\sqrt{2}} \left(\ket{0}^{\otimes 3} + \ket{1}^{\otimes 3} \right)$.
	The matrices are defined with respect to the ordered basis $(\ket{0}, \ket{1})$.
A table of generators for the GNS cohomology components of the tripartite GHZ and W-states is provided in App.~\ref{app:ghz_w_cohomology_generators}. \label{fig:cocycles_representation_intro}}
\end{figure}
A brute force computation of GNS or commutant cohomologies involves computing reduced density states associated to all subsystems.
If our initial state is pure and we are given the collection of all such reduced density states, one can in principle work out the factorizability properties of the original state.
In this sense, using cohomology (computed via brute force) to detect factorizability might not seem advantageous from a computational perspective.
On the other hand, the cohomologies provide a natural way of encoding all the partial trace information into one mechanism that allows us to answer factorizability questions more easily.
Moreover, when a state \textit{fails} to be factorizable, the actual cohomology components provide correlated tuples of operators that encode the obstructions to being factorizable; Poincar\'{e} polynomials can be used as quantitative measure of how badly certain factorizability fails (i.e.\ an entanglement measure for pure states).

As described in \S\ref{sec:multipartite_properties}, the multipartite complexes obey generalizations of the basic properties enjoyed by their bipartite specializations: once again they are equivariant under local invertible transformations, making their associated Poincar\'{e} polynomials local invertible invariants. 
Moreover, the multipartite complexes of a tensor product of multipartite density states (Def.~\ref{def:multipartite_state_tensor}) is given by the tensor product of chain complexes up to a shift in degree. 
\begin{theorem*}{(C.f.\ Thm.~\ref{thm:cochain_fact}) Cochains for Factorizable States}{}
	Let $\bdens{\rho}_{P}$ and $\bdens{\varphi}_{Q}$ be $N$ and $M$ partite density states, then there are canonical isomorphisms:
	\begin{align*}	
		\cpx{G}(\bdens{\rho}_{P} \otimes \bdens{\varphi}_{Q}) &\cong \left( \cpx{G}(\bdens{\rho}_{P}) \otimes \cpx{G}(\bdens{\varphi}_{Q}) \right)[1];\\
		\cpx{E}(\bdens{\rho}_{P} \otimes \bdens{\varphi}_{Q}) &\cong \left(\cpx{E}(\bdens{\rho}_{P}) \otimes \cpx{E}(\bdens{\varphi}_{Q}) \right)[1].
	\end{align*}
\end{theorem*}
Thus, combined with the K\"{u}nneth theorem for cochain complexes over a field of characteristic zero, the cohomology of a tensor product of states is the tensor product of states up to a shift.
Stated in terms of Poincar\'{e} polynomials
\begin{align*}
	P_{\cpx{G}}(\bdens{\rho}_{P}) &:= \sum_{k = 0}^{|P|-1} y^{k} \dim H^{k} \left[ \cpx{G}(\bdens{\rho}_{P}) \right] \in \mathbb{Z}[y],\\
	P_{\cpx{E}}(\bdens{\rho}_{P}) &:= \sum_{k = 0}^{|P|-1} y^{k} \dim H^{k} \left[ \cpx{E}(\bdens{\rho}_{P}) \right] \in \mathbb{Z}[y].
\end{align*}
we have the following corollary.
\begin{corollary*}{(C.f.\ Cor.~\ref{cor:poincare_factorization})}{}
	Let $\bdens{\rho}_{P}$ and $\bdens{\varphi}_{Q}$ be multipartite density states, then we have a factorization of Poincar\'{e} polynomials:
	\begin{align*} 
		P_{\cpx{G}}(\bdens{\rho}_{P} \otimes \bdens{\varphi}_{Q}) &= y P_{\cpx{G}}(\bdens{\rho}_{P}) P_{\cpx{G}}(\bdens{\varphi}_{Q}),\\
		P_{\cpx{E}}(\bdens{\rho}_{P} \otimes \bdens{\varphi}_{Q}) &= y P_{\cpx{E}}(\bdens{\rho}_{P}) P_{\cpx{E}}(\bdens{\varphi}_{Q}).
	\end{align*}
\end{corollary*}

Given the behavior of cohomologies under tensor product of multipartite density states, it becomes relatively straightforward to demonstrate show that cohomologies of the multipartite complexes can detect the presence of factorizability over the set of tensor factors.
\begin{corollary*}{(C.f.\ Cor.~\ref{cor:fully_fact_to_cohom})}{}
	Let $\bdens{\rho}_{P}$ be a pure, fully factorizable $N$-partite density state, then:
	\begin{itemize}
		\item The only non-vanishing cohomology component of $\cpx{G}(\bdens{\rho}_{P})$ is in degree $N-1$ with
			\begin{align*}
				\dim H^{N-1}[\cpx{G}(\bdens{\rho}_{P})] = \prod_{p \in P} (\dim \hilb_{p} -1);
			\end{align*}

		\item The complex $\cpx{E}(\bdens{\rho}_{P})$ has vanishing cohomology: $\dim H^{k}[\cpx{E}(\bdens{\rho}_{P})] \equiv 0$ for all $k \in \mathbb{Z}$.
	\end{itemize}
\end{corollary*}
This a corollary of Thm.~\ref{thm:support_fact_multipartite_cohomologies}: a general statement for support factorizable mixed density states.
One can also try to detect factorizability with respect to an arbitrary partition.
Given a multipartite density state $\bdens{\rho}_{P}$ and a partition $\lambda$ of $P$ of length $L$, we can construct an $L$-partite ``coarsening" $\lambda[\bdens{\rho}_{P}]$ of the $|P|$-partite density state $\bdens{\rho}_{P}$: a multipartite density state over the set of tensor factors labelled by the $L$ components of the partition $\lambda$.  
Thus, given a multipartite density state over a set of tensor factors $P$, one can compute GNS/commutant cohomologies of any coarsening associated to a given partition of $P$.
\begin{theorem*}{(C.f.\ Thm.~\ref{thm:lambda_fact_bipartite_scan})}{}
	Let $\bdens{\rho}_{P}$ be a pure $N$-partite density state and $\lambda$ a partition of $P$ of length $\geq 2$.
	\begin{enumerate}
		\item If there exists a coarsening $\eta \geq \lambda$ such that $H^{k}[\cpx{G}(\eta[\bdens{\rho}_{P}])]$ or $H^{k}[\cpx{E}(\eta[\bdens{\rho}_{P}])] \neq 0$ for some $k < |\eta|-1$, then $\bdens{\rho}_{P}$ is $\lambda$-entangled.

		\item $\bdens{\rho}_{P}$ is $\lambda$-entangled if and only if there exists a length 2 coarsening $\eta \geq \lambda$ such that $H^{0}[\cpx{G}(\eta[\bdens{\rho}_{P}])] \neq 0$.

		\item $\bdens{\rho}_{P}$ is $\lambda$-entangled if and only if there exists a length 2 coarsening $\eta \geq \lambda$ such that $H^{0}[\cpx{E}(\eta[ \bdens{\rho}_{P} ])] \neq 0$.
	\end{enumerate}
\end{theorem*}
The first statement is just an application of the previous corollary, and the latter two statements follow from the fact that a state is factorizable with respect to a partition if and only if it is factorizable with respect to every bipartite coarsening of that partition; moreover, bipartite cohomologies provide necessary and sufficient conditions for the presence of factorizability.
Hence, if we are trying to answer questions about factorizability with respect to a given partition, then it suffices to focus on GNS or commutant cohomologies of bipartite coarsenings of that partition.
On the other hand, if one is interested in factorizability with respect to \textit{any} partition, one can compute cohomologies for the finest partition, keeping the following theorem in mind.
\begin{theorem*}{(C.f.\ Thm.~\ref{thm:pure_nonvanishing_complete_entanglement})}{}
	If $\bdens{\rho}_{P}$ is $\lambda$-(support) factorizable for some $|\lambda| = L$, then $H^{k}[\cpx{G}(\bdens{\rho}_{P})] =0$ and $H^{k}[\cpx{E}(\bdens{\rho}_{P})] = 0$ for all $k \leq L-2$.
\end{theorem*}
Equivalently, if $\bdens{\rho}_{P}$ is $\lambda$-factorizable for some $|\lambda|=L$, then the Poincar\'{e} polynomials 
\begin{align*}
	P_{\cpx{G}}(\bdens{\rho}_{P}) &:= \sum_{k = 0}^{|P|-1} y^{k} \dim H^{k} \left[ \cpx{G}(\bdens{\rho}_{P}) \right] \in \mathbb{Z}[y],\\
	P_{\cpx{E}}(\bdens{\rho}_{P}) &:= \sum_{k = 0}^{|P|-1} y^{k} \dim H^{k} \left[ \cpx{E}(\bdens{\rho}_{P}) \right] \in \mathbb{Z}[y].
\end{align*}
are divisible by $y^{L-2}$.
When $\bdens{\rho}_{P}$ is pure we conjecture that the converse is true: i.e.\ we are guaranteed that $\lambda$-factorizability with respect to some $\lambda$ is detected by a non-trivial cohomology component.
\begin{conjecture*}{(C.f.\ Conj.~\ref{conj:complete_entanglement_pure_cohomologies})}{}
	Suppose $\bdens{\rho}_{P}$ is a pure multipartite density state.
	\begin{enumerate}
		\item $P_{\cpx{G}}(\bdens{\rho}_{P})$ is divisible by $y^{k} \Longleftrightarrow \bdens{\rho}_{P}$ is $\lambda$-factorizable for some $|\lambda| = k+2$.
		
		\item $P_{\cpx{E}}(\bdens{\rho}_{P})$ is divisible by $y^{k} \Longleftrightarrow \bdens{\rho}_{P}$ is $\lambda$-factorizable for some $|\lambda| = k+2$.
	\end{enumerate}
\end{conjecture*}
(This is the contrapositive of Conj.~\ref{conj:complete_entanglement_pure_cohomologies}, which is stated in terms of the notion of ``complete-$k$-entanglement" introduced in Def.~\ref{def:pure_multipartite_entanglement}.)
So, roughly speaking, non-zero elements of $k$th cohomology components---representable by $\binom{|P|}{k+1}$-tuples of $(k+1)$-body operators---can be thought of as obstructions to the factorizability of $\bdens{\rho}_{P}$ with respect to all length $\leq (k+2)$ partitions.

\subsubsection{Tripartite Computational Results}
\S\ref{sec:W_vs_GHZ} is a testing ground for multipartite cohomologies.
We begin with tripartite states noting that, for qubit systems and tripartite pure states, it is known \cite{Dur} that there are six distinct equivalence classes of states under local invertible transformations (known as \textit{SLOCC equivalence classes}): the class corresponding to factorizable states, three bipartite entangled classes (corresponding to the three possible bipartite partitions of three elements), the class represented by the GHZ state
\begin{align*}
	\frac{1}{\sqrt{2}}\left( \ket{000} + \ket{111} \right) \in \mathrm{span}_{\mathbb{C}} \{ \ket{0}, \ket{1} \}^{\otimes 3},
\end{align*}
and the class represented by the W state 
\begin{align*}
	\frac{1}{\sqrt{3}}\left( \ket{001} + \ket{010} + \ket{100} \right) \in \mathrm{span}_{\mathbb{C}} \{ \ket{0}, \ket{1} \}^{\otimes 3}.
\end{align*}
In this section we state explicit computations of associated cohomologies of the tripartite W-state and GHZ-states as well as higher generalizations.
Secondly, the tripartite mutual informations of both the W-state and GHZ state are vanishing: meaning that, if there is any shared information between all three tensor factors, the (tripartite) mutual information fails to detect it.
On the other hand, the Poincar\'{e} polynomials associated to these states, are non-trivial: 
\begin{align*}
	P_{\cpx{G}}(\boldsymbol{\mathrm{GHZ}}) = P_{\cpx{G}}(\boldsymbol{\mathrm{W}}) = 1 + 6 y.
\end{align*}
Because the polynomials are identical, GNS Poincar\'{e} polynomials cannot distinguish between the SLOCC equivalence classes of the W and GHZ states. 
Nevertheless, the presence of a non-vanishing cohomology component in degree 0 indicates that (modulo shifts by 3-tuples of identity operators) there is a one-dimensional family of $3$-tuples of operators encapsulating shared correlations/information among all three tensor factors, and (modulo coboundaries) there is a six-dimensional family of $3$-tuples of operators encapsulating correlations between pairs of tensor factors.
These are representative of the ``shared information" that the tripartite mutual information fails to detect.

Moreover, despite the fact that one cannot distinguish the SLOCC equivalence class of the GHZ state from the W-state via Poincar\'{e} polynomials of GNS cohomology, Poincar\'{e} polynomials associated to commutant cohomologies \textit{do} distinguish the SLOCC classes of these two states:
\begin{align*}
	P_{\cpx{E}}(\boldsymbol{\mathrm{GHZ}}) &= 7 + 7 y,\\
	P_{\cpx{E}}(\boldsymbol{\mathrm{W}}) &= 3 + 3 y.
\end{align*}
In \S\ref{sec:generalized_W_and_GHZ} we also formulate conjectures about the form of Poincar\'{e} polynomials of $(N \geq 3)$-partite versions of the GHZ and W-states based upon software aided computations.

\subsubsection{The State Index and a Path Toward a Categorification of Mutual Information}
In \S\ref{sec:categorification} we take a diversion to outline how these (co)chain complexes can be thought of as steps toward the goal of a categorification of entropy/mutual information.\footnote{One can view this section as an introductory motivation for a forthcoming paper.}
After some motivation, a quantity called the \textit{state index} is introduced: to a multipartite density state $\bdens{\rho}_{P}$,  the state index $\mathfrak{X}(\bdens{\rho}_{P})$ is a function on $\mathbb{C}^{3}$ that is valued in the ring of polynomials $\mathbb{C}[w]$ in a formal variable $w$:\footnote{Although the state index and its properties are very precisely defined, most of \S\ref{sec:categorification} is spent on its motivation. 
The motivational discussion is largely schematic to keep with the spirit of avoiding use of sophisticated technology.}
\begin{equation*}
	\begin{array}{lccl}
		\mathfrak{X}(\bdens{\rho}_{P}): & \mathbb{C}^{3} & \longrightarrow & \mathbb{C}[w]\\
		{}                                        & (\alpha,q,r) & \longmapsto & w^{|P|} \sum_{\emptyset \subseteq T \subseteq P} (-1)^{|T|} \dim(\hilb_{T})^{\alpha} \left[ \Tr(\dens{\rho}_{T})^{q} \right]^{r},
	\end{array}
\end{equation*}
where, $\hilb_{T} := \bigotimes_{t \in T} \hilb_{t}$. The formal variable $w$ can be thought of as a book-keeping device that keeps track of the number of tensor factors under consideration, although for the purposes of this paper it suffices to evaluate it at, e.g. $w = \pm 1$. 
The state index has the advantage of having properties one would expect from an ``Euler characteristic" of some associated (non-commutative) geometry associated to each density state.
\begin{theorem*}{(C.f.\ Thm.~\ref{thm:index_properties})}{}
	Let $\bdens{\rho}_{P}$ and $\bdens{\varphi}_{Q}$ denote multipartite density states.
	\begin{enumerate}
		\item For any fixed $w \in \mathbb{C}$, the index $\mathfrak{X}^{w}(\bdens{\rho}_{P})$ is an entire function in in the parameters $\alpha,\,q,$ and $r$.

		\item The state index is invariant under local invertible transformations of multipartite states. 

		\item $\mathfrak{X}^{w}(\bdens{\rho}_{P} \otimes \bdens{\varphi}_{Q}) = \mathfrak{X}^{w}(\bdens{\rho}_{P}) \mathfrak{X}^{w}(\bdens{\varphi}_{Q})$.
		
		\item For any $0 \leq i < j \leq |P|-1$, we have 
			\begin{align*}	
				\mathfrak{X}^{w}(\bdens{\rho}_{P}) = -w \left[ \mathfrak{X}^{w}(\bdens{\rho}_{\partial_{i} P}) + \mathfrak{X}^{w}(\bdens{\rho}_{\partial_{j} P}) - \mathfrak{X}^{w}(\lambda_{ij} [\bdens{\rho}_{P}]) \right],
		\end{align*}
		where $\bdens{\rho}_{\partial_{k}P}$ is the reduction of $\bdens{\rho}_{P}$ to a $(|P|-1)$-partite density state after tracing over the $(k+1)$th tensor factor in $P$, and $\lambda_{ij} [\bdens{\rho}_{P}]$ is the $(|P|-1)$-partite coarsening one obtains after merging together the $(i+1)$th and $(j+1)$th elements of $P$ into a single tensor factor.
	\end{enumerate}
\end{theorem*}

One can recover the Euler characteristics of complexes defined in this paper by studying the $q \rightarrow 0$ limit of the state index; on the other hand, multipartite mutual information emerges from a study of the $q \rightarrow 1$ limit (more precisely, by studying derivatives, or $q$-derivatives in this limit).
The situation is summarized in the diagram below.

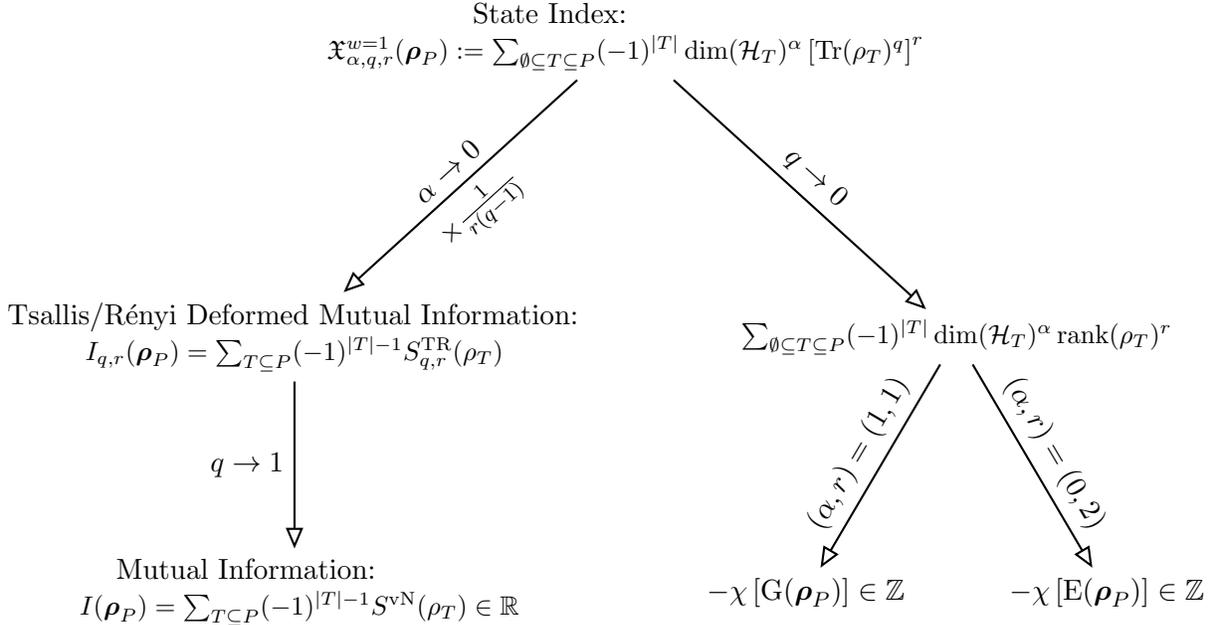
\begin{center}
	\begin{tikzpicture}
		\node at (0,0) (gi) {\parbox{21em}{\centering State Index:\newline \mbox{\small $\mathfrak{X}^{w=1}_{\alpha,q,r}(\bdens{\rho}_{P}) := \sum_{\emptyset \subseteq T \subseteq P} (-1)^{|T|} \dim(\hilb_{T})^{\alpha} \left[ \Tr(\dens{\rho}_{T})^{q} \right]^{r}$}}}; 

		\node at (-4.4,-4) (tsmi) {\parbox{20em}{\centering Tsallis/R\'{e}nyi Deformed Mutual Information: \newline \mbox{\small $I_{q,r}(\bdens{\rho}_{P}) = \sum_{T \subseteq P} (-1)^{|T|-1} S_{q,r}^{\mathrm{TR}}(\dens{\rho}_{T})$}}};

		\node at (-4.4,-7.4) (mi) {\parbox{19em}{\centering Mutual Information:\newline \mbox{ \small $I(\bdens{\rho}_{P}) = \sum_{T \subseteq P} (-1)^{|T|-1} S^{\mathrm{vN}}(\dens{\rho}_{T}) \in \mathbb{R}$}}};

		\node at (4.4, -4) (ranks) {\parbox{18em}{ \centering \mbox{\small $\sum_{\emptyset \subseteq T \subseteq P} (-1)^{|T|} \dim(\hilb_{T})^{\alpha} \rank(\dens{\rho}_{T})^{r}$}}};

		\node at (2.4,-7.4) (gnseulerchar) {$-\chi \left[\cpx{G}(\bdens{\rho}_{P}) \right] \in \mathbb{Z}$}; 

		\node at (6.4,-7.4) (comeulerchar) {$-\chi \left[\cpx{E}(\bdens{\rho}_{P}) \right] \in \mathbb{Z}$};

		\draw[->, >=open triangle 45, decorate, thick] (gi) to node[sloped,above,midway]{$\alpha \rightarrow 0$} node[sloped,below,midway]{$\times \frac{1}{r(q-1)}$} (tsmi);

		\draw[->, >=open triangle 45, decorate, thick] (tsmi) to node[left,midway]{$q \rightarrow 1$}  (mi);

		\draw[->, >=open triangle 45, decorate, thick] (gi) to node[sloped,above,midway]{$q \rightarrow 0$} (ranks);

		\draw[->, >=open triangle 45, decorate, thick] (ranks) to node[sloped,above,midway]{$(\alpha,r) = (1,1)$} (gnseulerchar);

		\draw[->, >=open triangle 45, decorate, thick] (ranks) to node[sloped,above,midway]{$(\alpha,r) = (0,2)$} (comeulerchar);
	\end{tikzpicture}
\end{center}
In \S\ref{sec:W_vs_GHZ} the state index is computed for (generalized) W and GHZ states.

\subsection{Related Work}
This work is closely related to that of Baudot and Bennequin \cite{bb:homent} (J.P.\ Vigneaux provides an excellent detailed exposition in \cite{vigneaux}), who constructed cochain complexes of functions on spaces of probability measures such that mutual informations (and their Tsallis $q$-deformations, which also appear naturally in our story) arise as generators of the first cohomology component.
The technology outlined here, however, is associated to a fixed measure rather than the space of measures---hinting at an interpretation of mutual informations as Euler characteristics rather than cocycles.
Yet, the two theories are undoubtedly intimately related; an approach via the perspective of obstruction theory and classifying stacks might help formalize a precise connection.

A large part of the ideas behind the categorification of mutual information are inspired by the work of Baez, Fritz, and Leinster \cite{bfl:paper, bfl:nlab1} who realize entropy of measures on a finite set as continuous functors out of a suitable topological category of such measures. 
There are also certainly connections with the work of Drummond-Cole, Park, and Terilla \cite{terilla:1,terilla:2,terilla:3,terilla:4} who approach (non-commutative) probability theory from an $A_{\infty}/L_{\infty}$-perspective.
In fact, one sophisticated version of our (co)chain complexes actually admit the structure of a differential graded module for a differential graded algebra (at least when working with finite dimensional Hilbert spaces); hence, the resulting cohomology should be an $A_{\infty}$-module of some $A_{\infty}$-algebra.
Very little is known about these higher algebraic structures at the moment, but---drawing vague analogies with the way that Massey products can identify subtle linkage properties of knot complements in the three sphere---one might speculatively hope, for instance, that such higher structures can detect the Borromean-like entanglement properties of the famous tripartite Greenberger-Horne-Zeilinger (GHZ) state.

\subsection{Future Directions}
Beyond the search for new measures of entanglement and correlations, the author is compelled to mention some other possible future directions with a bias toward interests that overlap with some high energy physicists.

\begin{enumerate}
	\item \textbf{Link invariants in (Quantum) Chern-Simons with compact gauge group}:  Chern-Simons theory is a topological field theory specified by the data of compact Lie Group $G$ and a level $k \in H^{4}(BG,\mathbb{Z})$ \cite{dijkgraaf1990,Freed:CS}.
		Given a (framed) $N$-component link $L$ on the three sphere $S^{3}$, one can study its complement $\mathcal{L} := S^{3} \backslash N(L)$, where $N(L)$ is a tubular neighborhood of $L$ specified by the framing; the boundary of $\mathcal{L}$ is a disjoint union of $N$-tori.
		Chern Simons assigns to $\mathcal{L}$ a pure multipartite state $\ket{\mathcal{L}} \in \left(\hilb_{T^{2}}\right)^{\otimes N}$, where $\hilb_{T^2}$ is the (finite-dimensional) Hilbert space assigned to the two-torus.
		Entanglement properties of these states were studied in \cite{Balasubramanian:2016sro, Balasubramanian:2018por, Salton:2016qpp}. 
		By taking Poincar\'{e} polynomials of GNS/commutant cohomologies of such states, produce (framing-independent) link invariants (given by positive integer coefficient polynomials).
		Even for Abelian Chern-Simons (i.e.\ with $G = U(1)^{r}$) these polynomials are non-trivial.
		The relationship of these invariants to known link invariants remains unexplored.

	\item \textbf{A possible cohomological/geometric approach to a proof of Strong Subadditivity}: 
		Strong subadditivity of quantum entropies is a non-trivial statement, first proven by Leib and Ruskai \cite{lieb}.
		However, Ryu-Takayanagi formulae lead to a very simple geometric interpretation of quantum strong subadditivity; the caveat is that these formulae only hold for small class of states that have a holographic interpretation.
		Such an approach only holds for states that have a holographic interpretation \cite{Hirata:2006jx}.
		The homological techniques touched upon in this paper are not bound to this limitation.

	\item \textbf{The study of how cohomology varies in families of states}: any good geometric structure should be studied in families.
		A deeper understanding of the cohomological constructs of this paper would likely follow from a study of how cohomology varies with families of states.
		For instance: as mentioned below, if one varies mixed states within particular families, our cohomologies are ``piecewise constant", with possible jumps across loci where the ranks of the density states change.
		Although such a jumping behavior might be simple, it would be interesting to explore if it can be quantified.
		One might also wish to understand cohomologies of non-commutative families of density states \cite{moore:qmna}.
		Such an approach is particularly adapted to the $C^{*}/W^{*}$-algebraic origins of this work: a family over a topological space $X$, can be rephrased in terms of, for instance, modules for the commutative algebra of functions on $X$; it is not too far fetched to speculate that passing to non-commutative algebras might lead to the emergence of new phenomena and significant insight.

	\item \textbf{The study of cohomologies associated to completely positive maps}: following the discussion of \cite{moore:qmna}, the study of non-commutative families is the study of generalizations of the results here from density states---thought of as defining a completely positive map from a $C^{*}$-algebra into $\mathbb{C}$---to completely positive maps.
		The Kasparov bimodule is a generalization of the GNS module to completely positive maps; one can possibly generalize the results here using such a bimodule as a basic building block.
\end{enumerate}

\subsubsection*{Disclaimer}
When working with mixed density states we avoid the use of the word ``entanglement" as our homological computations are better adapted to the detection of the failure a very weak form of factorizability (or lack-thereof) called ``support factorizability".
For pure states, support factorizability and the usual notion of factorizability coincide.
However, support factorizability does not imply factorizability, or even separability. 
Nevertheless, the homological technology developed here is quite general; so it would not be surprising if future work demonstrated that it can be used to understand the entanglement of mixed states (defined as the failure of separability).

\subsection{Acknowledgements}
This work is the result of approximately a decade's worth of discussions with various colleagues. 
Detailed comments, suggestions, and encouragement by Gregory Moore influenced the majority of this paper.
Throughout Daniel Carney and Victor Chua were extremely helpful in pointing out relevant and related work. 
Brief discussions with Pierre Baudot, Daniel Bennequin, and Juan Pablo Vigneaux were helpful in inspiring possible future connections with their work.  
For discussions at the early stages I would like to thank: David Ben-Zvi, Yingyue Boretz, Aaron Fenyes, Daniel S. Freed, Andrew Neitzke, and E.C.G. Sudarshan; particularly Aaron Fenyes whose early involvement inspired a productive viewpoint in terms of multi-body correlated operators.
For discussions at later stages I thank Richard Derryberry, Gregory Moore, and John Terilla.
Jackie Shadlen and James Stokes were helpful in pointing out an interesting class of examples that might be incorporated in future versions.
My canine colleagues Thelma and Louise were helpful in pointing out relevant smells that could not be incorporated into the electronic version of this paper.
The author acknowledges support from the US Department of Energy under grant DOE-SC0010008.

\section{Terminology and Notation \label{sec:notation}}
In this section we introduce some basic (mostly standard) notation.
It can be skipped and returned to as needed.
The first occurrences of particularly important definitions throughout the paper will be highlighted in \newword{blue}.
\begin{enumerate}
	\item The algebra of operators Hilbert/vector space $\hilb$ will be denoted as $\algebra{\hilb}$.
		The space of trace class operators on $\hilb$ will be denoted as $\states{\hilb}$.
		Because we restrict our attention to finite dimensions in the majority of this paper the adjectives ``bounded" and ``trace-class" are not necessary.

	\item If $\hilb$ is a Hilbert space, then $\Dens(\hilb)$ is the (convex) set of density states on $\hilb$.  When $\hilb \neq 0$ this is the set of positive semidefinite trace 1 operators.  When $\hilb = 0$ we take $\Dens(\hilb) = \{0\}$.

	\item The symbols $\dens{\rho}$ and $\dens{\varphi}$, along with subscript decorations, will be reserved for density states.

	\item Elements of $\algebra{\hilb}$ will be denoted with lowercase Latin letters (e.g.\ $r,a$ and $b$); elements of $\states{\hilb}$ will be denoted with Greek letters (e.g.\ $\dens{\gamma},\, \dens{\alpha},\, \dens{\beta}$).
\end{enumerate}
The following are some linear algebraic remarks.
\begin{enumerate}
	\item Beyond \S\ref{sec:alg_and_states}, all vector spaces in this paper will be finite dimensional vector spaces over $\mathbb{C}$.
		With this in mind, we will occasionally say ``vector spaces" instead of ``finite dimensional complex vector spaces".

	\item For any two vector (or Hilbert) spaces $V$ and $W$, the vector space of homomorphisms from $V$ to $W$ will be denoted as $\Hom(V,W)$.

	\item The dual $\Hom(V,\mathbb{C})$ of a vector space $V$ will be denoted $\newmath{V^{\vee}}$.  (In the infinite dimensional generalizations of the results of this paper, we work with Banach spaces instead of vector spaces and $V^{\vee}$ is replaced with the continuous dual of bounded linear maps from $V$ to $\mathbb{C}$.)
\end{enumerate}

The reader is expected to have a passing familiarity with the notion of (co)chain complexes, but we review some relevant definitions here to set up appropriate notation.
Complexes and cochain complexes will be denoted in non-italicized font.
Instead of using the common technique of notationally distinguishing cohomological grading from homological grading via the placement of a bullet (e.g.\ $\cpx{C}^{\bullet}$ vs $\cpx{C}_{\bullet}$) we will write chain complexes in lowercase and cochain complexes in upper case.
\begin{enumerate}
	\item A \newword{cochain complex} (of complex vector spaces) $\cpx{C}$ is the data of a collection of vector spaces $(C_{n})_{n \in	\mathbb{Z}}$ (referred to as \textit{components}) and linear maps (referred to as \textit{coboundary maps}) $(d^{n}: \cpx{C}^{n} \rightarrow \cpx{C}^{n+1})_{n \in \mathbb{Z}}$ satisfying $d^{n + 1} \circ d^{n} = 0$ for all $n \in \mathbb{Z}$ (equivalently $\image(d^{n}) \leq \ker(d^{n+1})$ for all $n$).  As a notational shorthand, sometimes we will simply write:
		\begin{align*}
			\cpx{C} = \cdots
			\overset{d^{-3}}{\longrightarrow} \cpx{C}^{-2}
			\overset{d^{-2}}{\longrightarrow} \cpx{C}^{-1}
			\overset{d^{1}}{\longrightarrow} \cpx{C}^{0}
			\overset{d^{0}}{\longrightarrow} \cpx{C}^{1}
			\overset{d^{1}}{\longrightarrow} \cpx{C}^{2}
			\overset{d^{2}}{\longrightarrow} \cdots.
		\end{align*}
		Elements of $\ker(d^{k})$ are called \textit{$k$-cocycles}.

	\item A \newword{chain complex} (of complex vector spaces) $\cpx{c}$ is the	data of a collection of vector spaces $(\cpx{c}_{n})_{n \in \mathbb{Z}}$ and linear maps (referred to as \textit{boundary maps}) $\partial_{n}:\cpx{c}_{n} \rightarrow c_{n -1}$ such that $\partial_{n-1} \circ \partial_{n} = 0$ for all $n \in \mathbb{Z}$ (equivalently $\image(\partial_{n}) \leq \ker(\partial_{n-1})$ for all $n$).  As with cochain complexes, we will occasionally use the shorthand
		\begin{align*}
			\cpx{c} = \cdots \overset{\partial^{-2}}{\longleftarrow}
			\cpx{c}^{-2} \overset{\partial^{-1}}{\longleftarrow} \cpx{c}^{0}
			\overset{\partial{1}}{\longleftarrow} \cpx{c}^{0}
			\overset{\partial^{1}}{\longleftarrow} \cpx{c}^{1}
			\overset{\partial^{2}}{\longleftarrow} \cpx{c}^{2}
			\overset{\partial^{3}}{\longleftarrow} \cdots
		\end{align*}

	\item A (co)chain complex is \newword{bounded} if only finitely many components are non-trivial.  In this paper we will only consider bounded complexes.

	\item Given cochain complexes $\cpx{C}$ and $\cpx{D}$, a \newword{cochain morphism} $\sigma: \cpx{C} \rightarrow \cpx{D}$ is a collection of linear maps $(\sigma^{k}: \cpx{C}^{k} \rightarrow \cpx{D}^{k})_{k \in \mathbb{Z}}$ such that $\sigma^{k+1} \circ d^{k} = d^{k} \circ \sigma^{k}$ for all $k \in \mathbb{Z}$.  A chain morphism between chain complexes is defined similarly.  A (co)chain isomorphism is a cochain morphism that is componentwise an isomorphism.

	\item The cohomology of a cochain complex $\cpx{C}$ with coboundary $d$ is the graded vector space $H[\cpx{C}] = \bigoplus_{k \in \mathbb{Z}}H^{k}[\cpx{C}]$ whose $k$th component $H^{k}[\cpx{C}]$---referred to as the $k$th cohomology---is defined as
		\begin{align*}
			\newmath{H^{k}[\cpx{C}]} := \ker(d^{k})/\image(d^{k-1})
		\end{align*}
		similarly the homology of a chain complex $\cpx{c}$ is the graded vector space $H[\cpx{c}] = \bigoplus_{k \in \mathbb{Z}} H_{k}[\cpx{c}]$ whose $k$th component $H^{k}[\cpx{c}]$---referred to as the $k$th homology---is defined as
		\begin{align*}
			\newmath{H_{k}[\cpx{c}]} := \ker(\partial_{k})/\image(\partial_{k+1}).
		\end{align*}

	\item The \newword{Euler characteristic} of a cochain complex $\cpx{C}$ is defined as
		\begin{align*}
			\newmath{\chi(\cpx{C})} := \sum_{k \in \mathbb{Z}}(-1)^{k} \dim_{\mathbb{C}} H^{k} [\cpx{C}],
		\end{align*}
		although (as a	corollary of the rank-nullity theorem for vector spaces), one can calculate it directly from the dimensions of cochain complexes:
		\begin{align*}
			\chi(\cpx{C}) = \sum_{k \in \mathbb{Z}} (-1)^{k}\dim_{\mathbb{C}} \cpx{C}^{k}.
		\end{align*}
		The definition of the Euler characteristic of a chain complex is syntactically identical:
		\begin{align*}
			\newmath{\chi(\cpx{c})} := \sum_{k \in \mathbb{Z}} (-1)^{k} \dim_{\mathbb{C}} H_{k} [\cpx{c}] = \sum_{k \in \mathbb{Z}} (-1)^{k}\dim_{\mathbb{C}} \cpx{c}_{k}.
		\end{align*}

	\item The \newword{Poincar\'{e} polynomial} of a cochain complex $\cpx{C}$ is a polynomial (which we take to be in the variable ``$y$") defined as:
		\begin{align*}
			P_{\cpx{C}} := \sum_{k \in \mathbb{Z}}y^{k} \dim_{\mathbb{C}} H^{k} [\cpx{C}] \in \mathbb{Z}[y].
		\end{align*}
		with the definition for chain complexes once again defined by lowering indices.

\end{enumerate}
We also offer some remarks for the overly pedantic reader (or for the overly pedantic writer)
\begin{enumerate}
	\item  It is notationally and computationally convenient to work throughout with embedded subspaces	rather than abstract vector spaces.
		An \textit{embedded subspace} $\iota_{W}: W \hookrightarrow V$ of a vector space $V$ is a linear injection from $W$ into $V$.
		For any vector space $V$, we can think of it as an embedded subspace in a trivial way:  $\mathrm{id}_{V}: V \rightarrow V$.
		The tensor product of two embedded subspaces $\iota_{W}: W	\hookrightarrow V$  and $\iota_{Y}: Y \hookrightarrow X$ is the embedded subspace $\iota_{W} \otimes \iota_{Y}: W \otimes Y \hookrightarrow V \otimes X$.
		Every vector space in this paper (including cohomology components) should be secretly thought of as an embedded subspace.
		However, to prevent a proliferation of unnecessary notation, we will employ the	convenient (and widespread) abuse of conflating an embedded subspace with its image.

	\item The direct sum of vector spaces is denoted as $V \oplus W$ and the Cartesian product as $V \times W$.
		These two notions give precisely the same vector space, however when we write $V \oplus W$ we are emphasizing its properties as a categorical coproduct, and when we write $V \times W$ we are emphasizing its properties as a categorical product.
		Roughly speaking this means that $V \oplus W$ should be secretly thought of as a vector space along with the two maps $\iota_{V}: V \rightarrow V \oplus W$ and $\iota_{W}: W \rightarrow V \oplus W$ given by $\iota_{V}: v \mapsto v \oplus 0$ and $\iota_{W}: w \mapsto 0 \oplus w$.
		On the other hand, $V \times W$ should be secretly thought of as a vector space along with the two projection maps $V	\times W \rightarrow V$ and $V \times W \rightarrow W$.
\end{enumerate}

\section{Algebras and States \label{sec:alg_and_states}}
Let $\hilb$ be a Hilbert space; we allow $\hilb$ to be infinite dimensional in this section.
Denote the algebra of bounded endomorphisms (``operators'') on $\hilb$ as $\algebra{\hilb}$, the subset of self-adjoint bounded operators (``observables'') as $\algebra{\hilb}_{\mathrm{s.a.}} = \{a \in \algebra{\hilb}: a = a^{*} \}$, and the (convex) set of of density states (positive semidefinite operators with unit trace\footnote{If $\hilb$ is the trivial Hilbert space then ``unit trace" should be replaced with zero trace.}) on $\hilb$ as $\Dens(\hilb)$.
The trace supplies a pairing:
\begin{align*}
	\Dens(\hilb) \times  \algebra{\hilb}_{\mathrm{s.a.}}  &\longrightarrow \mathbb{R}\\
	(\dens{\rho},r)			         &\longmapsto     \Tr[\dens{\rho}r].
\end{align*}
However, it is often easier complexify things and work with $\algebra{\hilb}$---which might be thought of as the space of ``complexified observables" via the decomposition $a  = a^{R}+ i a^{I}$ into self-adjoint $a^{R}$ and $a^{I}$)---and the space of trace-class operators:\footnote{$\algebra{\hilb}$ is a $C^{*}$-algebra, in particular it is an associative $\mathbb{C}$-algebra.
The space of self-adjoint operators on the other hand do not have an associative algebra structure, but rather the structure of a Jordan $C^{*}$-algebra, a somewhat less popular notion.}
\begin{align*}
	\states{\hilb} := \{\dens{\gamma} \in \algebra{\hilb}: \Tr[\dens{\gamma}]< \infty \},
\end{align*}
which might be thought of as the space of ``complexified" density states either via the polar decomposition theorem, or a combination of the self-adjoint and Jordan decompositions: $\dens{\gamma} = \Tr(\dens{\gamma}) \left[ (\dens{\gamma}^{R}_{+} - \dens{\gamma}^{R}_{-}) + i(\dens{\gamma}^{I}_{+} - \dens{\gamma}^{I}_{-}) \right]$ for $\dens{\gamma}^{R}_{\pm},\, \dens{\gamma}^{I}_{\pm} \in \Dens(\hilb)$.
The trace pairing extends to a bilinear map:
\begin{align*}
	\states{\hilb} \times  \algebra{\hilb}  &\longrightarrow \mathbb{\mathbb{C}}\\
	(\dens{\gamma},r)	    &\longmapsto     \Tr[\dens{\gamma}r].
\end{align*}
If we equip $\algebra{\hilb}$ with the operator norm, and $\states{\hilb}$ with the trace norm, then we can use this trace pairing to define an isometric isomorphism (of Banach spaces)
\begin{align*}
	(-)^{\Tr}: \algebra{\hilb} & \overset{\sim}{\longrightarrow} \states{\hilb}^{\vee}\\
	r        & \longmapsto r^{\Tr} := (\Tr[(-) r]: \dens{\gamma} \mapsto \Tr[\dens{\gamma}r])
\end{align*}
where $(-)^{\vee}$ is denoting the continuous dual (bounded linear maps into $\mathbb{C}$ equipped with its usual norm).
On the other hand, the map
\begin{align*}
	\mathbb{E}_{-}: \states{\hilb} & \longrightarrow \algebra{\hilb}^{\vee}\\
	\dens{\gamma}  & \longmapsto (\mathbb{E}_{\dens{\gamma}}: r \mapsto \Tr[\dens{\gamma} r])
\end{align*}
is an isometric embedding, but is \textit{not} an isomorphism when $\hilb$ is infinite dimensional (elements in the image of this map are called normal linear functionals).

In this paper we will only work with finite dimensional $\hilb$; all subtle analytic conditions vanish and one can work purely in the realm of linear algebra.
In particular, all endomorphisms of $\hilb$ are trace class and bounded; hence, at the level of vector spaces, then we can write $\algebra{\hilb} = \End(\hilb) = \states{\hilb}$.
Secondly, the map $\mathbb{E}_{-}$ is also an isomorphism; giving an identification between $\End(\hilb)^{\vee}$ and $\End(\hilb)$.
Nevertheless, we will continue to write $\algebra{\hilb}$ and $\states{\hilb}$ and avoid using $\mathbb{E}_{-}$ as an isomorphism in order to make the generalizations to infinite dimensions obvious and distinguish when we are secretly thinking of things as ``operators" vs. ``complexified states''.

This paper is about quantum mechanics.
However, this notation might also allow some readers to guess at the correct classical, or mixed quantum-classical versions of the constructions here.
In general $\algebra{\hilb}$ can be replaced by a $W^{*}$-algebra (or von Neumann algebra with choice of a particular representation) and $\states{\hilb}$ with its predual (which are canonically embedded into a subspace of continuous linear functionals of the $W^{*}$ algebra in a representation there is an isometric surjection into $\states{\hilb}$ from the Banach space of trace class operators.
In the purely commutative situation we can replace $\algebra{\hilb}$ with the (commutative) algebra of
essentially bounded functions on some measurable space and $\states{\hilb}$ with the space of measures on that
measurable space (if we only care about measure spaces given by a finite collection of points, we can replace
$\algebra{\hilb}$ with the algebra of $\mathbb{C}$-valued functions on a finite set).

\section{Building Blocks \label{sec:building_blocks}}
Throughout the remainder of this paper $\hilb$ will always be a finite dimensional Hilbert space, $\newmath{\algebra{\hilb}}$ its algebra of endomorphisms, $\newmath{\states{\hilb}}$ the vector space of complexified density states (which is just the underlying vector space of the algebra of endomorphisms in finite dimensions), and $\newmath{\Dens(\hilb)}$ the (convex) set of density states on $\hilb$.
We begin with a definition critical to the construction of our (co)chain complexes.
\begin{definition}{}{}
	Let $\dens{\rho} \in \Dens(\hilb)$, then the \newword{support projection of $\dens{\rho}$}--denoted $\newmath{\supp_{\dens{\rho}}} \in \algebra{\hilb}$---is the orthogonal projection onto $\image(\dens{\rho}) \subseteq \hilb$.
\end{definition}
From the support projection we will build cochain complexes using following building blocks of vector subspaces of $\algebra{\hilb}$.
(The names of these building blocks will be justified in \S\ref{sec:interpretation}.)
\begin{definition}{}{GNS_Com_def}
	Let $\dens{\rho} \in \Dens(\hilb)$, then define the embedded subspaces.\footnote{For any algebra $A$ and element $x \in A$, we use the common notation $Ax := \{a x : a \in A \}$ and $y A x = \{y a x: a \in A\}$.}
	\begin{equation*}
		\begin{array}{lll}
			\newmath{\mathtt{GNS}(\dens{\rho})}  := & \supp_{\dens{\rho}}  \algebra{\hilb} 				       & \leq \algebra{\hilb}\\
			\newmath{\mathtt{Com}(\dens{\rho})}  := & \supp_{\dens{\rho}}  \algebra{\hilb} \supp_{\dens{\rho}} & \leq \algebra{\hilb}.
		\end{array}
	\end{equation*}
\end{definition}
As a somewhat more explicit description: note that a density state $\dens{\rho}$ supplies an orthogonal decomposition $\hilb \cong \mathcal{I} \oplus \mathcal{K}$ where $\mathcal{I}:= \image(\dens{\rho})$ and $\mathcal{K} := \ker(\dens{\rho}) = \image(\dens{\rho})^{\perp}$.  Using this decomposition we can identify the $\mathbb{C}$-algebra of endomorphisms $\algebra{\hilb}$ with the algebra of $2 \times 2$ block matrices of the form
\begin{align}
	r &=
	\begin{blockarray}{lcc}
		\mathcal{I} & \mathcal{K} & \\
		\begin{block}{(cc)r}
			r_{\mathcal{I} \mathcal{I}} & r_{\mathcal{I} \mathcal{K}} & \mathcal{I}\\
			r_{\mathcal{K} \mathcal{I}} & r_{\mathcal{K} \mathcal{K}} & \mathcal{K}\\
		\end{block}
	\end{blockarray}
	\label{eq:block_form}
\end{align}
where $r_{\mathcal{I} \mathcal{I}} \in \End(\mathcal{I}), r_{\mathcal{K} \mathcal{I}} \in \Hom(\mathcal{I}, \mathcal{K}),\, r_{\mathcal{I} \mathcal{K}} \in \Hom(\mathcal{K}, \mathcal{I})$ and $r_{\mathcal{K} \mathcal{K}} \in \End(\mathcal{K})$.
Then we have
\begin{align*}
	\mathtt{GNS}(\dens{\rho}) &= \left\{ \blockmat{*}{0}{*}{0}\right \} \leq \algebra{\hilb},\\
	\mathtt{Com}(\dens{\rho})  &= \left\{ \blockmat{*}{0}{0}{0} \right \} \leq \algebra{\hilb}.
\end{align*}
Note, furthermore, that there are canonical isomorphisms (of $\mathbb{C}$-vector spaces)
\begin{equation*}
	\begin{array}{lclcl}
		\mathtt{GNS}(\dens{\rho}) &\cong& \Hom[\image(\dens{\rho}), \hilb] &\cong& \hilb \otimes \image(\dens{\rho})^{\vee}\\
		\mathtt{Com}(\dens{\rho}) &\cong& \End[\image(\dens{\rho})] &\cong& \image(\dens{\rho}) \otimes \image(\dens{\rho})^{\vee}.
	\end{array}
\end{equation*}
However it is helpful (both computationally and notationally) to consider $\mathtt{GNS}(\dens{\rho})$ and $\mathtt{Com}(\dens{\rho})$ as embedded subspaces in $\algebra{\hilb}$.

The building blocks for chain complexes are defined as:
\begin{definition}{}{}
	Let $\dens{\rho} \in \Dens(\hilb)$,
	\begin{equation*}
		\begin{array}{lll}
			\newmath{\mathtt{gns}(\dens{\rho})}  := & \supp_{\dens{\rho}} \states{\hilb} 					 & \leq \states{\hilb}\\
			\newmath{\mathtt{com}(\dens{\rho})}  := & \supp_{\dens{\rho}} \states{\hilb} \supp_{\dens{\rho}} & \leq \states{\hilb}.
		\end{array}
	\end{equation*}
\end{definition}
In finite dimensions, $\algebra{\hilb} = \states{\hilb} = \End(\hilb)$ (as vector spaces), so we immediately have $\mathtt{com}(\hilb) = \mathtt{Com}(\hilb)$; however, as alluded to in \S\ref{sec:alg_and_states} this finite-dimensional accident is somewhat misleading in the context of interpretation and generalizability; so we continue to distinguish between our usage of $\mathtt{com}(\dens{\rho})$ and $\mathtt{Com}(\dens{\rho})$

Using the trace pairing, the building blocks for cochain complexes are canonically isomorphic to the duals of the building blocks for chain complexes.  Indeed, recall the map
\begin{align}
	(-)^{\Tr}: \algebra{\hilb} &\overset{\sim}{\longrightarrow} \states{\hilb}^{\vee}\\
	r        &\longmapsto (\dens{\gamma} \mapsto \Tr[\dens{\gamma}r]),
\end{align}
introduced in \S\ref{sec:alg_and_states}.  Its restrictions provide isomorphisms:\footnote{We can also use the map $\mathbb{E}_{-}$ defined in \S\ref{sec:alg_and_states} to get maps $\mathtt{gns}(\dens{\rho}) \rightarrow \mathtt{GNS}(\dens{\rho})^{\vee}$ and $\mathtt{com}(\dens{\rho}) \rightarrow \mathtt{Com}(\dens{\rho})^{\vee}$.  In finite dimensions such maps are isomorphisms; in infinite dimensions they are not isomorphisms, but (isometric) embeddings.}
\begin{align}
	(-)^{\Tr}|_{\mathtt{GNS}(\dens{\rho})}: \mathtt{GNS}(\dens{\rho}) &\overset{\sim}{\longrightarrow} \mathtt{gns}(\dens{\rho})^{\vee}
	\label{eq:GNS_dual_iso}
\end{align}
and
\begin{align}
	(-)^{\Tr}|_{\mathtt{Com}}: \mathtt{Com}(\dens{\rho}) &\overset{\sim}{\longrightarrow} \mathtt{com}(\dens{\rho})^{\vee}.
	\label{eq:Com_dual_iso}
\end{align}

\section{Building Blocks and Right Essential Equivalence Classes \label{sec:interpretation}}
It is worthwhile to take a moment to give a meaning to the cochain building blocks of the previous section. We begin with a sensible definition.
\begin{definition}{}{right_essential_equivalence}
	Two operators $x,\, y \in \algebra{\hilb}$ are \newword{right essentially equivalent with respect to $\dens{\rho}$} if
	\begin{align*}
		\Tr[\dens{\rho} r^{*} x ] = \Tr[\dens{\rho} r^{*} y],
	\end{align*}
	for all $r \in \algebra{\hilb}$.
	Equivalently, $x$ and $y$ are right essential equivalent if $x -y \in \mathfrak{N}_{\dens{\rho}}$ where
	\begin{align*}
		\newmath{\mathfrak{N}_{\dens{\rho}}} := \{z \in \algebra{\hilb}: \text{$\Tr(\dens{\rho} r^{*} z) = 0$ for all $r \in \algebra{\hilb}$} \}.
	\end{align*}
	When $\dens{\rho}$ is understood, we will simply say $x$ and $y$ are ``right essentially equivalent".
\end{definition}
From the definition we see that $\mathfrak{N}_{\dens{\rho}}$ is a vector space; this observation allows us to verify that right essential equivalence satisfies the axioms of an equivalence relation on the set $\algebra{\hilb}$.  The notion of right essential equivalence is a non-commutative generalization of the notion of almost everywhere (a.e.) equivalence of measurable functions (``commutative random variables") on a measurable space equipped with some fixed measure, and $\mathfrak{N}_{\dens{\rho}}$ is the generalization of functions a.e.\ equivalent to zero.  Indeed, a Cauchy-Schwarz argument shows that
\begin{align*}
	\mathfrak{N}_{\dens{\rho}} = \{z \in \algebra{\hilb}: \Tr(\dens{\rho} z^{*} z) = 0 \}
\end{align*}
which is the non-commutative version of the statement that a.e.\ zero functions are precisely those whose absolute square integrates to zero.

Right essential equivalence is compatible with left multiplication by elements of $\algebra{\hilb}$ as $\mathfrak{N}_{\dens{\rho}}$ is a left ideal: given any  $z \in \mathfrak{N}_{\dens{\rho}}$ we have $a z \in \mathfrak{N}_{\dens{\rho}}$ for any $a \in \algebra{\hilb}$.
As a result, the set of right essential equivalence classes\footnote{This is the non-commutative generalization of the space of essentially bounded functions on a measurable space modulo the two-sided ideal of measure zero essentially bounded functions.
	In the non-commutative world, we have either left or right ideals of measure zero functions.
The distinction between left and right ideals comes from the non-equivalent choices of what is meant by the absolute square of a function: either $z^{*} z$ or $z z^{*}$.}
\begin{align*}
	\algebra{\hilb}/\mathfrak{N}_{\dens{\rho}}
\end{align*}
forms a left
$\algebra{\hilb}$-module.
We claim that this module is canonically isomorphic as a left module to $\mathtt{GNS}(\dens{\rho})$, after equipping the latter with a left-module structure given by left multiplication by elements of $\algebra{\hilb}$.
To verify this claim we begin with the following proposition.
\begin{proposition}{}{vanishing_ideal_characterizations}
	\begin{center}
		$\mathfrak{N}_{\dens{\rho}} = \algebra{\hilb} (1 -\supp_{\dens{\rho}})$.
	\end{center}
\end{proposition}
\begin{proof}
	We leave the proof as an exercise.
	As a hint: one can use non-degeneracy of the trace pairing to show that $\mathfrak{N}_{\dens{\rho}} = \{x \in \algebra{\hilb}: x \sqrt{\dens{\rho}} = 0 \}$.
\end{proof}
This proposition meshes with our intuition if we think of $\mathfrak{N}_{\dens{\rho}}$ as generalizing the ideal of functions that are a.e.\ equivalent to zero: in classical measure theory, the ideal of such functions only depends on where the measure is non-zero (i.e.\ the support of the measure), as opposed to its finer details.
Moreover, the proposition gives us an explicit way to verify right essential equivalence.
\begin{corollary}{}{}
	Let $x,y \in \algebra{\hilb}$, then $x$ and $y$ are right essentially equivalent with respect to $\dens{\rho}$ if and only if $x \supp_{\dens{\rho}} = y \supp_{\dens{\rho}}$.
\end{corollary}
Moreover, for every equivalence class $x + \mathfrak{N}_{\dens{\rho}} \in \algebra{\hilb}/\mathfrak{N}_{\dens{\rho}}$, the element $x \supp_{\dens{\rho}} \in
\mathtt{GNS}(\dens{\rho})$ is independent of the choice of representative $x \in \mathfrak{N}_{\dens{\rho}}$.
Thus, in combination with the above corollary, we have that the quotient map
\begin{align*}
	q: \algebra{\hilb} &\longrightarrow \algebra{\hilb}/\mathfrak{N}_{\dens{\rho}}\\
		x			   &\longmapsto      x + \mathfrak{N}_{\dens{\rho}}
\end{align*}
has a (well-defined) section:\footnote{I.e.\ a map such that $q \circ s$ is the identity map on $\algebra{\hilb}/\mathfrak{N}_{\dens{\rho}}$.}

\begin{align*}
	s: \algebra{\hilb}/\mathfrak{N}_{\dens{\rho}} & \longrightarrow \algebra{\hilb} \\
		x + \mathfrak{N}_{\dens{\rho}}            & \longmapsto     x \supp_{\dens{\rho}}.
\end{align*}
Moreover $q$ and $s$ are equivariant with respect to the left action of $\algebra{\hilb}$.
This gives us the following:
\begin{proposition}{}{}
	$\mathtt{GNS}(\dens{\rho})$ and $\algebra{\hilb}/\mathfrak{N}_{\dens{\rho}}$ are isomorphic as left $\algebra{\hilb}$ modules using the maps $s$ and $q|_{\mathtt{GNS}(\dens{\rho})}$.
\end{proposition}

The following remark describes the relationship of the module of right essential equivalence classes to the GNS representation.
\begin{remark}{Right Essential Equivalence Classes and the GNS Representation}{GNS_as_GNS}
	The module $\algebra{\hilb}/\mathfrak{N}_{\dens{\rho}}$ is a precursor to the \textit{Gelfand-Neumark-Segal (GNS) representation}: we can equip $\algebra{\hilb}/\mathfrak{N}_{\dens{\rho}}$ with a non-degenerate Hermitian inner product given by the descent of the pairing $(x,y) \mapsto \Tr[\dens{\rho} x^{*} y]$.
	The result is a Hilbert space;\footnote{In infinite dimensions we would need to take a completion of $\algebra{\hilb}/\mathfrak{N}_{\dens{\rho}}$} the left module structure gives this Hilbert space the structure of a $*$-representation of $\algebra{\hilb}$.
	In this sense, giving $\algebra{\hilb}/\mathfrak{N}_{\dens{\rho}}$ the name \newword{GNS module} might be appropriate.
	The results of this paper most easily generalize to infinite dimensions when we think of  $\algebra{\hilb}/\mathfrak{N}_{\dens{\rho}}$ as a Banach module equipped with the operator norm.
	In fact, there is a whole family of ``$L^{p}$-norms" one can place on $\algebra{\hilb}/\mathfrak{N}_{\dens{\rho}}$: the $L^{2}$-norm gives the GNS representation, while the $L^{\infty}$-norm is morally the structure we consider in this paper.
\end{remark}

Just as $\mathtt{GNS}(\dens{\rho}) = \algebra{\hilb} \supp_{\dens{\rho}}$ is naturally a left $\algebra{\hilb}$-submodule of $\algebra{\hilb}$, the subspace $\mathtt{Com}(\dens{\rho}) = \supp_{\dens{\rho}} \algebra{\hilb} \supp_{\dens{\rho}}$ is naturally a $\mathbb{C}$-subalgebra.\footnote{Indeed, the product of any two elements in $\supp_{\dens{\rho}} \algebra{\hilb} \supp_{\dens{\rho}}$ is in $\supp_{\dens{\rho}} \algebra{\hilb} \supp_{\dens{\rho}}$.}
As an (non-embedded) algebra it is canonically isomorphic to $\End[\image(\dens{\rho})]$.
To give the relationship of this algebra to right essential equivalence classes we introduce the following lemma.
\begin{lemma}{}{equivariant_endo}
	Let $k$ be a field, $A$ be an associative $k$-algebra, $L \subseteq A$ a left ideal of $A$, and ${}_{A}(A/L)$ the
	left $A$ module\footnote{Assuming the usual postcomposition of endomorphisms as our multiplication: $m(\phi_{1}, \phi_{2}) = \phi_{1} \circ \phi_{2}$.} given by action via left multiplication.
	Define $I := \{a \in A: L a \subset L\}$, then the $k$-algebra of equivariant endomorphisms $\End_{A} \left[{}_{A}(A/L) \right]$ is isomorphic to $(I/L)^{\mathrm{op}}$, via the descent of the right multiplication map\footnote{If $A$ is an $\mathbb{C}$-algebra with multiplication $(x,y) \mapsto x \cdot y$, then $A^{\op}$ is the $\mathbb{C}$-algebra with the same underlying $\mathbb{C}$-vector space, but with multiplication $(x,y) \mapsto y \cdot x$.}
	\begin{align*}
		\mathpzc{r} :I^{\mathrm{op}} &\longrightarrow \End_{A} \left[{}_{A}(A/L) \right]\\
		c &\longmapsto (n + L \mapsto n c + L).
	\end{align*}
\end{lemma}
\begin{proof}
	The proof is a straightforward generalization of Schur's lemma.
	See Appendix~\ref{app:equivariant_endo_lemma}.
\end{proof}
Applying this to our situation where $k = \mathbb{C},\,A = \algebra{\hilb}$ and $L = \mathfrak{N}_{\dens{\rho}} = \algebra{\hilb} (1-\supp_{\dens{\rho}})$ we have $I = \supp_{\dens{\rho}} \algebra{\hilb} \supp_{\dens{\rho}} + \mathfrak{N}_{\dens{\rho}}$.  We can see this explicitly if we recall the block matrix decomposition of an element $r \in \algebra{\hilb}$ (described above \eqref{eq:block_form}):
\begin{align*}
	r &=
	\begin{blockarray}{lcc}
		\mathcal{I} & \mathcal{K} & \\
		\begin{block}{(cc)r}
			r_{\mathcal{I} \mathcal{I}} & r_{\mathcal{I} \mathcal{K}} & \mathcal{I}\\
			r_{\mathcal{K} \mathcal{I}} & r_{\mathcal{K} \mathcal{K}} & \mathcal{K}\\
		\end{block}
	\end{blockarray}.
\end{align*}
Using this decomposition:
\begin{align*}
	L = \left \{\blockmat{*}{*}{0}{0} \right \} = \algebra{\hilb}(1-\supp_{\dens{\rho}})
\end{align*}
so it is easy to see that
\begin{align*}
	I = \left \{ \blockmat{*}{0}{*}{*} \right \} = \supp_{\dens{\rho}} \algebra{\hilb} \supp_{\dens{\rho}} + \algebra{\hilb}(1-\supp_{\dens{\rho}})
\end{align*}
In the same manner that the quotient module $\algebra{\hilb}/\mathfrak{N}_{\dens{\rho}}$ was canonically identified with the submodule $\algebra{\hilb} \supp_{\dens{\rho}}$ via application of support projections, we can identify the quotient algebra $I/L$ with the subalgebra $\supp_{\dens{\rho}} \algebra{\hilb} \supp_{\dens{\rho}}$.  As a result, we have the following.
\begin{proposition}{}{equivariant_endos}
	The right action of the subalgebra $\supp_{\dens{\rho}} \algebra{\hilb} \supp_{\dens{\rho}}$ on $\algebra{\hilb}$ provides a canonical isomorphism between the algebra of equivariant isomorphisms of the submodule $\algebra{\hilb} \supp_{\dens{\rho}}$ and $(\supp_{\dens{\rho}} \algebra{\hilb} \supp_{\dens{\rho}})^{\op}$.
\end{proposition}

In the $C^{*}$-algebraic world where one considers $*$-representations, the algebra of equivariant endomorphisms usually goes by the name ``commutant" (this is the origin of our notation $\mathtt{Com}$ and $\mathtt{com}$); the commutant of a $*$-representation of a $C^{*}$-algebra is also a $C^{*}$-algebra.\footnote{Furthermore, it is a von Neumann algebra}
In fact one can show that the endomorphisms given by the right action of $\supp_{\dens{\rho}} \algebra{\hilb} \supp_{\dens{\rho}}$, thought of as a $C^{*}$-algebra, is isomorphic to the commutant of the GNS representation associated to $\dens{\rho}$ (see Rmk.~\ref{rmk:GNS_as_GNS}).

Unfortunately, our tools in this paper are too brutish, and when we pass to (co)homology (and Poincar\'{e} polynomials) we will forget any underlying module structures and only remember the structure of underlying vector spaces.  However, Prop.~\eqref{prop:equivariant_endos} shows us that the complexes of vector spaces built out of $\mathtt{Com}(\dens{\rho})$ depend on the module structures that were forgotten when  constructing complexes out of $\mathtt{GNS}(\dens{\rho})$; in this way, we can still encode partial information about forgotten module structure while only working with vector spaces.

\section{Bipartite Complexes \label{sec:bipartite_complexes}}
The following definition and notation for bipartite (density) states emphasizes that the Hilbert spaces in the factorization are part of the data.
\begin{definition}[label=def:bipartite_density_states]{}{}
	\begin{enumerate}
		\item A \newword{bipartite density state} is a tuple $(\hilb_{\sA}, \hilb_{\sB}, \dens{\rho}_{\sAB})$ where $\hilb_{\sA}$ and $\hilb_{\sB}$ are non-zero Hilbert spaces and $\dens{\rho}_{\sAB} \in \Dens(\hilb_{\sA} \otimes \hilb_{\sB})$.

		\item Given a bipartite density state $(\hilb_{\sA}, \hilb_{\sB}, \dens{\rho}_{\sAB})$ we define the associated reduced density states:
			\begin{align*}
				\newmath{\dens{\rho}_{\sA}} &:= \Tr_{\sB} \left[\dens{\rho}_{\sAB} \right] \in \Dens(\hilb_{\sA}),\\
				\newmath{\dens{\rho}_{\sB}} &:= \Tr_{\sA} \left[\dens{\rho}_{\sAB} \right] \in \Dens(\hilb_{\sB}).
			\end{align*}

		\item When no confusion can arise, the support projection $\supp_{\dens{\rho}_{\sX}}$ will be denoted by $\newmath{\supp_{\sX}}$ and left ideal $\mathfrak{N}_{\dens{\rho}_{\sX}}$ will be denoted $\mathfrak{N}_{\sX}$ for $\sX \in \{\sAB, \sA, \sB \}$.

		\item A bipartite state $(\hilb_{\sA}, \hilb_{\sB}, \dens{\rho}_{\sAB})$ is \newword{pure} if $\dens{\rho}_{\sAB}$ is pure (i.e.\ $\dens{\rho}_{\sAB} = \psi \otimes \psi^{\vee}$ for some $\psi \in \hilb_{\sA} \otimes \hilb_{\sB}$).  A bipartite state is \newword{mixed} if it is not pure.

		\item  In the following sections a bipartite density state will be denoted by boldface version of its associated density state, i.e.\ $\bdens{\rho}_{\sAB}$ will denote $(\hilb_{\sA}, \hilb_{\sB}, \dens{\rho}_{\sAB})$, with the Hilbert spaces $\hilb_{\sA}$ and $\hilb_{\sB}$ being understood.
	\end{enumerate}
\end{definition}
Our definition of a bipartite density state includes the situation where one of $\hilb_{\sA}$ or $\hilb_{\sB}$ might be a one-dimensional Hilbert space.  Such an extreme situation might not be considered a bipartite state in practice, however all of our results follow easily using the definition provided above.  One can further generalize the above definition to allow for zero Hilbert spaces, but this only complicates the statements of the results in this paper without much payoff.

\subsection{Cochain Complexes}
Begin with a bipartite density state $\bdens{\rho}_{\sAB} =(\hilb_{\sA}, \hilb_{\sB}, \dens{\rho}_{\sAB})$.  Then we define the \newword{GNS cochain complex} (or just \newword{GNS complex}): 
\begin{align}
	\newmath{\cpx{G}(\bdens{\rho}_{\sAB})} := \cdots
	\rightarrow 0 \longrightarrow
	\mathbb{C}
	\overset{d_{\cpx{G}}^{-1}}{\longrightarrow}
	\mathtt{GNS}(\dens{\rho}_{\sA}) \times \mathtt{GNS}(\dens{\rho}_{\sB})
	\overset{d_{\cpx{G}}^{0}}{\longrightarrow}
	\mathtt{GNS}(\dens{\rho}_{\sAB})
	\longrightarrow 0 \rightarrow \cdots
	\label{eq:GNS_bipartite_complex}
\end{align}
and the \newword{commutant cochain complex} (or just \newword{commutant complex}):
\begin{align}
	\newmath{\cpx{E}(\bdens{\rho}_{\sAB})} := \cdots
	\rightarrow 0 \longrightarrow
	\mathbb{C}
	\overset{d_{\cpx{E}}^{-1}}{\longrightarrow}
	\mathtt{Com}(\dens{\rho}_{\sA}) \times \mathtt{Com}(\dens{\rho}_{\sB})
	\overset{d_{\cpx{E}}^{0}}{\longrightarrow}
	\mathtt{Com}(\dens{\rho}_{\sAB})
	\longrightarrow 0 \rightarrow \cdots
	\label{eq:Com_bipartite_complex}
\end{align}
where the coboundaries on $\cpx{G}(\bdens{\rho}_{\sAB})$ are given by
\begin{align*}
	d_{\cpx{G}}^{-1} : \lambda &\longmapsto (\lambda \supp_{\sA}, \lambda \supp_{\sB}),\\ d_{\cpx{G}}^{0}: (a,b) &\longmapsto \left(1_{\sA} \otimes b - a \otimes 1_{\sB}	\right) \supp_{\sAB},
\end{align*}
and the coboundaries on $\cpx{E}(\bdens{\rho}_{\sAB})$ are given by
\begin{align*}
	d_{\cpx{E}}^{-1} : \lambda &\longmapsto (\lambda \supp_{\sA}, \lambda \supp_{\sB}),\\
	d_{\cpx{E}}^{0}  : (a,b)   &\longmapsto \supp_{\sAB} \left(1_{\sA} \otimes b - a \otimes 1_{\sB} \right) \supp_{\sAB}.
\end{align*}
It remains to verify that these are actually cochain complexes: i.e.\ the coboundaries satisfy $d^{n+1} \circ d^{n} = 0$.  This is a straightforward corollary of the following lemma.
\begin{lemma}{Compatibility of supports}{compat_of_supports}
	\begin{center}
		$(\supp_{\sA} \otimes \supp_{\sB})\supp_{\sAB} = \supp_{\sAB}$.
	\end{center}
\end{lemma}
\begin{proof}
	A straightforward proof follows by showing that for any $a \in \mathfrak{N}_{\dens{\rho}_{\sA}}$ we have $a \otimes 1_{\sB} \in \mathfrak{N}_{\dens{\rho}_{\sAB}}$ and then applying Prop.~\ref{prop:vanishing_ideal_characterizations}.  Alternatively, in finite dimensions, one can prove this statement by showing:
	\begin{align*}
		\image(\dens{\rho}_{\sAB}) &\leq \image(\dens{\rho}_{\sA}) \otimes \hilb_{\sB},\\
		\image(\dens{\rho}_{\sAB}) &\leq \hilb_{\sA} \otimes \image(\dens{\rho}_{\sB}).
	\end{align*}
	which is an exercise in linear algebra.\footnote{One possible way is to verify this fact using Schmidt decompositions and properties of positive semidefinite operators.}
\end{proof}

\subsection{Chain complexes}
We define chain complexes\footnote{As remarked in \S\ref{sec:notation}, the switch between the notation $\times$ and $\oplus$ is notation is a hint at the manner in which the coboundaries of the complexes are constructed (and provides insight into generalizations of the theory here). However, one should keep in mind that $V \times W$ and $V \oplus W$ are canonically isomorphic as vector spaces.
While an element of $V \times W$ is written as a pair $(v,w)$ we will denote elements of $V \oplus W$ as $v \oplus w$.}
\begin{align*}
	\newmath{\cpx{g}(\bdens{\rho}_{\sAB})} := \cdots
	\leftarrow 0 \longleftarrow
	\mathbb{C}
	\overset{\partial^{\cpx{g}}_{0}}{\longleftarrow}
	\mathtt{gns}(\dens{\rho}_{\sA}) \oplus \mathtt{gns}(\dens{\rho}_{\sB})
	\overset{\partial^{\cpx{g}}_{1}}{\longleftarrow}
	\mathtt{gns}(\dens{\rho}_{\sAB})
	\longleftarrow 0 \leftarrow \cdots,
\end{align*}
named the \newword{GNS chain complex} and
\begin{align*}
	\newmath{\cpx{e}(\bdens{\rho}_{\sAB})} := \cdots
	\leftarrow 0 \longleftarrow
	\mathbb{C}
	\overset{\partial^{\cpx{e}}_{0}}{\longleftarrow}
	\mathtt{com}(\dens{\rho}_{\sA}) \oplus \mathtt{com}(\dens{\rho}_{\sB})
	\overset{\partial^{\cpx{e}}_{1}}{\longleftarrow}
	\mathtt{com}(\dens{\rho}_{\sAB})
	\longleftarrow 0 \leftarrow \cdots,
\end{align*}
named the \newword{commutant chain complex}; the boundary maps of either complex are given by the restricting the following maps to the appropriate domain:
\begin{equation*}
	\begin{array}{llcll}
		\partial_{0} &:& \dens{\alpha} \oplus \dens{\beta} &\longmapsto \Tr[\dens{\alpha}] + \Tr[\dens{\beta}],\\
		\partial_{1} &:&  \dens{\gamma}  &\longmapsto 0_{\sA} \oplus \Tr_{\sA}[\dens{\gamma}] - \Tr_{\sB}[\dens{\gamma}] \oplus 0_{\sB} = \left(-\Tr_{\sB}[\dens{\gamma}] \right) \oplus \left(\Tr_{\sA}[\dens{\gamma}] \right),
	\end{array}
\end{equation*}
where $0_{\sX} \in \states{\hilb_{\sX}}$ is the zero element.  In the construction of the cochain complexes it was immediately clear that the coboundaries landed in the appropriate spaces, but it was not obvious that they squared to zero. On the other hand, the maps above clearly satisfy the square zero condition ``$\partial_{-1} \circ \partial_{0} = 0$", but their restrictions do not obviously land in the appropriate spaces (one might expect that the expressions should be compressed by left/right multiplications of support projections).  The following lemma---which can be taken to be the ``predual" of the compatibility of supports lemma---shows that these coboundaries do indeed land in the appropriate spaces.
\begin{lemma}{}{compat_of_supports_predual}
	Let $\dens{\gamma} \in \mathtt{gns}(\dens{\rho}_{\sAB})$; define $\sA^{c} := \sB$ and $\sB^{c} := \sA$.  Then $\Tr_{\sX^{c}}[\dens{\gamma}] \in \mathtt{gns}(\dens{\rho}_{\sX})$ for $\sX \in \{\sA, \sB \}$. Similarly if $\dens{\gamma} \in \mathtt{com}(\dens{\rho}_{\sAB})$, then $\Tr_{\sX^{c}}[\dens{\gamma}] \in \mathtt{com}(\dens{\rho}_{\sX})$.
\end{lemma}
\begin{proof}
	Begin with $\dens{\gamma} \in \mathtt{gns}(\dens{\rho}_{\sAB})$, we will show that $\Tr_{\sA}[\dens{\gamma}] \in \mathtt{gns}(\dens{\rho}_{\sB})$ by verifying $\supp_{\sB} \Tr_{\sA}[\dens{\gamma}] = \Tr_{\sA}[\dens{\gamma}]$.
	Via Lemma~\ref{lem:compat_of_supports}: $\supp_{\sAB} = (\supp_{\sA} \otimes \supp_{\sB})\supp_{\sAB} = (1_{\sA} \otimes \supp_{\sB})\supp_{\sAB}$; so, $(1_{\sA} \otimes \supp_{\sB}) \dens{\gamma} = \dens{\gamma}$.  Hence, $\Tr_{\sA}[\dens{\gamma}] = \Tr_{\sA}[(1_{\sA} \otimes \supp_{\sB})\dens{\gamma}] = \supp_{\sB} \Tr_{\sA}[\dens{\gamma}]$. To verify the last equality note that
	\begin{align*}
		\Tr \{\Tr_{\sA}[(1_{\sA} \otimes \supp_{\sB}) \gamma] b\}
		&= \Tr[(1_{\sA} \otimes \supp_{\sB}) \dens{\gamma} (1_{\sA} \otimes b)]\\
		&= \Tr[(1_{\sA} \otimes b \supp_{\sB}) \gamma] \\
		&= \Tr[\supp_{\sB} \Tr_{\sA}(\dens{\gamma}) b],
	\end{align*}
	for all $b \in \algebra{\hilb_{\sB}}$.  The proof that $\dens{\gamma} \in \mathtt{com}(\dens{\rho}_{\sAB}) \Rightarrow \Tr_{\sX^{c}}[\dens{\gamma}] \in \mathtt{com}(\dens{\rho}_{\sX})$ follows similar reasoning.
\end{proof}

\subsection{Basic Properties of the Bipartite Complexes \label{sec:bipartite:properties}}
We first explore a few simple properties and isomorphisms of the complexes we have constructed.

\subsubsection{Descent to Support Equivalence Classes}
Our (co)chain complexes are defined directly in terms of the support projections of density states.
\begin{definition}{}{support_equiv_vanilla}
	Two density states $\dens{\rho}, \dens{\varphi} \in \Dens(\hilb)$ are said to be \newword{support equivalent} if $\supp_{\dens{\rho}} = \supp_{\dens{\varphi}}$.
\end{definition}
It is not hard to see that support equivalence is an honest equivalence relation on density states.  Moreover, the following shows that this equivalence relation is compatible with partial traces.

\begin{lemma}{}{support_equivalence_partial_trace}
	If $\dens{\rho}_{\sAB},\, \dens{\varphi}_{\sAB} \in \Dens(\hilb_{\sA} \otimes \hilb_{\sB})$ are support equivalent, then their reduced density matrices are support equivalent.
\end{lemma}
\begin{proof}
	Let $a \in \mathfrak{N}_{\sA}$, then by definition:
	\begin{align*}
		0 &= \Tr[\dens{\rho}_{\sA} a^{*} a] = \Tr[\dens{\rho}_{\sAB} a^{*} a \otimes 1_{\sB}]
	\end{align*}
	but the right hand side vanishes if and only if $a \otimes 1_{\sB} \in \mathfrak{N}_{\dens{\rho}_{\sAB}}= \algebra{\hilb}(1-\supp_{\dens{\rho}_{\sAB}}) = \algebra{\hilb}(1- \supp_{\dens{\varphi}_{\sAB}}) = \mathfrak{N}_{\dens{\varphi}_{\sAB}}$ (where we have used Prop.~\ref{prop:vanishing_ideal_characterizations}).  Using this we can show $\mathfrak{N}_{\dens{\rho}_{\sA}} = \mathfrak{N}_{\dens{\varphi}_{\sA}}$; it is then straightforward to verify that $\supp_{\dens{\rho}_{\sA}} = \supp_{\dens{\varphi}_{\sA}}$.  Repeating the argument we have $\supp_{\dens{\rho}_{\sB}} = \supp_{\dens{\varphi}_{\sB}}$.
\end{proof}

With Lem.~\ref{lem:support_equivalence_partial_trace} in mind, we introduce a definition of support equivalence for bipartite density states that live on the same Hilbert space decomposition.
\begin{definition}{}{}
	Two bipartite density states $\bdens{\rho}_{\sAB} = (\hilb_{\sA}, \hilb_{\sB}, \dens{\rho}_{\sAB})$ and $\bdens{\rho}_{\sAB}' = (\hilb_{\sA}, \hilb_{\sB}, \dens{\rho}_{\sAB}')$ are support equivalent if $\dens{\rho}_{\sAB}$ is support equivalent to $\dens{\rho}_{\sAB}'$.
\end{definition}

The following is immediate via definition and Lem.~\ref{lem:support_equivalence_partial_trace}.
\begin{proposition}{}{support_invariance_bipartite}
	The (co)chain complexes constructed above only depend on support equivalence classes: i.e.\ if $\bdens{\rho}_{\sAB}$ and $\bdens{\varphi}_{\sAB}$ are support equivalent then $\cpx{G}(\bdens{\rho}_{\sAB}) = \cpx{G}(\bdens{\varphi}_{\sAB})$ and $\cpx{E}(\bdens{\rho}_{\sAB}) = \cpx{E}(\bdens{\varphi}_{\sAB})$ (and similarly for the corresponding chain complexes).
\end{proposition}

As a result, our (co)chain complexes, and hence their (co)homologies, will not depend on some of the finer details of density states.  In particular, if we have a density state $\dens{\rho}_{\sAB}$ with spectral decomposition
\begin{align*}
	\dens{\rho}_{\sAB} = \sum_{i \in I} \lambda_{i} \mathbbm{P}_{i}
\end{align*}
for some eigenvalues $\{\lambda_{i}\}_{i \in I} \subseteq (0,1]_{>0}$ and projections $\{\mathbbm{P}_{i}\}_{i \in I} \in \algebra{\hilb_{\sA} \otimes \hilb_{\sB}}$, then because our cohomologies will only depend on the projection
\begin{align*}
	\supp_{\sAB} = \sum_{i \in I}  \mathbbm{P}_{i},
\end{align*}
we can vary the eigenvalues $\lambda_{i}$ inside of $(0,1]_{>0}$ without changing the associated chain complexes. Furthermore, for density states of rank $>1$, there are non-commuting density states with the same support projection; hence, there are support equivalent density states with incompatible spectral decompositions.
Generalizing the results here in a manner that depends on finer details of the spectral decomposition will require a more sophisticated approach than presented in this paper.
Alternatively, this ``bug" can turn, be considered a ``feature": namely our (co)homologies are invariants of \textit{families} of density states that vary within a given support equivalence class (a linear/convex subspace of the convex space of density states).

\begin{remark}{}{}
	Every subspace of a finite dimensional Hilbert space can be realized as the image of the support projection of some density state.  With this observation we have a bijective correspondence:
	
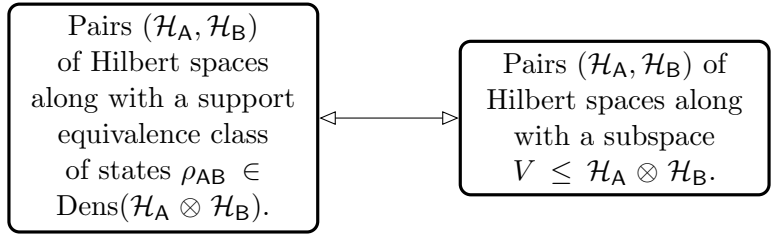
\begin{center}
		\begin{tikzpicture}
			\tikzstyle{block} = [rectangle, thick,
			text width=10em, text centered, rounded corners, minimum height=2em]

			\node at (-3,0)[block,draw=black,very thick] (supeq) {Pairs ($\hilb_{\sA}, \hilb_{\sB})$ of Hilbert spaces along with a support equivalence class of states $\rho_{\sAB} \in \Dens(\hilb_{\sA} \otimes \hilb_{\sB})$.};

			\node at (3,0)[block,draw=black,very thick] (subspaces) {Pairs $(\hilb_{\sA}, \hilb_{\sB})$ of Hilbert spaces along with a subspace $V \leq \hilb_{\sA} \otimes \hilb_{\sB}$.};

			\draw[<->, >=open triangle 45] (supeq) to (subspaces);
		\end{tikzpicture}
	\end{center}
	so by Prop.~\ref{prop:support_invariance_bipartite} we can think of our (co)chain complexes as associated to the data of a pair $(\hilb_{\sA}, \hilb_{\sB})$ along with a subspace $V \leq \hilb_{\sA} \otimes \hilb_{\sB}$.
	However, to emphasize the applicability to individual states in quantum mechanics (and with an eye toward generalizability) we continue to think of our (co)chain complexes as associated to density states.\footnote{In infinite dimensions this bijective correspondence holds true with subspaces $V \subseteq \hilb_{\sA} \otimes \hilb_{\sB}$ that are closed and of countable dimension (separable).  From a von Neumann algebraic perspective, such subspaces are in correspondence with (weakly closed) left ideals of $\algebra{\hilb_{\sA}} \otimes_{\mathrm{sp}} \algebra{\hilb_{\sB}} := \algebra{\hilb_{\sA} \otimes \hilb_{\sB}}$, where ``$\otimes_{\mathrm{sp}}$" is the spatial tensor product of von Neumann algebras.}
\end{remark}

\subsubsection{Trace duality for Chains and Cochains}
Recall that the dual of a chain complex (of $\mathbb{C}$-vector spaces):
\begin{align*}
	\cpx{c} = \cdots \overset{\partial_{-2}}{\longleftarrow}
	\cpx{c}_{-2} \overset{\partial_{-1}}{\longleftarrow} \cpx{c}_{-1}
	\overset{\partial_{0}}{\longleftarrow} \cpx{c}_{0}
	\overset{\partial_{1}}{\longleftarrow} \cpx{c}_{1}
	\overset{\partial_{2}}{\longleftarrow} \cpx{c}_{2}
	\overset{\partial_{3}}{\longleftarrow} \cdots
\end{align*}
is the cochain complex $\cpx{c}^{\vee}$ given by:
\begin{align*}
	\cpx{c}^{\vee} = \cdots \overset{\left(\partial_{-2}\right)^{\vee}}{\longrightarrow}
	(\cpx{c}_{-2})^{\vee} \overset{\left(\partial_{-1}\right)^{\vee}}{\longrightarrow} (\cpx{c}_{-1})^{\vee}
	\overset{\left(\partial_{0}\right)^{\vee}}{\longrightarrow} (\cpx{c}_{0})^{\vee}
	\overset{\left(\partial_{1}\right)^{\vee}}{\longrightarrow} (\cpx{c}_{1})^{\vee}
	\overset{\left(\partial_{2}\right)^{\vee}}{\longrightarrow} (\cpx{c}_{2})^{\vee}
	\overset{\left(\partial_{3}\right)^{\vee}}{\longrightarrow} \cdots
\end{align*}
i.e.\ the cochain complex with components $(\cpx{c}^{\vee})^{k} := (\cpx{c}_{k})^{\vee} = \mathrm{Hom}(\cpx{c}^{k}, \mathbb{C})$ and coboundaries 
\begin{align*}
	d^{k} := (\partial_{k+1})^{\vee}: (\cpx{c}^{\vee})^{k} &\longrightarrow (\cpx{c}^{\vee})^{k+1}\\
	\omega &\longmapsto \omega \circ \partial_{k+1}.
\end{align*}
At the end of \S\ref{sec:building_blocks} we showed that the building blocks for cochain complexes were isomorphic to the duals of chain complexes using restrictions of the map $(-)^{\Tr}$ defined by partial application of the trace pairing.  We can use these isomorphisms to define componentwise maps:
\begin{equation}
	\xymatrix{
		\cdots \ar[r] &
		0 \ar[r] \ar[d] & \mathbb{C} \ar[r] \ar[d]^{\sigma_{\cpx{G}}^{-1}=\mathrm{id}} & 
		\mathtt{GNS}(\dens{\rho}_{\sA}) \times \mathtt{GNS}(\dens{\rho}_{\sB}) \ar[r] \ar[d]^{\sigma_{\cpx{G}}^{0}= (-)_{\sA}^{\Tr} \times (-)_{\sB}^{\Tr} } & 
		\mathtt{GNS}(\dens{\rho}_{\sAB}) \ar[r] \ar[d]^{\sigma_{\cpx{G}}^{1}=(-)^{\Tr}_{\sAB}} &
		0 \ar[r] \ar[d] &
		\cdots\\
		\cdots \ar[r] &
		0 \ar[r] & 
		\mathbb{C} \ar[r]  & 
		\left[\mathtt{gns}(\dens{\rho}_{\sA}) \right]^{\vee} \times \left[\mathtt{gns}(\dens{\rho}_{\sB}) \right]^{\vee} \ar[r] & 
		\left[\mathtt{gns}(\dens{\rho}_{\sAB})\right]^{\vee} \ar[r]  &
		0 \ar[r] & \cdots
	}
\end{equation}
and
\begin{equation}
	\xymatrix{
		\cdots \ar[r] &
		0 \ar[r] \ar[d] & \mathbb{C} \ar[r] \ar[d]^{\sigma_{\cpx{E}}^{-1}=\mathrm{id}} & 
		\mathtt{Com}(\dens{\rho}_{\sA}) \times \mathtt{Com}(\dens{\rho}_{\sB}) \ar[r] \ar[d]^{\sigma_{\cpx{E}}^{0}= (-)_{\sA}^{\Tr} \times (-)_{\sB}^{\Tr} } & 
		\mathtt{Com}(\dens{\rho}_{\sAB}) \ar[r] \ar[d]^{\sigma_{\cpx{E}}^{1}=(-)^{\Tr}_{\sAB}} &
		0 \ar[r] \ar[d] &
		\cdots\\
		\cdots \ar[r] &
		0 \ar[r] & 
		\mathbb{C} \ar[r]  & 
		\left[\mathtt{com}(\dens{\rho}_{\sA}) \right]^{\vee} \times \left[\mathtt{com}(\dens{\rho}_{\sB}) \right]^{\vee} \ar[r] & 
		\left[\mathtt{com}(\dens{\rho}_{\sAB})\right]^{\vee} \ar[r]  &
		0 \ar[r] & \cdots
	}
\end{equation}
One can check that all squares above commute: i.e.\ our componentwise defined maps commute with coboundaries and so are honest chain maps (which is a fancy manifestation of the trace-pairing duality/adjunction between tensoring by the identity maps and partial traces); so we have the following.\footnote{In Prop.~\ref{prop:trace_duality_multipartite} one can find a sketch of a proof of a generalization of Prop.~\ref{prop:bipartite_cocomplex_to_dualcomplex} to the multipartite situation.}
\begin{proposition}{}{bipartite_cocomplex_to_dualcomplex}
	Componentwise application of restrictions of $(-)^{\Tr}$ defines a cochain isomorphisms.
	\begin{align*}
		\sigma_{\mathrm{G}}: \cpx{G}(\bdens{\rho}_{\sAB}) &\overset{\sim}{\longrightarrow} \left[\cpx{g}(\bdens{\rho}_{\sAB}) \right]^{\vee}\\
		\sigma_{\mathrm{E}}: \cpx{E}(\bdens{\rho}_{\sAB}) &\overset{\sim}{\longrightarrow} \left[\cpx{e}(\bdens{\rho}_{\sAB}) \right]^{\vee}.
	\end{align*}
\end{proposition}
The cohomology of a dual of a chain complex (of $\mathbb{C}$-vector spaces) is isomorphic to the dual of the homology of the chain complex: i.e.\ $H^{k}[\cpx{c}^{\vee}] \overset{\sim}{\longrightarrow} (H_{k}[\cpx{c}])^{\vee}$ via the (well-defined) map\footnote{If we were working with chain complexes $R$-modules for some ring $R$, this map would have a kernel given by $\mathrm{Ext}^{1}(\mathrm{c}_{k-1},R)$ This group vanishes if $R$ is a field (i.e.\ complexes of vector spaces over field).}
\begin{equation*}
	\begin{array}{lccl}
		h_{\cpx{c}}^{k}:& H^{k}[\cpx{c}^{\vee}] & \longrightarrow & \left(H_{k}[\cpx{c}]\right)^{\vee} \\
		{}&    [f: c \mapsto f(c)] & \longmapsto & ([c] \mapsto f(c))
	\end{array}
\end{equation*}
where the square brackets on the second line denote classes in (co)homology (e.g.\ $[c]$ denotes the homology class of the chain $c$).  Hence, we have isomorphisms. 
\begin{align*}
	h_{\cpx{G}}^{k} \circ (\sigma^{k}_{\cpx{G}})_{*}: H^{k}\left[\cpx{G}(\bdens{\rho})\right] &\overset{\sim}{\longrightarrow} H_{k}\left[\cpx{g}(\bdens{\rho})\right]^{\vee},\\
	h_{\cpx{E}}^{k} \circ (\sigma^{k}_{\cpx{E}})_{*}: H^{k}\left[\cpx{G}(\bdens{\rho})\right] &\overset{\sim}{\longrightarrow} H_{k}\left[\cpx{e}(\bdens{\rho})\right]^{\vee}.
\end{align*}
As a result of this isomorphism, we will specialize most of our results to the cochain/cohomological realm, with the understanding that the appropriate version is given by taking the appropriate duals (at least in finite dimensions).

\subsubsection{Equivariance Under Permutation of Tensor Factors}
The following is a straightforward exercise.
\begin{proposition}{}{}
	Let $\bdens{\rho}_{\sAB} = ((\hilb_{\sA}, \hilb_{\sB}), \dens{\rho}_{\sAB})$ be a bipartite density state; suppose $u: \hilb_{\sA} \otimes \hilb_{\sB} \rightarrow \hilb_{\sB} \otimes \hilb_{\sA}$ is the obvious reshuffling map and define the bipartite density state $\bdens{\rho}_{\sB \sA} := ((\hilb_{\sB}, \hilb_{\sA}), u \dens{\rho} u^{*})$.  Then there are (canonical) chain isomorphisms 
	\begin{align*}
		U_{\cpx{G}}: \cpx{G}(\bdens{\rho}_{\sAB}) &\overset{\sim}{\longrightarrow} \cpx{G}(\bdens{\rho}_{\sB \sA})\\
		U_{\cpx{E}}: \cpx{E}(\bdens{\rho}_{\sAB}) &\overset{\sim}{\longrightarrow} \cpx{E}(\bdens{\rho}_{\sB \sA}). 
	\end{align*}
\end{proposition}

\subsubsection{Equivariance under Local Unitary/Invertible Transformations}
If the cochain complexes defined so far encapsulate interesting information about the entanglement properties of bipartite states, then one should expect that they transform nicely under ``local" unitary transformations: i.e.\ unitary transformations on each tensor factor.\footnote{The adjective ``local" is being used synonymously here with ``tensor factor", however nothing about our setup implies that the tensor factors are associated to any specific underlying geometry.}  We formalize the definition of such transformations with the notation we have developed.
As we will see, it is helpful to give a generalize our definition to include arbitrary invertible transformations that might not be unitary.
\begin{definition}[label=def:local_invertible_bipartite]{}{}
	Let $\bdens{\rho}_{\sAB} := (\hilb_{\sA}, \hilb_{\sB}, \dens{\rho}_{\sAB})$ and $\bdens{\varphi}_{\sAB} = (\mathcal{K}_{\sA}, \mathcal{K}_{\sB}, \dens{\varphi}_{\sAB})$ be bipartite density states.  
	A \newword{local invertible transformation} $\boldsymbol{l}: \bdens{\rho_{\sAB}} \rightarrow \bdens{\varphi}_{\sAB}$ is a pair of invertible linear maps\footnote{In infinite dimensions we require these to be bounded.} $(l_{\sA}: \hilb_{\sA} \rightarrow \mathcal{K}_{\sB},\, l_{\sB}: \hilb_{\sA} \rightarrow \mathcal{K}_{\sA})$ such that $\dens{\varphi}_{\sAB} = l \dens{\rho}_{\sAB} l^{-1}$.
	When $l_{\sA}$ and $l_{\sB}$ are unitary ($l_{\sX}^{-1} = l_{\sX}^{*}$), we call the pair $(l_{\sA}, l_{\sB})$ a \newword{local unitary transformation}.
\end{definition}

\begin{lemma}{}{building_block_unitary_equivariance}
	Let $l: \hilb \rightarrow \mathcal{K}$ be an invertible linear transformation\footnote{Bounded for infinite dimensional Hilbert spaces.} and $\dens{\rho} \in \Dens(\hilb)$, then the appropriate restrictions of the map
	\begin{align*}
		\algebra{\hilb} &\longrightarrow \algebra{\mathcal{K}}\\
		r &\longmapsto l r l^{-1}
	\end{align*} 
	gives isomorphisms of vector spaces: 
	\begin{align*}
		\mathtt{GNS}(l): \mathtt{GNS}(\dens{\rho}) &\overset{\sim}{\longrightarrow} \mathtt{GNS}(l \dens{\rho} l^{-1}),\\
		\mathtt{Com}(l): \mathtt{Com}(\dens{\rho}) &\overset{\sim}{\longrightarrow} \mathtt{Com}(l \dens{\rho} l^{-1}).
	\end{align*}
\end{lemma}
\begin{proof}
	Note that $\supp_{l \dens{\rho} l^{-1}} = l \supp_{\dens{\rho}} l^{-1}$.  Hence, if $r \supp_{\dens{\rho}} = r$ we must have
	\begin{align*}
		l r l^{-1} \supp_{l \dens{\rho} l^{-1}} = l r l^{-1} l\supp_{\dens{\rho}} l^{-1} =  l r l^{-1}
	\end{align*} 
	thus, $r \in \mathtt{GNS}(\dens{\rho}) \Rightarrow l r l^{-1} \in \mathtt{GNS}(l \dens{\rho} l^{-1})$; by a similar argument, we have $r \in \mathtt{Com}(\dens{\rho}) \Rightarrow l r l^{-1} \in \mathtt{Com}(l \dens{\rho} l^{-1})$.  Showing these maps are isomorphisms is straightforward.
\end{proof}

\begin{theorem}{}{bipartite_invertible_equivariance}
	Let $\bdens{\rho}_{\sAB} := (\hilb_{\sA}, \hilb_{\sB}, \dens{\rho}_{\sAB})$ and $\bdens{\varphi}_{\sAB} = (\mathcal{K}_{\sA}, \mathcal{K}_{\sB}, \dens{\varphi}_{\sAB})$ be bipartite density states; suppose $\boldsymbol{l}: \bdens{\rho_{\sAB}} \rightarrow \bdens{\varphi}_{\sAB}$ is a local unitary transformation given by $(l_{\sA}: \hilb_{\sA} \rightarrow \mathcal{K}_{\sA}, l_{\sB}: \hilb_{\sB} \rightarrow \mathcal{K}_{\sB})$, then there are induced chain isomorphisms
	\begin{align*}
		\cpx{G}(\boldsymbol{l}): \cpx{G}(\bdens{\rho}_{\sAB}) &\overset{\sim}{\longrightarrow} \cpx{G}(\bdens{\varphi}_{\sAB})\\
		\cpx{E}(\boldsymbol{l}): \cpx{E}(\bdens{\rho}_{\sAB}) &\overset{\sim}{\longrightarrow} \cpx{E}(\bdens{\varphi}_{\sAB})
	\end{align*}
	given componentwise by
	\begin{equation*}
		\xymatrix{
			\cdots \ar[r] &
			0 \ar[r] \ar[d] & \mathbb{C} \ar[r] \ar[d]|-{\cpx{G}^{-1}(\boldsymbol{l})=\mathrm{id}} & 
			\mathtt{GNS}(\dens{\rho}_{\sA}) \times \mathtt{GNS}(\dens{\rho}_{\sB}) \ar[r] \ar[d]|{\cpx{G}^{0}(\boldsymbol{l})= \mathtt{GNS}(l_{\sA}) \times \mathtt{GNS}(l_{\sB})} & 
			\mathtt{GNS}(\dens{\rho}_{\sAB}) \ar[r] \ar[d]|-{\cpx{G}^{1}(\boldsymbol{l}) =\mathtt{GNS}(l_{\sA} \otimes l_{\sB})} &
			0 \ar[r] \ar[d] &
			\cdots\\
			\cdots \ar[r] &
			0 \ar[r] & 
			\mathbb{C} \ar[r]  & 
			\mathtt{GNS}(l_{\sA} \dens{\rho}_{\sA} l_{\sA}^{-1}) \times \mathtt{GNS}(l_{\sB} \dens{\rho}_{\sAB} l_{\sB}^{-1}) \ar[r] & 
			\mathtt{GNS}(l_{\sAB} \dens{\rho}_{\sAB} l_{\sAB}^{-1}) \ar[r]  &
			0 \ar[r] & \cdots
		}
	\end{equation*}
	and
	\begin{equation*}
		\xymatrix{
			\cdots \ar[r] &
			0 \ar[r] \ar[d] & \mathbb{C} \ar[r] \ar[d]|-{\cpx{E}^{-1}(\boldsymbol{l})=\mathrm{id}} & 
			\mathtt{Com}(\dens{\rho}_{\sA}) \times \mathtt{Com}(\dens{\rho}_{\sB}) \ar[r] \ar[d]|-{\cpx{E}^{0}(\boldsymbol{l})= \mathtt{Com}(l_{\sA}) \times \mathtt{Com}(l_{\sB})} & 
			\mathtt{Com}(\dens{\rho}_{\sAB}) \ar[r] \ar[d]|-{\cpx{E}^{1}(\boldsymbol{l}) = \mathtt{Com}(l_{\sA} \otimes l_{\sB})} &
			0 \ar[r] \ar[d] &
			\cdots\\
			\cdots \ar[r] &
			0 \ar[r] & 
			\mathbb{C} \ar[r]  & 
			\mathtt{Com}(l_{\sA} \dens{\rho}_{\sA} l_{\sA}) \times \mathtt{Com}(l_{\sB} \dens{\rho}_{\sAB} l_{\sB}) \ar[r] & 
			\mathtt{Com}(l_{\sAB} \dens{\rho}_{\sAB} l_{\sAB}) \ar[r]  &
			0 \ar[r] & \cdots
		}
	\end{equation*}
\end{theorem}

More generally, one can define a local completely positive map between density states as a pair of completely positive maps that plays nicely with support projections; such maps induce chain morphisms.\footnote{More generally, completely positive maps that admit a decomposition into Kraus operators that factorize (known as \textit{separable operations} \cite{Bengtsson}), and play nicely with support projections, will induce chain morphisms.}
However, because we will not explicitly need such technology we leave the appropriate definitions and statements (which should make an appearance in future work) as an exercise for the interested reader.

Because every cochain map induces a map in cohomology we have the following corollary. 
\begin{corollary}{}{bipartite_invertible_polynomial_invariance}
	The $\cpx{G}$ and $\cpx{E}$ cohomology groups of bipartite density states related by a local invertible transformation are canonically isomorphic.  In particular, the Poincar\'{e} polynomials
	\begin{align*}
		\newmath{P_{\cpx{G}}(\bdens{\rho}_{\sAB})} &:= \dim H^{0} [ \cpx{G}(\bdens{\rho}_{\sAB})] + \dim H^{1} [ \cpx{G}(\bdens{\rho}_{\sAB})]y \in \mathbb{Z}[y],\\
		\newmath{P_{\cpx{E}}(\bdens{\rho}_{\sAB})} &:= \dim H^{0} [ \cpx{E}(\bdens{\rho_{\sAB}})] +  \dim H^{1} [ \cpx{E}(\bdens{\rho_{\sAB}})]y \in \mathbb{Z}[y] 
	\end{align*}
	are invariant under local invertible transformations of $\bdens{\rho}_{\sAB}$.
\end{corollary}

\subsubsection{Poincar\'{e} Polynomials as SLOCC Invariants}
Because local unitary transformations are a special case of local invertible transformations, the discussion above establishes equivariance of our cohomology groups (Poincar\'{e} polynomials) under local unitary transformations.
In particular, the corresponding Poincar\'{e} polynomials only depend on local unitary equivalence classes of density states (two density states being local unitary equivalent if there exists a local unitary transformation between them).
However, because Thm.~\ref{thm:bipartite_invertible_equivariance} and Cor.~\ref{cor:bipartite_invertible_polynomial_invariance} actually pertain to invertible linear transformations, they only depend equivalence classes of density states defined by local invertible transformations.
In the context of pure density states, such equivalence classes can be interpreted as equivalence classes of (pure) states that are (stochastically) related by a special class of quantum operations dubbed LOCC (local operations and classical communication) (see, e.g.\ \cite[\S 16.4]{Bengtsson}); in this context, local invertible equivalence classes of pure density states are known as \textit{SLOCC equivalence classes} and Cor.~\ref{cor:bipartite_invertible_polynomial_invariance} shows that the Poincar\'{e} polynomials $P_{\cpx{G}}$ and $P_{\cpx{E}}$ are \newword{SLOCC invariants} for pure bipartite density states (i.e.\ they only depend on SLOCC equivalence classes).\footnote{Note that these polynomials might not be invariant polynomials associated to the standard algebraic structure on the space of SLOCC equivalence classes, thought of as quotient of a projective space.}
Although we will not make use of such an operational perspective, we will adopt this terminology for the remainder of this section if only for the reason that ``SLOCC equivalence class" is a more succinct phrase than ``local invertible equivalence class".

Restricting our attention to pure bipartite density states, it is easy to see that the Schmidt rank is also a invariant under local invertible transformations; moreover, two states with the same Schmidt rank are equivalent under local invertible transformations.
Hence, the Schmidt rank is a \textit{complete} local SLOCC invariant: its value determines a unique equivalence class.
Thus, there are $\mathrm{min}(d_{\sA}, d_{\sB})$ SLOCC equivalence classes of pure bipartite density states, each uniquely labelled by the corresponding Schmidt rank.  
As a result, we should only see at most $\min(d_{\sA}, d_{\sB})$ possible Poincar\'{e} polynomials associated to pure bipartite density states on with Hilbert space tensor factors of dimensions $d_{\sA}$ and $d_{\sB}$.
As we will see (c.f. Rmk.~\ref{rmk:bipartite_poincare}), these Poincar\'{e} polynomials are explicit functions of the Schmidt rank, and are complete SLOCC invariants: there are precisely $\min(d_{\sA}, d_{\sB})$ possible polynomials, each uniquely labelling an SLOCC equivalence class.

This might not seem particularly impressive: the Schmidt rank, which is easily computable, already forms a very good SLOCC invariant.
The computation of Poincar\'{e} polynomials for bipartite pure states, offers no further advantage.
However, multipartite SLOCC/local-invertible equivalence classes have a much more complicated structure.
In the generalization of our cohomological technology to multipartite density states, the Poincar\'{e} polynomials associated to the pure multipartite density states are SLOCC invariants and it is not unreasonable to speculate that they might form a useful tool for distinguishing SLOCC equivalence classes.
As proof of concept of such an application, in \S\ref{sec:tripartite_cohomology_distinguish} (and \S\ref{sec:multipartite_cohomology_distinguish}) we distinguish the SLOCC equivalence classes of tripartite (and multipartite) GHZ and W-states via their commutant Poincar\'{e} polynomials.

\subsubsection{Comparison of the Bipartite GNS and Commutant (Co)homologies}
There are natural chain morphisms between the GNS and commutant complexes.  For any $\dens{\rho} \in \Dens(\hilb)$ we automatically have an inclusion of building blocks
\begin{align}
	\newmath{i_{\dens{\rho}}}: \mathtt{com}(\dens{\rho}) \hooklongrightarrow \mathtt{gns}(\dens{\rho}).
	\label{eq:com_inc_gns}
\end{align}
Taking the dual map $i_{\dens{\rho}}^{\vee}$ and utilizing the canonical identifications $\mathtt{com}(\dens{\rho})^{\vee} \cong \mathtt{Com}(\dens{\rho})$ and $\mathtt{gns}(\dens{\rho}) \cong \mathtt{GNS}(\dens{\rho})$ (c.f.\ equations \eqref{eq:Com_dual_iso} and \eqref{eq:GNS_dual_iso}) we have the surjective linear map:
\begin{align}
	\newmath{\Pi_{\dens{\rho}}}: \mathtt{GNS}(\dens{\rho}) &\longrightarrow \mathtt{Com}(\dens{\rho})\\
	r &\longmapsto \supp_{\dens{\rho}} r.	
	\label{eq:GNS_prj_Com}	
\end{align}
If we work with a bipartite state $\dens{\rho}_{\sAB} \in \Dens(\hilb_{\sA} \otimes \hilb_{\sB})$ and define $i_{\sX} := i_{\dens{\rho}_{\sX}}$ and $\Pi_{\sX} := \Pi_{\dens{\rho}_{\sX}}$ for $\sX \in \{\sA, \sB, \sAB \}$, then we then have componentwise maps
\begin{equation}
	\xymatrix{
		\cdots &
		\ar[l] 0 & \ar[l] \mathbb{C}  & 
		\ar[l] \mathtt{gns}(\dens{\rho}_{\sA}) \oplus \mathtt{gns}(\dens{\rho}_{\sB}) & 
		\ar[l] \mathtt{gns}(\dens{\rho}_{\sAB}) &
		\ar[l] 0  &
		\ar[l] \cdots\\
		\cdots &
		\ar[l] \ar[u] 0 & \ar[l] \ar[u]^{\newmath{i_{-1}} := \mathrm{id}} \mathbb{C}  & 
		\ar[l] \ar[u]^{\newmath{i_{0}} := i_{\sA} \oplus i_{\sB}} \mathtt{com}(\dens{\rho}_{\sA}) \oplus \mathtt{com}(\dens{\rho}_{\sB}) & 
		\ar[l] \ar[u]^{\newmath{i_{1}} := i_{\sAB}} \mathtt{com}(\dens{\rho}_{\sAB}) &
		\ar[l] \ar[u] 0  &
		\ar[l] \cdots\\
	}
\end{equation}
and
\begin{equation}
	\xymatrix{
		\cdots \ar[r] &
		0 \ar[r] \ar[d] & \mathbb{C} \ar[r] \ar[d]^{\newmath{\Pi^{-1}} :=\mathrm{id}} & 
		\mathtt{GNS}(\dens{\rho}_{\sA}) \times \mathtt{GNS}(\dens{\rho}_{\sB}) \ar[r] \ar[d]^{\newmath{\Pi^{0}} := \Pi_{\sA} \times \Pi_{\sB}} & 
		\mathtt{GNS}(\dens{\rho}_{\sAB}) \ar[r] \ar[d]^{\newmath{\Pi^{1}} :=\Pi_{\sAB}} &
		0 \ar[r] \ar[d] &
		\cdots\\
		\cdots \ar[r] &
		0 \ar[r] & 
		\mathbb{C} \ar[r]  & 
		\mathtt{Com}(\dens{\rho}_{\sA}) \times \mathtt{Com}(\dens{\rho}_{\sB}) \ar[r] & 
		\mathtt{Com}(\dens{\rho}_{\sAB}) \ar[r]  &
		0 \ar[r] & \cdots
	}
\end{equation}
(with $i_{k}$ and $P_{k}$ being defined as the zero map for $k \neq -1,0,1$).
\begin{proposition}{}{GNS_to_Com}
	The maps above assemble into a (co)chain morphisms
	\begin{align*}
		i: \cpx{e}(\dens{\rho}_{\sAB}) &\longrightarrow \cpx{g}(\dens{\rho}_{\sAB}),\\
		\Pi: \cpx{G}(\dens{\rho}_{\sAB}) &\longrightarrow \cpx{E}(\dens{\rho}_{\sAB}).
	\end{align*}
	where, moreover $i$ (resp. $P$) is componentwise injective (resp.\ surjective).
\end{proposition}
\begin{proof}
	We need only show commutativity with the coboundary.  This follows by working through the definitions of the coboundaries along with an application of Lem.~\ref{lem:compat_of_supports_predual} (for $i$) or Lem.~\ref{lem:compat_of_supports} (for $\Pi$).
\end{proof}
The kernel of a chain morphism $M: \cpx{C} \rightarrow \cpx{D}$ is a cochain complex $\ker(M)$ formed by taking componentwise kernels and restrictions of coboundaries of $\cpx{C}$.  Let $\ker(\Pi)$ denote the kernel of the cochain morphism $P$ defined above. Then, by componentwise surjectivity of $P$, we have a short exact sequences of cochain complexes:
\begin{align*}
	0 \longrightarrow \ker(\Pi) \overset{\iota}{\longrightarrow} \cpx{G}(\dens{\rho}_{\sAB}) \overset{\Pi}{\longrightarrow} \cpx{E}(\dens{\rho}_{\sAB}) \longrightarrow 0 
\end{align*}
where $\iota$ is the obvious inclusion.  We can think of $\ker(\Pi)$ as a complex built\footnote{In fact, one can also build a cochain complex using the ideals $\mathfrak{N}_{\dens{\rho}}$ as building blocks.} from the building block $\ker(\Pi_{\dens{\rho}}) = (1-\supp_{\dens{\rho}}) \algebra{\hilb} \supp_{\dens{\rho}}$; using the block matrix decomposition of \S \ref{sec:building_blocks}:
\begin{align*}
	\ker(\Pi_{\dens{\rho}}) = \left\{\blockmat{0}{0}{*}{0} \right \} \cong \Hom(\image(\dens{\rho}),\ker(\dens{\rho})).
\end{align*}
The induced long exact sequence in cohomology is:
\begin{equation}
	\begin{tikzpicture}[descr/.style={fill=white,inner sep=1.5pt},baseline={([yshift=-.5ex]current bounding box.center)}]
		\matrix (m) [
		matrix of math nodes,
		row sep=2em,
		column sep=2.5em,
		text height=1.5ex, text depth=0.25ex
		]
		{ 0 & H^{0}[\ker(\Pi)] & H^{0}[\cpx{G}(\bdens{\rho}_{\sAB})] & H^{0}[\cpx{E}(\bdens{\rho}_{\sAB})] & {} \\
			& H^{1}[\ker(\Pi)] & H^{1}[\cpx{G}(\bdens{\rho}_{\sAB})]  & H^{1}[\cpx{E}(\bdens{\rho}_{\sAB})] & 0, \\
		};

		\path[overlay,->, font=\scriptsize]
			(m-1-1) edge (m-1-2)
			(m-1-2) edge node[above] {$\iota^{0}_{*}$} (m-1-3)
			(m-1-3) edge node[above] {$P^{0}_{*}$} (m-1-4)
			(m-1-4) edge[out=355,in=175] node[descr,yshift=0.3ex] {$\delta^0$} (m-2-2)
			(m-2-2) edge node[above] {$\iota^{1}_{*}$} (m-2-3)
			(m-2-3) edge node[above] {$P^{1}_{*}$} (m-2-4)
			(m-2-4) edge (m-2-5);
	\end{tikzpicture}
	\label{eq:com_les_cohomology}
\end{equation}
\noindent where $\delta^{0}$ is the Bockstein homomorphism.\footnote{As an easy corollary of this long exact sequence, we have that $\dim H^{1}[\cpx{E}(\bdens{\rho}_{\sAB})] \leq \dim H^{1}[\cpx{G}(\bdens{\rho}_{\sAB})]$.}

Similarly we can define the cokernel of any map of complexes; because $i$ is just realizing $\cpx{e}$ as a subcomplex of $\cpx{g}$ then $\coker(i)$ is the quotient complex $\cpx{g}(\dens{\rho}_{\sAB})/ \cpx{e}(\dens{\rho}_{\sAB})$ and we end up with a short exact sequence:
\begin{align*}
	0 \longleftarrow  \cpx{g}(\dens{\rho}_{\sAB})/ \cpx{e}(\dens{\rho}_{\sAB}) \overset{\pi}{\longleftarrow} \cpx{g}(\dens{\rho}_{\sAB}) \overset{i}{\longleftarrow} \cpx{e}(\dens{\rho}_{\sAB}) \longleftarrow 0, 
\end{align*}
(where $\pi$ is the componentwise projection map) and an associated long exact sequence
\begin{center}
	\begin{tikzpicture}[descr/.style={fill=white,inner sep=1.5pt}]
		\matrix (m) [
		matrix of math nodes,
		row sep=2em,
		column sep=2.5em,
		text height=1.5ex, text depth=0.25ex
		]
		{ 0 & H_{0}[ \cpx{e}(\dens{\rho}_{\sAB})] & H_{0}[\cpx{g}(\bdens{\rho}_{\sAB})] & H_{0}[\cpx{e}(\bdens{\rho}_{\sAB})] & {} \\
			& H_{1}[ \cpx{e}(\dens{\rho}_{\sAB})] & H_{1}[\cpx{g}(\bdens{\rho}_{\sAB})]  & H_{1}[\cpx{e}(\bdens{\rho}_{\sAB})] & 0. \\
		};

		\path[overlay,<-, font=\scriptsize]
			(m-1-1) edge (m-1-2)
			(m-1-2) edge node[above] {$(\pi_{0})_{*}$} (m-1-3)
			(m-1-3) edge node[above] {$(i_{0})_{*}$} (m-1-4)
			(m-1-4) edge[out=355,in=175] node[descr,yshift=0.3ex] {$\delta_0$} (m-2-2)
			(m-2-2) edge node[above] {$(\pi_{1})_{*}$} (m-2-3)
			(m-2-3) edge node[above] {$(i_{1})_{*}$} (m-2-4)
			(m-2-4) edge (m-2-5);
	\end{tikzpicture}
\end{center}
By taking duals we recover the story in cohomology.

\begin{remark}{}{}
	The fact that one gets (co)chain maps is deeply related the fact that the process of taking (the underlying vector space of) equivariant endomorphisms of a module (for a $\mathbb{C}$-algebra) with distinguished cyclic point is functorial.  On the other hand, the surjective map $\mathtt{gns}(\dens{\rho}) \longrightarrow \mathtt{com}(\dens{\rho})$ given by $\dens{\gamma} \mapsto \dens{\gamma} \supp_{\dens{\rho}}$ and its dual inclusion $\mathtt{Com}(\dens{\rho}) \hookrightarrow \mathtt{GNS}(\dens{\rho})$ do not produce (co)chain maps.
\end{remark}

\subsection{Factorizability and Cohomology}

Next we would like to understand the extent that the (co)homologies of our (co)chain complexes detect factorizability of a density state.
For pure density states, the notion of factorizability relevant to our purposes is the standard one.
However, for mixed density states the situation is more complex, and the notion of factorizability relevant to our (co)homological approach will be a weaker notion we call \textit{support factorizability}.  In the next section we introduce some relevant definitions of factorizability and its generalizations. 

\subsubsection{Factorizabilia \label{sec:factorizabilia}}
We begin with the most straightforward definition of factorizability.
\begin{definition}{}{}
	A bipartite density state $\bdens{\rho}_{\sAB}$ is \newword{factorizable} if $\dens{\rho}_{\sAB} = \dens{\rho}_{\sA} \otimes \dens{\rho}_{\sB}$.
\end{definition}
This definition of factorizability is equivalent to the factorizability of the expectation value functional $\mathbb{E}_{\dens{\rho}_{\sAB}}$, i.e.\ $\mathbb{E}_{\dens{\rho}_{\sAB}}(a \otimes b) = \mathbb{E}_{\dens{\rho}_{\sA}}(a) \mathbb{E}_{\dens{\rho}_{\sB}}(b)$ for all $(a,b) \in \algebra{\hilb_{\sA}} \times \algebra{\hilb_{\sB}}$ (see also Lem.~\ref{lem:vanishing_cov_factorizability}).  Moreover, it generalizes the natural notion of factorizability for vectors in a Hilbert space.
\begin{proposition}{}{}
	If $\dens{\rho}_{\sAB} = \psi \otimes \psi^{\vee}$ for some $\psi \in \hilb_{\sA} \otimes \hilb_{\sB}$, then $\bdens{\rho}_{\sAB}$ is factorizable if and only if $\psi = \psi_{\sA} \otimes \psi_{\sB}$ for some $\psi_{\sX} \in \hilb_{\sX}$ (i.e.\ $\psi$ is factorizable).
\end{proposition}
Quantum informatics students learn early that factorizability of a pure state can be accomplished by a computation of partial traces.
\begin{proposition}{}{}
	A unit norm vector $\psi \in \hilb_{\sA} \otimes \hilb_{\sB}$ can be written as $\psi  = \psi_{\sA} \otimes \psi_{\sB}$ if and only if $\rank \left(\Tr_{\sX}[\psi \otimes \psi^{\vee}] \right) = 1$, where $\sX$ is one of $\sA$ or $\sB$. 
\end{proposition}

Because our (co)chain complexes can only detect properties of density states up to support equivalence, it is more appropriate to focus our attention on the following variant of factorizability (which we will see is a strictly weaker than factorizability for mixed states). 
\begin{definition}{}{}
	A bipartite density state $\bdens{\rho}_{\sAB}$ is \newword{support factorizable} if $\dens{\rho}_{\sAB}$ is support equivalent to $\dens{\rho}_{\sA} \otimes \dens{\rho}_{\sB}$ (equivalently, $\supp_{\sAB} = \supp_{\sA} \otimes \supp_{\sB}$).
\end{definition}

Similar to factorizability for pure states, we can easily check for support factorizability by computing ranks of partial traces.\footnote{Although, if our state is mixed, it is not guaranteed that the ranks of partial traces over disjoint subsystems are equal; so really we must compute partial traces over \textit{both} subsystems if our state is not pure.}
\begin{proposition}{}{}
	If $\bdens{\rho}_{\sAB}$ is a (finite rank) bipartite density state, then $\rank(\dens{\rho}_{\sAB}) \leq \rank(\dens{\rho}_{\sA}) \rank(\dens{\rho}_{\sB})$ with equality if and only if $\dens{\rho}_{\sAB}$ is support factorizable.\footnote{This also holds for finite rank density states in infinite dimensions.}  
\end{proposition}
\begin{proof}
This is an immediate consequence of the compatibility of supports lemma (Lem.~\ref{lem:compat_of_supports}).
\end{proof}

The following is immediate. 
\begin{proposition}{}{fact_then_supp_fact}
	If a bipartite density state is factorizable, then it is support factorizable.
\end{proposition}

Restricting our attention to pure states, the converse of Prop.~\ref{prop:fact_then_supp_fact} holds (we leave the proof of the following as a simple exercise).
\begin{proposition}{}{}
	Suppose $\bdens{\rho}_{\sAB}$ is a bipartite density state such that $\rank(\dens{\rho}_{\sAB}) = 1$ (i.e.\ $\dens{\rho}_{\sAB}$ is a pure) then $\bdens{\rho}_{\sAB}$ is factorizable if and only if it is support factorizable.
\end{proposition}
However, the converse of Prop.~\ref{prop:fact_then_supp_fact} does not hold in general for mixed factorizable states as the following counter-example demonstrates.
\begin{example}{Support Factorizable but not Factorizable}{supp_fact_not_fact}
	Let $\hilb_{\sA}$ and $\hilb_{\sB}$ be single-qubit Hilbert spaces: two-dimensional Hilbert spaces equipped with (orthonormal) computational bases $\{\ket{0_{\sX}}, \ket{1_{\sX}} \} \subseteq \hilb_{\sX}$.  Consider the three parameter family of ``diagonal'' bipartite states $(\hilb_{\sA}, \hilb_{\sB}, \rho_{\sAB})$ given by: 
	\begin{align*}
		\dens{\rho}_{\sAB} &= \lambda_{1} \ket{0_{\sA}0_{\sB}}\bra{0_{\sA}0_{\sB}} + \lambda_{2} \ket{0_{\sA}1_{\sB}}\bra{0_{\sA}1_{\sB}} + \lambda_{3} \ket{1_{\sA}0_{\sB}}\bra{1_{\sA}0_{\sB}} +\\
						   &+ (1-\lambda_{1} - \lambda_{2} - \lambda_{3}) \ket{1_{\sA}1_{\sB}}\bra{1_{\sA}1_{\sB}}
	\end{align*}
	where $\{\lambda_{i}\}_{i = 1}^{3} \in [0,1]$.  It is a straightforward exercise to show that this is factorizable if and only if $\left(\lambda_{1} + \lambda_{2} \right) \left(\lambda_{1} \lambda_{3} + 1 \right) = 0$.
	Thus, a generic choice of $\{\lambda_{i}\}_{i = 1}^{3} \subseteq [0,1]$ will provide a non-factorizable state that is support equivalent to a factorizable state.  (For concreteness we can choose e.g.\ $\lambda_{1} = \lambda_{2} = 1/3$ and $\lambda_{3} = 1/4$).
\end{example}

More generally, we might ask if there are any mixed bipartite density states that are support factorizable but not \textit{separable}.
Recall a density state is \newword{separable} if it can be written as a convex linear combination of factorizable density states.
The failure to be separable is usually taken to be the generalization of entanglement to mixed states.\footnote{The action of local operations and classical communication (LOCC) on a factorizable density state produces a separable density states; conversely, every separable density state can be created by acting on a factorizable density state with LOCC\@.
This characterizes entangled states as precisely those that cannot be generated from factorizable states via LOCC \cite{Werner}.}
The density states in the counter-example above are all clearly separable from the way we have presented them. 
There are support factorizable states that not only fail to be factorizable, but also fail to be separable: i.e.\ cannot be written as a convex linear combination of factorizable states.
\begin{example}{Support Factorizable but Entangled (Not Separable)}{supp_fact_but_entangled}
	Let $\hilb_{\sA}$ and $\hilb_{\sB}$ be single-qubit Hilbert spaces as in Ex~\ref{ex:supp_fact_not_fact}.
	Consider the one-parameter family of bipartite density states $(\hilb_{\sA}, \hilb_{\sB}, \rho_{\sAB})$ given by: 	
	\begin{align*}
		\rho_{\sAB} &= \lambda \left(\ket{00}+\ket{11} \right)\left(\bra{00}+\bra{11} \right) + \frac{\lambda}{4} \mathbbm{1}
	\end{align*}
	where $\mathbbm{1}=\mathbbm{1}_{\sA} \otimes \mathbbm{1}_{\sB}$ is the identity operator on $\hilb_{\sA} \otimes \hilb_{\sB}$, and $\lambda \in (0,1)$.
	These are commonly known as ``Werner states".
	When $\lambda > 1/3$ this state fails to satisfy the Peres-Horodecki/PPT (positive partial transpose) criterion and so must be non-separable.
	On the other hand, it is clear that $\supp_{\sAB} = \supp_{\sA} \otimes \supp_{\sB}$ for any $\lambda \in (0,1)$.
\end{example}

So, not surprisingly, support factorizability is a rather weak property. Nevertheless, it is a natural property to ask for when studying families of bipartite density states.

\subsubsection{Pure States \label{sec:bipartite_factorizability_pure}}
As a first step toward understanding what cohomology can tell us about factorizability, we provide the following easy theorem.
\begin{theorem}{}{pure_commutant_cohomology}
	Let $\bdens{\rho}_{\sAB}$ be a bipartite density state with $\dens{\rho}_{\sAB} = \psi \otimes \psi^{\vee}$ for some $\psi \in \hilb_{\sA} \otimes \hilb_{\sB}$, then
	\begin{equation*}
		\dim H^{k}[\cpx{E}(\bdens{\rho}_{\sAB})] =
		\left \{
			\begin{array}{ll} 
				2(S^2 -1) &, \text{if $k=0$} \\
				0 &, \text{otherwise}
			\end{array}
		\right.
	\end{equation*}
	where $S$ is the Schmidt rank of $\psi$ (equivalently given as $\rank \left( \Tr_{\sX}[\psi \otimes \psi^{\vee}] \right)$ for $\sX \in \{\sA, \sB \})$.  In particular $\psi$ is factorizable if and only if $H^{k}[\mathrm{E}(\bdens{\rho}_{\sAB})] = 0$ for all $k$.
\end{theorem}
\begin{proof}
	By working through the definitions it is easy to see that the degree one cohomology vanishes.
	Hence, the dimension of the zeroth cohomology group must equal to the Euler characteristic which is easily computed to be $2S^2 - 2$.
\end{proof}
In fact, it is straightforward to explicitly compute (co)homology groups: let $\dens{\rho}_{\sAB}$ be a pure state as in Thm.~\ref{thm:pure_commutant_cohomology} then:
\begin{align*}
	H_{1}[\cpx{e}(\bdens{\rho}_{\sAB})] &= 0,\\
	H_{0}[\cpx{e}(\bdens{\rho}_{\sAB})] &= \{(\dens{\gamma}_{\sA}, \dens{\gamma}_{\sB}) \in \mathtt{com}(\dens{\rho}_{\sA}) \times \mathtt{com}(\dens{\rho}_{\sB}): \Tr[\dens{\gamma}_{\sA}] = -\Tr[\dens{\gamma}_{\sB}] \}/\mathrm{span}_{\mathbb{C}}\{(-\dens{\rho}_{\sA}, \dens{\rho}_{\sB}) \}.
\end{align*}
Via the trace pairing duality, then:
\begin{align*}
	H^{1}[\cpx{E}(\bdens{\rho}_{\sAB})] &= 0,\\
	H^{0}[\cpx{E}(\bdens{\rho}_{\sAB})] &= \{(a,b) \in \mathtt{Com}(\dens{\rho}_{\sA}) \times \mathtt{Com}(\dens{\rho}_{\sB}): \Tr[\dens{\rho}_{\sA} a] = \Tr[\dens{\rho}_{\sB} b] = 0 \}/\mathrm{span}_{\mathbb{C}}\{(\supp_{\sA}, \supp_{\sB}) \},
\end{align*}
and it follows that the degree 0 (co)homology groups have dimension $2S^2 -2$.

The following proposition shows that the zeroth cohomology of $\cpx{G}(\bdens{\rho}_{\sAB})$ also detects factorizability of pure states.
\begin{theorem}{}{pure_GNS_cohomology}
	Let $\bdens{\rho}_{\sAB}$ be a bipartite density state with $\dens{\rho}_{\sAB} = \psi \otimes \psi^{\vee}$ for some $\psi \in \hilb_{\sA} \otimes \hilb_{\sB}$, then $H^{0}[\cpx{G}(\bdens{\rho}_{\sAB})] = 0$ if and only if $\psi$ is factorizable.
\end{theorem}
\begin{proof}
	If $\psi$ is factorizable then $\supp_{\sX}$ for $\sX \in \{\sA, \sB, \sAB \}$ is a rank 1 projection.  Suppose $H^{0}[\cpx{G}] = 0$, then by the long exact sequence \eqref{eq:com_les_cohomology} we have $H^{0}[\ker(\Pi)] = 0$; as a result, the Euler characteristic $\chi = \chi[\ker(\Pi)]$ of $\ker(\Pi)$ must be $-\dim H^{1}[\ker(\Pi)]$. 
	It is a straightforward exercise to calculate
	\begin{align*}
		\chi &= (1- r_{\sA})^2 + (1-r_{\sB})^2 - (1-r_{\sAB})^2
	\end{align*}
	where $r_{\sX} := \rank(\dens{\rho}_{\sX})$. Because $H^{0}[\ker(\Pi)] = 0 $, then we must have
	\begin{align*}
		\dim H^{1}[\ker(\Pi)] = -\chi = -(1- r_{\sA})^2 - (1-r_{\sB})^2 + (1-r_{\sAB})^2.
	\end{align*}
	Because $\dens{\rho}_{\sAB}$ is pure, then $r_{\sAB} = 1$.  Further, a dimension must be $\geq 0$; so we must have $r_{\sA} = r_{\sB} = 1$; hence, $\psi$ is factorizable.
\end{proof}

Theorem~\ref{thm:pure_GNS_cohomology} can also be considered as a corollary of Thm.~\ref{thm:pure_bipartite_cohomology} below, where we construct non-trivial cohomology classes and provide the dimension of $H^{0}[\cpx{G}(\bdens{\rho}_{\sAB})]$ in terms of the Schmidt rank $S$.
However, it is instructive to supply a quicker construction of non-trivial cohomology classes (providing another proof of Thm.~\ref{thm:pure_GNS_cohomology}).  
To this end, suppose $x$ is a self-adjoint operator on a finite dimensional Hilbert space, then $x$ has the spectral decomposition 
\begin{align*}
	x &= \sum_{\lambda \in \sigma_{\neq 0}(x)} \lambda \mathbbm{P}_{\lambda} 
\end{align*}
where $\sigma_{\neq 0}(x)$ is the non-vanishing set of eigenvalues of $x$, and $\mathbbm{P}_{\lambda}$ is the projection onto the eigenspace associated to eigenvalue $\lambda \in \sigma_{\neq 0}(x)$. 
Then for any function $f: \mathbb{R} \rightarrow \mathbb{C}$ we can define a new operator\footnote{Note that the right hand side is a well-defined operator for any function of underlying sets $f: \mathbb{R} \rightarrow \mathbb{C}$. This means, we require no, e.g.\ continuity properties.
	(In fact, if we can work with functions only be defined on $\sigma_{\neq 0}(x)$, and extend arbitrarily, if we wish.)
	In infinite dimensions one can perform similar procedures on self-adjoint operators with respect to a restricted class of functions (e.g.\ the continuous functional calculus, holomorphic functional calculus, and Borel functional calculus).
	The proper generalization of our discussion here relies on taking $f$ to be Borel measurable and applying the Borel functional calculus.}
\begin{align*}
	f(x) &= \sum_{\lambda \in \sigma_{\neq 0}} f(\lambda) \mathbbm{P}_{\lambda}.
\end{align*}

The following lemma allows us to provide alternative proofs of Theorems \ref{thm:pure_commutant_cohomology} and \ref{thm:pure_GNS_cohomology}.
\begin{lemma}{}{explicit_cocycles}
	Let $\bdens{\rho}_{\sAB}$ be a pure bipartite density state and $f: \mathbb{R} \rightarrow \mathbb{C}$ a function\footnote{In infinite dimensions we use the Borel function calculus and require $f$ to be Borel-measurable and essentially bounded.}, then $(f(\dens{\rho}_{\sA}), f(\dens{\rho}_{\sB})) \in \algebra{\hilb_{\sA}} \times \algebra{\hilb_{\sB}}$ defines a degree 0 cocycle of both $\cpx{E}(\bdens{\rho}_{\sAB})$ and $\cpx{G}(\bdens{\rho}_{\sAB})$.
\end{lemma}
\begin{proof}
	Because $\bdens{\rho}_{\sAB}$ is pure, then $\dens{\rho}_{\sA} = \psi \otimes \psi^{\vee}$ for some unit vector $\psi \in \hilb_{\sA} \otimes \hilb_{\sB}$ (unique up to multiplication by an element of the unit circle).
	It is clear that $f(\bdens{\rho}_{\sX}) = \supp_{\sX} f(\bdens{\rho}_{\sAB}) \supp_{\sX}$ so $f(\bdens{\rho}_{\sX})$ is an element of both $\mathtt{Com}(\bdens{\rho}_{\sX})$ and $\mathtt{GNS}(\dens{\rho}_{\sX})$.
	Note that $(\supp_{\sA} \otimes f(\dens{\rho}_{\sB}) - f(\dens{\rho}_{\sA}) \otimes s_{\sA}) \supp_{\sAB} = 0$ if and only if
	\begin{align*}
		0 &= (\supp_{\sA} \otimes f(\dens{\rho}_{\sB}) - f(\dens{\rho}_{\sA}) \otimes s_{\sA}) \psi.
	\end{align*}
	With this in hand, the result follows easily by choosing a Schmidt decomposition of $\psi$ and expressing $f(\dens{\rho}_{\sA})$ and $f(\dens{\rho}_{\sB})$ as weighted sums of projections onto the orthonormal vectors in this Schmidt decomposition. 
\end{proof}

As a special family of cases: note that, by choosing $f$ to be characteristic functions with sufficiently small support, we can select out pairs of projection operators.
When $\dens{\rho}_{\sAB}$ is pure, its reduced states have the same spectrum, and these pairs of projection operators are those pairs of projections onto eigenspaces that are associated to a given eigenvalue. 
For a generic pure bipartite state with Schmidt rank $S$ (where ``generic" means that the Schmidt coefficients have multiplicity at most 1), this provides $(S-1)$-independent elements of $H^{0}[\cpx{E}(\bdens{\rho}_{\sAB})]$ and $H^{0}[\cpx{G}(\bdens{\rho}_{\sAB})]$.
Hence, when the state is generic, we have just provided an alternative proof to both Thm.~\ref{thm:pure_commutant_cohomology} and Thm.~\ref{thm:pure_GNS_cohomology}.
To extend to the case when the state is \textit{not} generic we use Thm.~\ref{thm:bipartite_invertible_equivariance}: suppose $\dens{\rho}_{\sAB} = \psi \otimes \psi^{\vee}$ and the Schmidt decomposition of $\psi$ has Schmidt coefficients with multiplicity $\geq 1$; then we can always find linear transformations $l_{\sX}: \hilb_{\sX} \rightarrow \hilb_{\sX},\, \sX \in \{\sA, \sB \}$, such that $\left(l_{\sA} \otimes l_{\sB} \right) \psi$ is a unit norm vector with generic Schmidt decomposition (note that this requires $l_{\sA}$ and $l_{\sB}$ to be non-unitary), by an application of Thm.~\ref{thm:bipartite_invertible_equivariance}.

As yet another way to prove Theorems \ref{thm:pure_commutant_cohomology} and \ref{thm:pure_GNS_cohomology} using Lem.~\ref{lem:explicit_cocycles}, we can choose $f: \mathbb{R} \rightarrow \mathbb{C}$ to be the function $\log_{+}$ given by $\log_{+}(\lambda) = \log(\lambda)$ for $\lambda > 0$ and $0$ otherwise; the result is that $(\log_{+}(\dens{\rho}_{\sA}), \log_{+}(\dens{\rho}_{\sB}))$ defines a cocycle.\footnote{Our choice of modifying the logarithm to be finite at zero is somewhat immaterial: no matter how we wish to extend the notion of logarithm $\log(\dens{\rho})$ as an operator on $\image(\dens{\rho}) \oplus \ker(\dens{\rho})$, the extension results in the same element $\log(\dens{\rho}) \supp_{\dens{\rho}} \in \mathtt{GNS}(\dens{\rho})$ and $\supp_{\dens{\rho}} \log(\dens{\rho}) \supp_{\dens{\rho}} \in \mathtt{Com}(\dens{\rho})$.}
It is straightforward to show that, outside of the highly non-generic case that that all Schmidt coefficients are equal to $1/S$ (where $S$ is the Schmidt rank), this cocycle defines a non-trivial cohomology class if $S > 1$.
Once again, one can treat the generic situation by an application of Thm.~\ref{thm:bipartite_invertible_equivariance}.

\begin{remark}{}{}
	The operator $-\log_{+}(\dens{\rho})$ is related to the \textit{modular Hamiltonian}, alluding to its role in the modular flow of Tomita-Takesaki.\footnote{The flow here is on the GNS representation associated to $\dens{\rho}$; the modular flow acts on the GNS representation (given by $\hilb \otimes \image(\dens{\rho})^{\vee}$ for finite rank $\dens{\rho})$) by right multiplication by the 1-parameter of unitary elements $t \mapsto \dens{\rho}^{it}$ for real $t$; it commutes with left multiplication by elements of $\algebra{\hilb}$.}
		Traditionally the modular flow is only defined for full rank states; however, there is nothing wrong with defining it for non-full rank states.
		In this sense, the modular Hamiltonians of reduced density states of a pure bipartite state always form a cocycle, which is (generically) represents a non-trivial cohomology class of either GNS or commutant cohomology whenever the state is entangled.
\end{remark}

\subsubsection{Mixed States \label{sec:factorizability_mixed}}
We turn our attention to generic, possibly mixed (non-pure) states.
In this more general setup, we have the following weaker version of Thms.~\ref{thm:pure_GNS_cohomology} and \ref{thm:pure_commutant_cohomology}.
\begin{theorem}{}{support_fact_then_vanishing}
	If $\bdens{\rho}_{\sAB}$ is support factorizable then $H^{0}[\cpx{E}(\bdens{\rho}_{\sAB})] = H^{0}[\cpx{G}(\bdens{\rho}_{\sAB})] = 0$.
\end{theorem}
\begin{proof}
	The proof follows from the observation that $a \otimes \supp_{\sB} - \supp_{\sA} \otimes b = 0$ for $(a,b) \in \mathtt{GNS}(\dens{\rho}_{\sA}) \times \mathtt{GNS}(\dens{\rho}_{\sB})$ if and only if $a = \supp_{\sA}$ and $b = \supp_{\sB}$.
\end{proof}
So, even for mixed states, cohomologies can be used to detect factorizability: the non-vanishing of either GNS or commutant cohomologies contradicts support factorizability, and hence factorizability. 
However, the converse to Thm.~\ref{thm:support_fact_then_vanishing} is not true for a general mixed state: without purity, vanishing of cohomologies does not imply factorizability.
The first example below provides a counter-example.
The second example shows that there exist separable states (recall that a density state is \textit{separable} if it can be written as a convex linear combination of factorizable states) whose cohomologies are non-vanishing.\footnote{One can use the included software to calculate explicit cohomologies of the examples.}
\begin{example}{}{mixed_state_counterexamples}
	In the following we suppress the indices $\sA$ and $\sB$ and $\{\ket{0}, \ket{1} \}$ denotes an orthonormal basis for a qubit Hilbert space (i.e.\ a two-dimensional Hilbert space).
	\begin{enumerate}
		\item \textit{$\supp_{\sAB}$ is a strict subprojection of $\supp_{\sA} \otimes \supp_{\sB}$ (i.e.\ not support factorizable) but $H^{0}[\cpx{G}(\bdens{\rho}_{\sAB})] = H^{0}[\cpx{E}(\bdens{\rho}_{\sAB})] = 0$}:
			\begin{enumerate}
				\item $\dens{\rho}_{\sAB} = \lambda \ket{00} \bra{00} + (1-\lambda) \ket{01}\bra{01}$ for $\lambda \in (0,1)$;

				\item $\dens{\rho}_{\sAB} = \lambda \ket{00}\bra{00} + \eta \ket{01}\bra{01} + (1- \lambda - \eta) \ket{10} \bra{10}$ for $\lambda, \eta \in (0,1]$.
			\end{enumerate}

		\item \textit{Separable but $H^{0}[\cpx{G}(\bdens{\rho}_{\sAB})],\, H^{0}[\cpx{E}(\bdens{\rho}_{\sAB})] \neq 0$}:
			\begin{align*}
				\dens{\rho}_{\sAB} &= \lambda \ket{00}\bra{00} + (1-\lambda) \ket{11}\bra{11} 
			\end{align*}
			for $\lambda >0$.
			The associated Poincar\'{e} polynomials are $P_{\cpx{G}}(\bdens{\rho}_{\sAB}) = 1 + 2y$ and $P_{\cpx{E}}(\bdens{\rho}_{\sAB}) = 5 + 2y$.
	\end{enumerate}
\end{example}
In \S\ref{sec:cor_and_cohom} we present an interpretation of $H^{0}[\cpx{G}(\bdens{\rho}_{\sAB})]$ in terms of pairs of correlated operators.
The fact that there are separable states with non-vanishing cohomology is related to the fact that (non-factorizable) separable states can still be associated with correlations between local operators; because separable states can be constructed from pure factorizable states by local operations and classical communication \cite{Werner}, one might consider correlation between operators with respect to a separable state as ``classical".

\subsection{Correlation and Cohomology \label{sec:cor_and_cohom}}
Next we explore the relationship between the cohomology of $\cpx{G}(\bdens{\rho}_{\sAB})$ and pairs of operators with ``non-local" correlations.

\subsubsection{EPR Pairs \label{sec:EPR_pairs}}
Given a self adjoint operator $r \in \algebra{\hilb}$ (a.k.a. an \textit{observable}), and a density state $\dens{\rho} \in \Dens(\hilb)$, the Born-rule is a way of assigning a probability measure on the spectrum of $r$.
This probability measure is thought of as the outcome of a ``projective measurement" of $r$: an experiment whose outcome is an element of the spectrum of $r$; and modifying the density state associated to the system to be compatible with this outcome.
Given two non-commuting operators $q,r \in \algebra{\hilb}$; the Born rule derived probability measure associated to the outcome of a projective measurement of $q$, then $r$ is different than the Born-rule derived probability measure for the projective measurements in the reversed order.
However, for commuting operators, the Born-rule derived probability distributions are independent of order (an one can even speak of a simultaneous measurement).
With this in mind, we can unambiguously speak of the result of a projective measurement of a pair of operators $(q,r) \in \algebra{\hilb} \times \algebra{\hilb}$, the outcome of such a measurement is a pair of real numbers $(\lambda, \eta)$ where $\lambda$ is in the spectrum of $q$, $\eta$ is in the spectrum of $r$, and such outcomes are distributed according to the probability measure derived through the Born-rule.

Now given a bipartite density state $\bdens{\rho}_{\sAB} = (\hilb_{\sA}, \hilb_{\sB}, \dens{\rho}_{\sAB})$, any operators of the form $a \otimes 1_{\sB} \in \algebra{\hilb_{\sAB}}$ commute with operators of the form $1_{\sA} \otimes b \in \algebra{\hilb_{\sB}}$ and we can speak of the result of a projective measurement of the pair $(a \otimes 1_{\sB}, 1_{\sA} \otimes b)$; for simplicity we will simply say a projective measurement of $(a,b)$.
If we interpret the tensor factors $\sA$ and $\sB$ as spatially disjoint ``local" subsystems, then a measurement of $a \otimes 1_{\sB}$ can be thought of as a measurement of $a$ by an observer sitting at the subsystem $\sA$ and a measurement of $1_{\sA} \otimes b$ as a measurement of $b$ by an observer sitting at $\sB$.
If $\dens{\rho}_{\sAB}$ is a pure entangled (non-factorizable) state, then there are pairs of operators $(a,b)$ where the outcome of a projective measurement of $a$ is sufficient to determine the outcome of a projective measurement of $b$ and vice-versa.

To see this we begin with a typical modern version of the Einstein-Podolsky-Rosen gedankenexperiment one begins with a bipartite system given by qubit Hilbert spaces $\hilb_{\sA}$ and $\hilb_{\sB}$ (equipped with computational bases $\{ \ket{0_{\sX}}, \ket{1_{\sX}} \}$) along with the pure bipartite state given by $\dens{\rho}_{\sAB} = \psi_{\mathrm{bell}} \otimes \psi_{\mathrm{bell}}^{\vee}$ where
\begin{align*}
	\psi_{\mathrm{Bell}} := \frac{1}{\sqrt{2}} \left( \ket{0_{\sA}0_{\sB}} + \ket{1_{\sA}1_{\sB}} \right).
\end{align*}
Define $p_{\sX} := \ket{0_{\sX}} \bra{0_{\sX}}$ (i.e.\ the projection operator onto the span of $\ket{0_{\sX}}$) for $\sX \in \{\sA, \sB \}$.
Then a projective measurement of the pair of operators $(p_{\sA}, p_{\sB})$ in the presence of the state $\dens{\rho}_{\sAB}$ results in either the pair $(0,0)$ (with probability $1/2$) or $(1,1)$ (with probability $1/2$).
If observers at $\sA$ and $\sB$ know the initial state on the system $\sAB$ before measurement, the value of the measurement of $p_{\sA}$ by an observer at subsystem $\sA$ automatically and instantaneously knows the result of the measurement of $p_{\sB}$ at subsystem $\sB$ (even if the two are separated by large distance).
This sort of phenomenon is characteristic of pure entangled states.

\begin{remark}{}{}
	There is a similar, purely classical, phenomenon that also exists for joint probability measures.
	For instance, let $X$ be a finite set and $\mu: X \times X \rightarrow [0,1]$ a point-wise defined probability measure concentrated along the diagonal.
	Random variables are simply taken to be $\mathbb{C}$-valued functions on $X$.
	The entanglement phenomenon of quantum mechanics, however, is distinguished by the fact that there are EPR pairs that do not commute: i.e.\ there are pairs $(a_1, b_1)$ and $(a_2, b_{2})$ such that $[a_1,b_1] \neq 0$ and $[b_1,b_2] \neq 0$.
	In contrast, all classical observables/random variables are commutative.
	See, e.g.\ \cite{bell} for an exposition of the differences between classical correlations and quantum correlations due to entanglement.
	When generalized and restricted to the purely commutative/classical case, the cohomology theory defined in this paper actually takes into account classical EPR pairs; however, as a reflection of the fact that entanglement is a quantum mechanical phenomenon, the dimensions of cohomology groups of, e.g.\ the Bell state, are different than the dimensions of cohomology groups associated to a non-factorizable joint probability measure on the set $\{0,1 \}^{\times 2}$.
	In particular, the cohomology groups of the Bell state contain non-commuting EPR pairs in the sense described above.
\end{remark}

Using the example above as motivation, we define the notion of an EPR pair in the presence of an arbitrary bipartite density state.
\begin{definition}{}{EPR_pair}
	Let $\bdens{\rho}_{\sAB} = (\hilb_{\sA}, \hilb_{\sB}, \dens{\rho}_{\sAB})$ be a bipartite density state.
	A pair of self-adjoint operators $(a,b) \in \algebra{\hilb_{\sA}} \times \algebra{\hilb_{\sB}}$ is an \newword{EPR Pair} (associated to $\bdens{\rho}_{\sAB}$) if the result of any projective measurement of $(a,b)$ lies on the diagonal of $\mathbb{R} \times \mathbb{R}$. 
\end{definition}
I.e.\ suppose that $(a,b) \in \ker(d^{0}_{\msf{G}})$ with $a$ and $b$ self-adjoint, if an observer associated to the tensor factor/subsystem $\sA$ performs a projective measurement of the variable $a$, obtaining the eigenvalue $\lambda$, and an observer associated to $\sB$ performs a simultaneous projective measurement of $b$, obtaining the eigenvalue $\eta$, then we are guaranteed that $\lambda = \eta$.
Requiring the result of any projective measurement of $a$ to \textit{equal} the projective measurement of $b$ might appear too strict: more generally we would like the result of a projective measurement of $a$ to simply determine the value of the projective measurement of $b$ (encompassing observables that are, perhaps correlated in some non-linear manner).
However, in Appendix~\ref{app:projective_indistinguishability} we explain how any pair $(a,b)$ of self-adjoint operators such that the value of a projective measurement of $a$ determines the value of a projective measurement of $b$ can be used to construct a new pair $(\twid{a}, b)$ that is an EPR pair by our above definition.

Using the following theorem, we can extract EPR pairs of $\bdens{\rho}_{\sAB}$ (up to shifts by the ``trivial" pair of observables $(\supp_{\sA}, \supp_{\sB})$) from $H^{0}[\cpx{G}(\bdens{\rho}_{\sAB})]$. 
\begin{theorem}{}{projective_measurement_indistinguishability}
	Let $(a,b) \in \mathtt{GNS}(\dens{\rho}_{\sA}) \times \mathtt{GNS}(\dens{\rho}_{\sB})$ be self-adjoint operators.\footnote{In particular we must have $a \in \supp_{\sA} \algebra{\hilb_{\sA}} \supp_{\sA}$ and $b \in \supp_{\sB} \algebra{\hilb_{\sB}} \supp_{\sB}$.}  Then $(a,b) \in \ker(d^{0}_{\cpx{G}})$ if and only if $(a,b)$ is an EPR pair.
\end{theorem}
\begin{proof}
	See Appendix~\ref{app:projective_indistinguishability}.
\end{proof}
Note that all of the explicit cocycles in Lemma.~\ref{lem:explicit_cocycles} are pairs of self-adjoint operators; hence, this Lemma supplies us with a way of generating EPR pairs.\footnote{Of course, this can be also proven directly from the definition of an EPR pair, without any intermediate reference to cocycles.}
\begin{proposition}{}{}
	Let $\bdens{\rho}_{\sAB}$ be a pure bipartite density state, then $(f(\bdens{\rho}_{\sA}), f(\bdens{\rho}_{\sB}))$ is an EPR pair for any function\footnote{Borel measurable in infinite dimensions} $f:\mathbb{R} \rightarrow \mathbb{R}$.
\end{proposition}
More generally, one can generate EPR pairs of positive operators from an arbitrary (possibly non-self adjoint) pair in $\ker(d^{0}_{\cpx{G}})$ using the following observation.
\begin{proposition}{}{}
	Let $(a,b) \in \ker(d^{0}_{\cpx{G}})$, then $(a^{*} a, b^{*} b) \in \ker(d^{0}_{\cpx{G}})$.
\end{proposition}
\begin{proof}
	This follows from the identity $(a^{*} a \otimes 1_{\sB} - 1_{\sA} \otimes b^{*}b)\supp_{\sAB} = (a \otimes 1_{\sB} + 1_{\sA} \otimes b)^{*} (a \otimes 1_{\sB} - 1_{\sA} \otimes b) \supp_{\sAB}$ for any $(a,b) \in \algebra{\hilb_{\sA}} \times \algebra{\hilb_{\sB}}$.
\end{proof}

\subsubsection{Covariance and Cohomology}
Not all elements of $\ker(d_{\cpx{G}}^{0})$ are self-adjoint.\footnote{Or even normal (commute with their adjoint): a more relaxed condition for the notion of a spectral decomposition.}
Nevertheless, we can give a meaning to ``maximally correlated" to any pair of operators by thinking of saturation of their covariance (a linear measure of correlation).
For any two operators $x,y \in \algebra{\hilb_{\sA} \otimes \hilb_{\sB}}$ we define their covariance as $\Tr[\dens{\rho}_{\sAB}x^{*} y] - \Tr[\dens{\rho}_{\sAB}x^{*}] \Tr[\dens{\rho}_{\sAB}y]$; the variance is defined by specializing to $x =y$.
Restricting our attention to operators of the form $x = a \otimes 1_{\sB}$ and $y = 1_{\sA} \otimes b$, we present the following definitions.
\begin{definition}[label=def:cov_var]{}{}
	Suppose $\bdens{\rho}_{\sAB}$ is a bipartite density state; for any $r \in \algebra{\hilb_{\sX}}$ define $\newmath{r'} := r - \Tr[\dens{\rho}_{\sX} r]$, then define the sesquilinear map
	\begin{align*}
		\newmath{\Cov}:  \algebra{\hilb_{\sA}}  \times  \algebra{\hilb_{\sB}} &\longrightarrow \mathbb{C},
	\end{align*}
	and the quadratic forms
	\begin{align*}
		\newmath{\Var_{\sA}}: \algebra{\hilb_{\sA}} &\longrightarrow \mathbb{R}_{\geq 0},\\
		\newmath{\Var_{\sB}}: \algebra{\hilb_{\sB}} &\longrightarrow \mathbb{R}_{\geq 0},	
	\end{align*}
	via:
	\begin{align*}
		\Cov: (a,b) &\longmapsto \Tr[\dens{\rho}_{\sAB}(a')^{*} b' ],\\
		\Var_{\sX}: r &\longmapsto \Tr[\dens{\rho}_{\sX} (r')^{*} r'].
	\end{align*}
\end{definition}
The following is easy to show using non-degeneracy of the trace pairing between operators and states along with the observation that $\Cov(a,b) = \Tr[(\dens{\rho}_{\sAB} - \dens{\rho}_{\sA} \otimes \dens{\rho}_{\sB})(a^{*} \otimes b) ]$.
\begin{lemma}{}{vanishing_cov_factorizability}
	$\Cov \equiv 0$ if and only if $\dens{\rho}_{\sAB} = \dens{\rho}_{\sA} \otimes \dens{\rho}_{\sB}$.
\end{lemma}
In this sense, any pair of operators $(a,b)$ such that $\Cov(a,b) \neq 0$ can be thought of as an obstruction to factorizability. We claim that the cohomology group $H^{0}[\cpx{G}(\bdens{\rho}_{\sAB})]$ outputs (right essential equivalence classes of) such obstructing pairs where $\Cov$ is maximized in a precise way. 
Just as in commutative probability theory, the (absolute) square of the covariance is bounded above by the product of variances.
\begin{lemma}{}{covariance_saturation}
	For any $(a,b) \in \algebra{\hilb_{\sA}} \times \algebra{\hilb_{\sB}}$ we have
	\begin{align}
		\left |\Cov(a,b)  \right|^{2} \leq \Var_{\sA}(a) \Var_{\sB}(b)
		\label{eq:covariance_cauchy_schwarz}
	\end{align}
	with saturation if and only if $a' \otimes 1_{\sB} = \lambda (1_{\sA} \otimes b') + z$ for some $\lambda \in \mathbb{C}$ and $z \in \mathfrak{N}_{\sAB}$ (i.e.\ $a' \otimes 1_{\sB}$ and $\lambda (1_{A} \otimes b')$ are right essentially equivalent).
\end{lemma}
\begin{proof}
	The lemma follows by application of the Cauchy-Schwarz inequality (for possibly degenerate sesquilinear forms).
\end{proof}

\begin{remark}{}{}
	Notice that $\Cov, \,\Var_{\sA}, \Var_{\sB}$ are invariant under shifts by elements of the left ideals $\mathfrak{N}_{\sA}$ and $\mathfrak{N}_{\sB}$ defined in Def.~\ref{def:right_essential_equivalence}; so all three of these functions descend to well-defined functions on right essential equivalence classes; because the $\algebra{\hilb}$-module $\algebra{\hilb}/\mathfrak{N}_{\dens{\rho}}$ of right essential equivalence classes is canonically isomorphic to $\mathtt{GNS}(\dens{\rho})$ (as described in \S\ref{sec:interpretation}), these descents are canonically identifiable with the restrictions of $\Cov,\, \Var_{\sA},$ and $\Var_{\sB}$ to the subspaces $\mathtt{GNS}(\dens{\rho}_{\sX})$.
\end{remark}

The following theorem relates the zeroth cohomology group of $\mathrm{G}$ to pairs of operators that saturate the bound above. 
\begin{theorem}{}{kernel_covariance} 
	\begin{center}
		$\ker(d^{0}_{\cpx{G}}) =
		\left \{(a,b) \in \mathtt{GNS}(\dens{\rho}_{\sA}) \times \mathtt{GNS}(\dens{\rho}_{\sB}):
			\text{\parbox{15em}{
					\centering $\Cov(a,b) = \Var_{\sA}(a) = \Var_{\sB}(b)$ \\
					and \\ 
				$\Tr[\dens{\rho}_{\sA}a] = \Tr[\dens{\rho}_{\sB}b]$}
			}
		\right \}$.
	\end{center}
\end{theorem}
\begin{proof}
	Suppose $(a,b) \in \ker(d^{0}_{\cpx{G}})$; then $a \otimes 1_{\sB} - 1_{\sA} \otimes b \in \mathfrak{N}_{\sAB}$ (i.e.\ $a \otimes 1_{\sB}$ and $1_{\sA} \otimes b$ are right essentially equivalent); it follows that
	\begin{align*}
		\Tr[\dens{\rho}_{\sA} a] =  \Tr[\dens{\rho}_{\sAB} (1_{\sAB})^{*} (a \otimes 1_{\sB})] =  \Tr[\dens{\rho}_{\sAB} (1_{\sAB})^{*} (b \otimes 1_{\sA})] = \Tr[\dens{\rho}_{\sB} b]
	\end{align*}
	with this observation it follows that $a \otimes 1_{\sB} - 1_{\sA} \otimes b = a' \otimes 1_{\sB} - 1_{\sA} \otimes b'$; hence,
	\begin{align*}
		0 &= (a \otimes 1_{\sB} - 1_{\sA} \otimes b) \supp_{\sAB} = \left(a' \otimes 1_{\sB} - 1_{\sA} \otimes b' \right)\supp_{\sAB};
	\end{align*} 
	equivalently, $a' \otimes 1_{\sB} - 1_{\sA} \otimes b' \in \mathfrak{N}_{\sAB}$.  Thus, by Lemma~\ref{lem:covariance_saturation} we have $|\Cov(a,b)|^{2} = \Var_{\sA}(a) \Var_{\sB}(b) \Rightarrow \Cov(a,b) = \lambda \sqrt{\Var_{\sA}(a) \Var_{B}(b)}$ for some $\lambda \in \mathbb{C}$ with $|\lambda | = 1$ (and we take the positive square root of $\Var_{\sA}(a) \Var_{\sB}(b) \geq 0$). Moreover,
	\begin{align}
		\Tr \left[\dens{\rho}_{\sAB}(a' \otimes 1_{B} - 1_{\sA} \otimes b')^{*} (a' \otimes 1_{B} - 1_{\sA} \otimes b') \right] = \Var_{\sA}(a) + \Var_{\sB}(b) - 2 \Cov(a,b).
		\label{eq:cov_sat_prime_null}
	\end{align}
	The left-hand-side of the above must vanish as $a' \otimes 1_{B} - 1_{\sA} \otimes b' \in \mathfrak{N}_{\sAB}$.  Because $\Var_{\sA}(a)$ and $\Var_{B}(b)$ are non-negative, the only way the right-hand-side can vanish is if $\lambda =1$; it follows that $\Cov(a,b) = \Var_{\sA}(a) = \Var_{B}(b)$.
	Thus, we have shown
	\begin{align*}
		\ker(d^{0}_{\cpx{G}}) \subseteq  \left \{(a,b) \in \mathtt{GNS}(\dens{\rho}_{\sA}) \times \mathtt{GNS}(\dens{\rho}_{\sA}):
		\text{$\Cov(a,b) = \Var_{\sA}(a) = \Var_{\sB}(b)$ and $\Tr[\dens{\rho}_{\sA}a] = \Tr[\dens{\rho}_{\sB}b]$} \right \}.
	\end{align*}
	For the reverse inclusion, suppose $(a,b) \in \mathtt{GNS}(\dens{\rho}_{\sA}) \times \mathtt{GNS}(\dens{\rho}_{\sB})$ is such that $\Cov(a,b) = \Var_{\sA}(a) = \Var_{\sB}(b)$ and $\Tr[\dens{\rho}_{\sA} a] = \Tr[\dens{\rho}_{\sB}b]$.
	By Lem.~\ref{lem:covariance_saturation} we have $a' \otimes 1_{\sB} - 1_{\sA} \otimes b' \in \mathfrak{N}_{\sAB}$.  But $a \otimes 1_{\sB} - 1_{\sA} \otimes b = a' \otimes 1_{\sB} - 1_{\sA} \otimes b'$ as $\Tr[\dens{\rho}_{\sA} a] = \Tr[\dens{\rho}_{\sB} b]$; it follows that $a \otimes 1_{\sB} - 1_{\sA} \otimes b = 0 \Rightarrow (a,b) \in \ker(d^{0}_{\cpx{G}})$.  
\end{proof}

Note that, the only elements with vanishing variances in $\mathtt{GNS}(\dens{\rho}_{\sX})$ are elements of the line $\mathrm{span}_{\mathbb{C}}\{\supp_{\sX}\}$; as a result 
\begin{align*}
	\{(a,b): \Cov(a,b) = 0 \}  \cap \ker(d^{0}_{\cpx{G}}) &= \mathrm{span}_{\mathbb{C}}\{(\supp_{\sA}, \supp_{\sB})\},
\end{align*} 
this is precisely what we quotient by to get $H^{0}[\cpx{G}(\bdens{\rho}_{\sAB})]$.  We summarize this in the following corollary.
\begin{corollary}[label=cor:GNS_cohom_saturated_cov]{}{}
	\begin{align*}
		H^{0}[\cpx{G}(\bdens{\rho}_{\sAB})] &= \left \{
			(a,b) \in \mathtt{GNS}(\dens{\rho}_{\sA}) \times \mathtt{GNS}(\dens{\rho}_{\sB}) \colon 
			\text{\parbox{8.2em}{
					\centering \small $\Cov(a,b) = \Var_{\sA}(a),$ \\
					$\Cov(a,b) = \Var_{\sB}(b),$\\
					$\Tr[\dens{\rho}_{\sA}a] = \Tr[\dens{\rho}_{\sB}b]$}
			}
		\right \} 
		\Bigg/ \text{\small $\{(a,b): \Cov(a,b)=0 \}$ }.
	\end{align*}
\end{corollary}
Alternatively, we can observe that $\Cov,\,\Var_{\sA}$, and $\Var_{\sB}$ descend to a well-defined functions
\begin{align*}
	\underline{\Cov}: \left(\mathtt{GNS}(\dens{\rho}_{\sA})/\mathrm{span}_{\mathbb{C}} \{\supp_{\sA} \} \right) \times  \left(\mathtt{GNS}(\dens{\rho}_{\sB})/\mathrm{span}_{\mathbb{C}} \{\supp_{\sB} \}\right) &\longrightarrow \mathbb{C}, \\
	\underline{\Var}_{\sA}: \mathtt{GNS}(\dens{\rho}_{\sA})/\mathrm{span}_{\mathbb{C}} \{\supp_{\sA} \} \longrightarrow \mathbb{C},\\
	\underline{\Var}_{\sB}: \mathtt{GNS}(\dens{\rho}_{\sB})/\mathrm{span}_{\mathbb{C}}\{\supp_{\sB} \} \longrightarrow \mathbb{C}.
\end{align*}
Then, we have
\begin{align*}
	H^{0}[\cpx{G}(\bdens{\rho}_{\sAB})] &\cong \left \{([a],[b]): \underline{\Cov}([a],[b]) = \underline{\Var}_{\sA}([a]) = \underline{\Var}_{\sB}([b]) >0 \right \},   
\end{align*}
where $[r]$ is denoting the class $r + \mathrm{span}_{\mathbb{C}}\{\supp_{\sX} \} \in \mathtt{GNS}(\dens{\rho}_{\sX})/\mathrm{span}_{\mathbb{C}}\{\supp_{\sX} \}$.  Moreover, using the well-defined map 
\begin{align*}
	\left(\mathtt{GNS}(\dens{\rho}_{\sX})/\mathrm{span}_{\mathbb{C}} \{\supp_{\sA} \} \right) &\overset{\sim}{\longrightarrow} \{r \in \mathtt{GNS}(\dens{\rho}_{\sX}): \Tr[\dens{\rho}_{\sX} r] = 0 \}\\
	[r] &\longmapsto r - \Tr[\dens{\rho}_{\sX} r] \supp_{\sX}
\end{align*}
we have a canonical isomorphism from $H^{0}[\cpx{G}(\bdens{\rho}_{\sAB})]$ to an explicit subspace of $\mathtt{GNS}(\dens{\rho}_{\sA}) \times \mathtt{GNS}(\dens{\rho}_{\sB}) \leq \algebra{\hilb_{\sA}} \times \algebra{\hilb_{\sB}}$ consisting of pairs of operators with zero expectation values.

\subsection{Bipartite Pure States and Schmidt Decompositions \label{sec:pure_state_schmidt}}
Using the discussion of the previous section, we can explicitly calculate the zeroth cohomology group of $\cpx{G}(\bdens{\rho}_{\sAB})$, when $\dens{\rho}_{\sAB}$ is pure.
\begin{theorem}[label=thm:pure_bipartite_cohomology]{}{}
	Let $\bdens{\rho}_{\sAB}$ be a pure bipartite density state with $\dens{\rho}_{\sAB} = \psi \otimes \psi^{\vee}$ for some $\psi \in \hilb_{\sA} \otimes \hilb_{\sB}$. Decompose $\psi$ as:
	\begin{align*}
		\psi = \sum_{i = 1}^{S} \sqrt{p_{i}} \xi^{\sA}_{i} \otimes \xi^{\sB}_{i}
	\end{align*} 
	for positive coefficients $\{p_{i} \}_{i = 1}^{S} \subseteq \mathbb{R}_{>0}$ and orthonormal vectors $\{\xi^{\sX}_{i}\}_{i = 1}^{S} \subset \hilb_{\sX}$---i.e.\ a Schmidt decomposition of $\psi$---then 
	\begin{align*}
		H^{0}[\cpx{G}(\bdens{\rho}_{\sAB})] = \operatorname{span}_{\mathbb{C}} \left[(\mathbbm{e}^{\sA}_{ij},\mathbbm{e}^{\sB}_{ij}): i,j = 1,\cdots, S \right]/\operatorname{span}_{\mathbb{C}} \{(\supp_{\sA}, \supp_{\sB}) \}
	\end{align*}
	where 
	\begin{align*}
		\mathbbm{e}^{\sX}_{ij} := \xi^{\sX}_{i} \otimes \left(\xi^{\sX}_{j} \right)^{\vee} \in \mathtt{GNS}(\dens{\rho}_{\sX})
	\end{align*}
	for $\sX \in \{\sA, \sB \}$. In particular, $\dim H^{0}[\cpx{G}(\bdens{\rho}_{\sAB})] = S^2 - 1$, where $S$ is the Schmidt rank of $\psi$ (equivalently, $\dim H^{0}[\cpx{G}(\bdens{\rho}_{\sAB})] = \mathrm{rank}(\rho_{\sA})^2 -1 = \mathrm{rank}(\rho_{\sB})^2 -1$).
\end{theorem}
\begin{proof}
	The proof involves a usage of Thm.~\ref{thm:kernel_covariance}. We supply the details in Appendix~\ref{app:pure_bipartite_cohomology}.
\end{proof}

In general, the first cohomology group of $\cpx{G}(\bdens{\rho}_{\sAB})$ for pure $\bdens{\rho}_{\sAB}$ will not vanish.  Indeed, let $\bdens{\rho}_{\sAB}$ be pure with $\dens{\rho}_{\sAB} = \psi \otimes \psi^{\vee}$.  The Euler characteristic of $\cpx{G}(\bdens{\rho}_{\sAB})$ can be calculated easily at the level of cochains and is given by
\begin{align*}
	\chi \left[ \cpx{G}(\bdens{\rho}_{\sAB}) \right] = S(d_{\sA} + d_{\sB}) - (d_{\sA} d_{\sB}+1),
\end{align*}
where $d_{\sX} := \dim \hilb_{\sX}$, and $S$ the Schmidt rank of $\psi$. Because this must be equal to the alternating sum of cohomology groups, and the only non-vanishing cohomology groups are in degrees $0$ and $1$, then
\begin{align*}
	\dim H^{0}\left[ \cpx{G}(\bdens{\rho}_{\sAB}) \right] - \dim H^{1}\left[ \cpx{G}(\bdens{\rho}_{\sAB}) \right] = S(d_{\sA} + d_{\sB}) - (d_{\sA} d_{\sB}+1)
\end{align*}
It follows from Thm.~\ref{thm:pure_bipartite_cohomology} that
\begin{align*}
	\dim H^{1} \left[\cpx{G}(\bdens{\rho}_{\sAB})  \right] = (d_{\sA} - S)(d_{\sB} - S)
\end{align*}
which might be thought of as a very coarse measure of how far $\psi$ is from being ``maximally entangled" (smaller means closer): the Schmidt rank $S$ of a maximally entangled state should be equal to the smallest of $d_{\sA}$ or $d_{\sB}$ (i.e.\ the ranks of the reduced density matrices should be their maximal possible values).
Being $\mathbb{Z}$-valued, of course, this measure cannot detect the intricate details about how large the coefficients of the Schmidt decomposition are---only whether they are zero or not.

\begin{remark}[label=rmk:bipartite_poincare]{}{}
	Suppose $\bdens{\rho}_{\sAB}$ is a pure bipartite density state with Schmidt rank $S$ then, by the above discussion, its Poincar\'{e} polynomials are given by: 
	 \begin{align*}
		 P_{\cpx{G}}(\bdens{\rho}_{\sAB}) &= (S^2 - 1) + (d_{\sA} - S)(d_{\sB} - S) y\\
		 P_{\cpx{E}}(\bdens{\rho}_{\sAB}) &= 2(S^2-1).
	 \end{align*}
	 Thus, there are $\mathrm{min}(d_{\sA}, d_{\sB})$ possible polynomials, each uniquely labelled by the Schmidt rank. In particular the Poincar\'{e} polynomials associated to pure bipartite density states are complete SLOCC invariants.
\end{remark}
The following remark describes why one should not expect the dimensions of cohomology groups for generic mixed bipartite states to be simple functions of Schmidt ranks.
\begin{remark}{}{}
	As mentioned in Example~\ref{ex:mixed_state_counterexamples} the mixed bipartite density state $\bdens{\rho}_{\sAB}$ with 
	\begin{align*}
		\dens{\rho}_{\sAB} &= \lambda \ket{00} \bra{00} + (1-\lambda) \ket{01}\bra{01}
	\end{align*}
	($\lambda \in (0,1)$) has $H^{0}[\cpx{E}(\bdens{\rho}_{\sAB})] = H^{0}[\cpx{G}(\bdens{\rho}_{\sAB})] = 0$, but the bipartite density state with
	\begin{align*}
		\dens{\rho}_{\sAB} &= \lambda \ket{00} \bra{00} + (1-\lambda) \ket{11}\bra{11}
	\end{align*}
	($\lambda \in (0,1)$) has non-trivial cohomology groups.
	This shows that, in contrast to the situation for pure states established in Thm~\ref{thm:pure_bipartite_cohomology}, the dimensions of cohomology groups associated to a generic bipartite mixed state are not simply functions of ranks of reduced density matrices and Schmidt ranks of orthogonal decompositions, and depend sensitively on the relations between the subspaces $\image(\dens{\rho}_{\sA}),\, \image(\dens{\rho}_{\sB})$, and $\image(\dens{\rho}_{\sAB})$.
\end{remark}

\section{Multipartite Complexes \label{sec:multipartite_complexes}}
We next generalize (co)chain complexes above in the situation above for the case of a density state on a finite dimensional Hilbert space equipped with an $N$-partite ($N \geq 1$) tensor product decomposition.
It is helpful to define some notation for the data of such a decomposition.

\begin{definition}{}{multipartite_dens}
	Let $N$ be a positive integer.  Then a \newword{$N$-partite density state} (or \newword{multipartite density state} if we suppress mention of $N$) is a tuple of data $(P,(\hilb_{p})_{p \in P},\dens{\rho})$ where
	\begin{enumerate}
		\item $P$---the set of \newword{tensor factors}---is a totally ordered set of size $|P| = N$: i.e.\ a set equipped with a bijection to $\{1,\cdots,|P|\}$ that allows one to place a (total) order on $P$;

		\item $(\hilb_{p})_{p \in P}$ is a tuple of (non-zero) Hilbert spaces indexed by $P$;

		\item $\dens{\rho} \in \Dens{(\bigotimes_{p \in P} \hilb_{p})}$ (the order of the tensor product induced by the order on $P$).
	\end{enumerate}
	When the data $(\hilb_{p})_{p \in P}$ are understood, we will denote an $N$-partite density state $(P,(\hilb_{p})_{p \in P},\dens{\rho})$ by the bold symbol $\bdens{\rho}_{P}$.
	An $N$-partite density state $\bdens{\rho}_{P}$ is \newword{pure} if $\dens{\rho}$ is pure; otherwise it is \newword{mixed}.
\end{definition}
In examples, such as a multipartite density state associated to a lattice of spins, or a state on a space-time system divided up into causal diamonds, the elements of the set $P$ of tensor factors of a local assignment can be thought of as primitive/irreducible ``subsystems" (hence the letter ``$P$"), each equipped with a ``local degrees of freedom"\footnote{We do not necessarily assume $P$ is associated to any underlying geometric system, however.} encapsulated by $\hilb_{p}$; a more general ``subsystem" is given by an arbitrary subset of $P$.

\begin{remark}{}{}
	The total order on the set of tensor factors $P$ required in Def.~\ref{def:multipartite_dens} is for computational convenience.
	Any two choices of total order on $P$ (related by a permutation of $P$) result in (co)chain isomorphic cochain complexes (c.f.\ \S\ref{sec:multipartite_permutation_equivariance}.
	It is possible to work with a definition of a multipartite state that does not require a total order on the set of tensor factors, but this would only complicate the discussion and remain unfaithful to the computationally-oriented nature of this paper. 
\end{remark}

\begin{notation}{}{}
When convenient, totally ordered sets are denoted are denoted in tuple notation: e.g.\ $P = (p_{1}, \cdots p_{n})$ denotes the set $P = \{p_{1}, \cdots, p_{n} \}$ equipped with the total order given by $p_{i} < p_{j}$ if $i < j$.
\end{notation}

From the data of an $N$-partite density state it is useful to develop notation for a few more quantities.

\begin{definition}{}{multipartite_induced_data}
Let $(P,(\hilb_{p})_{p \in P},\dens{\rho})$ be an $N$-partite density state.
\begin{enumerate}
	\item For each subset $T \subseteq P$, define:
		\begin{align*}
			\newmath{\hilb_{T}} := 
			\left \{
				\begin{array}{ll}
					\bigotimes_{t \in T} \hilb_{t}, & \text{if $T \neq \emptyset$}\\
					\mathbb{C}, & \text{if $T = \emptyset$}
				\end{array}
			\right. ,
		\end{align*}
		where the order of the tensor product being given by the induced order on elements of $T$ (e.g.\ if $T = \{t_1, t_2 \}$ with $t_{1} < t_{2}$ then $\hilb_{T} = \hilb_{t_{1}} \otimes \hilb_{t_{2}}$).

		\item The reduced density state assigned to $T \subseteq P$ is
			\begin{align*}
				\newmath{\dens{\rho}_{T}} := \Tr_{P \backslash T} (\dens{\rho}) \in \Dens(\hilb_{T}),
			\end{align*}
			where $P \backslash T$ denotes the complement of $T$ in $P$, and $\Tr_{P \backslash T}$ denotes the partial trace over $\hilb_{P \backslash T}$ (note that $\rho_{\emptyset} = 1 \in \mathrm{Dens}(\mathbb{C})$ by definition).
			When $\dens{\rho} = \dens{\rho}_{P}$ is understood, the support projection $\supp_{\dens{\rho}_{T}}$ of $\dens{\rho}_{T}$ will be denoted $\newmath{\supp_{T}}$ and the left ideal $\mathfrak{N}_{\dens{\rho}_{T}}$ (c.f.\ Def.~\ref{def:right_essential_equivalence}) as $\newmath{\mathfrak{N}_{T}}$.
	\end{enumerate}
\end{definition}

To add to our zoo of definitions, the following straightforward notation is also helpful.

\begin{definition}[label=def:ordered_set_ops]{}{}
	Let $P$ be a finite set equipped with a total order and $T \subseteq P$.  
	\begin{enumerate} 
		\item $\newmath{T(l)}$ is the $(l+1)$th element of $T$ ( $0 \leq l \leq |T|-1$) using the induced order on elements of $T \subseteq P$.

		\item $\newmath{\partial_{l} T} := T \backslash T(l)$: i.e.\ the set of order $|T|-1$ obtained by eliminating $T(l)$ from $T$.
	\end{enumerate}
\end{definition}

\subsection{Cochain Complexes \label{sec:multipartite_cochain_complexes}}
Given an $N$-partite density state $\bdens{\rho}_{P} := (P, (\hilb_{p})_{p \in P}, \dens{\rho})$, we define the cochain complexes
\begin{align*}
	\newmath{\cpx{G}(\bdens{\rho}_{P})} :=
	\cdots \longrightarrow 0
	\longrightarrow \mathbb{C}
	\overset{d_{\cpx{G}}^{-1}}{\longrightarrow}
	\cpx{G}^{0}(\bdens{\rho}_{P})
	\overset{d_{\cpx{G}}^{0}}{\longrightarrow}
	\cpx{G}^{1}(\bdens{\rho}_{P})
	\overset{d_{\cpx{G}}^{2}}{\longrightarrow} \cdots
	\overset{d_{\cpx{G}}^{N-2}}{\longrightarrow}
	\cpx{G}^{N-1}(\bdens{\rho}_{P})
	\longrightarrow 0 \longrightarrow
	\cdots
\end{align*}
and
\begin{align*}
	\newmath{\cpx{E}(\bdens{\rho}_{P})} :=
	\cdots \longrightarrow 0
	\longrightarrow \mathbb{C}
	\overset{d_{\cpx{E}}^{-1}}{\longrightarrow}
	\cpx{E}^{0}(\bdens{\rho}_{P})
	\overset{d_{\cpx{E}}^{0}}{\longrightarrow}
	\cpx{E}^{1}(\bdens{\rho}_{P})
	\overset{d_{\cpx{E}}^{2}}{\longrightarrow} \cdots
	\overset{d_{\cpx{E}}^{N-2}}{\longrightarrow}
	\cpx{E}^{N-1}(\bdens{\rho}_{P})
	\longrightarrow 0 \longrightarrow
	\cdots 
\end{align*}
with components
\begin{align*}
	\newmath{\cpx{G}^{k}(\bdens{\rho}_{P})} = \prod_{\{T \subseteq P : |T| = k +1 \} } \mathtt{GNS}(\dens{\rho}_{T}),
\end{align*}
and 
\begin{align*}
	\newmath{\cpx{E}^{k}(\bdens{\rho}_{P})} = \prod_{\{T \subseteq P: |T| = k +1 \} } \mathtt{Com}(\dens{\rho}_{T}).
\end{align*}	
For $k < -1$ and $k > N-1$ the coboundaries are trivial; the coboundaries for $k = -1$ are given by:
\begin{align*}
	d_{\cpx{G}}^{-1}: \lambda &\longmapsto (\lambda \supp_{p})_{p \in P} =: \prod_{p \in P} (\lambda \supp_{p}),\\
	d_{\cpx{E}}^{-1}: \lambda &\longmapsto (\lambda \supp_{p})_{p \in P} =: \prod_{p \in P} (\lambda \supp_{p}),
\end{align*} 
(the far right hand side just offering an alternative notation for tuples $(\supp_{p})_{p \in P})$; the coboundaries for $0 \leq k \leq N-1$ defined via alternating sums of linear maps:
\begin{align*} 
	d_{\cpx{G}}^{k} = \sum_{j = 0}^{k+1} (-1)^{j} (\Delta_{\cpx{G}})\indices{^{k}_{j}},\\ 
	d_{\cpx{E}}^{k} = \sum_{j = 0}^{k+1} (-1)^{j} (\Delta_{\cpx{E}})\indices{^{k}_{j}}.
\end{align*}
To define the maps
\begin{align*} 
	(\Delta_{\cpx{G}})\indices{^{k}_{j}}: \prod_{\{V \subseteq P : |V| = k +1 \} } 
	\mathtt{GNS}(\dens{\rho}_{V}) \longrightarrow \prod_{\{T \subseteq P : |T| = k + 2 \} } \mathtt{GNS}(\dens{\rho}_{T}) 
\end{align*} 
and
\begin{align*} 
	(\Delta_{\cpx{G}})\indices{^{k}_{j}}:  \prod_{\{V \subseteq P : |V| = k +1 \} }
	\mathtt{Com}(\dens{\rho}_{V})  \longrightarrow \prod_{\{T \subseteq P : |T| = k + 2 \} }  \mathtt{Com}(\dens{\rho}_{T}), 
\end{align*}
for $0 \leq j \leq k+1$, we first define the ``extension maps" that take in an assignment of operators to each subset of size $k+1$ (an element of $\prod_{|V| = k + 1} \algebra{\hilb_{V}}$) and extend them to an assignment of operators on sets of size $k+2$ by tensoring by the identity at certain positions.
Explicitly these maps are given by   
\begin{align*}
	\epsilon\indices{^{k}_{j}}: \prod_{\{V \subseteq P : |V| = k +1 \} } 
	\algebra{\hilb_{V}} &\longrightarrow \prod_{\{T \subseteq P : |T| = k+2 \} } \algebra{\hilb_{T}}\\
	R &\longmapsto \prod_{\{T \subseteq P : |T| = k+2 \} } (\epsilon\indices{^{k}_{j}} R)_{T},
\end{align*}
where\footnote{We are using the following notation: given a product of vector spaces $\prod_{l \in S}V_{l}$ and an element $v \in \prod_{l \in S}V_{l}$ we let $v_{s} \in V_{s}$ denote the image of the projection $\prod_{l \in S}V_{l} \rightarrow V_{s}$.
Moreover, given a sequence of maps $(f_{l}: W \rightarrow V_{l})_{l \in S}$, we denote the induced map into the product $\prod_{l \in S} V_{l}$ as $\prod_{l \in S} f_{l}: W \rightarrow \prod_{l \in S} V_{l}$.}
\begin{align*}
	(\epsilon\indices{^{k}_{j}} R)_{T} := \Sigma_{(T,j)}\left(R_{\partial_{j} T} \otimes 1_{t_{j}} \right) \in \algebra{\hilb_{T}},
\end{align*}
and $\Sigma_{(T,j)}$ is the reshuffling map
\begin{align*}
	\Sigma_{(T,j)}: \algebra{\hilb_{\partial_{j} T}} \otimes \algebra{\hilb_{t}} & \longrightarrow \algebra{\hilb_{T}}.
\end{align*}
This reshuffling map does ``the obvious thing" and is not worthy of much attention, but for total clarity we define it explicitly: first define the reshuffling of Hilbert spaces
\begin{align*} 
	\sigma_{(T,j)}: \left( \bigotimes_{t \in \partial_{j} T} \hilb_{t} \right) \otimes \hilb_{t_{j}} \rightarrow  \bigotimes_{t \in T} \hilb_{t}
\end{align*}
by linearization of
\begin{align*}
	\phi_{t_{1}} \otimes \phi_{t_{2}} \otimes \cdots \phi_{t_{j-1}} \otimes \cancel{\phi_{t_{j}}} \otimes \phi_{t_{j+1}} \otimes \cdots \otimes \phi_{t_{|T|}} \otimes \phi_{t_{j}} \longmapsto \phi_{t_{1}} \otimes \phi_{t_{2}} \otimes \cdots \phi_{t_{|T|}} 
\end{align*}
where $T = \{t_{1}, \cdots, t_{|T|} \}$ with $t_{l} < t_{m}$ for $l < m$, $t_{r} \in \hilb_{t_{r}}$, and $\cancel{\phi_{t_{j}}}$ is a placeholder emphasizing the absence of the expected tensor factor at the $j$th position.
Then we define:
\begin{align*}
	\Sigma_{(T,j)}: \algebra{\hilb_{\partial_{j} T}} \otimes \algebra{\hilb_{t}} & \longrightarrow \algebra{\hilb_{T}}\\
	r               & \longmapsto \left(\sigma_{(T,j)} \right) r \left( \sigma_{(T,j)} \right)^{*}.
\end{align*}
With this in hand, the evaluation of $(\Delta_{\cpx{G}})\indices{^{k}_{j}}$ and $(\Delta_{\cpx{E}})\indices{^{k}_{j}}$ on a $k$-cochain $R$ (living in the appropriate domain) are defined by restrictions and support-projection-compressions of the extension maps $\epsilon\indices{^{k}_{j}}$: 
\begin{align*}
	\left[(\Delta_{\cpx{G}})\indices{^{k}_{j}} R\right]_{T} &:= \left(\epsilon\indices{^{k}_{j}} R\right)_{T} \supp_{T},\\
	\left[(\Delta_{\cpx{E}})\indices{^{k}_{j}} R\right]_{T} &:= \supp_{T} \left(\epsilon\indices{^{k}_{j}} R \right)_{T} \supp_{T},
\end{align*}
for any subset $T \subseteq P$ with $|T| = k + 2$.

Summarizing this in a different notation, we have:
\begin{align*}
	(\Delta_{\cpx{G}})\indices{^{k}_{j}} : R &\longmapsto \prod_{|T| = k + 2} \left[\Sigma_{(T,j)}\left(R_{\partial_{j} T} \otimes 1_{t_{j}} \right) \right] \supp_{T}.
\end{align*}
and
\begin{align*}
	(\Delta_{\cpx{E}})\indices{^{k}_{j}} : R &\longmapsto \prod_{|T| = k + 2} \supp_{T}\left[\Sigma_{(T,j)}\left(R_{\partial_{j} T} \otimes 1_{t_{j}} \right) \right] \supp_{T}
\end{align*}

\begin{proposition}{}{}
	The above description defines a cochain complex.
\end{proposition}
\begin{proof}
	One must show that the coboundaries square to zero.
	This is a computational exercise that relies on the compatibility of supports lemma Lem~\ref{lem:compat_of_supports}.
\end{proof}

\subsection{Chain Complexes \label{sec:multipartite_chain_complexes}}
We define the chain complexes 
\begin{align*}
	\newmath{\cpx{g}(\bdens{\rho}_{P})} :=
	\cdots \longleftarrow 0
	\longleftarrow \mathbb{C}
	\overset{\partial^{\cpx{g}}_{0}}{\longleftarrow}
	\cpx{g}_{0}(\bdens{\rho}_{P})
	\overset{\partial^{\cpx{g}}_{1}}{\longleftarrow}
	\cpx{g}_{1}(\bdens{\rho}_{P})
	\overset{\partial^{\cpx{g}}_{2}}{\longleftarrow} \cdots
	\overset{\partial^{\cpx{g}}_{N-1}}{\longleftarrow}
	\cpx{g}_{N-1}(\bdens{\rho}_{P})
	\longleftarrow 0 \longleftarrow
	\cdots
\end{align*}
and
\begin{align*}
	\newmath{\cpx{e}(\bdens{\rho}_{P})} :=
	\cdots \longleftarrow 0
	\longleftarrow \mathbb{C}
	\overset{\partial^{\cpx{e}}_{0}}{\longleftarrow}
	\cpx{e}_{0}(\bdens{\rho}_{P})
	\overset{\partial^{\cpx{e}}_{1}}{\longleftarrow}
	\cpx{e}_{1}(\bdens{\rho}_{P})
	\overset{\partial^{\cpx{e}}_{2}}{\longleftarrow} \cdots
	\overset{\partial^{\cpx{e}}_{N-1}}{\longleftarrow}
	\cpx{e}_{N-1}(\bdens{\rho}_{P})
	\longleftarrow 0 \longleftarrow
	\cdots 
\end{align*}
with components
\begin{align*}
	\newmath{\cpx{g}_{k}(\bdens{\rho}_{P})} = \bigoplus_{\{T \subseteq P : |T| = k +1 \} } \mathtt{gns}(\dens{\rho}_{T}),
\end{align*}
and 
\begin{align*}
	\newmath{\cpx{e}_{k}(\bdens{\rho}_{P})} = \bigoplus_{\{T \subseteq P: |T| = k +1 \} } \mathtt{com}(\dens{\rho}_{T}).
\end{align*}	
For $k < 0$ and $k > N -1$ the boundary maps are trivial; postponing the definition of the $k=0$ boundary maps, the non-trivial boundary maps for $k > 0$ are defined via alternating sums of linear maps
\begin{align*} 
	\partial^{\cpx{g}}_{k} = \sum_{j = 0}^{k-1} (-1)^{j} (\delta^{\cpx{g}})\indices{_k^j}\\ 
	\partial^{\cpx{e}}_{k} = \sum_{j = 0}^{k-1} (-1)^{j} (\delta^{\cpx{e}})\indices{_k^j}.
\end{align*} 
where
\begin{align*} 
	(\delta^{\cpx{g}})\indices{_k^j}: \bigoplus_{\{V \subseteq P : |V| = k + 1 \} } 
	\mathtt{gns}(\dens{\rho}_{V}) \rightarrow \bigoplus_{\{W \subseteq P : |W| = k \} } \mathtt{gns}(\dens{\rho}_{W}) 
\end{align*} 
and
\begin{align*} 
	(\delta^{\cpx{e}})\indices{_k^j}:  \bigoplus_{\{V \subseteq P : |V| = k +1 \} }
	\mathtt{com}(\dens{\rho}_{V})  \rightarrow \bigoplus_{\{W \subseteq P : |W| = k \} }  \mathtt{com}(\dens{\rho}_{W}), 
\end{align*}
are given by restrictions of a map:
\begin{align*}
	\tau\indices{_k^j}: \bigoplus_{\{V \subseteq P : |V| = k +1 \} } 
	\states{\hilb_{V}} &\longrightarrow \bigoplus_{\{W \subseteq P : |W| = k \} } \states{\hilb_{W}}
\end{align*}
to the proper domains.
To define $\tau\indices{_k^j}$, note that every element $\dens{\Gamma} \in \bigoplus_{|V| = k +1} 
\states{\hilb_{V}}$ can be written as $\bigoplus_{|V| = k=1} \dens{\Gamma}_{W}$ where $\dens{\Gamma}_{W} \in \states{\hilb_{W}}$; with this in mind we have:
\begin{align*}
	\tau\indices{_k^j}: \bigoplus_{|V|=k+1} \dens{\Gamma}_{V} &\longmapsto \bigoplus_{|W|=k} \left( \sum_{\{V \subseteq P: W = \partial_{j} V \}} \Tr_{P \backslash (\partial_{j} V)} \left[\dens{\Gamma}_{V}\right] \right),
\end{align*}
(where the summation in parenthesis is taken to be $0_{W} \in \states{\hilb_{W}}$ if $\{V \subseteq P: W= \partial_{j}V \} = \emptyset$ for a fixed $W$ and $j$). The $k=0$ boundary maps are given by restrictions of $\tau\indices{_0^0}$; noting that $\states{\hilb_{\emptyset}} = \mathbb{C}$, $\tau\indices{_0^0}$ is given by:
\begin{align*}
	\tau\indices{_0^0}: \bigoplus_{p \in P} \states{\hilb_{p}} &\longrightarrow \mathbb{C},\\ 
	\bigoplus_{p \in P} \dens{\Gamma}_{p} &\longmapsto \sum_{p \in P} \Tr\left[\dens{\Gamma}_{p}\right].
\end{align*}

\begin{proposition}{}{}
	The above defines a chain complex.
\end{proposition}
\begin{proof}
	The hard part here is showing that the maps $(\delta^{\cpx{e}})\indices{_k^j}$ and $(\delta^{\cpx{g}})\indices{_k^j}$ defined above are indeed well-defined: restricted to the appropriate domain, the differentials as defined above land in the appropriate targets/codomains.  This can be shown using Lem~\ref{lem:compat_of_supports_predual} (the ``predual" of the compatibility of supports lemma).
	After this is shown, the fact that the differentials square to zero is a straightforward computation in terms of the maps $\tau\indices{_k^j}$.
\end{proof}

\subsection{Unravelling Definitions}
It is worthwhile to unravel some of the definitions above for $N$-partite density states with small $N$.  In the following sections we focus on cochain complexes. 

\subsubsection{Unipartite Complexes}
Nothing in our definitions prevents us from considering the case of a ``unipartite" density state $(*, (\hilb_{*}), \dens{\rho} )$: a density state on a single tensor factor ``$*$" (here we abuse notation by denoting both the one-point set and its single element by $*$).  It should be straightforward to see the following proposition from the definitions.

\begin{proposition}{Unipartite complex}{unipartite}
	Let $\bdens{\rho}_{*} = (*, (\hilb_{*}), \dens{\rho} )$ be a unipartite density state then $\cpx{G}(\bdens{\rho}_{*})$ is the complex concentrated in degrees $-1$ and $0$ with 
	\begin{align*}
		\cpx{G}(\bdens{\rho}_{*})^{-1} &= \mathbb{C},\\
		\cpx{G}(\bdens{\rho}_{*})^{0} &= \mathtt{GNS}(\dens{\rho});
	\end{align*}	
	the map $\delta^{-1}_{0}$ being the map $\lambda \mapsto \lambda \supp_{\dens{\rho}}$.  As a result, its associated cohomology is concentrated in degree 0 with
	\begin{align*}
		H^{0}\left[\cpx{G}(\bdens{\rho}_{*}) \right] &= \mathtt{GNS}(\dens{\rho})/\mathrm{span}_{\mathbb{C}}\{\supp_{\dens{\rho}}\}.
	\end{align*}
	Similarly, $\cpx{E}(\bdens{\rho}_{*})$ has cohomology concentrated in degree 0 with
	\begin{align*}
		H^{0}\left[\cpx{E}(\bdens{\rho}_{*}) \right] &= \mathtt{Com}(\dens{\rho})/\mathrm{span}\{\supp_{\dens{\rho}}\};
	\end{align*}
	if $\bdens{\rho}_{*}$ is of finite rank, the dimension of this latter vector space is $\rank(\dens{\rho})^2 - 1$, so can be considered as a measure of the ``purity" of the state.
\end{proposition}
As often in mathematics, the trivial case of a good definition is far from having trivial applicability.  In \S\ref{sec:Kunneth} we will see that the unipartite example is an essential ingredient used to understand the cohomology of factorizable density states.

\subsubsection{Bipartite Complexes Revisited}
Next we compare the constructions in \S\ref{sec:multipartite_cochain_complexes} to the bipartite constructions of the previous sections.
We can work with a bipartite density state $\bdens{\rho}_{\sAB} = ((\sA,\sB), (\hilb_{\sA}, \hilb_{\sB}), \dens{\rho}_{\sAB})$ with $\sA < \sB$.
For simplicity we abuse notation and write the one point sets $\{\sA \}$ and $\{\sB \}$ as simply $\sA$ and $\sB$ (respectively) and denote the two element (totally ordered) set $\{\sA, \sB\}$ as $\sAB$.
Then running through our definitions, we start with the diagram that encodes the non-trivial part of our complex:
\begin{align*}
	\underbrace{\mathbb{C}}_{\cpx{G}^{-1}(\bdens{\rho}_{\sAB})}
	\xrightarrow{d^{-1}_{\cpx{G}}} 
	\underbrace{\mathtt{GNS}(\dens{\rho}_{\sA}) \times \mathtt{GNS}(\dens{\rho}_{\sB})}_{\cpx{G}^{0}(\bdens{\rho}_{\sAB})}
	\substack{\xrightarrow{(\Delta_{\cpx{G}})\indices{^{0}_{0}}} \\ \xrightarrow[(\Delta_{\cpx{G}})\indices{^{0}_{1}}]{}} 
	\underbrace{\mathtt{GNS}(\dens{\rho}_{\sAB})}_{\cpx{G}^{1}(\bdens{\rho}_{\sAB})}.
\end{align*}
to decode the maps $(\Delta_{\cpx{G}})\indices{^{0}_{1}}$, let $R \in \cpx{G}(\bdens{\rho}_{\sAB})^{0}$ which we will think of as an assignment of operators to (the primitive) subsystems: $R_{\sA} = a \in \mathtt{GNS}(\dens{\rho}_{\sA})$ and $R_{\sB} = b \in \mathtt{GNS}(\dens{\rho}_{\sB})$.
Then
\begin{align*}
	\left[(\Delta_{\cpx{G}})\indices{^{0}_{0}}R\right]_{\sAB} &= (1_{\sA} \otimes R_{\partial_{0}(\sAB)}) \supp_{\sAB} = (1_{\sA} \otimes b ) \supp_{\sAB}\\
	\left[(\Delta_{\cpx{G}})\indices{^{0}_{1}}R\right]_{\sAB} &= (R_{\partial_{1}(\sAB)} \otimes 1_{\sB}) \supp_{\sAB} = (a \otimes 1_{\sB} ) \supp_{\sAB},
\end{align*}
or, writing $R$ as a tuple $(a,b)$:
\begin{align*}
	(\Delta_{\cpx{G}})\indices{^{0}_{0}}: (a,b) &\longmapsto (1_{\sA} \otimes b ) \supp_{\sAB}\\
	(\Delta_{\cpx{G}})\indices{^{0}_{1}}: (a,b) &\longmapsto (a \otimes 1_{\sB} ) \supp_{\sAB};
\end{align*}	
and we have
\begin{align*}
	d^{0}_{\cpx{G}} = (\Delta_{\cpx{G}})\indices{^{0}_{0}} - (\Delta_{\cpx{G}})\indices{^{0}_{1}}: (a,b) \mapsto \left[1_{\sA} \otimes b  -  a \otimes 1_{\sB} \right]\supp_{\sAB}. 
\end{align*}		

\subsubsection{Tripartite Complexes}
In this section we work with a tripartite density state $\bdens{\rho}_{\sABC} = ((\sA,\sB, \sC), (\hilb_{\sA}, \hilb_{\sB}, \hilb_{\sABC}) , \dens{\rho}_{\sABC})$ with $\sA < \sB < \sC$, as in the previous section we abuse notation and let, e.g.\ $\sAB$ denote the set $\{\sA, \sB \}$.
The non-trivial part of the cochain complex $\cpx{G}(\bdens{\rho}_{\sABC})$ is given by studying the diagram: 
\begin{align*}
	\underbrace{\mathbb{C}}_{\cpx{G}^{-1}(\bdens{\rho}_{\sAB})}
	\xrightarrow{\Delta\indices{^{-1}_{0}}} 
	\underbrace{\mathtt{GNS}(\dens{\rho}_{\sA}) \times \mathtt{GNS}(\dens{\rho}_{\sB}) \times \mathtt{GNS}(\dens{\rho}_{\sC}) }_{\cpx{G}^{0}(\bdens{\rho}_{\sABC})}
	\substack{\xrightarrow{\Delta\indices{^{0}_{0}}} \\ \xrightarrow[\Delta\indices{^{0}_{1}}]{}} 
	\underbrace{\mathtt{GNS}(\dens{\rho}_{\sAB}) \times \mathtt{GNS}(\dens{\rho}_{\sAC}) \times \mathtt{GNS}(\dens{\rho}_{\sBC})}_{\cpx{G}^{1}(\bdens{\rho}_{\sABC})}
	\substack{\xrightarrow{\Delta\indices{^{1}_{0}}} \\ \xrightarrow{\Delta\indices{^{1}_{1}}} \\ \xrightarrow{\Delta\indices{^{1}_{2}}} }
	\underbrace{\mathtt{GNS}(\dens{\rho}_{\sABC})}_{\cpx{G}^{2}(\bdens{\rho}_{\sABC})},
\end{align*}
where $\Delta \indices{^{j}_{k}} := (\Delta_{\cpx{G}})\indices{^{j}_{k}}$; the non-trivial differentials $d_{\cpx{G}}^{k} = \sum_{j=0}^{k-1} (-1)^{k} \Delta\indices{^{j}_{k}}: \cpx{G}^{k}(\bdens{\rho}_{\sABC}) \rightarrow \cpx{G}^{k+1}(\bdens{\rho}_{\sABC})$ act via
\begin{align*}
	d_{\cpx{G}}^{-1}:
	\left \{
		\begin{array}{ll}
			\emptyset \mapsto \lambda
		\end{array}
	\right.
	&\longmapsto
	\left \{
		\begin{array}{ll}
			\sA \mapsto \lambda \supp_{\sA}\\
			\sB \mapsto \lambda \supp_{\sB}\\
			\sC \mapsto \lambda \supp_{\sC}
		\end{array}
	\right.,\\
	d_{\cpx{G}}^{0}:	
	\left \{
		\begin{array}{ll}
			\sA \mapsto a\\
			\sB \mapsto b\\
			\sC \mapsto c
		\end{array}
	\right.
	&\longmapsto
	\left \{
		\begin{array}{ll}
			\sAB \mapsto (1_{\sA} \otimes b - a \otimes 1_{\sB}) \supp_{\sAB}\\
			\sAC \mapsto (1_{\sA} \otimes c - a \otimes 1_{\sC}) \supp_{\sAC} \\
			\sBC \mapsto (1_{\sB} \otimes c - b \otimes 1_{\sC}) \supp_{\sBC}
		\end{array}
	\right.,\\
	d_{\cpx{G}}^{1}:	
	\left \{
		\begin{array}{ll}
			\sAB \mapsto x\\
			\sAC \mapsto y\\
			\sBC \mapsto z
		\end{array}
	\right.
	&\longmapsto
	\left \{
		\begin{array}{ll}
			\sABC \mapsto \left[1_{\sA} \otimes z - \Sigma_{(\sABC, \sB)} (y \otimes 1_{\sB}) + x \otimes 1_{\sC} \right]\supp_{\sABC}
		\end{array}
	\right. ,
\end{align*}
where---in lieu of tuple notation---we are denoting elements of Cartesian products via maps out of the indexing set of the product.

\subsection{An Interpretative Aside}
In this section we take a pause to give a geometric picture of the construction of the chain complexes above. 
We will use this picture to aid in our understanding of the $k$th cohomology component associated to an $N$-partite density state (where $k \leq N-2$) as encoding non-trivial non-local correlations between operators associated to a collection subsets of size $(k+1)$: a property related to the non-factorizability of the multipartite density state in question. 

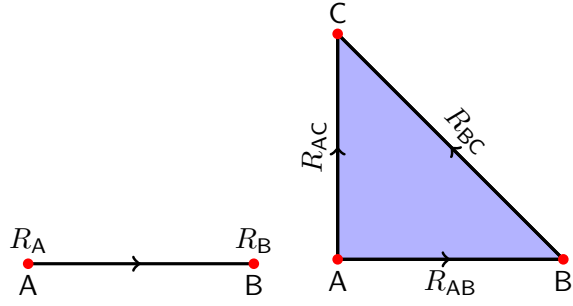
\begin{figure}
	\centering 
	\begin{tikzpicture}[very thick,decoration={markings, mark=at position 0.5 with {\arrow{>}}}]
		\coordinate [label=below:$\sA$, label=above:$R_{\sA}$] (A) at (0,0);
		\coordinate [label=below:$\sB$, label=above:$R_{\sB}$] (B) at (3,0);

		\draw[postaction={decorate}] (A) to (B);

		\fill[red]  (A) circle [radius=2pt]; 
		\fill[red]  (B) circle [radius=2pt]; 
	\end{tikzpicture}
	\begin{tikzpicture}[very thick,decoration={markings, mark=at position 0.5 with {\arrow{>}}}]
		\coordinate [label=below:$\sA$] (A) at (0,0);
		\coordinate [label=below:$\sB$] (B) at (3,0);
		\coordinate [label=above:$\sC$] (C) at (0,3);

		\draw[fill=blue!30] (A)--(B)--(C)--cycle;

		\draw[postaction={decorate}] (A) to node[sloped,below,midway] {$R_{\sAB}$} (B);

		\draw[postaction={decorate}] (B) to node[sloped,above,midway] {$R_{\sBC}$} (C);

		\draw[postaction={decorate}] (A) to node[sloped,above,midway] {$R_{\sAC}$} (C);

		\fill[red]  (A) circle [radius=2pt]; 
		\fill[red]  (B) circle [radius=2pt]; 
		\fill[red]  (C) circle [radius=2pt]; 
	\end{tikzpicture}
	\caption{Left: Representation of a 0-cochain $R$ in the computation of GNS or commutant cohomologies for a bipartite density state.
		One associates operators to the 0-simplices/nodes of a 1-simplex.
		The coboundary of $R$, which assigns an operator to the 1-simplex, is computed by taking the alternating sum of values on the boundary of the 1-simplex, with signs determined by the oriented boundary of the simplex.
		Right: situation representing a 1-cochain $R$ associated to GNS/commutant cohomologies of a tripartite density state.
		Oriented boundaries are computed by using the standard orientation of the 2-simplex on the page.
	\label{fig:cocycles_representation}} 
\end{figure}
\subsubsection{Visualization of (Co)-Chains as Sections Over Simplices \label{sec:visualization}}
The (co)boundary maps of the previous sections can be visualized geometrically.
Indeed, suppose we are handed $\bdens{\rho}_{P}$ an $N$-partite density state.
Then we begin by visualizing a standard topological $(N-1)$-simplex whose vertices are labelled by elements of the ordered set $P$.
With this labelling, faces of dimension $k$ (a.k.a. $k$-faces) of the simplex are in bijective correspondence with subsets $T \subseteq P$ of order $k+1$.
We will denote the face whose vertices are labelled by the subset $T$ by $F_{T}$; the ordering of elements of $P$ induces an orientation on each such face.

Suppose we wish to consider the cochain complex $\cpx{C}(\bdens{\rho}_{P})$, where $\cpx{C}$ is the GNS complex $\cpx{G}$ or the commutant complex $\cpx{E}$.
For $k \geq 1$, a $k$-cochain $R \in \cpx{C}^{k}(\bdens{\rho}_{P})$ is precisely the data of an assignment of an operator (living in the appropriate subspace) to the faces of dimension $k$.
To see this let $\mathtt{B}$ denote either the GNS building block $\mathtt{GNS}$ or commutant building block $\mathtt{Com}$ depending on which cochain complex we are considering; then we begin with the assignment
\begin{align}
	F_{T} &\longmapsto \mathtt{B}(\dens{\rho}_{T})
	\label{eq:face_assignment}
\end{align}
over all faces $\{F_{T}: T \subseteq P \}$.
A $k$-cochain $R$, which is a choice of $R_{T} \in \mathtt{B}(\bdens{\rho}_{T})$ for each subset $T$ with $|T| = k + 1$, is a section of this assignment over the \textit{$k$-skeleton} of the $N$-simplex: the ordered collection of dimension $k$-faces.

To understand the coboundary map $d^{k}: \cpx{C}^{k}(\bdens{\rho}_{P}) \rightarrow \cpx{C}^{k+1}(\bdens{\rho}_{P})$, we begin by trying to understand the intermediate maps
\begin{align}
	\left(\Delta_{\cpx{C}}\right) \indices{^{l}_{k}}: \cpx{C}^{k}(\bdens{\rho}_{P}) \rightarrow  \cpx{C}^{k+1}(\bdens{\rho}_{P}),
	\label{eq:coface_general}
\end{align}
for $l \in \{0, \cdots, k+1 \}$.
Indeed, fix $l$, then application of the map \eqref{eq:coface_general} can be thought of as the process of extending a section $R$ over $k$-faces to a section over on $(k+1)$ faces by:
\begin{enumerate}
	\item Supposing that the extension of $R$ to a $(k+1)$-face $F_{V}$ is given by looking at its value $R_{\partial_{l} V}$ on the $k$-face $F_{\partial_{l} V}$ (the $(l+1)$th-face of $F_{V}$ using the induced lexicographical ordering on subsets of $P$);

	\item Tensoring by $1_{V(l)}$ and then reshuffling tensor components to get an element of $\algebra{\hilb_{V}}$; 

	\item Projecting to the appropriate subspace $\mathtt{B}\left(\dens{\rho}_{V}\right) \leq \algebra{\hilb_{V}}$ by either right multiplication by $\supp_{V}$ (in the case of the GNS complex) or left/right multiplication by $\supp_{V}$ (in the case of the commutant complex).
\end{enumerate}
In order to combine the maps \eqref{eq:coface_general} into the alternating sum defining the differential, we note that the oriented boundary of $F_{V}$ is given as a disjoint union over its boundary faces with alternating choices of orientation: 
\begin{align*}
	\partial F_{V} = \coprod_{l = 0}^{k+1} (-1)^{l} F_{\partial_{l}V} 
\end{align*}
where $-F_{T}$ denotes the face $F_{T}$ with the opposite orientation to the one induced by the ordering of elements of $T$.
The signs used in the decomposition above are precisely the same signs used in the definition of the differential: 
\begin{align*}
	\left( d^{k}_{\cpx{C}} R \right)_{V} &= \sum_{l = 0}^{k+1} (-1)^{l} (\Delta \indices{^{l}_{k}} R)_{V}
\end{align*}
for all $V \subseteq P$ with $|V| = k+2$.

The picture for $k$-chains is similar: for GNS or commutant chain complexes we begin with an assignment of a vector space $\mathtt{b}(\bdens{\rho}_{T})$ to each face $F_{T}$, where $\mathtt{b}$ is one of $\mathtt{gns}$ or $\mathtt{com}$.
A $k$-chain $\Gamma \in \cpx{c}^{k}(\bdens{\rho}_{P})$ (where $\cpx{c} = \cpx{g}$ or $\cpx{e}$) is a section of this assignment over the $k$-faces; each of the maps $(\delta^{\cpx{c}})\indices{_k^j},\, j =0,\cdots, k-1$ are given by using partial traces to determine a section over the $(k-1)$-faces. 

\begin{remark}{}{}
	The above geometric picture gives an intuitive proof of the statement that commutant (co)homology must vanish for pure factorizable states. 
	Indeed, note that if $\bdens{\rho}_{P}$ is pure and factorizable, then $\mathtt{com}(\bdens{\rho}_{T}) \cong \mathbb{C}$.
	The commutant \textit{chain} complex is just the augmented simplicial chain complex, with $\mathbb{C}$-coefficients, of the standard $(|P|-1)$-simplex. This is the complex whose $k$th chain component $0 \leq k \leq (|P|-1)$ is the $\mathbb{C}$-linear vector space spanned by formal linear combinations of $k$-faces, with boundary maps $\partial_{k}$ for $0 \leq k \leq |P|-1$ defined by formal linear combinations of oriented boundaries of faces, and with $\partial_{-1}$ from $\mathbb{C}$ (in degree $-1$) to the degree 0 component sending $1 \in \mathbb{C}$ to the formal sum of all 0-simplices.
	For the standard $|P|$-simplex; its homology is the reduced simplicial homology of the standard $(|P|-1)$-simplex with $\mathbb{C}$-coefficients which is vanishing in all degrees as the $(|P|-1)$-simplex is contractible.
	Similarly the dual commutant cochain complex is then the augmented simplicial cochain complex where $k$-cochains are $\mathbb{C}$-linear functions on the set of $k$-faces of the $(|P|-1)$-simplex.
\end{remark}

\subsubsection{Cohomology classes and Multipartite Correlations \label{sec:classes_and_correlations}}
Let us now try to understand in what sense our cohomologies might encapsulate non-local correlations.
Recall that $k$-cochains (elements of $\cpx{C}^{k}(\bdens{\rho}_{P})$) can be understood as sections of the assignment
\begin{align*}
	F_{T} \mapsto \mathtt{B}(F_{T})
\end{align*}
over the $k$-skeleton (the union of all $k$-faces) of the $(|P|-1)$-simplex.
With this in mind, a $k$-cocycle $R \in \cpx{C}^{k}(\bdens{\rho}_{P})$ is a section over the $k$-skeleton of the $(|P|-1)$-simplex such that, for each $(k+1)$-face $F_{V}$, the operators on the boundary of $F_{V}$ satisfy the relation:
\begin{align}
	\sum_{l = 0}^{k+1} (-1)^{l} \underline{R_{\partial_{l}V}} \overset{\mathtt{B}}{\sim}_{V} 0
	\label{eq:cocycle_relation}
\end{align}
where:
\begin{itemize}
	\item $\underline{R_{\partial_{l}V}}$ is a shorthand notation for tensoring by the support projection and reshuffling:
		\begin{align*}
			\underline{R_{\partial_{l}V}}  := \Sigma_{(l,V)} (R_{\partial_{l}V} \otimes \supp_{V(l)}) \in \algebra{\hilb_{V}},
		\end{align*}
	\item and the equivalence relation $\overset{\mathtt{B}}{\sim}_{V}$ is an equivalence relation on elements of $\algebra{\hilb_{V}}$ defined by: 
		\begin{align*}
			r \overset{\mathtt{B}}{\sim}_{V} m \Leftrightarrow 
			\left \{	
				\begin{array}{ll}
					0 =(r-m)\supp_{V}, &  \text{if $\mathtt{B} = \mathtt{GNS}$ (used when computing $\cpx{G}(\bdens{\rho}_{P})$)}\\
					0 =\supp_{V} (r - m) \supp_{V}, &  \text{if $\mathtt{B} = \mathtt{Com}$ (used when computing $\cpx{E}(\bdens{\rho}_{P})$)}
				\end{array}
			\right. ,
		\end{align*}
		for any $r,m \in \algebra{\hilb_{V}}$.
\end{itemize}
As the definition suggests, the relation $\sim^{\mathtt{B}}_{V}$ depends on the density state $\dens{\rho}_{V}$ (up to support equivalence), but we suppress this dependence to keep our notation reasonable.
For those who prefer extending by the identity operator instead of support projections, it is worth noting that:
\begin{align*}
	\underline{R_{\partial_{l}V}} \overset{\mathtt{B}}{\sim}_{V} \Sigma_{(l,V)} (R_{\partial_{l}V} \otimes 1_{V(l)})
\end{align*}
so we can think of \eqref{eq:cocycle_relation} as equivalent to the statement that: 
\begin{align*}
	\sum_{l = 0}^{k} (-1)^{l} \Sigma_{(l,V)} (R_{\partial_{l}V} \otimes 1_{V(l)}) \overset{\mathtt{B}}{\sim}_{V} 0.
\end{align*}

\begin{remark}{}{}
	Let $r,m \in \algebra{\hilb_{V}}$.
	Note that $r \overset{\mathtt{GNS}}{\sim}_{V} m$ is equivalent to the statement that $r$ is right essentially equivalent to $m$: c.f.\ Def.~\ref{def:right_essential_equivalence}.
	In an attempt to give one possible interpretation to $\overset{\mathtt{Com}}{\sim}_{V}$ we note that $r \overset{\mathtt{Com}}{\sim}_{V}m$ is equivalent to the statement that the right multiplication action of $r$ and $m$ on elements of $\mathtt{GNS}(\dens{\rho}_{V}) = \algebra{\hilb_{V}} \supp_{V}$ (which might be thought of as the space of canonical representatives of right essential equivalence classes) are indistinguishable up to right essential equivalence, i.e.\ 
	\begin{align*}
		a r \overset{\mathtt{GNS}}{\sim}_{V} a m
	\end{align*}
	for all $a \in \mathtt{GNS}(\dens{\rho}_{V})$.
\end{remark}
We can think of the relation \eqref{eq:cocycle_relation} in the following way: suppose the $(k+1)$-face $V$ is fixed and we understand the value of the $k$-cocycle $R$ on every boundary face of $V$ except one of them; let this be the $m$th face (where $m \in \{0,\cdots,k + 1 \}$), then \eqref{eq:cocycle_relation} allows us to determine this missing value (tensored by a support projection) up to the equivalence relation $\overset{\mathtt{B}}{\sim}_{V}$:
\begin{align*}
	\underline{R_{\partial_{m} V}} \overset{\mathtt{B}}{\sim}_{V} \sum_{l \neq m} (-1)^{l + m} \underline{R_{\partial_{l}V}}.
\end{align*}
For instance, referring to Fig.~\ref{fig:cocycles_representation}: 
\begin{enumerate}
	\item If $R$ is a 0-cycle associated to the GNS/commutant cohomology of a bipartite density state $\bdens{\rho}_{\sAB}$, then:
		\begin{align}
			\underline{R_{\sA}} \overset{\mathtt{B}}{\sim}_{\sAB} \underline{R_{\sB}};
			\label{eq:bipartite_cocycle_relation}
		\end{align}
		In the case that $\cpx{C} = \cpx{G}$, by Thm.~\ref{thm:kernel_covariance}, the relation \eqref{eq:bipartite_cocycle_relation} is equivalent to the statement that $\Cov(R_{\sA}, R_{\sB}) = \Var_{\sA}(R_{\sA}) = \Var_{\sB}(R_{\sB})$ and $\Tr[\dens{\rho}_{\sA} R_{\sA}] = \Tr[\dens{\rho}_{\sB} R_{\sB}]$. 

	\item if $R$ is a 1-cycle associated to the GNS/commutant cohomology of a tripartite density state $\bdens{\rho}_{\sABC}$, we have:
		\begin{align*}
			\underline{R_{\sAC}} \overset{\mathtt{B}}{\sim}_{\sABC} \underline{R_{\sAB}} + \underline{R_{\sBC}}.	
		\end{align*}
\end{enumerate}
In general, the right hand side of \eqref{eq:cocycle_relation} includes operators associated to subsystems that include the tensor factor $V(m)$---a tensor factor that is not included in $\partial_{m}V$.
Thus, the non-zero elements of the set 
\begin{align*}
	\mathrm{Sol}_{V}(\bdens{\rho}_{P}) := \left\{R \in \prod_{l=0}^{k+1} \mathtt{B}\left(\dens{\rho}_{\partial_{l} V} \right) \colon \sum_{l = 0}^{k+1} (-1)^{l} \underline{R_{\partial_{l}V}} \overset{\mathtt{B}}{\sim}_{V} 0  \right\}
\end{align*}
can be interpreted as those sections over the face $V$ that exhibit \textit{possible} non-local correlations between $\overset{\mathtt{B}}{\sim}_{V}$ equivalence classes of elements in the collection $(\underline{R_{\partial_{l}V}})_{l = 0}^{k+1}$.  We emphasize the word ``possible" as it might happen that the relation \eqref{eq:cocycle_relation} does not convey any interesting information about the multipartite state in question.
For instance:
\begin{enumerate}
	\item When dealing with a bipartite density state $\bdens{\rho}_{\sAB}$: the pair of operators $(\supp_{\sA}, \supp_{\sB})$ satisfies:
		\begin{align*}
			\underline{\supp_{\sA}} = \supp_{\sA} \otimes \supp_{\sB} = \underline{\supp_{\sB}}
		\end{align*}
		so, regardless of what the density state $\dens{\rho}_{\sAB}$ associated to the full system $\sAB$ might be, we have: 
		\begin{align}
			\underline{\supp_{\sA}} \overset{\mathtt{B}}{\sim}_{\sAB} \underline{\supp_{\sB}}.
			\label{eq:bipartite_identity_relation}
		\end{align}
		However, this relation does not convey any non-local information between tensor factors: thinking of $\supp_{\sX}$ as a ``constant" random variable, it requires the collaboration of zero observers/tensor factors to measure; moreover, its variance with respect to the reduced density state $\dens{\rho}_{\sX}$ is vanishing and the value of its measurement is already known.

	\item Suppose we are studying cochain complexes associated to a tripartite state on the (ordered) set of tensor factors $(\sA, \sB, \sC)$.
		Then for any 0-cochain $F$, define its coboundary $R = d^{1}_{\cpx{C}} F$; because every coboundary is a cocycle $R$ must satisfy \eqref{eq:cocycle_relation}. However, this relation descends from the equality of operators 
		\begingroup
		\renewcommand{\arraystretch}{1.5}
		\begin{align*}
			\begin{array}{llc}
				\underbrace{\supp_{\sA} \otimes \supp_{\sB} \otimes F_{\sC} - F_{\sA} \otimes \supp_{\sB} \otimes \supp_{\sC}}_{\underline{R_{\sAC}}} & = & 
				\underbrace{\supp_{\sA} \otimes F_{\sB} \otimes \supp_{\sC} - F_{\sA} \otimes \supp_{\sB} \otimes \supp_{\sC}}_{\underline{R_{\sAB}}} \\ 
				{} & {} & + \\
				{} & {} & \underbrace{\supp_{\sA} \otimes \supp_{\sB} \otimes F_{\sC}  - \supp_{\sA} \otimes F_{\sB} \otimes \supp_{\sC}}_{\underline{R_{\sBC}}}.
			\end{array}
		\end{align*}
		\endgroup
		It is difficult to interpret this equality as some profound non-local relation between operators.
\end{enumerate}

If we are interested in non-trivial non-local correlations, we wish to quotient out by the subspace spanned by the ``trivial" elements of $\mathrm{Sol}(\bdens{\rho}_{P})$: those that do not exhibit any correlations inherent to the non-factorizability of $\bdens{\rho}_{P}$. 
As suggested by the examples such trivial elements seem to be related to those solutions to \eqref{eq:cocycle_relation} that descend from an \textit{equality} of operators
\begin{align}
	\sum_{l = 0}^{k+1} (-1)^{l} \underline{R_{\partial_{l}V}} = 0.
	\label{eq:cocycle_relation_trivial}
\end{align}
To make this statement more precise and convincing, begin by letting $\bdens{\rho}_{P} = ((\hilb_{p})_{p \in P}, \dens{\rho}_{P})$ be an arbitrary multipartite density state; using the reduced density states $\dens{\rho}_{\{p\}} := \Tr_{P \backslash \{P\}}(\dens{\rho}_{p})$ on the one element subsets $\{p \} \subseteq P$ we define the associated ``fully factorizable" multipartite density state\footnote{Using the notation developed below in Def.~\ref{def:multipartite_state_tensor}, we have $\twid{\bdens{\rho}}_{P} = \bigotimes_{p \in P} \bdens{\rho}_{\{p\}}$.}
\begin{align}
	\twid{\bdens{\rho}}_{P} = \left((\hilb_{p})_{p \in P}, \bigotimes_{p \in P} \dens{\rho}_{\{p\}}\right).
	\label{eq:fully_factorizable_form}
\end{align}
The assignment \eqref{eq:face_assignment} associated to $\twid{\bdens{\rho}}_{P}$ is given by $F_{T} \mapsto \mathtt{B}\left(\bigotimes_{t \in T} \dens{\rho}_{\{t\}}\right)$.
Suppose $R$ is a section of such an assignment over the boundary faces of a fixed face $F_{V}$; using the fact that our support projections factorize as $\supp_{\twid{\dens{\rho}}_{T}} = \bigotimes_{t \in T} \supp_{\dens{\rho}_{\{t\}}}$ for any $T \subseteq V$, it is straightforward to show that $R$ satisfies the condition \eqref{eq:cocycle_relation} if and only if the equality \eqref{eq:cocycle_relation_trivial} is true.
That is:
\begin{align*}
	\mathrm{Sol}_{V}(\twid{\bdens{\rho}}_{P}) = \left\{R \in \prod_{l=0}^{k+1} \mathtt{B}\left(\bigotimes_{m \neq l} \dens{\rho}_{V(m)} \right) \colon \sum_{l = 0}^{k+1} (-1)^{l} \underline{R_{\partial_{l}V}} = 0  \right\}.
\end{align*}
All elements of this set should be considered ``trivial" as, $\twid{\bdens{\rho}}_{P}$ is fully factorizable; so no operators should exhibit non-local correlations due to non-factorizability.

Now let us return to studying sections of $F_{T} \mapsto \mathtt{B}(\dens{\rho}_{T})$ over boundary faces of the face $F_{V}$ for the assignment given by the arbitrary density state $\bdens{\rho}_{P}$.
By the compatibility of supports lemma (Lem.~\ref{lem:compat_of_supports}), we have the inclusion: 
\begin{align*}
	\mathtt{B}(\dens{\rho}_{T}) \subseteq \mathtt{B}\left(\bigotimes_{t \in T} \dens{\rho}_{\{t\}}\right) 
\end{align*}
for any subset $T \subseteq P$.
Thus, we can always lift sections of the assignment associated to an arbitrary multipartite density state $\bdens{\rho}_{P}$ to sections of the assignment of its fully factorized form $\twid{\bdens{\rho}}_{P}$.
Hence, we can make sense of the intersection: 
\begin{align*}
	\mathrm{Sol}_{V}(\bdens{\rho}_{P}) \cap \mathrm{Sol}_{V}(\twid{\bdens{\rho}}_{P}),
\end{align*}
whose elements should be considered to be trivial: any element of this intersection is \textit{equal} to an element of $\mathrm{Sol}_{V}(\twid{\bdens{\rho}}_{P})$, whose elements cannot be indicators of non-factorizability/non-local correlations because $\twid{\bdens{\rho}}_{P}$ is factorizable. 
However, comparing sections via the relations $\sim^{\mathtt{B}}_{T},\, T \subseteq P$ is more natural than comparison by equality; after all, elements of the subspace $\mathtt{B}(\dens{\rho}_{V}) \leq \algebra{\hilb_{V}}$ are canonical representatives of $\overset{\mathtt{B}}{\sim}_{V}$ equivalence classes of elements in $\algebra{\hilb_{V}}$.
So really, it is more natural to consider elements of 
\begin{align}
	\mathrm{Triv}_{V}(\bdens{\rho}_{P}) :=  \left\{R \in \mathrm{Sol}(\bdens{\rho}_{P}) \colon \text{$\exists Q \in \mathrm{Sol}(\twid{\bdens{\rho}}_{P})$ such that $R_{\partial_{l}V} \overset{\mathtt{B}}{\sim}_{\partial_{l}V} Q_{\partial_{l}V}$ for $0 \leq l \leq k+1$} \right\}
	\label{eq:Triv_V_def}
\end{align}
to be considered trivial.
By the discussion above, these are precisely the solutions in $\mathrm{Sol}(\bdens{\rho}_{P})$ that ``descend" (in the sense of $\overset{\mathtt{B}}{\sim}_{T}$-equivalence) from solutions of \eqref{eq:cocycle_relation_trivial}.

\begin{remark}[label=rmk:triv_simplification]{}{}
	It is worth noting that the condition that $R \in \mathrm{Sol}(\bdens{\rho}_{P})$ in the definition \eqref{eq:Triv_V_def} is extraneous: using the compatibility of supports lemma one can show that any section $R$ of the assignment $F_{T} \mapsto \mathtt{B}(\bdens{\rho}_{T})$ over the face $V$ is an element of $\mathrm{Sol}(\bdens{\rho}_{P})$ if there exists  $Q \in \mathrm{Sol}(\twid{\bdens{\rho}}_{P})$ such that $R_{\partial_{l} V} \overset{\mathtt{B}}{\sim}_{\partial_{l}V} Q_{\partial_{l} V}$ for $0 \leq l \leq k+1$.
	As a result, we can write:
	\begin{align*}
		\mathrm{Triv}_{V}(\bdens{\rho}_{P}) =  \left\{R \in \prod_{l=0}^{k+1} \mathtt{B}\left(\dens{\rho}_{\partial_{l} V} \right) \colon \text{$\exists Q \in \mathrm{Sol}(\twid{\bdens{\rho}}_{P})$ s.t. $R_{\partial_{l}V} \overset{\mathtt{B}}{\sim}_{\partial_{l}V} Q_{\partial_{l}V},\,\forall l \in \{0,\cdots,k\}$} \right\}.
	\end{align*}
\end{remark}
Thus, if we are interested in studying sections over the boundary of a fixed face $F_{V}$ that exhibit non-trivial, non-local correlations due to non-factorizability, we wish to study non-trivial representatives of elements of the quotient vector space:
\begin{align*}
	\mathrm{Cor}_{V}(\bdens{\rho}_{P}) := \mathrm{Sol}_{V}(\bdens{\rho}_{P})/\mathrm{Triv}_{V}(\bdens{\rho}_{P}).
\end{align*}
However, our primary interest is in $k$-cocycles, which define sections over the entire $k$-skeleton (not just a single face) and satisfy \eqref{eq:cocycle_relation} over \textit{all} faces; that is, we want to look at the space
\begin{align*}
	\ker(d_{\cpx{C}}^{k}) = \left \{R \in \prod_{|T| = k+1} \mathtt{B}(\dens{\rho}_{V}): \text{$\sum_{l = 0}^{k+1} (-1)^{l} \underline{R_{\partial_{l}V}} \overset{\mathtt{B}}{\sim}_{V} 0$ for all $(k+1)$-faces $F_{V}$} \right \}.
\end{align*}
Once again, we would like to identify the ``trivial elements" of this subspace of operators: those that do not convey any information about non-factorizability/non-local correlations. 
These should be the operators that do not exhibit correlations along the boundaries of \textit{any} $(k+1)$-face $F_{V}$; to define this notion formally, note that for every $|V| = k + 2$, the projection map
\begin{align*}
	\mathrm{pr}_{V}: \prod_{|T| =k + 1} \mathtt{B}_{T} \rightarrow \prod_{l=0}^{k+1} \mathtt{B}_{\partial_{l} V} 
\end{align*}
restricts to a map
\begin{align*}
	\mathrm{pr}_{V}|_{\ker(d^{k}_{\cpx{C}})}: \ker(d^{k}_{\cpx{C}}) \rightarrow \mathrm{Sol}_{V}(\bdens{\rho}_{P});
\end{align*}
so, composing with the quotient map $\mathrm{Sol}_{V}(\bdens{\rho}_{P}) \rightarrow \mathrm{Cor}_{V}(\bdens{\rho}_{P})$ we have a map 
\begin{align*}
	\mathrm{cor}_{V}: \ker(d^{k}_{\cpx{C}}) \rightarrow \mathrm{Cor}_{V}(\bdens{\rho}_{P}).
\end{align*}
Elements of $\ker(d^{k})$ that are in the kernel of $\mathrm{cor}_{V}$ have trivial correlations along the face $V$.
Thus, the subspace of trivial elements of $\ker(d^{k}_{\cpx{C}})$ can be defined as: 
\begin{align*}
	\mathrm{Triv}_{k}(\bdens{\rho}_{P}) := \left \{R \in \ker(d^{k}_{\cpx{C}}) : \text{$\mathrm{cor}_{V}(R) = 0,\, \forall V \subseteq P$ with $|V| = k+2$} \right\}.
\end{align*}
This can be explicitly written as (c.f.\ Remark~\ref{rmk:triv_simplification} for the second equality)
\begin{align*}
	\mathrm{Triv}^{k}(\bdens{\rho}_{P}) &=  \left\{R \in \ker(d^{k}) \colon \text{$\exists Q \in \ker(\twid{d}^{k})$ s.t. $R_{T} \overset{\mathtt{B}}{\sim}_{T} Q_{T},\, \forall T \subseteq P$ with $|T| = k+1$} \right\}\\
										&=	\left\{R \in \prod_{|T|=k+1} \mathtt{B}(\dens{\rho}_{T}) \colon \text{$\exists Q \in \ker(\twid{d}^{k})$ s.t. $R_{T} \overset{\mathtt{B}}{\sim}_{T} Q_{T},\, \forall T \subseteq P$ with $|T| = k+1$} \right\},
\end{align*}
where, to remove some notational ambiguity, we are letting $d_{\cpx{C}}$ denote the coboundary associated to the complex $\cpx{C}(\bdens{\rho}_{P})$ and $\twid{d}_{\cpx{C}}$ denote the coboundary associated to the complex $\cpx{C}(\twid{\bdens{\rho}}_{P})$.
As it so happens, elements of $\mathrm{Triv}^{k}(\bdens{\rho}_{P})$ are precisely coboundaries.
\begin{proposition}[label=prop:triviality_means_coboundary]{}{}
	$\mathrm{Triv}^{k}(\bdens{\rho}_{P}) = \image \left(d_{\cpx{C}}^{k-1}: \cpx{C}^{k-1}(\bdens{\rho}_{P}) \rightarrow \cpx{C}^{k}(\bdens{\rho}_{P})  \right)$ for $k < |P|-1$.
\end{proposition}
\begin{proof}	
	The proof is provided in Appendix.~\ref{app:proof_prop_triviality_means_coboundary}.
	The proof involves a reference to Thm.~\ref{thm:support_fact_multipartite_cohomologies} (a result in \S\ref{sec:Kunneth} below).
\end{proof}
Summarizing, for any $k \leq |P|-2$, we have:
\begin{align*}
	H^{k} \left[ \cpx{C}(\bdens{\rho}_{P}) \right] &= \ker(d_{\cpx{C}}^{k})/\image(d_{\cpx{C}}^{k-1})\\[5pt]
													 &=\left\{\text{\parbox{15.5em}{Sections of $F_{T} \mapsto \mathtt{B}(\bdens{\rho}_{T})$ over the $k$-skeleton of the $(|P|-1)$-simplex that exhibit \textit{possible} non-local correlations along each face: i.e.\ solutions to \eqref{eq:cocycle_relation}.}} \right\} \Big / \left\{\text{\parbox{15em}{Trivial solutions to \eqref{eq:cocycle_relation}: i.e.\ those solutions that do not encode correlations due to non-factorizability of $\bdens{\rho}_{P}$ along the boundaries of any face.}} \right\}.\\
\end{align*}

Moreover, if $R$ is a representative of a non-zero class $[R] \in H^{k}[\cpx{C}(\bdens{\rho}_{P})]$, then $\mathrm{cor}_{V}(R) \neq 0$ for at least one $V \subseteq P$ with $|V|=k+2$: i.e.\ representatives of non-zero equivalence classes of the $k$th cohomology groups are those sections over the assignment $F_{T} \rightarrow \mathtt{B}(\bdens{\rho}_{P})$ that exhibit non-trivial, non-local correlations along the boundary of at least one $(k+1)$-face.

\subsubsection{The Highest Cohomology Component \label{sec:highest_multipartite_cohomology}}
So far we have neglected to offer an interpretation to the highest (possibly) non-trivial cohomology component $H^{|P|-1} \left[ \cpx{C}(\bdens{\rho}_{P}) \right]$. 
At the time of writing the author does not have an interpretation as satisfying for the components of degree $\leq |P|-1$, however it is worthwhile to make some observations.
First note that
\begin{align*}
	\ker(d_{\cpx{C}}^{|P|-1}) = \cpx{C}^{k}(\bdens{\rho}_{P}) = \mathtt{B}(\bdens{\rho}_{P}).
\end{align*}
That is, every section over the $(|P|-1)$-face defines a cocycle. 
Moreover, we have:
\begin{align*}
	\mathrm{Image}(d^{|P|-2}) = \mathrm{span}_{\mathbb{C}} \left( \bigcup_{l = 0}^{|P|-1} \left \{\underline{F_{l}} \supp_{P} \colon F_{l} \in \mathtt{B}(\bdens{\rho}_{\partial_{l} P}) \right \} \right)
\end{align*}
As a result, any $R \in \mathtt{B}(\bdens{\rho}_{P})$ passes to the zero class in $H^{|P|-1}\left[ \cpx{C}(\bdens{\rho}_{P}) \right]$ if
\begin{align*}
	R \overset{\mathtt{B}}{\sim}_{P} \sum_{l = 0}^{|P|-1} \Sigma_{P,l} \left[F_{l} \otimes 1_{l} \right]
\end{align*}
for any tuple of operators $(F_{l})_{l=1}^{|P|-1}$ with $F_{l} \in \mathtt{B}(\bdens{\rho}_{\partial_{l} P})$.
In other words, trivial cohomology classes consist of those $|P|$-body operators $R$ that, up to $\overset{\mathtt{B}}{\sim}_{P}$-equivalence, can be reduced to a sum of of lifts of $N-1$-body operators.
A generic operator is certainly not equal to such a linear combination, and in the fully factorizable situation, $\overset{\mathtt{B}}{\sim}_{P}$-equivalence is synonymous with equality; so in this sense, the dimension of the component of degree $|P|-1$ can be indicative of ``how badly" $\dens{\rho}_{P}$ is \textit{not} factorizable (reflected in the ability to reduce $N$-body operators up to $N-1$ body operators using our equivalence relation).
This phenomenon was already seen at the end of \S\ref{sec:pure_state_schmidt}, where it is shown that the dimension of the first GNS cohomology component for a pure bipartite state is a (very coarse) measure of how far that state is from being maximally entangled.\footnote{On the other hand, the dimension of pure bipartite commutant cohomology is always zero.}

\subsection{Basic Properties of the Multipartite Complexes \label{sec:multipartite_properties}}
The multipartite complexes obey natural generalizations of the properties outlined for the bipartite situation.  At the risk of sounding repetitive, we state these generalizations and the appropriate definitions that go along with them for the sake of clarity.

\subsubsection{Descent to Support Equivalence Classes}

\begin{definition}{}{}
	Two multipartite density states $(P,(\hilb_{p})_{p \in P}, \dens{\rho})$ and $(P,(\hilb_{p})_{p \in P}, \dens{\rho}')$ are \newword{support equivalent} if $\supp_{\dens{\rho}} = \supp_{\dens{\rho}'}$ (i.e.\ $\dens{\rho}$ is support equivalent to $\dens{\rho}'$ in the sense of Def.~\ref{def:support_equiv_vanilla}).
\end{definition}

\begin{proposition}{}{}
	The multipartite (co)chain complexes only depend on support equivalence classes: if $\bdens{\rho}_{P}$ and $\bdens{\varphi}_{P}$ are support equivalent then $\cpx{G}(\bdens{\rho}_{P}) = \cpx{G}(\bdens{\varphi}_{P})$ and $\cpx{E}(\bdens{\rho}_{P}) = \cpx{E}(\bdens{\varphi}_{P})$ (and similarly for the corresponding chain complexes). 
\end{proposition}

\subsubsection{Trace Duality for Chains and Cochains}
The following proposition is just the multipartite version of Prop.~\ref{prop:bipartite_cocomplex_to_dualcomplex}.
\begin{proposition}{}{trace_duality_multipartite}	
	Using appropriate restrictions of the maps $(-)^{\Tr}: \algebra{\hilb} \rightarrow \states{\hilb}^{\vee}$, there are cochain isomorphisms of cochain complexes 
	\begin{align*}
		\cpx{G}(\bdens{\rho}_{P}) &\overset{\sim}{\longrightarrow} [\cpx{g}(\bdens{\rho}_{P})]^{\vee}\\
		\cpx{E}(\bdens{\rho}_{P}) &\overset{\sim}{\longrightarrow} [\cpx{e}(\bdens{\rho}_{P})]^{\vee}. 
	\end{align*}	
\end{proposition}
\begin{proof}
	For $-1 \leq k \leq |P|-1$, define the linear maps 
	\begin{align*}
		\sigma^{k} := \prod_{|V|=k+1} (-)_{V}^{\Tr}: \prod_{|V|=k+1} \algebra{\hilb_{V}} &\longrightarrow \prod_{|V| = k+1} \states{\hilb_{V}}^{\vee} \cong \left(\bigoplus_{|V| = k+1} \states{\hilb_{V}} \right)^{\vee}\\
		R &\longmapsto  \prod_{|V| = k+1} (R_{V})^{\Tr},
	\end{align*}
	and for $k < -1$ and $k>N$ let $\sigma^{k}$ be the zero map between the zero vector space.
	Let $\sigma|_{\cpx{G}}^{k}$ and $\sigma|_{\cpx{E}}^{k}$ denote the restrictions of $\sigma^{k}$ to $\cpx{G}^{k}(\bdens{\rho}_{P})$ and $\cpx{E}^{k}(\bdens{\rho}_{P})$ respectively; one can verify that these restrictions respectively land in $\cpx{G}^{k+1}(\bdens{\rho}_{P})$ and $\cpx{E}^{k+1}(\bdens{\rho}_{P})$.
	To show that these restrictions define chain isomorphisms (taking $\sigma^{k}$ to be the identity map for $k < -1$ or $k > N$) it remains to show commutativity with the appropriate differentials; to show this we begin by observing that, for any $\Gamma \in \bigoplus_{|T| = k+2} \states{\hilb_{T}}$ and $R \in \prod_{|V| = k + 1} \algebra{\hilb_{V}}$, we have:	
	\begin{align*}
		\Tr \left \{ \left[\tau\indices{_{(k+1)}^j}(\Gamma)\right]_{\partial_{j}T} R_{\partial_{j}T} \right \} &= \Tr \left \{\Gamma_{T} \left[\epsilon\indices{^{k}_{j}}(R) \right]_{T} \right \}
	\end{align*}	
	for all $T \subseteq P$ with $|T| = k+2$ and $0 \leq j \leq k+1$.
	The above equation is equivalently written as 
	\begin{equation}
		\begin{aligned}
			\left \langle \left(R_{\partial_{j}T} \right)^{\Tr}, \left[\tau\indices{_{(k+1)}^j}(\Gamma)\right]_{\partial_{j}T} \right \rangle_{\partial_{j}T} &= \left \langle \left(\left[\epsilon\indices{^{k}_{j}}(R)\right]_{T} \right)^{\Tr}, \Gamma_{T} \right \rangle_{T},
		\end{aligned}
		\label{eq:trace_duality_multipartite}
	\end{equation}
	where, to reduce an overload of parenthesis, $\langle - , - \rangle_{V}: \states{\hilb_{V}}^{\vee} \times \states{\hilb_{V}} \rightarrow \mathbb{C}$ denotes the dual pairing.
	Using \eqref{eq:trace_duality_multipartite} one can verify: 
	\begin{align*}
		\sigma|_{\cpx{G}}^{k+1} \circ (\Delta_{\cpx{G}})\indices{^{k}_{j}} &= \left[(\delta^{\cpx{g}}) \indices{_{(k+1)}^j}\right]^{\vee} \circ \sigma|_{\cpx{G}}^{k}\\
		\sigma|_{\cpx{E}}^{k+1} \circ (\Delta_{\cpx{E}})\indices{^{k}_{j}} &= \left[(\delta^{\cpx{e}}) \indices{_{(k+1)}^j}\right]^{\vee} \circ \sigma|_{\cpx{E}}^{k}.
	\end{align*}
	from which commutativity with the appropriate differentials follows.
\end{proof}

As a corollary we have induced isomorphisms 
\begin{align*}
	H^{k}\left[\cpx{G}(\bdens{\rho}_{P}) \right] &\overset{\sim}{\longrightarrow} \left(H_{k}\left[\cpx{g}(\bdens{\rho}_{P}) \right] \right)^{\vee}\\
	H^{k}\left[\cpx{E}(\bdens{\rho}_{P}) \right] &\overset{\sim}{\longrightarrow}  \left(H_{k}\left[\cpx{e}(\bdens{\rho}_{P}) \right] \right)^{\vee}.  
\end{align*}

\subsubsection{Equivariance Under Permutations of Tensor Factors \label{sec:multipartite_permutation_equivariance}}
\begin{definition}{}{}
	\begin{enumerate}
		\item Let $P = (p_{1}, \cdots, p_{N})$ be a totally ordered set and $\sigma$ a permutation of $N$-elements (a bijective map $\{1, \cdots, N \} \rightarrow \{1, \cdots, N \}$), then $\sigma \cdot P$ is the totally ordered set $(p_{\sigma(1)}, \cdots p_{\sigma(N)})$.

		\item Let $\bdens{\rho}_{P} = (P, (\hilb_{p})_{p \in P}, \dens{\rho})$ be a multipartite density state, then 
			\begin{align*}
				\newmath{\sigma \cdot \bdens{\rho}_{P}} := (\sigma \cdot P, (\hilb_{\sigma(p)})_{p \in P}, u_{\sigma} \dens{\rho} u_{\sigma}^{*})
			\end{align*}
			where 
			\begin{align}
				u_{\sigma}: \bigotimes_{p \in P} \hilb_{p} \overset{\sim}{\longrightarrow} \bigotimes_{q \in \sigma \cdot P} \hilb_{q} = \bigotimes_{p \in P} \hilb_{\sigma(p)}
				\label{eq:permutation_to_induced_unitary}
			\end{align}
			(i.e.\ $u_{\sigma}: \hilb_{P} \rightarrow \hilb_{\sigma \cdot P}$) is the (unitary) reshuffling isomorphism given by linearization of
			\begin{align*}
				\psi_{1} \otimes \psi_{2} \otimes \cdots \otimes \psi_{N} \longmapsto \psi_{\sigma(1)} \otimes \psi_{\sigma(2)} \otimes \cdots \otimes \psi_{\sigma(N)},
			\end{align*}
			where $\psi_{i} \in \hilb_{p_{i}}$.
	\end{enumerate}
\end{definition}
\begin{proposition}{}{multipartite_permutation_equivariance}
	Let $\bdens{\rho}_{P}$ be an $N$-partite density state and $\sigma$ a permutation of $N$-elements, then there are induced cochain isomorphisms:
	\begin{align*}
		U_{\cpx{G}}(\sigma): \cpx{G}(\bdens{\rho}_{P}) &\overset{\sim}{\longrightarrow} \cpx{G}(\sigma \cdot \bdens{\rho}_{P})\\
		U_{\cpx{E}}(\sigma): \cpx{E}(\bdens{\rho}_{P}) &\overset{\sim}{\longrightarrow} \cpx{E}(\sigma \cdot \bdens{\rho}_{P}).
	\end{align*}
	and similarly for chain complexes.
\end{proposition}
\begin{proof}
	The proof is nearly an exercise in notation. For $0 \leq k \leq |P| - 1$ define the map 
	\begin{align*}
		U^{k}: \prod_{\{T \subseteq P: |T|=k+1 \}} \algebra{\hilb_{T}} &\longrightarrow \prod_{\{S \subseteq \sigma \cdot P: |S|=k+1 \}}\algebra{\hilb_{S}} \\
		\prod_{\{T \subseteq P: |T|=k+1 \}} x_{T} &\longmapsto \prod_{\{S \subseteq \sigma \cdot P: |S|=k+1 \}} (u^{\sigma^{-1} \cdot S}_{\sigma}) x_{\sigma^{-1} \cdot S} (u^{\sigma^{1} \cdot S}_{\sigma})^{*}, 
	\end{align*}
	where (to be explicit) we have ordered our products via the induced lexicographical ordering on subsets, and  $u_{\sigma}^{T}: \hilb_{T} \rightarrow \hilb_{\sigma \cdot T}$ is the reshuffling map defined akin to \eqref{eq:permutation_to_induced_unitary}.
	By further composing this map with left/right multiplications by support projections, we can construct isomorphisms $U_{\cpx{G}}(\sigma)$ and $U_{\cpx{E}}(\sigma)$ as indicated above.
\end{proof}

\subsubsection{Equivariance under Local Unitary/Invertible Transformations}
We define the notion of a local unitary transformation between multipartite density states.\footnote{In this definition, the data of an order preserving bijection simply allows for a relabelling of the tensor factors.} 
\begin{definition}[label=def:local_invertible_multipartite]{}{}
	Let $\bdens{\rho}_{P} := (P, (\hilb_{p})_{p \in P}, \dens{\rho})$ and $\bdens{\varphi}_{P'} = (P', (\mathcal{K}_{p'})_{p' \in P'}, \dens{\varphi})$ be $N$-partite density states ($|P| = |P'| = N$).
	A \newword{local invertible transformation} $\boldsymbol{l}: \bdens{\rho}_{P} \rightarrow \bdens{\varphi}_{P'}$ is a pair $(f, (l_{p})_{p \in P})$ of an order-preserving bijection $f: P \rightarrow P'$, and a tuple of invertible linear maps $l_{p}: \mathcal{H}_{p} \rightarrow \mathcal{K}_{f(p)}$ such that $\dens{\varphi} = l \dens{\rho} l^{-1}$ where $l := \bigotimes_{p \in P} l_{p}: \bigotimes_{p \in P} \hilb_{p} \rightarrow \bigotimes_{p' \in Q} \mathcal{K}_{p'}$.  
	When $l_{p}$ is unitary for each $p \in P$, then we say $\boldsymbol{l}$ is a \newword{local unitary transformation}. 
\end{definition}

Recall Thm.~\ref{thm:bipartite_invertible_equivariance}, which states that every local invertible transformation of bipartite density states canonically induces isomorphisms of the corresponding GNS and commutant complexes.
This theorem can be extended to the multipartite situation; the proof follows as an exercise in explicitly writing down induced isomorphisms. 
\begin{theorem}{Invariance/Equivariance Under Local Invertible Transformations}{}
	Let $\boldsymbol{u}: \bdens{\rho}_{P} \rightarrow \bdens{\varphi}_{P'}$ be a local invertible transformation between $N$-partite density states, then there are induced cochain isomorphisms:
	\begin{align*}
		\cpx{G}(\boldsymbol{u}): \cpx{G}(\bdens{\rho}_{P}) &\overset{\sim}{\longrightarrow} \cpx{G}(\bdens{\varphi}_{P'})\\
		\cpx{E}(\boldsymbol{u}): \cpx{E}(\bdens{\rho}_{P}) &\overset{\sim}{\longrightarrow} \cpx{G}(\bdens{\varphi}_{P'}).	
	\end{align*}
\end{theorem}

As two consequences of the theorem we have:
\begin{enumerate}
	\item Cohomologies are equivariant under local unitary transformations.
		In particular, the associated Poincar\'{e} polynomials
		\begin{align*}
			\newmath{P_{\cpx{G}}(\bdens{\rho}_{P})} &:= \sum_{k = 0}^{N-1} \left(\dim H^{k}[\cpx{G}(\bdens{\rho}_{P})] \right) y^{k},\\
			\newmath{P_{\cpx{E}}(\bdens{\rho}_{P})} &:= \sum_{k = 0}^{N-1} \left(\dim H^{k}[\cpx{E}(\bdens{\rho}_{P})] \right) y^{k}
		\end{align*}
		are invariant under local invertible transformations.

	\item Suppose each local Hilbert space $\hilb_{p}$ is a representation for some group $G$, so that each $\hilb_{T} = \bigotimes_{t \in T} \hilb_{t}$ is a tensor product representation of $G$.
		If $\dens{\rho} \in \hilb_{P}$ is fixed under the conjugation action of $G$, then each of the cohomology groups are $G$-representations.
\end{enumerate}

\subsubsection{Comparing Commutant to GNS complexes}
As in the bipartite situation, one can construct long exact sequences relating the homology of the complexes $\cpx{e}(\bdens{\rho}_{P})$ and $\cpx{g}(\bdens{\rho}_{P})$ cohomology of the complexes $\cpx{E}(\bdens{\rho}_{P})$ and $\cpx{G}(\bdens{\rho}_{P})$ coming from short exact sequences of complexes obtained by componentwise application of the maps \eqref{eq:com_inc_gns} and \eqref{eq:GNS_prj_Com}.
The cohomological version has the form:
\begin{center}
	\begin{tikzpicture}[descr/.style={fill=white,inner sep=1.5pt}]
		\matrix (m) [
		matrix of math nodes,
		row sep=1em,
		column sep=2.5em,
		text height=1.5ex, text depth=0.25ex
		]
		{ 0 & H^0[\ker(\Pi)] & H^0[\cpx{G}(\bdens{\rho}_{P})] & H^0[\cpx{E}(\bdens{\rho}_{P})] & \\
			& H^1[\ker(\Pi)] & H^1[\cpx{G}(\bdens{\rho}_{P})] & H^1[\cpx{E}(\bdens{\rho}_{P})] & \\
			& H^2[\ker(\Pi)] & H^2[\cpx{G}(\bdens{\rho}_{P})] & H^2[\cpx{E}(\bdens{\rho}_{P})] & \\
			& \makebox[\widthof{$H^{3}[\ker(\Pi)]$}][c]{} &            \vdots     & \makebox[\widthof{$H^3[\cpx{E}(\bdens{\rho}_{P})]$}][c]{}  & \\
			& H^{N-1}[\ker(\Pi)] & H^{N-1}[\cpx{G}(\bdens{\rho}_{P})] & H^{N-1}[\cpx{E}(\bdens{\rho}_{P})] & 0. \\
		};

		\path[overlay,->, font=\scriptsize,>=to,line width = 0.5pt]
			(m-1-1) edge (m-1-2)
			(m-1-2) edge (m-1-3)
			(m-1-3) edge (m-1-4)
			(m-1-4) edge[out=355,in=175] node[descr,yshift=0.3ex] {$b^0$} (m-2-2)
			(m-2-2) edge (m-2-3)
			(m-2-3) edge (m-2-4)
			(m-2-4) edge[out=355,in=175] node[descr,yshift=0.3ex] {$b^1$} (m-3-2)
			(m-3-2) edge (m-3-3)
			(m-3-3) edge (m-3-4)
			(m-3-4) edge[out=355,in=175] node[descr,yshift=0.3ex] {$b^2$} (m-4-2)
			(m-4-4) edge[out=355,in=175] node[descr,yshift=0.3ex] {$b^{N-2}$} (m-5-2)
			(m-5-2) edge (m-5-3)
			(m-5-3) edge (m-5-4)
			(m-5-4) edge (m-5-5);
	\end{tikzpicture}
\end{center}

\subsection{Tensor Products and Factorizability \label{sec:Kunneth}}
In the following sections we will explore what can be said about the factorizability of a multipartite density state by studying the cohomology of its associated cochain complexes.
We begin by providing concrete definitions for the various notions of factorizability of multipartite density states.

\subsubsection{The Tensor Product of Multipartite States}
Before giving a precise definition of what it means for a multipartite density state to be factorizable, we give a precise definition of a tensor product of multipartite density states. 
\begin{definition}[label=def:multipartite_state_tensor]{}{}
	\begin{enumerate}
		\item Let $P$ and $Q$ be totally ordered sets, then $\newmath{P \vee Q}$ is the totally ordered set given by the underlying disjoint union set $P \coprod Q$ equipped with the total order given by the usual ordering on the subsets $P,Q \subseteq P \coprod Q$ and such that $p < q$ for every $p \in P$ and $q \in Q$.

		\item Let $\bdens{\rho}_{P} := (P, (\hilb_{p})_{p \in P}, \dens{\rho})$ and $\bdens{\varphi}_{Q} := (Q, (\mathcal{K}_{q})_{q \in Q}, \dens{\varphi})$ be $N$ and $M$-partite density states respectively, then $\newmath{\bdens{\rho}_{P} \otimes \bdens{\varphi}_{Q}}$ is the $NM$-partite density state 
			\begin{equation*}
				(P \vee Q, (\mathcal{L}_{r})_{r \in P \vee Q}, \dens{\rho}_{P} \otimes \dens{\varphi}_{Q})
			\end{equation*}
			where 
			\begin{align*}
				\mathcal{L}_{r} = 
				\left \{
					\begin{array}{cc}
						\hilb_{r}, & \text{if $r \in P$}\\
						\mathcal{K}_{r}, & \text{if $r \in Q$}
					\end{array}
				\right. .
			\end{align*}
	\end{enumerate}
\end{definition}

Next we recall the notion of tensor products of cochain complexes, and the K\"{u}nneth theorem, which expresses the cohomology of a tensor product of cochain complexes as the tensor product of cochain complexes of their cohomologies as graded vector spaces.

\subsubsection{Tensor Products of Complexes and The K\"{u}nneth theorem}
\begin{definition}{}{}
	Let $\cpx{L}$ and $\cpx{M}$ be cochain complexes with coboundaries $d_{\cpx{L}}$ and $d_{\cpx{M}}$.
	The tensor product of two cochain complexes $\cpx{L}$ and $\cpx{M}$ is the cochain complex $\cpx{L} \otimes \cpx{M}$ with components
	\begin{align*}
		(\cpx{L} \otimes \cpx{M})^{n} := \bigoplus_{i + j = n} \cpx{L}^{i} \otimes_{\mathbb{C}} \cpx{M}^{j},
	\end{align*}
	and coboundary defined by linearization of
	\begin{align*}
		d^{i + j}_{\cpx{L} \otimes \cpx{M}} (x \otimes y) := (d^{i}_{\cpx{L}} x) \otimes y + (-1)^{i} x \otimes (d^{j}_{\cpx{M}} y).
	\end{align*}
	for $x \in \cpx{L}^{i}$ and $y \in \cpx{M}^{j}$.
\end{definition}	
Note that we can regard a $\mathbb{Z}$-graded vector space as a complex with vanishing coboundary.
In particular, the cohomology $H(\cpx{L})$ of a cochain complex---which is a $\mathbb{Z}$-graded vector space---can be regarded as a complex in this manner.
The following theorem states that cohomology (of cochain complexes of $\mathbb{C}$-vector spaces) ``preserves" tensor products.
\begin{theorem}{The Weak K\"{u}nneth Theorem}{weak_kunneth}
	Suppose $\cpx{L}$ and $\cpx{M}$ are cochain complexes of $\mathbb{C}$-vector spaces.\footnote{This theorem generalizes to a statement for cochain complexes of abelian groups if we require the cohomologies of $\cpx{L}$ and $\cpx{M}$ to be levelwise torsion-free.}
	Then there is a canonical isomorphism:
	\begin{align*}
		H^{k}\left[\cpx{L} \otimes \cpx{M} \right] \cong \bigoplus_{k = l + m} H^{l}[\cpx{L}] \otimes H^{m}[\cpx{M}].
	\end{align*}
\end{theorem}
\begin{proof}
	See, e.g.\ \cite[\S 58]{munkres} or \cite{nlab:kuenneth_theorem}. 
\end{proof}

\subsubsection{From Tensor Products of States to Tensor Products of Complexes}
The cochain complexes associated to a multipartite density state are compatible with tensor products up to a shift in grading.
\begin{definition}{}{}
	Let $\cpx{L}$ be a cochain complex with coboundary $d_{\cpx{L}}$ then the $k$-shifted cochain complex $\cpx{L}[k]$ is the complex with components $\cpx{L}[k]^{n} = \cpx{L}^{n+k}$ and  coboundary defined componentwise by:
	\begin{align*}
		d^{n}_{\cpx{L}[k]}: \cpx{L}[k]^{n} &\longrightarrow \cpx{L}[k]^{n+k}\\
		c &\longmapsto (-1)^{k} d^{n+k}(c).
	\end{align*}
\end{definition}

\begin{theorem}{Cochains for Factorizable States}{cochain_fact}
	Let $\bdens{\rho}_{P}$ and $\bdens{\varphi}_{Q}$ be $N$ and $M$ partite density states, then there are canonical isomorphisms:
	\begin{align*}	
		\cpx{G}(\bdens{\rho}_{P} \otimes \bdens{\varphi}_{Q}) &\cong \left( \cpx{G}(\bdens{\rho}_{P}) \otimes \cpx{G}(\bdens{\varphi}_{Q}) \right)[1];\\
		\cpx{E}(\bdens{\rho}_{P} \otimes \bdens{\varphi}_{Q}) &\cong \left(\cpx{E}(\bdens{\rho}_{P}) \otimes \cpx{E}(\bdens{\varphi}_{Q}) \right)[1].
	\end{align*}
\end{theorem}
\begin{proof}
	We begin with the observation that, for Hilbert spaces $\hilb$ and $\mathcal{K}$ along with states $\dens{\alpha} \in \Dens(\hilb)$ and $\dens{\beta} \in \Dens(\mathcal{K})$, we have canonical isomorphisms (of underlying $\mathbb{C}$-vector spaces)
	\begin{align*}
		\mathtt{GNS}(\dens{\alpha} \otimes \dens{\beta}) &\cong \left(\hilb \otimes \mathcal{K} \right) \otimes \left( \image(\alpha)^{\vee} \otimes \image(\beta)^{\vee} \right)\\
														 &\cong \left(\hilb \otimes \image(\beta)^{\vee} \right) \otimes \left(\hilb \otimes \image(\alpha)^{\vee} \right)\\
														 &\cong \mathtt{GNS}(\dens{\alpha}) \otimes \mathtt{GNS}(\dens{\beta}).
	\end{align*}	
	With this observation, we can write a string of canonical isomorphisms
	\begin{align*}
		\cpx{G}(\bdens{\rho}_{P} \otimes \bdens{\varphi}_{Q})^{l} &= \prod_{|T|=l+1} \mathtt{GNS}\left[(\dens{\rho} \otimes \dens{\varphi})_{T} \right]\\
																  &\cong \prod_{|T|=l+1} \mathtt{GNS}\left[\dens{\rho}_{T \cap P} \right] \otimes \mathtt{GNS}\left[\dens{\varphi}_{T \cap Q} \right]\\
																  &\cong \prod_{T = U \vee V} \left(\prod_{\{U \subseteq P: |U|=m+1 \}}\mathtt{GNS}\left[\dens{\rho}_{U} \right] \right) \otimes \left(\prod_{\{V \subseteq Q: |V|=n+1 \}}\mathtt{GNS}\left[\dens{\rho}_{V} \right] \right)\\
																  &\cong \prod_{l=m+n+1} \cpx{G}(\bdens{\rho}_{P})^{m} \otimes \cpx{G}(\bdens{\varphi}_{Q})^{n};
	\end{align*} 
	where, on the second line, we used $(\dens{\rho} \otimes \dens{\varphi})_{T} = \dens{\rho}_{T \cap P} \otimes \dens{\varphi}_{T \cap J}$.
	Thus, at the level of graded vector spaces we have $\cpx{G}(\bdens{\rho}_{P} \otimes \bdens{\varphi}_{Q}) = (\cpx{G}(\bdens{\rho}_{P}) \otimes \cpx{G}(\bdens{\varphi}_{Q}))[1]$.
	Verifying that the coboundary is the same as the coboundary on the shifted tensor product cochain complex is an exercise in unravelling definitions and keeping track of signs.
	The proof of $\cpx{E}(\bdens{\rho}_{P} \otimes \bdens{\varphi}_{Q}) \cong \left(\cpx{E}(\bdens{\rho}_{P}) \otimes \cpx{E}(\bdens{\varphi}_{Q}) \right)[1]$ follows nearly identical reasoning.
\end{proof}

Combining this result with the weak K\"{u}nneth theorem (Thm.~\ref{thm:weak_kunneth}) we have:
\begin{align*}
	H^{k}\left[ \cpx{G}(\bdens{\rho}_{P} \otimes \bdens{\varphi}_{Q})  \right]&\cong H^{k}\left[ \left( \cpx{G}(\bdens{\rho}_{P}) \otimes \cpx{G}(\bdens{\varphi}_{Q}) \right) \right][1];\\
	H^{k}\left[ \cpx{E}(\bdens{\rho}_{P} \otimes \bdens{\varphi}_{Q})  \right]&\cong H^{k}\left[ \left(\cpx{E}(\bdens{\rho}_{P}) \otimes \cpx{E}(\bdens{\varphi}_{Q}) \right) \right][1].
\end{align*}
for all $k \in \mathbb{Z}$.  Stated in terms of Poincar\'{e} polynomials, we have the following corollary.
\begin{corollary}{}{poincare_factorization}
	Let $\bdens{\rho}_{P}$ and $\bdens{\varphi}_{Q}$ be multipartite density states, then we have a factorization of Poincar\'{e} polynomials:
	\begin{align*} 
		P_{\cpx{G}}(\bdens{\rho}_{P} \otimes \bdens{\varphi}_{Q}) = y P_{\cpx{G}}(\bdens{\rho}_{P}) P_{\cpx{G}}(\bdens{\varphi}_{Q}),\\
		P_{\cpx{E}}(\bdens{\rho}_{P} \otimes \bdens{\varphi}_{Q}) = y P_{\cpx{E}}(\bdens{\rho}_{P}) P_{\cpx{E}}(\bdens{\varphi}_{Q}).
	\end{align*}
\end{corollary}

The irritating shift ``$[1]$" on the right hand side of the equations in Thm.~\ref{thm:cochain_fact} (which manifests itself by multiplication by $y$ on the right hand side of the equations in Cor.~\ref{cor:poincare_factorization}) can be eliminated if we focus our attention on the shifted complexes $\cpx{E}(\bdens{\rho}_{P})[1]$ and $\cpx{G}(\bdens{\rho}_{P})[1]$.
It is worthwhile to define some simplified notation for these shifted complexes.

\begin{definition}{}{shifted_complexes}
	Let $\bdens{\rho}_{P}$ be a multipartite density state; define
	\begin{align*}
		\newmath{\twid{\cpx{G}}(\bdens{\rho}_{P})} &:= \cpx{E}(\bdens{\rho}_{P})[1],\\
		\newmath{\twid{\cpx{E}}(\bdens{\rho}_{P})} &:= \cpx{E}(\bdens{\rho}_{P})[1].
	\end{align*}
	Moreover, denote the associated Poincar\'{e} polynomials by
	\begin{align*}
		\newmath{P_{\twid{G}}(\bdens{\rho}_{P})} &= \sum_{k = 0}^{|P|-1} (-1)^{k} \left(\dim H^{k}[\twid{\cpx{G}}(\bdens{\rho}_{P})] \right) y^{k} = y P_{\cpx{G}}(\bdens{\rho}_{P}),\\
		\newmath{P_{\twid{E}}(\bdens{\rho}_{P})} &= \sum_{k = 0}^{|P|-1} (-1)^{k} \left(\dim H^{k}[\twid{\cpx{E}}(\bdens{\rho}_{P})] \right) y^{k} = y P_{\cpx{E}}(\bdens{\rho}_{P}).
	\end{align*}	
\end{definition}
Then Thm.~\ref{thm:cochain_fact} is equivalent to the existence of canonical isomorphisms
\begin{align*}
	\twid{\cpx{G}}(\bdens{\rho}_{P} \otimes \bdens{\varphi}_{Q}) &\cong \twid{\cpx{G}}(\bdens{\rho}_{P}) \otimes \twid{\cpx{G}}(\bdens{\rho}_{Q}),\\
	\twid{\cpx{E}}(\bdens{\rho}_{P} \otimes \bdens{\varphi}_{Q}) &\cong \twid{\cpx{E}}(\bdens{\rho}_{P}) \otimes \twid{\cpx{E}}(\bdens{\varphi}_{Q}).
\end{align*}
At the level of Poincar\'{e} polynomials, we have:
\begin{align*}
	P_{\twid{\cpx{G}}}(\bdens{\rho}_{P} \otimes \bdens{\varphi}_{Q}) &= P_{\twid{\cpx{G}}}(\bdens{\rho}_{P}) P_{\twid{\cpx{G}}}(\bdens{\varphi}_{Q}),\\
	P_{\twid{\cpx{E}}}(\bdens{\rho}_{P} \otimes \bdens{\varphi}_{Q}) &= P_{\twid{\cpx{E}}}(\bdens{\rho}_{P}) P_{\twid{\cpx{E}}}(\bdens{\varphi}_{Q}).
\end{align*}

\begin{remark}{}{}
	Motivated by  Thm.~\ref{thm:cochain_fact}, we can ask if there is an operation on multipartite density states that induces Cartesian products/direct sums of (co)chain complexes and sums of Poincar\'{e} polynomials. 
	Such an operation exists, but takes us outside of the realm of the purely quantum mechanical.
	First of all, it requires the notion of quantum-classical mixtures (e.g.\ tuples of density states on tuples of Hilbert spaces)\footnote{C.f.\ Def.~\ref{def:finstate}.}.
	But more importantly, defining such an operation requires a generalization of our notion of ``multipartite density state".
	Such a generalized notion requires a treatment in terms of presheaves of vector spaces, which is slightly beyond the scope of this paper and will be treated in future work.
\end{remark}

\subsubsection{Full Support Factorizability and Cohomology}

Given a multipartite density state $\bdens{\rho}_{P}$ we can define (reduced) multipartite density states associated to subsets $T \subseteq P$.
A multipartite state is factorizable if it can be written as a tensor product of the reduced states associated to the one-element subsets (i.e.\ the ``primitive" subsystems). 
\begin{definition}[label=def:full_support_factorizability]{}{}
	\begin{enumerate}
		\item Let $T \subseteq P$, then $\newmath{\bdens{\rho}_{T}}$ is the $|T|$-partite density state $(T, (\hilb_{t})_{t \in T}, \dens{\rho}_{T})$.

		\item $\bdens{\rho}_{P}$ is said to be \newword{fully factorizable} if $\bdens{\rho}_{P} = \bigotimes_{p \in P} \bdens{\rho}_{\{p\}}$.

		\item $\bdens{\rho}_{P}$ is said to be \newword{fully support factorizable} if $\bdens{\rho}_{P}$ is support equivalent to $\bigotimes_{p \in P} \bdens{\rho}_{\{p\}}$.
	\end{enumerate}
\end{definition}
Because our (co)chain complexes only depend on the support equivalence classes of density states, it is support factorizability that has the best hope of detection by our (co)chain complexes.
As in the bipartite situation, one can show that being fully support factorizable is generically a weaker condition than being fully factorizable.
Nevertheless, if we restrict our attention to pure multipartite density states, support factorizability suffices as described in the remark below.
\begin{remark}{}{}
	Let $\bdens{\rho}_{P}$ be a (pure) multipartite density state with $\dens{\rho} = \psi \otimes \psi^{\vee}$ for some $\psi \in \hilb_{P}$, then $\psi = \bigotimes_{p \in P} \psi_{p}$ for some $\psi_{p} \in \hilb_{p}$ if and only if $\bdens{\rho}$ is fully factorizable.
	Moreover, a pure multipartite density state is fully support factorizable if and only if it is fully factorizable.
\end{remark}
The notion of full (support) factorizability does not account coarser versions of factorizability that can occur for an $N$-partite density state for $N \geq 2$; we will return to these coarser notions of factorizability in \S\ref{sec:coarse_factorizability} (which require a few more definitions to make sense of in our already established notation).

Using results from the previous section, we can compute the cohomologies associated to fully support factorizable states.
\begin{theorem}[label=thm:support_fact_multipartite_cohomologies]{}{}
	If $\bdens{\rho}_{P}=(P,(\hilb_{p})_{p \in P}, \dens{\rho})$ is a fully support factorizable multipartite density state, then
	\begin{align}
		H^{k}\left[\cpx{G}(\bdens{\rho}_{P}) \right] \cong 		
		\left \{
			\begin{array}{ll}
				\mathbb{C}^{D}, & \text{if $k = |P|-1$}\\
				0, & \text{otherwise}
			\end{array}
		\right.
		\label{eq:factorizable_multipartite_GNS}
	\end{align}
	where $D = \prod_{p \in P} (\dim_{\mathbb{C}}(\hilb_{p}) \rank(\dens{\rho}_{\{p\}})-1)$, and 
	\begin{align}
		H^{k}\left[\cpx{E}(\bdens{\rho}_{P}) \right] \cong 		
		\left \{
			\begin{array}{ll}
				\mathbb{C}^{d}, & \text{if $k = |P|-1$}\\
				0, & \text{otherwise}
			\end{array}
		\right.
		\label{eq:factorizable_multipartite_commutant}
	\end{align}
	where $d = \prod_{p \in P} (\rank(\dens{\rho}_{\{p\}})^2-1)$.
\end{theorem}
\begin{proof}
	By Theorem~\ref{thm:cochain_fact} we have a decomposition in terms of unipartite complexes: 
	\begin{align*}	
		\cpx{G}(\bdens{\rho}_{P}) \cong \left(\bigotimes_{p \in P} \mathrm{G}(\bdens{\rho}_{\{p\}}) \right)[|P|-1].
	\end{align*}
	Using a combination of the Weak-K\"{u}nneth theorem, Prop.~\ref{prop:unipartite}, and the observation that $H(\cpx{C}[l]) = H(\cpx{C})[l]$ for any (co)chain complex $\cpx{C}$ (here we are thinking of the cohomology $H(\cpx{C})$ as a graded vector space/complex with zero differential), we have 
	\begin{align*}	 
		H\left[\cpx{G}(\bdens{\rho}_{P})\right] \cong \left(\bigotimes_{p \in P} \mathtt{GNS}(\dens{\rho}_{\{p\}})/\mathrm{span}_{\mathbb{C}} \{ \supp_{\{p \}} \} \right)[|P|-1]
	\end{align*}
	and \eqref{eq:factorizable_multipartite_GNS} follows.  The equation
	\eqref{eq:factorizable_multipartite_commutant} follows via similar reasoning.
\end{proof}

Specializing to bipartite states (i.e.\ $|P|=2$), we note that Thm.~\ref{thm:support_fact_then_vanishing} follows as a corollary of Thm.~\ref{thm:support_fact_multipartite_cohomologies}.  Moreover, for pure multipartite density states, we have the following corollary.
\begin{corollary}[label=cor:fully_fact_to_cohom]{}{}
	Let $\bdens{\rho}_{P}$ be a pure, fully factorizable $N$-partite density state, then:
	\begin{itemize}
		\item The only non-vanishing cohomology group of $\cpx{G}(\bdens{\rho}_{P})$ is in degree $N-1$ with
			\begin{align*}
				\dim H^{N-1}[\cpx{G}(\bdens{\rho}_{P})] = \prod_{p \in P} (\dim \hilb_{p} -1);
			\end{align*}

		\item The complex $\cpx{E}(\bdens{\rho}_{P})$ has vanishing cohomology: $\dim H^{k}[\cpx{E}(\bdens{\rho}_{P})] \equiv 0$ for all $k \in \mathbb{Z}$.
	\end{itemize}
\end{corollary}
The latter statement of this corollary should be compared to Thm.~\ref{thm:support_fact_then_vanishing}, which states that a pure state $\bdens{\rho}_{P}$ is (fully) factorizable if and only if $\cpx{E}(\bdens{\rho}_{P})$ has vanishing cohomology; the first statement should be compared to Thm.~\ref{thm:pure_commutant_cohomology}, which states that a bipartite pure state $\bdens{\rho}_{P}$ is (fully) factorizable if and only if the cohomology of $\cpx{G}(\bdens{\rho}_{P})$ is concentrated in degree 1.
The ``if" parts of both of these bipartite theorems---which are statements of the form: factorizability of a pure bipartite state $\Rightarrow$ vanishing of certain cohomology groups---follow as corollaries of the above.
We might then ask if there are any multipartite generalizations of the ``only if" parts: i.e.\ can the vanishing of certain cohomology groups say anything about full factorizability.
This motivates the following example, which shows that the converse to the second statement of Cor.~\ref{cor:fully_fact_to_cohom} does not hold.

\begin{example}{Tensoring by a Unipartite Pure States Destroys Commutant Cohomology}{}
	Let $\bdens{\alpha}_{P}$ be an $N$-partite density state and suppose $\beta_{*} = ((*), (\hilb_{*}), \psi \otimes \psi^{\vee})$ is a (pure) unipartite state with $\psi \in \hilb_{*}$, then by an application of Thms.~\ref{thm:cochain_fact} and \ref{thm:support_fact_multipartite_cohomologies} the $(N+1)$-partite density state $\bdens{\alpha}_{P} \otimes \bdens{\beta}_{*}$ has
	\begin{align*}
		H^{k}[\cpx{E}(\bdens{\alpha}_{P} \otimes \bdens{\beta}_{*})] \equiv 0	
	\end{align*}
	for all $k \in \mathbb{Z}$.
\end{example}
In particular, if we take $\bdens{\alpha}_{P}$ to be a pure, non-factorizable, \textit{bipartite} state, then by Thm.~\ref{thm:pure_commutant_cohomology}, the zeroth cohomology group of $\cpx{E}(\bdens{\rho}_{P})$ is non-vanishing; however, by tensoring with an arbitrary unipartite pure state $\bdens{\beta}_{}$ on disjoint tensor factor, the resulting cohomology of $\cpx{E}(\bdens{\alpha}_{P} \otimes \bdens{\beta}_{*})$ is trivial even though $\bdens{\alpha}_{P} \otimes \bdens{\beta}_{*}$ is not fully factorizable.

On the other hand, this example motivates us to look at coarser measures of factorizability: we can use the data of any tripartite density state $\bdens{\rho}_{\sABC} = ((\sA,\sB,\sC), (\hilb_{\sA}, \hilb_{\sB}, \hilb_{\sC}), \dens{\rho}_{\sABC})$ to define a ``coarse-grained'' bipartite density state with respect to the partition $(\sAB, \sC)$: i.e.\ a bipartite state $\rho_{(\sAB)\sC} := ((\sAB, \sC), (\hilb_{\sAB}, \hilb_{\sC}), \dens{\rho}_{\sABC})$.
If we have $\bdens{\rho}_{\sABC} = \bdens{\alpha}_{\sAB} \otimes \bdens{\beta}_{*}$, then the resulting bipartite density state $\bdens{\rho}_{(\sAB)\sC}$ is factorizable.

\subsection{Factorizability with Respect to a Partition \label{sec:coarse_factorizability}}
For a bipartite density state there was only one non-trivial notion of factorizability, coinciding with what was referred to as ``full factorizability" in the previous section.
However, for multipartite density states, there is an entire zoo of possible generalizations of this notion of factorizability.
In particular ---as indicated end of the previous section---one might wish to study if a multipartite density state is factorizable with respect to arbitrary partitions of the tensor factors.\footnote{Given a multipartite density state $\bdens{\rho}_{P}$ one might also ask if any of the reductions $\bdens{\rho}_{T},\, T \subseteq T$ are factorizable in some sense; we will not explore this latter question in this paper.}

\subsubsection{Partitions and their Coarsenings}
We begin by recalling the definition of a partition of a set, placing an order on the elements of the partition when the set is totally ordered.
\begin{definition}{}{}
	Let $P$ be a finite set.
	\begin{enumerate}
		\item A \newword{partition} of $P$ is a collection $\lambda \subseteq \mathrm{Power}(P)$ such that the union $\bigcup_{T \in \lambda} T = P$  and $T \cap V = \emptyset$ for any $T, V \in \lambda$ with $T \neq V$. 

		\item The number of elements in a partition $\lambda$ of $P$ is denoted $|\lambda|$ and called the \newword{length} of $\lambda$. 

		\item Suppose $P$ is totally ordered, and $\lambda$ is a partition of $P$.  Using the lexicographical ordering on elements of $\mathrm{Power}(P)$ (a total order), $\lambda$ is equipped with a total order and we occasionally write it as a tuple $(\lambda_{1}, \lambda_{2}, \cdots, \lambda_{|\lambda|})$ where $\lambda_{i} < \lambda_{j}$ for $i < j$.
	\end{enumerate}
\end{definition}
The set of partitions of a set form a poset under refinement; for our purposes it is more natural to speak of ``coarsenings".\footnote{The deep reason behind this is that every partition $\lambda = (\lambda_{1}, \cdots \lambda_{L})$ of a set $P$ forms a ``complementary" cover of $P$ by open sets $\mathfrak{U}^{\lambda} := \{P \backslash \lambda_{l} \}_{l = 1}^{L}$; the cochain complexes formed in this paper are constructed by using \v{C}ech techniques to a presheaf using such covers.  Coarsenings of partitions correspond to refinements of covers.}
\begin{definition}{}{}
	A partition $\eta = (\eta_{1}, \cdots, \eta_{M})$ is a \newword{coarsening} of a partition $\lambda = (\lambda_{1}, \cdots, \lambda_{L})$, written $\eta \geq \lambda$, if for each $l \in \{1, \cdots L \}$ we have $\lambda_{l} \subseteq \eta_{m}$ for some $m \in \{1, \cdots, M \}$. 
\end{definition}
This definition is a trivial modification of the more common usage of the word \newword{refinement}: saying $\eta$ is a coarsening of $\lambda$ is equivalent to saying $\lambda$ is a refinement of $\eta$.  Note that, in particular, if $\eta$ is a coarsening of $\lambda$ we have $|\eta| \geq |\lambda|$ (with equality if and only if $\eta = \lambda$).

\subsubsection{Factorizability With Respect to a Partition}
We now give a formal definition of the factorizability of a multipartite density state with respect to a particular partition.
\begin{definition}[label=def:coarsenings]{}{}
	Let $\bdens{\rho}_{P} = (P, (\hilb_{p})_{p \in P}, \dens{\rho})$ be a multipartite density state and $\lambda$ a partition of $P$ of length $L$.
	\begin{enumerate}
		\item The \newword{$\lambda$-coarsening} of $\bdens{\rho}_{P}$ is the $L$-partite density state\footnote{Note that we are using the notation of Def.~\ref{def:multipartite_induced_data}, i.e.\ $\hilb_{\lambda_{l}} = \bigotimes_{p \in \lambda_{l}} \hilb_{p}$.}
			\begin{align*}
				\newmath{\lambda[\bdens{\rho}_{P}]} := \left((1,\cdots,L), (\hilb_{\lambda_{l}})_{l=1}^{L}, \lambda[\dens{\rho}]\right) 
			\end{align*}
			where
			\begin{align*}
				\newmath{\lambda[\dens{\rho}]} := u_{\lambda} \rho u_{\lambda}^{*}	
			\end{align*}
			with 
			\begin{align*}
				u_{\lambda}: \bigotimes_{p \in P} \hilb_{p} \overset{\sim}{\longrightarrow} \bigotimes_{l=1}^{L} \hilb_{\lambda_{l}}
			\end{align*}
			being the (unitary) reshuffling isomorphism defined by linearization of the map
			\begin{align*}
				\psi_{p_{1}} \otimes \cdots \otimes \psi_{p_{N}} &\longmapsto \left(\bigotimes_{q^{1} \in \lambda_{1}} \psi_{q^{1}} \right) \otimes \cdots \otimes \left(\bigotimes_{q^{L} \in \lambda_{L}} \psi_{q^{L}} \right),
			\end{align*}
			where $P = (p_{1}, \cdots, p_{N}),\, \psi_{q} \in \hilb_{q}$, and the order of the tensor product of each $\bigotimes_{q^{i} \in \lambda_{i}} \psi_{q^{i}}$ is in turn given by the order on elements of $P$.

		\item $\bdens{\rho}_{P}$ is \newword{factorizable with respect to $\lambda$} (a.k.a. \newword{$\lambda$-factorizable}) if $\lambda[\bdens{\rho}_{P}]$ is fully factorizable (c.f.\ Def.~\ref{def:full_support_factorizability}).

		\item $\bdens{\rho}_{P}$ is \newword{support factorizable with respect to $\lambda$} (a.k.a. \newword{$\lambda$-support factorizable}) if $\lambda[\bdens{\rho}_{P}]$ is fully support factorizable. 
	\end{enumerate}
\end{definition}

We make some quick remarks about these definitions.
\begin{remark}{}{pure_states_drop_support}
	\begin{enumerate}
		\item Once again, $\lambda$-\textit{support} factorizability is a weaker notion of $\lambda$-factorizability for general mixed multipartite density states: i.e.\ $\lambda$-factorizability implies $\lambda$-support factorizability, but the converse is not true.  However, for \textit{pure} multipartite states, $\lambda$-support factorizability coincides with $\lambda$-factorizability and the adjective ``support" is an unnecessary decoration.

		\item Within the collection of all partitions of a (totally ordered) set $P = (p_{1}, \cdots, p_{N})$ there are two ``extreme" partitions: the ``coarsest" or ``trivial" partition of length 1: $\lambda^{P}_{\mathrm{triv}} = ((P))$, and the ``finest" or ``full" partition of length $N$: $\lambda^{P}_{\mathrm{full}} := ((p_{1}), \cdots, (p_{N}))$.  Any multipartite density state $\bdens{\rho}_{P}$ is automatically (support) factorizable with respect to $\lambda^{P}_{\mathrm{triv}}$, while $\bdens{\rho}_{P}$ is (support) factorizable with respect to $\lambda^{P}_{\mathrm{full}}$ if and only if it is fully (support) factorizable. 
	\end{enumerate}
\end{remark}
The following proposition is a straightforward exercise.
\begin{lemma}{}{partition_fact_to_coarsening_fact}
	If $\bdens{\rho}_{P}$ is $\lambda$-(support) factorizable for some $P$ then it is $\eta$-(support) factorizable for any coarsening $\eta \geq \lambda$.
\end{lemma}
Better yet, we can reduce the question of $\lambda$-(support) factorizability to a question about support factorizability with respect to any length 2 coarsening of $\lambda$. 
\begin{lemma}{}{lambda_fact_to_bipartite}
	$\bdens{\rho}_{P}$ is $\lambda$-(support) factorizable if and only if it is $\eta$-(support) factorizable for any coarsening $\eta \geq \lambda$ with $|\eta| = 2$. 
\end{lemma}
So if $|\lambda| = L$ this involves checking factorizability of $\binom{L}{2} = \frac{1}{2} L (L-1)$ bipartite density states.  If our initial multipartite density state is pure, then one can check for factorizability of each of its (pure) bipartite coarsenings by checking the purity of their partial traces over a single tensor factor.

Next, let us explore the consequences of $\lambda$-factorizability on the cohomologies of cochain complexes associated to multipartite density states. Given an $N$-partite density state, for each partition $\lambda$ of its tensor factors we can associate a complex of its associated coarsening; this produces $2^{N}$ associated complexes associated to a multipartite density state.  For the sake of readability we introduce a natural shorthand notation for these complexes.
\begin{definition}{}{}
	Let $\bdens{\rho}_{P}$ be a multipartite density state and $\lambda$ a partition of $P$, then we define the complexes
	\begin{align*}
		\newmath{\cpx{G}_{\lambda}(\bdens{\rho}_{P})} &:= \cpx{G}(\lambda[\bdens{\rho}_{P}])\\
		\newmath{\cpx{E}_{\lambda}(\bdens{\rho}_{P})} &:= \cpx{E}(\lambda[\bdens{\rho}_{P}]).
	\end{align*}
\end{definition}
If a state is $\lambda$-factorizable, then there is a massive vanishing of cohomology groups: the cohomologies of the complexes associated to any coarsening of $\lambda$ are concentrated in the highest possible degree.
\begin{theorem}{}{multipartite_factorizable_coarsening_vanishing}
	Let $\bdens{\rho}_{P}$ be a multipartite density state and $\lambda$ a partition of $P$.
	\begin{enumerate}
		\item If $\bdens{\rho}_{P}$ is $\lambda$-support factorizable, then for each coarsening $\eta \geq \lambda$ we have $H^{k}[\cpx{G}_{\eta}(\bdens{\rho}_{P})] = 0$ for $k < |\eta|-1$. 

		\item If $\bdens{\rho}_{P}$ is $\lambda$-support factorizable, then for each coarsening $\eta \geq \lambda$ we have $H^{k}[\cpx{E}_{\eta}(\bdens{\rho}_{P})] = 0$ for $k < |\eta|-1$. 
	\end{enumerate}
\end{theorem}
\begin{proof}
	These statements are consequences of Lemma~\ref{lem:partition_fact_to_coarsening_fact} and Theorem.~\ref{thm:support_fact_multipartite_cohomologies}.
\end{proof}

Specializing to pure multipartite density states, one can supply a converse statement to the above theorem (requiring only bipartite coarsenings in the converse directions).

\begin{theorem}{}{pure_multipartite_cohomology_factorizability}
	Let $\bdens{\rho}_{P}$ be a pure multipartite density state and $\lambda$ a partition of $P$.
	\begin{enumerate}
		\labitem{(A)}{list:lambda_fact_G} $\bdens{\rho}_{P}$ is $\lambda$-support factorizable if and only if $H^{0}[\cpx{G}_{\eta}(\bdens{\rho}_{P})] = 0$ for all length 2 coarsenings $\eta \geq \lambda$. 

		\labitem{(B)}{list:lambda_fact_E} $\bdens{\rho}_{P}$ is $\lambda$-support factorizable if and only if $H^{0}[\cpx{E}_{\eta}(\bdens{\rho}_{P})] = 0$ for all length 2 coarsenings $\eta \geq \lambda$.
\end{enumerate}
\end{theorem}
\begin{proof}
	To prove \ref{list:lambda_fact_G} we use Lem.~\ref{lem:lambda_fact_to_bipartite} along with Thm.~\ref{thm:pure_GNS_cohomology}.
	To prove \ref{list:lambda_fact_E} we use Thm.~\ref{thm:pure_commutant_cohomology}.
\end{proof}
Of course, if bipartite cohomologies are computed purely by brute force this theorem does not offer any computational advantage over checking partial traces with respect to bipartite coarsenings to verify factorizability.
However, it does guarantee that if a multipartite density state $\bdens{\rho}_{P}$ state is \newword{not} $\lambda$-factorizable, then there exists length 2 coarsenings $\eta \geq \lambda$ and $\chi \geq \lambda$ such that $H^{0}[\cpx{G}_{\eta}(\bdens{\rho}_{P})] \neq 0$ and $H^{0}[\cpx{E}_{\chi}(\bdens{\rho}_{P})] \neq 0$.
Using the statements in \S\ref{sec:bipartite_factorizability_pure} it is easy to see that one can take $\eta = \chi$.

\begin{remark}{}{}
	The relationship between cohomologies of coarsenings of a multipartite density state, which can be expressed in terms of the relationship between \v{C}ech cohomologies as one takes refinements, is an important mostly unexplored aspect of this paper.
	One might ask, for instance, what the cohomology of a multipartite density state says about the cohomologies of its coarsenings and vice-versa.  The author hopes to come back to this issue in a more sophisticated discussion of the constructs in this paper. 
\end{remark}
As a small step toward understanding the issues mentioned in the remark above, the following theorem expresses the consequence of $\lambda$-factorizability on cohomologies of complexes computed using the \textit{finest} partition (i.e.\ the multipartite complexes considered before this section).
\begin{theorem}{}{finest_partition_cohomology_factorizability}
	If $\bdens{\rho}_{P}$ is $\lambda$-support factorizable for some $|\lambda| = L$, then $H^{k}[\cpx{G}(\bdens{\rho}_{P})] =0$ and $H^{k}[\cpx{E}(\bdens{\rho}_{P})] = 0$ for all $k \leq L-2$.
\end{theorem}
\begin{proof}
	Every $\lambda$-support factorizable state is support equivalent to a $\lambda$-factorizable state and cohomologies only depend on support equivalence classes.
	Hence, wlog we begin by supposing that $\bdens{\rho}_{P}$ is $\lambda$-factorizable.
	By definition $\lambda[\dens{\rho}_{P}] = \bigotimes_{l = 1}^{L} \lambda[\dens{\rho}]_{l} = \bigotimes_{l = 1}^{L} \dens{\rho}_{\lambda_{l}}$.  Thus we have
	\begin{align*}
		\bdens{\rho}_{P} = \sigma_{\lambda} \cdot \left(\bigotimes_{l = 1}^{L}\bdens{\rho}_{\lambda}\right)	
	\end{align*}
	where $\sigma_{\lambda}$ is the permutation taking the ordered word $\lambda_1 \cdots \lambda_{n}$ to $p_{1} \cdots p_{n}$.  As a result of equivariance of cochain complexes under permutation (Prop.~\ref{prop:multipartite_permutation_equivariance}) and Cor.~\ref{cor:poincare_factorization} we have
	\begin{align*}
		P_{\cpx{G}}(\bdens{\rho}_{P}) = y^{L-1} \prod_{l = 1}^{L}P_{\cpx{G}}(\bdens{\rho}_{\lambda_{l}}),\\
		P_{\cpx{E}}(\bdens{\rho}_{P}) = y^{L-1} \prod_{l = 1}^{L}P_{\cpx{E}}(\bdens{\rho}_{\lambda_{l}}),
	\end{align*}
	and the proposition follows.
\end{proof}
Thus, if, e.g.\ $H^{k}[\cpx{G}(\bdens{\rho}_{P})] \neq 0$ for some $k \leq K$, then $\bdens{\rho}_{P}$ cannot be $\lambda$-factorizable for any partition of length $k+2$.

\subsection{Pure state Entanglement and Cohomology}
Readers of this paper are likely to be interested in entanglement rather than factorizability.
In this section we restrict our attention to pure density states, where the notions of factorizability, support factorizability, and separability all coincide: this coincidence gives a very straightforward notion of entanglement as the property of failing to factorize.  
In general, entanglement is typically defined as the failure of separability. When working with mixed density states, one does not have a coincidence between the notions of factorizability, support factorizability, and separability. 
For instance, as demonstrated in Ex.~\ref{ex:supp_fact_but_entangled}, there are (mixed) bipartite density states that are entangled yet support equivalent to a factorizable density states. 
Because our chain complexes only depend on support equivalence classes, this makes them ill-equipped (at least na\"{i}vely) for detecting entanglement of a particular mixed density state. 
Instead, one these chain complexes are better equipped to aid with questions about properties of families of mixed density states. 

With this in mind we provide a few natural definitions that allow one to rephrase the statements of \S\ref{sec:coarse_factorizability} in a more succinct way and to draw attention to their usefulness in detecting pure state entanglement.

\begin{definition}{}{pure_multipartite_entanglement}
	Let $\bdens{\rho}_{P}$ be a pure multipartite density state.
	\begin{enumerate}
		\item A pure multipartite density state is \newword{$\lambda$-entangled} if it is not $\lambda$-factorizable;

		\item A pure multipartite density state is \newword{completely $L$-entangled} if it is not $\lambda$-factorizable for all partitions $\lambda$ of length $L$.
	\end{enumerate}
\end{definition}
Suspiciously missing is a term for the situation where there \textit{exists} a partition $\lambda$ of length $L$ such that $\bdens{\rho}_{P}$ is $\lambda$-entangled; one might wish to call this simply ``$L$-entangled", although we avoid officially defining this terminology as we do not have anything useful to say about it with the cohomological technology we have developed thus far.
The following remark gives some straightforward observations about our definitions.
\begin{remark}{}{}
	\begin{enumerate}
		\item $\lambda$-entangled $\Rightarrow \mu$-entangled for all refinements $\mu \leq \lambda$;

		\item completely $L$-entangled $\Rightarrow$ completely $M$-entangled for all $M \geq L$.
	\end{enumerate}
\end{remark}
The failure of complete $L$-entanglement is due to the existence of partitions of a length $\leq L$ for which the state factorizes.\footnote{The minimal level $L$ of a completely $L$ entangled state can be thought of as how much we need to coarse-grain (by choosing an arbitrary coarsening of the finest partition) until we are guaranteed that every subsystem is entangled with another.}
\begin{example}{}{}
	Let $\psi_{\sAB} \in \hilb_{\sA} \otimes \hilb_{\sB}$ be an entangled state, let $\bdens{\rho}_{\sABC}$ be the tripartite density state associated to $\psi_{\sAB} \otimes \psi_{\sC}$ for some $\psi_{\sC} \in \hilb_{\sC}$.  Then $\bdens{\rho}_{\sABC}$ is not completely 2-entangled but is completely 3-entangled. 
\end{example}

We can now paraphrase the statements of \S\ref{sec:coarse_factorizability} from an entangled perspective.
When searching for entanglement with respect to a particular partition, the following rephrasing of Thms.~\ref{thm:multipartite_factorizable_coarsening_vanishing} and \ref{thm:pure_multipartite_cohomology_factorizability} might be useful.
\begin{theorem}[label=thm:lambda_fact_bipartite_scan]{Testing For Multipartite Entanglement of Pure States}{}
	Let $\bdens{\rho}_{P}$ be a pure $N$-partite density state and $\lambda$ a partition of $P$ of length $\geq 2$.
	\begin{enumerate}
		\item If there exists a coarsening $\eta \geq \lambda$ such that $H^{k}[\cpx{G}_{\eta}(\bdens{\rho}_{P})]$ or $H^{k}[\cpx{E}_{\eta}(\bdens{\rho}_{P})] \neq 0$ for some $k < |\eta|-1$, then $\bdens{\rho}_{P}$ is $\lambda$-entangled.

		\item $\bdens{\rho}_{P}$ is $\lambda$-entangled if and only if there exists a length 2 coarsening $\eta \geq \lambda$ such that $H^{0}[\cpx{G}_{\eta}(\bdens{\rho}_{P})] \neq 0$.

		\item $\bdens{\rho}_{P}$ is $\lambda$-entangled if and only if there exists a length 2 coarsening $\eta \geq \lambda$ such that $H^{0}[\cpx{E}_{\eta} (\bdens{\rho}_{P})] \neq 0$.
	\end{enumerate}
\end{theorem}
Additionally---as already stated below Thm.~\ref{thm:pure_multipartite_cohomology_factorizability}---if we find a coarsening $\eta \geq \lambda$ for which one of $\cpx{G}_{\eta}(\bdens{\rho}_{P})$ (where $\bdens{\rho}_{P}$ a pure multipartite density state as in the theorem) has non-vanishing zeroth cohomology, then we are guaranteed that $\cpx{E}_{\eta}(\bdens{\rho}_{P})$ has non-vanishing zeroth cohomology (and vice-versa). 

Moreover, if we wish to search for entanglement with respect to all partitions of a certain length $L$ (``complete $L$-entanglement"), the following theorem--which is just a rephrasing of Thm.~\ref{thm:finest_partition_cohomology_factorizability}--tells us we should compute cohomology groups of complexes associated to the finest partition.
\begin{theorem}{}{pure_nonvanishing_complete_entanglement}
	Let $\bdens{\rho}_{P}$ be a pure multipartite density state. If either $H^{k}[\cpx{G}(\bdens{\rho}_{P})] \neq 0$ or $H^{k}[\cpx{E}(\bdens{\rho}_{P})] \neq 0$ for some $k$, then $\bdens{\rho}_{P}$ is completely $(k+2)$-entangled. 
\end{theorem}
Stated in terms of Poincar\'{e} polynomials, the theorem consists of the following statements for pure $\bdens{\rho}_{P}$ 
\begin{enumerate}
	\item $P_{\cpx{G}}(\bdens{\rho}_{P})$ is not divisible by $y^{k} \Longrightarrow \bdens{\rho}_{P}$ is completely $(k+2)$-entangled.

	\item $P_{\cpx{E}}(\bdens{\rho}_{P})$ is not divisible by $y^{k} \Longrightarrow \bdens{\rho}_{P}$ is completely $(k+2)$-entangled.
\end{enumerate}
Based on computational exploration, we conjecture the ``$\Leftarrow$" directions of the above are also true.
\begin{conjecture}[label=conj:complete_entanglement_pure_cohomologies]{}{}
	Suppose $\bdens{\rho}_{P}$ is a pure multipartite density state.
	\begin{enumerate}
		\item $P_{\cpx{G}}(\bdens{\rho}_{P})$ is not divisible by $y^{k} \Longleftrightarrow \bdens{\rho}_{P}$ is completely $(k+2)$-entangled.
		
		\item $P_{\cpx{E}}(\bdens{\rho}_{P})$ is not divisible by $y^{k} \Longleftrightarrow \bdens{\rho}_{P}$ is completely $(k+2)$-entangled.
	\end{enumerate}
\end{conjecture}
For pure bipartite density states statement 1 of the conjecture is true as demonstrated by Thm.~\ref{thm:pure_GNS_cohomology} and statement 2 is true via Thm.~\ref{thm:pure_commutant_cohomology}.
A direct way of approaching the conjecture is to demonstrate the existence of a $k$-cocycle for either $\cpx{G}(\bdens{\rho}_{P})$ or $\cpx{E}(\bdens{\rho}_{P})$ that passes to something non-trivial in cohomology whenever $\bdens{\rho}_{P}$ is completely $(k+2)$-entangled.

\section{Categorification of Mutual Information \label{sec:categorification}}
In this section we touch upon one of the more intriguing claims in the introduction: that the Euler characteristics of the cochain complexes we produced above are related to (multipartite) mutual informations.
This claim should be considered a part of a more general art of ``categorification", instead of defining this term we summarize it in our very specialized context with the following diagram.

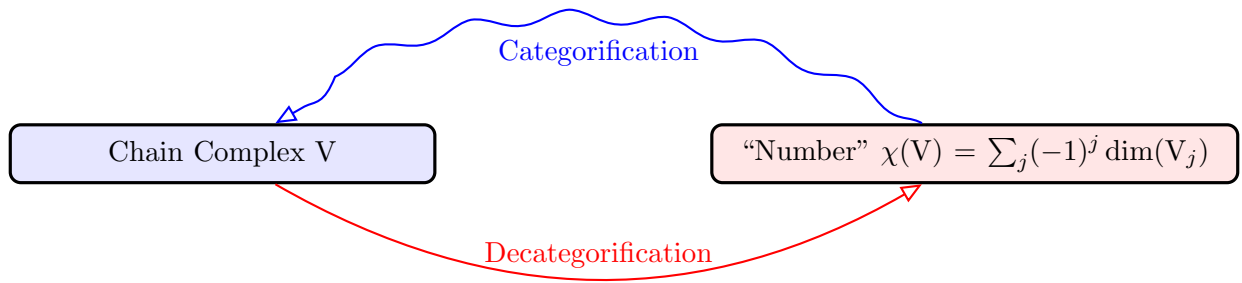
\begin{center}
	\begin{tikzpicture}
		\tikzstyle{block} = [rectangle, draw=blue, thick, fill=blue!10,
		text width=14em, text centered, rounded corners, minimum height=2em]

		\tikzstyle{blockred} = [rectangle, draw=blue, thick, fill=red!10,
		text width=17.5em, text centered, rounded corners, minimum height=2em]

		\node at (-5,0)[block,draw=black,very thick] (complex) {Chain Complex $\cpx{V}$};

		\node at (5,0)[blockred,draw=black,very thick] (dimension) {``Number" $\chi(\cpx{V}) = \sum_{j} (-1)^{j} \dim(\cpx{V}_{j})$};

		\draw[->, >=open triangle 45, decorate, decoration=snake, segment length=13mm, thick, color=blue] (dimension) to [bend right = 30] node[sloped,below,midway]{\parbox{18em}{\centering Categorification}} (complex) ;

		\draw[->, >=open triangle 45, decorate, thick, color=red] (complex) to [bend right = 30] node[sloped,above,midway]{\parbox{10em}{\centering Decategorification}} (dimension) ;
	\end{tikzpicture}
\end{center}
The left and right hand sides of this diagram should be thought of as invariants associated to some mathematical object of interest.\footnote{Here ``invariant" means that, up to an appropriate notion of isomorphism of our mathematical object, the chain complex changes canonically by e.g.\ chain (or quasi) isomorphism, and the number remains constant.}
As an example: the Euler characteristic of a topological space is a numerical quantity invariant under homotopy (or representation of that topological space via combinatorial data); its categorification can be thought of as a singular (or simplicial) chain complex associated to that topological space: another ``invariant" in the sense that it transforms via canonical (i.e.\ functorial) isomorphism under homotopy (or change in combinatorial representation).
The Euler characteristic is more easily calculable, but chain complexes and their homologies provide more information and can help distinguish between two topological spaces with the same Euler characteristic.

As another example that should please physicists: in super-quantum mechanics the $\mathbb{Z}_{2}$ graded Hilbert space of ground states (along with a supercharge acting as a differential) can be thought of as a categorification of the Witten index: both ``invariants" of some associated super quantum mechanics.
Famously, the Jones polynomial associated to a knot (an invariant under ambient isotopy) can be realized as the Euler characteristic of a cochain complex discovered by Khovanov.
While the process of decategorification of an invariant is usually straightforward (in our case, it is simply the passage to an Euler characteristic), the process of categorifying is akin to quantization in the sense that---beyond a few specialized contexts---it is more of an art than a functorial process.

For the purposes of this paper, our utopian dream is the categorification of multipartite mutual information: a numerical quantity associated to a multipartite density state, and invariant under local invertible transformations (hence, also invariant under local unitary transformations).
Let us take a brief diversion to motivate and define this quantity.

\subsection{Multipartite Mutual Information}
Before introducing multipartite mutual information, we should first recall the definition of the von Neumann entropy associated to a density state.
\begin{definition}{}{}
	The \newword{von Neumann entropy} of a density state $\dens{\rho} \in \Dens(\hilb)$ is the $\mathbb{R}_{\geq 0}$ valued quantity: 
	\begin{align*}
		\newmath{S^{\mathrm{vN}}(\dens{\rho})} := -\Tr[\dens{\rho} \log(\dens{\rho})].
	\end{align*}
\end{definition}
The von Neumann entropy is a reasonable unitary invariant associated to a unipartite density state: application of a unitary (or invertible linear) transformation leaves the von Neumann entropy undisturbed as von Neumann entropy of a density state only depends on the (unordered) set of eigenvalues.
The von Neumann entropy of a unipartite density state is, moreover, always positive and vanishes if and only if the state is pure. 

Handed a bipartite density state, we wish to ask for a generalization of the entropy to a numerical quantity that is invariant under local unitary (or invertible) transformations.
We begin by noticing that, given a bipartite density state $\bdens{\rho}_{\sAB}$, each of the three quantities $S(\dens{\rho}_{\sA}),\, S(\dens{\rho}_{\sB}),\, S(\dens{\rho}_{\sAB})$ is invariant under local invertible transformations.  A particularly famous linear combination of these is the (bipartite) mutual information:
\begin{align}
	I(\bdens{\rho}_{\sAB}) := S^{\mathrm{vN}}(\dens{\rho}_{\sA}) + S^{\mathrm{vN}}(\dens{\rho}_{\sB}) - S^{\mathrm{vN}}(\dens{\rho}_{\sAB}).
	\label{eq:bipartite_mutual_info_def}
\end{align}
One possible motivation for this quantity is its ability to be expressed as a relative entropy; we recall the definition of relative entropy for density states below.
\begin{definition}{}{}
	Let $\dens{\rho}, \dens{\varphi} \in \Dens(\hilb)$ be density states such that $\supp_{\dens{\rho}} \leq \supp_{\dens{\varphi}}$.
	Then the relative entropy of $\dens{\rho}$ with respect to $\dens{\varphi}$, is given by
	\begin{align*}
		\relent{\dens{\rho}}{\dens{\varphi}} := \Tr \left[ \dens{\rho} \left(\log(\dens{\rho}) - \log(\dens{\varphi}) \right) \right].
	\end{align*}
\end{definition}
One can show that the relative entropy $\relent{\dens{\rho}}{\dens{\varphi}}$ of density states is always non-negative and vanishes if and only if $\dens{\rho} = \dens{\varphi}$ (c.f.\ \cite[\S 1]{petz:quantum_ent}). 
In this sense, relative entropy provides a useful measure of ``how close" two density states are.
It is a straightforward exercise to show that:
\begin{align*}
	I(\bdens{\rho}_{\sAB}) = \relent*{\dens{\rho}_{\sAB}}{\dens{\rho}_{\sA} \otimes \dens{\rho}_{\sB}}.
\end{align*}
In other words, the bipartite mutual information is a measure of how close a density state is to being factorizable.

Given the discussion above, one natural way of generalizing the bipartite mutual information to multipartite density states $\bdens{\rho}_{P}$ with $|P| \geq 3$ is to consider the $\mathbb{R}_{\geq 0}$ quantity: 
\begin{equation}
	\begin{aligned}
		C(\bdens{\rho}_{P}) & := \relent[\bigg]{\dens{\rho}_{P}}{\bigotimes_{p \in P} \dens{\rho}_{p}}\\
                            & = \sum_{p \in P} S^{\mathrm{vN}}(\dens{\rho}_{\{p\}}) - S^{\mathrm{vN}}(\dens{\rho}_{P}).
	\end{aligned}
	\label{eq:total_correlation}
\end{equation}
Once again, because the entropies of all reduced density states are invariant under local invertible transformations, it follows that $C(\bdens{\rho}_{P})$ is invariant under local invertible transformations.
Moreover, this quantity provides an excellent measure of how far a multipartite density state is from being fully factorizable: it is non-negative and vanishing only on fully factorizable density states.
The classical analogue of this quantity (given by replacing density states with probability measures and quantum relative entropies with classical relative entropies) is known as the \textit{total correlation}.

An alternative way of generalizing the mutual information---which will be relevant to our applications---is to take the inclusion-exclusion sum appearance of \eqref{eq:bipartite_mutual_info_def} seriously and adopt the following definition.

\begin{definition}[label=def:mutual_info_def]{}{}
	The \newword{mutual information} of a multipartite density state $\bdens{\rho}_{P}$ is 
	\begin{align}
		\newmath{I(\bdens{\rho}_{P})} := \sum_{\emptyset \subseteq T \subseteq P} (-1)^{|T|-1} S^{\mathrm{vN}}(\dens{\rho}_{T}) \in \mathbb{R}.
		\label{eq:multipartite_mutual_info_def}
	\end{align}
\end{definition}
We will come back to a motivation for this definition in terms of a recursion relation in \S\ref{sec:inclusion_exclusion}.

The density state assigned to the empty set is the identity operator on $\mathbb{C}$---the one-dimensional Hilbert space assigned to $\emptyset$---the von Neumann entropy of this state is zero; so the empty set contribution to the sum of \eqref{eq:multipartite_mutual_info_def} vanishes.  In particular, we note that for a unipartite density state, we have
\begin{align*}
	I(\bdens{\rho}_{*}) &= S^{\mathrm{vN}}(\dens{\rho}_{P}),
\end{align*}
whereas, for a bipartite density state we recover \eqref{eq:bipartite_mutual_info_def}.
Unlike with bipartite mutual information, tripartite mutual information is not necessarily always non-negative (or non-positive).\footnote{However, one can argue that holographic states---state for which mutual informations can be computed by Ryu-Takayanagi formulae---the tripartite mutual information is non-positive \cite{Hayden:2011ag}.
In other words, the information shared between $A$ and the combined system $(BC)$ is at least as big as the sum of the information shared between $\sA$ and $\sB$ and the information shared between $\sA$ and $\sB$.}
One of the earliest references to the tripartite version of the classical multipartite mutual information (replacing von Neumann entropy with the entropy of a probability measure) appeared in \cite{McGill:1954}, with a full generalization to multipartite systems presented in \cite{han:1980}.
\begin{remark}[label=rmk:multipartite_mutual_info_sign]{}{}
	As in \cite{han:1980}, some authors define multipartite mutual information to include an overall global sign that takes into account the number of primitive subsystems, i.e.\ we could define: 
	\begin{align*}
		\newmath{I^{-}(\bdens{\rho}_{P})} := (-1)^{|P|} I(\bdens{\rho}_{P}). 
	\end{align*}
	Such a sign appears naturally if we think of entropy as defining an element of a module for the incidence algebra of subsets of $P$.  In this case $I^{-}(\bdens{\rho}_{P})$ arises as the right action of the M\"{o}bius-mu function of the incidence algebra on the element defined by entropy.  We outline a few details of this remark in App.~\ref{app:incidence_alg_mutual_info}.
\end{remark}

\subsection{Factorizability and Mutual Information}
We will not provide a detailed exposition of the properties or interpretation of multipartite mutual informations; such an exposition deserves a dedicated short paper to itself.
Nevertheless, it is worth giving a few remarks on how the mutual informations can be used to answer questions about factorizability of multipartite density states.
As discussed above, the mutual information associated to a bipartite density state is the total correlation (the relative entropy of a bipartite density state with respect to the tensor product of its reduced states); so the following is immediate via properties of relative entropies.
\begin{proposition}[label=prop:bipartite_mi_vanish_factorizable]{}{}
	Let $\bdens{\rho}_{\sAB}$ be a bipartite density state, then $\bdens{\rho}_{\sAB}$ is factorizable if and only if $I(\bdens{\rho}_{\sAB}) = 0$.
\end{proposition}
For higher mutual informations, the following proposition provides a generalization of the ``only if" direction:
\begin{proposition}[label=prop:lambda_fact_then_mi_vanish]{}{}
	If $\bdens{\rho}_{P}$ is $\lambda$-factorizable for some partition $\lambda$ of length $\geq 2$, then $I(\bdens{\rho}_{P}) = 0$.
\end{proposition}
\begin{proof}
	This follows by a straightforward computation that makes use of the property $S^{\mathrm{vN}}(\dens{\rho} \otimes \dens{\varphi}) = S^{\mathrm{vN}}(\dens{\rho}) + S^{\mathrm{vN}}(\dens{\varphi})$ for any density states $\dens{\rho} \in \Dens(\hilb)$ and $\dens{\varphi} \in \Dens(\mathcal{K})$ along with the property that $\sum_{\emptyset \subset T \subseteq P} (-1)^{|P|} = (1-1)^{|P|} = 0$.
\end{proof}
\noindent Prop.~\ref{prop:lambda_fact_then_mi_vanish} justifies the name ``mutual information": the fact that mutual information vanishes if any subsystem ``decouples" is an indication that mutual information provides a measure of the ``shared (or mutual) information" among \textit{all} subsystems.

One can ask about the converse direction to Prop.~\ref{prop:lambda_fact_then_mi_vanish}: does the vanishing of certain mutual informations necessarily imply factorizability?
Indeed, to verify $\lambda$-factorizability for some partition $\lambda$ of $P$ the vanishing of all mutual informations associated to the $\frac{1}{2}|\lambda|(|\lambda|-1)$ bipartite coarsenings of $\lambda$ suffices.\footnote{Recall that a $\lambda$-coarsening $\lambda[\bdens{\rho}_{P}]$ of a multipartite density state $\bdens{\rho}_{P}$---formally defined in Def.~\ref{def:coarsenings}---is simply a reorganization of the set of tensor factors $P$ according to a partition $\lambda$ of $P$.
The result is a $|\lambda|$-partite density state denoted $\lambda[\bdens{\rho}_{P}]$.}
\begin{proposition}[label=prop:lambda_fact_mi_vanish_bipartite]{}{}
	$\bdens{\rho}_{P}$ is $\lambda$-factorizable if and only if $I(\eta[\bdens{\rho}_{P}])=0$ for all length $2$ partitions $\eta \leq \lambda$.
\end{proposition}
\begin{proof}
	This follows from Prop.~\ref{prop:bipartite_mi_vanish_factorizable} and the fact that a multipartite density state $\bdens{\rho}_{P}$ is $\lambda$-factorizable if and only if it is factorizable with respect to any length 2 coarsening. 
\end{proof}
\noindent As a corollary of Prop.~\ref{prop:lambda_fact_mi_vanish_bipartite} and Prop.~\ref{prop:lambda_fact_then_mi_vanish}, we have
\begin{corollary}{}{}
	A density state $\bdens{\rho}_{P}$ is fully factorizable if and only if $I(\lambda[\bdens{\rho}_{P}]) = 0$ for all partitions $\lambda$ of $P$. 
\end{corollary}
which is a weaker statement than Prop.~\ref{prop:lambda_fact_mi_vanish_bipartite} (in particular, Prop.~\ref{prop:lambda_fact_mi_vanish_bipartite} suggests the computation of  $\frac{1}{2}|P|(|P|-1)$ mutual informations rather than $2^{|P|}$ mutual informations to verify full factorizability).

Of course, total correlation (defined in \eqref{eq:total_correlation}) also helps answer questions about factorizability: a state is $\lambda$-factorizable if and only if $C(\lambda[\bdens{\rho}_{P}])=0$. 
Lovers of total correlation should be pleased to note that we can recover it from the computation of mutual informations of reductions (rather than coarsenings) of multipartite density states.
\begin{proposition}{}{}
	The total correlation can be expressed in terms of the mutual informations associated to all reduced multipartite density states on subsets including at least 2 primitive elements:
	\begin{align*}
		C(\bdens{\rho}_{P}) &= \sum_{\{T \subseteq P: |T| \geq 2 \}} (-1)^{|T|} I(\bdens{\rho}_{T})\\
							&= \sum_{\{T \subseteq P: |T| \geq 2 \}} I^{-}(\bdens{\rho}_{T}).
	\end{align*}
\end{proposition}
\begin{proof}
	This follows from general properties of inclusion-exclusion formulae (or M\"{o}bius inversion formulae as described in App.~\ref{app:incidence_alg_mutual_info}).
	Namely suppose $P$ is a set and $f: \mathrm{Power}(P) \rightarrow \mathbb{R}$ is a function on the power set of $P$.
	Define
	\begin{align*}
		\twid{f}: \mathrm{Power}(P) & \longrightarrow \mathbb{R}\\
		T        & \longmapsto \sum_{\emptyset \subseteq V \subseteq T} (-1)^{|P| - |T|} f(V).
	\end{align*}
	Then one can verify that
	\begin{align*}
		f(T) = \sum_{\emptyset \subseteq V \subseteq T} \twid{f}(T). 
	\end{align*}
	The proposition follows by taking the function $f$ to be given by $T \mapsto S^{\mathrm{vN}}(\dens{\rho}_{T})$.
\end{proof}
Of course, one can also expression the total correlation of a coarsening of a multipartite density state in terms of a sum of mutual informations over the reductions of the coarsening.

\subsection{Inclusion-Exclusion, Recursion \label{sec:inclusion_exclusion}}
To shed more light on Def.~\ref{def:mutual_info_def}, we reverse engineer our definition: noting that the alternating sums of the sort present in \eqref{eq:multipartite_mutual_info_def} are the result of some deeper inclusion-exclusion principle.
Indeed, note that the tripartite mutual information of a tripartite state $\bdens{\rho}_{\sABC}$ can be expressed as a combination of bipartite mutual informations:
\begin{align}
	I(\bdens{\rho}_{\sABC}) = I(\bdens{\rho}_{\sAB}) + I(\bdens{\rho}_{\sAC}) - I\left(\bdens{\rho}_{\sA (\sB \sC)} \right),
	\label{eq:tripartite_recursion}
\end{align}
where $\dens{\rho}_{\sA (\sB \sC)}$ is the bipartite coarsening\footnote{Using the notation of Def.~\ref{def:coarsenings} this is the bipartite density state $(\sA,\sBC)[\bdens{\rho}_{\sABC}]$, for the length 2 partition $(\sA, \sB \sC)$  consisting of the subsystems $\sA$ and $\sB \sC$.}
$((\hilb_{\sA}, \hilb_{\sBC}), \dens{\rho}_{\sABC})$.
In general, the $|P|$-partite mutual information of a density state $\bdens{\rho}_{P}$ can be computed from a combination of three $(|P|-1)$-partite mutual informations.
Indeed, given an ordered set $P$ of size $|P| \geq 2$, for each pair of integers $(i,j)$ such that $0 \leq i < j \leq |P|-1$ define the length $(|P|-1)$-partition $\lambda_{ij}$ as the partition obtained from the finest partition of $P$ after merging the $(i+1)$th and $(j+1)$th elements of $P$ together.\footnote{We can define $\lambda_{ij}$ more formally: first recall that $\partial_{k}$ is the operation on ordered sets that eliminates the $(k+1)$th element (C.f.\ Def.~\ref{def:ordered_set_ops} in \S\ref{sec:multipartite_cochain_complexes}). Then $\lambda_{ij}$ is the partition whose components are given by the one-element subsets/primitive subsystems of $\partial_{i} \partial_{j} P$, along with the subset $\{p_{i}, p_{j} \}$.}
Then we have
\begin{align}
	I(\bdens{\rho}_{P})  = I(\bdens{\rho}_{\partial_{i}P}) + I(\bdens{\rho}_{\partial_{j} P}) - I\left( \lambda_{ij}[\bdens{\rho}_{P}] \right),
	\label{eq:mutual_info_recursion}
\end{align}
where $\lambda_{ij}[\bdens{\rho}_{P}]$ is the $(|P|-1)$-partite density state formally defined in Def.~\ref{def:coarsenings}.
That is, if we have a fixed density state $\bdens{\rho}_{P}$, and choose two subsystems $p, q \in P$, with $p \neq q$, then the total mutual information (the information shared by all subsets $P$) is the sum of:
\begin{enumerate}
	\item $I(\bdens{\rho}_{P \backslash \{p\}})$: The mutual information associated to the primitive subsystems of $P \backslash \{p \}$,

	\item $I(\bdens{\rho}_{P \backslash \{q\}})$: The mutual information associated to the primitive subsystems of $P \backslash \{q \}$,
\end{enumerate}	
and modulo any ``double-counting" encapsulated by the mutual information between elements of $\left(P \backslash \{p \} \right) \cap \left(P \backslash \{q \} \right) =  P \backslash \{p,q \}$ and the combined system $\{p,q\}$.

This is an inclusion-exclusion relation \eqref{eq:mutual_info_recursion} should be familiar from a combinatorial perspective: if we have a finite set $X$ with subsets $A, B \subseteq X$ such that $X = A \cup B$, then the cardinality of $X$ can be computed as: 
\begin{align}
	|X | = |A | + |B | - |A \cap B | 
	\label{eq:inclusion_exclusion_bipartite_cardinality}
\end{align}
the subtracted term $|A \cap B |$ being a correction for double counting.  
More generally, suppose $X$ is a finite set covered by a finite collection of subsets $(A_{p})_{p \in P}$ where $P$ is some (finite) indexing set, then by repeated application of \eqref{eq:inclusion_exclusion_bipartite_cardinality} we obtain: 
\begin{align*}
	| X | = \sum_{\emptyset \subseteq T \subseteq P} (-1)^{|T|-1} \left| \bigcap_{t \in T} A_{t} \right|. 
\end{align*}
In a similar manner, applying \eqref{eq:mutual_info_recursion}, we can reduce the computation of mutual information to an alternating sum over mutual informations of unipartite density states:
\begin{align*}
	I(\bdens{\rho}_{P}) = \sum_{\emptyset \subseteq T \subseteq P} (-1)^{|T|-1} I\left( \lambda_{\mathrm{trivial}}^{T}[\bdens{\rho}_{T}] \right)
\end{align*}
where $\lambda_{\mathrm{trivial}}^{T} = (T)$ is the trivial (length 1) partition of $T$; so $\lambda_{\mathrm{trivial}}[\bdens{\rho}_{T}]$ is the unipartite state $((\hilb_{T}),\dens{\rho}_{T})$. 
If we specify that the mutual information associated to any unipartite density state $\bdens{\rho}_{*} = ((\hilb_{*}), \dens{\rho}_{*})$ is the von Neumann entropy: 
\begin{align*}
	I(\bdens{\rho}_{*}) := S^{\mathrm{vN}}(\bdens{\rho}_{*}),
\end{align*}
then we can obtain the closed expression \eqref{eq:multipartite_mutual_info_def}.

\begin{remark}{}{}
	To aid in the suspicion that mutual information might behave as an Euler characteristic, we compare \eqref{eq:mutual_info_recursion} to identities for Euler characteristics under ``gluing".\footnote{For the experts: (homotopy) pushout diagrams.}
	For instance: recall that if $X$ is a topological space, covered by open sets $A$ and $B$ (satisfying sufficiently nice properties), then the Euler characteristic of $X = A \cup B$ obeys:
	\begin{align*}
		\chi(A \cup B) = \chi(A) + \chi(B) - \chi(A \cap B).
	\end{align*}
	$f: A \rightarrow X$ and $g: A \rightarrow X$.
	More relevant to our applications here, are the gluing identities for \v{C}ech cohomologies of topological spaces equipped with a presheaf.
	In future work this will be made more precise using more sophisticated techniques beyond the scope of this paper.
\end{remark}

The following remark is relevant to an interpretation of the recursion \eqref{eq:mutual_info_recursion} in the spirit of the approach taken in \cite{bb:homent,vigneaux}. 
\begin{remark}{}{}
	Practitioners of homological algebra might suspect \eqref{eq:mutual_info_recursion} as a possible manifestation of a coboundary relation: perhaps multipartite mutual information should be thought of as an element of a cochain complex\footnote{Or an $A_{\infty}$-algebra/module.} with higher mutual informations obtainable by application of a coboundary map.
	This is the perspective of \cite{bb:homent,vigneaux}, where multipartite mutual informations live inside of cochain complexes consisting of of functions on spaces of density states/measures. 
	The perspective in this paper, however, is to think of the density state as fixed and think of mutual information as some associated numerical invariant (with respect to local invertible transformations).
	One might be able to reconcile our perspective here and that of \cite{bb:homent} using the language of obstruction theory.  
	One might suspect the chain complexes of \cite{bb:homent} are computing the cohomology of some classifying space associated to the moduli of multipartite density states/measures on some fixed set of subsystems.
	Asking if a density state/measure factorizes in a particular way should be phrased as an obstruction problem, with multipartite mutual informations playing the role of obstruction classes.
\end{remark}

\subsection{The State Index}
As mentioned in the introduction, we can manipulate the multipartite mutual information to appear as: 
\begin{align*}
	I(\bdens{\rho}_{P}) = \sum_{k = -1}^{|P|-1} (-1)^{k} \left[ \sum_{\{T \subseteq P: |T| = k + 1\} } S^{\mathrm{vN}}(\rho_{T}) \right]
\end{align*}
so if $I(\bdens{\rho}_{P})$ were the Euler characteristic of some (co)chain complex $\cpx{C}(\bdens{\rho}_{P})$, then it is natural to suspect that the expressions in the square brackets should be identifiable with the ``dimensions" of the components of this complex: 
\begin{align*}
	\text{{\fontsize{20.74}{0}\selectfont{``}} $\dim \cpx{C}^{k}(\bdens{\rho}_{P}) = 	\sum_{\{T \subseteq P: |T| = k + 1\} } S^{\mathrm{vN}}(\rho_{T})$ {\fontsize{20.74}{0}\selectfont{''}}}.
\end{align*}

Multipartite mutual information does not have a chance of being directly realized as an Euler characteristic of a chain complex of vector spaces: the most obvious reason being that the Euler characteristic of a chain complex of vector spaces is an integer valued quantity, whereas multipartite mutual information is valued in $\mathbb{R}$.  However, one might ask if we can realize multipartite mutual information as a Lefschetz index (an alternating sum of traces of the induced endomorphism on cohomology from some chain endomorphism), or perhaps by working with something more sophisticated than vector spaces---allowing for a  generalized notion of ``dimension".  
These ideas also fail for a more sophisticated reason: if we can associate chain complexes to multipartite density states in any meaningful way, one would expect that tensor products of multipartite density states pass to tensor products of complexes.
If such a property is true, then the associated Euler characteristic should take tensor products of complexes to products of numbers (and direct sums/Cartesian products of complexes to sums of numbers).
As a result we expect:
\begin{align*}
	I(\bdens{\rho}_{P} \otimes \bdens{\varphi}_{Q}) \overset{\text{\large ?}}{=} I(\bdens{\rho}_{P}) I(\bdens{\varphi}_{Q}),
\end{align*}
However, by Prop.~\ref{prop:lambda_fact_then_mi_vanish}: given any multipartite density states $\bdens{\rho}_{P}$ and $\bdens{\varphi}_{Q}$, we have:\footnote{Note that \eqref{eq:mutual_inf_vanish_tensor} also precludes the mutual information from taking tensor products of vector spaces to sums of mutual informations.} 
\begin{align}
	I(\bdens{\rho}_{P} \otimes \bdens{\varphi}_{Q}) = 0.
	\label{eq:mutual_inf_vanish_tensor}
\end{align}
even if the individual mutual informations $I(\bdens{\rho}_{P})$ and $I(\bdens{\varphi}_{Q})$ are non-vanishing.
Although \eqref{eq:mutual_inf_vanish_tensor} results from a useful property for detecting factorizability, it is now somewhat discouraging if our goal was to na\"{i}vely realize mutual information as an Euler characteristic.
With this in mind, we can ask what kind of quantities \textit{do} have the expected properties of Euler characteristics, and if we can recover the mutual information from any of them.

\subsubsection{The Search for an Euler Characteristic \label{sec:search_for_ec}}
We begin with the diagram below.

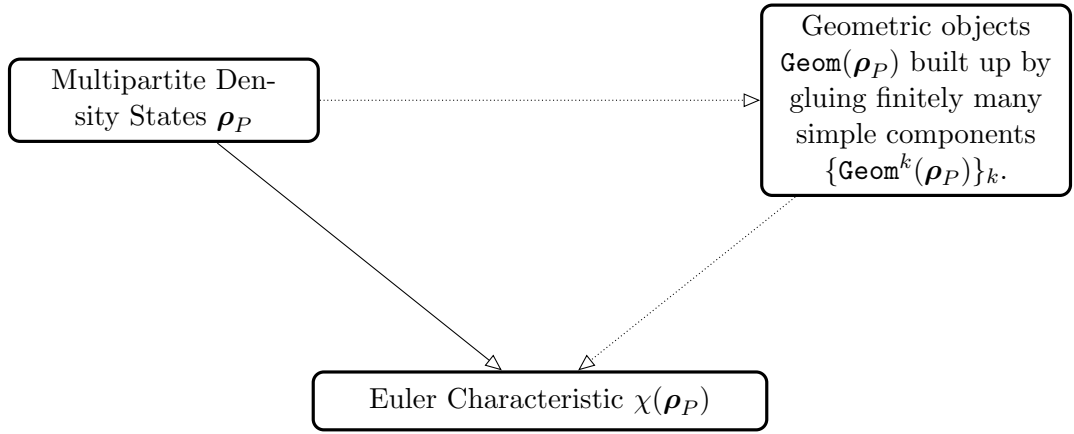
\begin{center}
	\begin{tikzpicture}
		\tikzstyle{block} = [rectangle, thick, text width=10em, text centered, rounded corners, minimum height=2em];
		\tikzstyle{bigblock} = [rectangle, thick, text width=15em, text centered, rounded corners, minimum height=2em];

		\node at (-5,0)[block, draw=black, very thick] (multipartite) {Multipartite Density States $\bdens{\rho}_{P}$}; 

		\node at (5,0)[block, draw=black, very thick] (geometricobject) {Geometric objects $\mathtt{Geom}(\bdens{\rho}_{P})$ built up by gluing finitely many simple components $\{\mathtt{Geom}^{k}(\bdens{\rho}_{P})\}_{k}$.}; 

		\node at (0,-4)[bigblock, draw=black, very thick] (eulerchar) {Euler Characteristic $\chi(\bdens{\rho}_{P})$}; 

		\draw[->, >=open triangle 45, densely dotted] (multipartite) to (geometricobject);

		\draw[->, >=open triangle 45, densely dotted] (geometricobject) to (eulerchar);

		\draw[->, >=open triangle 45] (multipartite) to (eulerchar);
	\end{tikzpicture}
\end{center}
That is, we begin by supposing that to any multipartite density state $\bdens{\rho}_{P}$ there is some sort of geometric object $\mathtt{Geom}(\bdens{\rho}_{P})$ built up by gluing together (finitely many) simple components $\{\mathtt{Geom}^{k}(\bdens{\rho}_{P}) \}_{k}$.
The multipartite cochain complexes $\cpx{E}(\bdens{\rho}_{P})$ and $\cpx{G}(\bdens{\rho}_{P})$ of the previous sections define concrete realizations of what $\mathtt{Geom}^{k}(\bdens{\rho}_{P})$ could be: the simple components $\mathtt{Geom}^{k}(\bdens{\rho}_{P})$ are the non-trivial components of the cochain complexes and the gluing is given by the coboundary maps. 
However, we will take the perspective that these complexes might be shadows of a more sophisticated geometric object: e.g.\ (as suggested above) a chain complex whose components are valued in a category of objects that are fancier than vector spaces.\footnote{In future work we will produce $\mathtt{Geom}(\bdens{\rho}_{P})$ as a simplicial object in an appropriate ``category of states".} 
Given such a hypothetical geometric object, we can try to search for candidates of numerical ``invariants" that behave like Euler characteristics $\chi(\bdens{\rho}_{P}) = \chi \left[\mathtt{Geom}(\bdens{\rho}_{P}) \right]$.

Even without understanding $\mathtt{Geom}(\bdens{\rho}_{P})$, if we require that any good Euler characteristic should preserve particular geometric properties, we can constrain the form of the Euler characteristic and obtain concrete formulae directly terms of the data of the multipartite density state.
First, whatever $\chi(\bdens{\rho}_{P})$ might be, it should play nicely with ``gluing".
To make this statement concrete, we appeal to the inclusion-exclusion principle \eqref{eq:mutual_info_recursion} satisfied by mutual information and require that $\chi(\bdens{\rho}_{P})$ satisfy the same relation: 
\begin{align*}
	\chi(\bdens{\rho}_{P}) = \chi(\bdens{\rho}_{\partial_{i}P}) + \chi(\bdens{\rho}_{\partial_{j}P}) - \chi(\lambda_{ij}[\bdens{\rho}_{P}]).
\end{align*}
This relation intuitively meshes with the vague idea that $\bdens{\rho}_{P}$ can be built by gluing together the two density states  $\bdens{\rho}_{\partial_{i}P}$ and $\bdens{\rho}_{\partial_{j}P}$ while keeping track of overlapping data encoded by the multipartite density state $\lambda_{ij}[\bdens{\rho}_{P}]$.

As with mutual information this inclusion-exclusion relation highly constrains the form of $\chi(\bdens{\rho}_{P})$: we can reduce the computations of $\chi(\bdens{\rho}_{P})$ to the computation of Euler characteristics of unipartite states:
\begin{align*}
	\chi(\bdens{\rho}_{P}) = \sum_{\emptyset \subseteq T \subseteq P} (-1)^{|T|-1} \chi( \lambda_{\mathrm{trivial}}[ \bdens{\rho}_{T} ] ) 
\end{align*}
where, as before, $\lambda_{\mathrm{trivial}}[\bdens{\rho}_{T}]$ is the unipartite state $((\hilb_{T}),\dens{\rho}_{T})$. 
In other words, $\chi(\lambda_{\mathrm{trivial}}[ \bdens{\rho}_{T} ])$ is a function of the pair of data $\hilb_{T}$ and a density state $\dens{\rho_{T}}$.

Secondly we require that $\chi(\bdens{\rho}_{P})$ is an invariant under suitable notion of isomorphism of multipartite density states (which should pass to isomorphisms of associated geometric objects).
In our case this means it should only depend on $\bdens{\rho}_{P}$ up to local invertible transformations.\footnote{it was also be natural to impose invariance under local unitary transformations, however, the larger class of local invertible transformations are easier to work with.} 
This requirement is equivalent to the requiring that $\chi(\lambda_{\mathrm{trivial}}[ \bdens{\rho}_{T} ])$ only depends on the unipartite density state up to invertible transformations (equivalently it only depends on the dimension of $\hilb_{T}$ and the conjugacy class of $\dens{\rho}_{T}$). 
With this in mind, we provide the following definition.
\begin{definition}{}{}
	\begin{enumerate}
		\item $\newmath{\catname{FinDens}}$ is the collection\footnote{We are explicitly avoiding the use of the term ``set" as the collection of all such pairs of Hilbert spaces and density states is not a set in any formal sense of the word.  However, in truth, we only care about isomorphism classes of such pairs, which does form a set.}
			of pairs $(\hilb,\dens{\rho})$ where $\hilb$ is a finite dimensional Hilbert space $\dens{\rho} \in \Dens(\hilb)$ (this is equivalently the collection of unipartite density states).
			We will say two pairs $(\hilb, \dens{\rho}), (\mathcal{K}, \varphi) \in \catname{FinDens}$ are invertibly isomorphic if there exists an invertible map $l: \hilb \rightarrow \mathcal{K}$ such that $\dens{\varphi} = l \dens{\rho} l^{-1}$.

		\item $\newmath{\mathrm{FinDens}}$ is the set of equivalence classes of $\catname{FinDens}$ under invertible isomorphism.\footnote{Because one can recover any pair $(\hilb, \dens{\rho}) \in \catname{FinDens}$ up to invertible isomorphism by the spectrum of $\dens{\rho}$ (thought of as an unordered multiset), there is a bijective correspondence between $\mathrm{FinDens}$ and $\bigsqcup_{n \geq 0} \mathbb{R}_{\geq 0}^{n}$---in particular, not only is $\mathrm{FinDens}$ a set, but we have the option of equipping it with a topology.}
	\end{enumerate}
\end{definition}
Now suppose $\mathcal{R}$ is some abelian group and we are supplied with a function 
\begin{align*}
	\mathpzc{G}: \mathrm{FinDens} \rightarrow \mathcal{R},
\end{align*}
then the above discussion motivates us to define: 
\begin{align}
	\newmath{\chi_{\mathpzc{G}}(\bdens{\rho}_{P})} := \sum_{\emptyset \subseteq T \subseteq P} (-1)^{|T|-1} \mathpzc{G}[ (\hilb_{T}, \dens{\rho}_{T} ) ] \in \mathcal{R}.
	\label{eq:ec_inclusion_exclusion_reduction}
\end{align}
This is precisely the hypothetical Euler characteristic defined by the unipartite assignments
\begin{align*}
	\chi_{\mathpzc{G}}(\lambda_{\mathrm{trivial}}[ \bdens{\rho}_{T} ]) &= \mathpzc{G}[(\hilb_{T},\dens{\rho}_{T})]. 
\end{align*}

Next we might wish to let $\mathcal{R}$ be a ring, and constrain the possible functions $\mathpzc{G}: \mathrm{FinDens} \rightarrow \mathcal{R}$ in such a way that the associated Euler characteristic $\chi_{\mathpzc{G}}$ takes tensor products of multipartite density states to products of Euler characteristics. 
But rather than doing this, we take a brief diversion and continue to use our vague intuition that $\chi_{\mathpzc{G}}(\bdens{\rho}_{P})$ is associated to some geometric object; the constraint that $\chi_{\mathpzc{G}}$ takes tensor products to products in $\mathcal{R}$ will emerge naturally.\footnote{Actually, as we will see, this is only true up to a sign.}

First reorder \eqref{eq:ec_inclusion_exclusion_reduction}:
\begin{align}
	\chi_{\mathpzc{G}}(\bdens{\rho}_{P}) = \sum_{k = -1}^{|P|-1} (-1)^{k} \left[\sum_{|T| = k+1} \mathpzc{G}[(\hilb_{T}, \dens{\rho}_{T})] \right].
	\label{eq:hypothetical_euler_char_1}
\end{align}
Our proposal is that the quantity in square brackets should be equivalent to a ``dimension" of the level $k$ component $\mathtt{Geom}^{k}(\bdens{\rho}_{P})$ of the geometric object associated to $\bdens{\rho}_{P}$.
Rather than elaborating on a formal definition of ``dimension", we lead by example.
\begin{enumerate}
	\item \textbf{Finite sets:} A good notion of dimension is the cardinality.  Note that if $A$ and $B$ are finite sets then $|A \times B| = |A| |B|$ and $|A \coprod B| = |A| + |B|$. 

	\item \textbf{Finite dimensional vector spaces over a field}: A good notion of dimension is the usual one; note that if $V$ and $W$ are vector spaces then $\dim (V \otimes W) = \dim(V) \dim(W)$ and $\dim(V \oplus W) = \dim(V) + \dim(W)$.

	\item \textbf{Graded vector spaces}: A good notion of dimension is the Euler characteristic $\chi$: if $V_{\bullet}$ and $W_{\bullet}$ are graded vector spaces and $\oplus$ and $\otimes$ are the usual direct sum and tensor product of graded vector spaces then, $\chi(V_{\bullet} \oplus W_{\bullet}) = \chi(V_{\bullet}) + \chi(W_{\bullet})$ and $\chi(V_{\bullet} \otimes W_{\bullet}) = \chi(V_{\bullet}) \chi(W_{\bullet})$.

	\item \label{list:endomorphism_dim} \textbf{Pairs of finite dimensional vector spaces over a field $k$ and an endomorphism}: there is a family of good notions of dimension indexed by the non-negative integers (or alternatively, there is a good notion of dimension valued in $k$-valued functions over the integers): given a pair $(V,L)$ of a vector space $V$ and an endomorphism $L: V \rightarrow V$ we can define 
		\begin{align*}
			D_{n}(V,L) &=
			\left \{
			\begin{array}{ll}		
				\Tr(L^{n}), & \text{if $n > 0$}\\
				\rank(L), & \text{if $n = 0$}
			\end{array}
			\right. .
		\end{align*}
		If we define $(V,L) \otimes (W, M) := (V \otimes W, L \otimes M)$ and $(V, L) \oplus (W,M) := (V \oplus W, L \oplus M)$, then $D_{n}(V \otimes W, L \otimes M) = D_{n}(L, V) D_{n}(W, M)$ and $D_{n}(V \oplus W, L \oplus M) = D_{n}(V, L) + D_{n}(W, M)$.
\end{enumerate}
In each of these cases, there is a natural notion of isomorphism of objects, and the set of isomorphism classes is equipped with product and sum operations (obeying the usual distributive identities) coming from the descent of operations $\otimes$ and $\oplus$ (if we are fans of amusing names we might say that the isomorphism classes form a ``rig": a ring without negatives, c.f.\ \cite{nlab:rig}).
The dimension is then taken to be a map out of the set of isomorphism classes into a ring, preserving the sum and product operations (i.e.\ a homomorphism).
The last case demonstrates a situation where a good notion of dimension is not necessarily an integer and bears some resemblance to our situation of interest.

Our concern is with the ``dimension" of the components of some hypothetical geometric object associated to a multipartite density state $\bdens{\rho}_{P}$.
As suggested by \eqref{eq:hypothetical_euler_char_1}, the $k$th component of such a hypothetical geometric object should determined by the data coming from the \textit{collection} of density states $\{(\hilb_{T}, \rho_{T}): |T| = k+1 \}$; so our concern should be with geometric quantities collections (or tuples) of pairs of Hilbert spaces and density states.
Because the geometric object is determined by the multipartite density state, we will assume that we only need to understand the multipartite density state, not necessarily what its associated geometric object might be.
Hence, we will attempt to give a notion of dimension of to tuples of pairs of Hilbert spaces and density states.
\begin{definition}[label=def:finstate]{}{}
	\begin{enumerate}
		\item $\catname{FinState}$ is the collection\footnote{We are explicitly avoiding the use of the term ``set" as the collection of all such pairs of Hilbert spaces and density states is not a set in any formal sense of the word.  However, in truth, we only care about isomorphism classes of such pairs, which does form a set.} of pairs $(\vec{\hilb}, \vec{\dens{\rho}})$ where
			\begin{enumerate}
				\item $\vec{\hilb} = [\hilb^{1}, \cdots, \hilb^{n}]$ is a tuple of finite dimensional Hilbert spaces of some arbitrary finite length $n \geq 1$;

				\item $\vec{\dens{\rho}} = [\dens{\rho}^{1}, \cdots, \dens{\rho}^{n}]$ is a tuple of length $n$ with $\dens{\rho}^{j}$ a positive semidefinite element of $\states{\hilb^{j}}$ for every $j \in \{0,\cdots,n \}$. 
			\end{enumerate}
			Two pairs $(\vec{\hilb}, \vec{\dens{\rho}})$ and $(\vec{\mathcal{K}}, \vec{\dens{\varphi}})$ are invertibly isomorphic if they consists of tuples of the same length and there exists:
			\begin{enumerate}
				\item A permutation $\sigma: \{1, \cdots, n \} \rightarrow  \{1, \cdots, n \}$

				\item A tuple of invertible linear transformations $(\ell^{j}: \hilb^{j} \overset{\sim}{\longrightarrow} \mathcal{K}^{\sigma(j)})_{j = 1}^{n}$ such that
					\begin{align*}
						\dens{\rho}^{j} = (\ell^{\sigma(j)})^{-1} \circ \dens{\varphi}^{\sigma(j)} \circ \ell^{j}.
					\end{align*}
			\end{enumerate}
		\item $\mathrm{FinState}$ is the collection of equivalence classes of $\catname{FinState}$ under invertible isomorphism.

		\item The operation $\boxplus$ on pairs of elements of $\catname{FinState}$ is defined by concatenation: 
			\begin{align*}
				&\newmath{\left( [\hilb^{1}, \cdots, \hilb^{n}], [\dens{\rho}^{1}, \cdots, \dens{\rho}^{n}] \right)  \boxplus \left([\mathcal{K}^{1}, \cdots, \mathcal{K}^{n}], [\dens{\varphi}^{1}, \cdots, \dens{\varphi}^{m}] \right)} := \\
				&\left( [\hilb^{1}, \cdots, \hilb^{n}, \mathcal{K}^{1}, \cdots, \mathcal{K}^{m}], [\dens{\rho}^{1}, \cdots, \dens{\rho}^{n}, \dens{\varphi}^{1}, \cdots, \dens{\varphi}^{m}] \right);
			\end{align*}
			The tensor product is defined as: 
			\begin{align*}
				\newmath{\left([\hilb^{i}]_{i=1}^{n}, [\dens{\rho}]_{i=1}^{n} \right)  \otimes  \left( [\mathcal{K}^{j}]_{j = 1}^{m}, [\varphi^{j}]_{j = 1}^{m} \right)} &:= \left( [\hilb^{i} \otimes \mathcal{K}^{j}]_{i=1,j=1}^{i=n,j=m} ,  [\dens{\rho}^{i} \otimes \varphi^{j} ]_{i=1,j=1}^{i=n,j=m} \right),
			\end{align*}
			where the tuples on the right hand side are equipped with the lexicographical ordering.
	\end{enumerate}	
\end{definition}

For those comfortable with $C^{*}$-algebras, the following remark gives a simplified description of $\catname{FinState}$.
\begin{remark}{}{}
	Using the classification of finite dimensional $C^{*}$-algebras, one can show that $\catname{FinState}$ is equivalent to the collection of pairs $(A, \mathbbm{E})$ where $A$ is a finite dimensional $C^{*}$-algebra and $\mathbbm{E}: A \rightarrow \mathbb{C}$ a positive linear functional on $A$.
	Two elements $(A, \mathbbm{E})$ and $(B, \mathbbm{F})$ are isomorphic if there exists a homomorphism $h: A \rightarrow B$ such that $\mathbbm{F} \circ h = \mathbbm{E}$.
	The operation $\boxplus$ is given by $(A, \mathbbm{E}) \boxplus (B, \mathbbm{F}) = (A \times B, \mathbbm{E} \times \mathbbm{F})$ (where $A \times B$ is the product $C^{*}$-algebra) and the tensor product by $(A, \mathbbm{E}) \otimes (B, \mathbbm{F}) = (A \otimes B, \mathbbm{E} \otimes \mathbbm{F})$.
\end{remark}

The discussion above motivates us to define the notion of a dimension function for $\catname{FinState}$: given a ring $\mathcal{R}$, an \newword{($\mathcal{R}$-valued) dimension function} is a function out of the set of equivalence classes $\mathrm{FinState}$ into $\mathcal{R}$:
\begin{align*}
	\dim: \mathrm{FinState} \rightarrow \mathcal{R} 
\end{align*}
such that
\begin{align}
		\dim \left[(\vec{\hilb}, \vec{\dens{\rho}}) \boxplus (\vec{\mathcal{K}}, \vec{\dens{\varphi}})\right] &= \dim \left[(\vec{\hilb},  \vec{\dens{\rho}}  )\right] +  \dim \left[(\vec{\mathcal{K}},  \vec{\dens{\varphi}}  )\right],
		\label{eq:dim_additive} \\
		\dim \left[(\vec{\hilb}, \vec{\dens{\rho}}) \otimes (\vec{\mathcal{K}}, \vec{\dens{\varphi}})\right] &= \dim \left[(\vec{\hilb},  \vec{\dens{\rho}}  )\right]  \dim \left[(\vec{\mathcal{K}},  \vec{\dens{\varphi}}  )\right]
		\label{eq:dim_multiplicative}
\end{align}
where we are abusing notation and writing $\dim (\vec{\hilb}, \vec{\dens{\rho}})$ to denote the value of $\dim$ on the equivalence class of $(\vec{\hilb}, \vec{\dens{\rho}})$.
Because every element of $\catname{FinState}$ can be thought of as a $\boxplus$-sum of 1-tuples of the form $([\hilb], [\dens{\rho}])$, then $\dim$ is determined how it acts on the data of tuples $([\hilb], [\dens{\rho}])$.
With this in mind, suppose we have a function $\mathpzc{G}: \mathrm{FinDens} \rightarrow R$ that satisfies
\begin{align*}
	\mathpzc{G}[(\hilb \otimes \mathcal{K}, \dens{\rho} \otimes \dens{\varphi})] = \mathpzc{G}[(\hilb, \dens{\rho})] \mathpzc{G}[(\hilb, \dens{\varphi})], 
\end{align*}
then the associated function 
\begin{align*}
	\dim_{\mathpzc{G}}: \mathrm{FinState} &\longrightarrow \mathcal{R}\\
	\bigboxplus_{i =1}^{n} ([\hilb^{i}], [\dens{\rho}^{i}]) &\longmapsto \sum_{i = 1}^{n} \mathpzc{G}[ (\hilb^{i}, \dens{\rho}^{i}) ]
\end{align*}
is a dimension function, and every dimension function (valued in $\mathcal{R}$) is uniquely determined by such a $\mathpzc{G}$.
With this notation we can rewrite \eqref{eq:hypothetical_euler_char_1} as: 
\begin{align}
	\chi_{\mathpzc{G}}(\bdens{\rho}_{P}) = \sum_{k=-1}^{|P|-1} (-1)^{k} \dim_{\mathpzc{G}}\left[ \bigboxplus_{|T|=k+1} ([\hilb_{T}], [\dens{\rho}_{T}] ) \right]
	\label{eq:chi_dim_simplification}
\end{align}
which looks more like the Euler characteristic of some (co)chain complex/geometric object whose $k$th component is associated to $\bigboxplus_{|T|=k+1} ([\hilb_{T}], [\dens{\rho}_{T}] )$.

Next we wish to understand if we can resolve the issue that prevented us from claiming mutual information is an Euler characteristic: namely, given a dimension function $\dim$ and the expression \eqref{eq:chi_dim_simplification} for a hypothetical Euler characteristic, is the hypothetical Euler characteristic of a multipartite density state the product of hypothetical Euler characteristics?

\begin{lemma}{}{}
	$\mathpzc{G}: \mathrm{FinDens} \rightarrow \mathcal{R}$ be a function into some ring $\mathcal{R}$. 
	Given a multipartite density state $\bdens{\rho}_{P}$, define 
	\begin{align*}
		\mathcal{X}^{w}_{\mathpzc{G}}(\bdens{\rho}_{P}) =  w^{|P|}\sum_{\emptyset \subseteq T \subseteq P} (-1)^{|T|} \mathpzc{G}\left[(\hilb_{T}, \dens{\rho}_{T}) \right] \in R[w].
	\end{align*}
	where $w$ is a formal variable and $R[w]$ denotes the ring of polynomials in $w$ with coefficients in $R$.
	Then $\mathcal{X}_{\mathpzc{G}}$ is invariant under local invertible transformations.
	Moreover, if
	\begin{align}
		\mathpzc{G}[(\hilb \otimes \mathcal{K}, \dens{\rho} \otimes \dens{\varphi})] = \mathpzc{G}[(\hilb, \dens{\rho})] \mathpzc{G}[(\hilb, \dens{\varphi})] 
		\label{eq:general_factorization}
	\end{align}
	then
	\begin{align*}
		\mathcal{X}_{\mathpzc{G}}(\bdens{\rho}_{P} \otimes \bdens{\varphi}_{Q}) = \mathcal{X}_{\mathpzc{G}}(\bdens{\rho}_{P}) \mathcal{X}_{\mathpzc{G}}(\bdens{\varphi}_{Q}). 
	\end{align*}
\end{lemma}
The function in the lemma above differs from our hypothetical Euler characteristic in a mild way:
\begin{align*}
	\mathcal{X}^{w}_{\mathcal{G}}(\bdens{\rho}_{P}) &= -w^{|P|} \chi_{\mathcal{G}}(\bdens{\rho}_{P})\\
													&=w^{|P|} \sum_{l = 0}^{|P|} (-1)^{l} \dim_{\mathcal{G}}\left[ \bigboxplus_{|T|=l} ([\hilb_{T}], [\dens{\rho}_{T}] ) \right].
\end{align*}
As a result of the lemma, if we have constructed some dimension function $\dim_{\mathcal{G}}: \mathrm{FinState} \rightarrow \mathcal{R}$ and take $\chi_{\mathcal{G}}(\bdens{\rho}_{P})$ defined by \eqref{eq:chi_dim_simplification}, then we have 
\begin{align*}
	\chi_{\mathcal{G}}(\bdens{\rho}_{P} \otimes \bdens{\varphi}_{Q}) = -\chi_{\mathcal{G}}(\bdens{\rho}_{P}) \chi_{\mathcal{G}}(\bdens{\varphi}_{Q}).
\end{align*}
The extra sign is a mild deviation from our expectation and should not be considered discouraging.
Such signs emerge naturally in topology: if we let $\twid{\chi}(X)$ denote the Euler characteristic of the cylinder $X \times [0,1]$ for any topological space $X$, then we have $\twid{\chi}(X \times Y) = -\twid{\chi}(X) \twid\chi(Y)$; similarly the Euler characteristics of shifted complexes takes the tensor product of two complexes to the product of Euler characteristics of tensor factors up to an overall minus sign.\footnote{In both instances we are taking ``degree-shifting" functorial transformations before computing Euler characteristics.}

\subsubsection{The State Index}
It remains to construct an explicit dimension function, or equivalently a function $\mathrm{FinDens} \rightarrow \mathbb{C}$ satisfying \eqref{eq:general_factorization}.
One possibility is the exponential of von Neumann entropy:
\begin{align*}
	\Omega: \mathrm{FinDens} & \longrightarrow \mathbb{C}\\
	\dens{\rho}             & \longmapsto \exp \left[ S^{\mathrm{vN}}(\dens{\rho}) \right]
\end{align*}
Which defines a dimension function
\begin{align*}
	\dim_{\Omega}: \bigboxplus_{i = 1}^{n} ([\hilb^{i}], [\dens{\rho}^{i}]) &\longmapsto \sum_{i =1}^{n} \exp \left[ S^{\mathrm{vN}}(\dens{\rho}^{i}) \right].
\end{align*}
However, this extension of $\Omega$ to tuples has some tension with what one would expect: the usual generalization of entropy to tuples or $\boxplus$-sums of (unnormalized) density states is also additive: if $([\hilb]^{i}, [\dens{\rho}^{i}]_{i =1}^{n}) = \bigboxplus_{i = 1}^{n} ([\hilb^{i}], [\dens{\rho}^{i}])$ is an element of $\catname{FinState}$, then one typically takes its entropy to be given as a sum of von Neumann entropies:
\begin{align*}
	S \left(\bigboxplus_{i = 1}^{n} ([\hilb^{i}], [\dens{\rho}^{i}]) \right) := \sum_{i = 1}^{n} S^{\mathrm{vN}}(\dens{\rho}^{i}),
\end{align*}
see, for instance, \cite[\S1]{petz:quantum_ent}.
(To convince oneself that $S$ is a good definition of entropy for elements of $\mathrm{FinState}$, notice that the specialization to $\mathcal{H}^{i} \cong \mathbb{C}$ for all $i =1, \cdots, n$, recovers the Shannon entropy of a measure on an $n$-element set.)
For a generic tuple in $\mathrm{FinState}$ we have:
\begin{align*}
	\exp\left[S\left(\bigboxplus_{i = 1}^{n} ([\hilb^{i}], [\dens{\rho}^{i}])\right)\right] \neq \dim_{\Omega}\left[\left(\bigboxplus_{i = 1}^{n} ([\hilb^{i}], [\dens{\rho}^{i}])\right)\right].
\end{align*}
So while $\dim_{\Omega}$ is certainly a well-defined dimension function, it is not the exponential of entropy for a generic element of $\mathrm{FinState}$. 

On the other hand, the study of pairs $(\hilb, \dens{\rho})$ of a Hilbert space and density state bears similarity to the study of pairs of (finite-dimensional) vector spaces and endomorphisms (c.f.\ Ex.~\ref{list:endomorphism_dim}, Page~\pageref{list:endomorphism_dim} in the previous section) where a good notion of dimension consists of traces of powers of the endomorphism. 
This observation motivates working with the following definition.

\begin{definition}{}{}
	\begin{enumerate}
		\item Let $\hilb$ be a (finite-dimensional) Hilbert space, $\dens{\rho} \in \Dens(\hilb)$, and $(\alpha, q, r) \in \mathbb{C}$; then define 
			\begin{align*}
				\newmath{\mathpzc{D}_{\alpha, q ,r} \left[ \left( \hilb, \rho \right) \right]} := \dim( \hilb )^{\alpha} \left( \Tr \left[ \dens{\rho}^{q} \right] \right)^{r}
			\end{align*}
			where $\dens{\rho}^{q}$ is the $q$th power of $\dens{\rho}$ (defined, e.g.\ via the holomorphic functional calculus).

		\item Let $\bdens{\rho}_{P}$ be a multipartite density state; then define the \newword{state index} as the function:
			\begin{equation*}
				\begin{array}{lccl}
					\newmath{\mathfrak{X}(\bdens{\rho}_{p})}: & \mathbb{C}^{3} & \longrightarrow & \mathbb{C}[w]\\
					{}                                        & (\alpha,q,r) & \longmapsto & \mathfrak{X}^{w}_{\alpha,q,r}(\bdens{\rho}_{P})
				\end{array}
			\end{equation*}
			where
			\begin{align*}
				\newmath{\mathfrak{X}^{w}_{\alpha,q,r}(\bdens{\rho}_{P})} := w^{|P|} \sum_{T \subseteq P} (-1)^{|T|} \mathpzc{D}_{\alpha, q, ,r} \left[ \left(\hilb_{T}, \dens{\rho}_{T} \right)\right]. 
			\end{align*}
			 and $w$ is thought of as either a formal parameter, or an element of $\mathbb{C}$.
	\end{enumerate}
\end{definition}
Let us rephrase this definition in terms of the language used in the previous section:
For any fixed $(\alpha,q,r)$, the quantity $\mathpzc{D}_{\alpha,q,r}(\hilb, \dens{\rho})$ only depends on the isomorphism class of $(\hilb, \dens{\rho})$ in $\catname{FinDens}$; hence, it defines a function 
\begin{align*}
	\mathpzc{D}_{\alpha,q,r}: \mathrm{FinDens} \rightarrow \mathbb{C},
\end{align*}
which satisfies the multiplicativity property \eqref{eq:general_factorization}:
\begin{align*}
	\mathpzc{D}_{\alpha, q, r}(\hilb \otimes \mathcal{K}, \dens{\rho} \otimes \dens{\varphi}) &=  \mathpzc{D}_{\alpha, q, r}(\hilb, \dens{\rho}) \mathpzc{D}_{\alpha, q, r}(\mathcal{K}, \dens{\varphi}).
\end{align*}
Thus, there is an associated three parameter family of $\mathbb{C}$-valued dimension functions on $\catname{FinState}$.
In the notation of the previous section, this three parameter family is given by:
\begin{align*}
	\mathfrak{X}_{\alpha,q,r}^{w}(\bdens{\rho}_{P}) = \mathcal{X}^{w}_{\mathpzc{D}_{\alpha,q,r}}(\bdens{\rho}_{P}) \in \mathbb{C}[w].
\end{align*}
We can combine these three-parameter quantities into holomorphic functions on $\mathbb{C}^{3}$.
Indeed, note that, because $\alpha,\,q,\,$ and $r$ always appear as powers of positive real numbers, then $\mathpzc{D}_{\alpha,q,r}(\hilb, \dens{\rho})$ is an everywhere-holomorphic (entire) function separately in the parameters $\alpha, q,$ and $r$; hence, it defines an entire function on $\mathbb{C}^{3}$.
Letting $\mathcal{O}(\mathbb{C}^{3})$ denote the ring of entire functions on $\mathbb{C}^{3}$ (under pointwise addition and multiplication), we have a function:
\begin{align*}
	\mathpzc{D}: \mathrm{FinDens} &\longrightarrow \mathcal{O}(\mathbb{C}^{3})\\
	 [(\hilb, \dens{\rho})] &\longmapsto \left( (\alpha, q, r) \mapsto \mathpzc{D}_{\alpha, q, r}[(\hilb,\dens{\rho})] \right)
\end{align*}
that satisfies the multiplicativity property \eqref{eq:general_factorization}.
As a result, we have an associated $\mathcal{O}(\mathbb{C}^{3})$-valued dimension function and, in the notation of the previous section,
\begin{align*}
	\mathfrak{X}^{w}(\bdens{\rho}_{P}) = \mathcal{X}^{w}_{\mathpzc{D}}(\bdens{\rho}_{P}) \in \mathcal{O}(\mathbb{C}^{3})[w].
\end{align*}

The following remark shows how working with the $\alpha \rightarrow 0$ specialization of the $\mathcal{O}(\mathbb{C}^{3})$-valued dimension function $\dim_{\mathpzc{D}}$ solves our original issue with working with the exponential of entropy.
\begin{remark}[label=rmk:entropy_from_dimension]{}{}
	One can extract the entropy $S$ of a generic element of $\mathrm{FinState}$ from the $\mathcal{O}(\mathbb{C}^{3})$-valued dimension function $\dim_{\mathpzc{D}}$.
	\begin{align*}
		-\frac{\partial}{\partial q} \left. \left[ \dim_{D_{0,q,r}} \left( \bigboxplus_{i =1}^{n} ([\hilb^{i}, [\dens{\rho}^{i}] ) \right) \right] \right|_{q=1} &= \sum_{i = 1}^{n} r S^{\mathrm{vN}}(\dens{\rho}^{i})
	\end{align*}
	Thus, taking $r=1$, we have:
	\begin{align*}
		S \left( \bigboxplus_{i =1}^{n} ([\hilb^{i}, [\dens{\rho}^{i}] ) \right) = -\frac{\partial}{\partial q}  \left. \left[\dim_{D_{0,q,1}} \left( \bigboxplus_{i =1}^{n} ([\hilb^{i}, [\dens{\rho}^{i}] ) \right) \right] \right|_{q = 1} 
	\end{align*}
	With this observation it becomes apparent that one can extract the multipartite mutual information 
	\begin{align*}
		I(\bdens{\rho}_{P}) = \sum_{k = 0}^{|P|-1} (-1)^{k} S\left( \bigboxplus_{|T|=k+1} ([\hilb_{T}], [\dens{\rho}_{T}] ) \right)
	\end{align*}
	as
	\begin{align*}
		I(\bdens{\rho}_{P}) = \left. \left[\frac{\partial}{\partial q}	\mathfrak{X}^{w=1}_{0,q,r}(\bdens{\rho}_{P}) \right] \right|_{(q,r) = (1,1)}.
	\end{align*}
	We will return to this statement in \S\ref{sec:mutual_info_as_limit}. 
\end{remark}

The following proposition is a summary of our discussion above.
\begin{theorem}[label=thm:index_properties]{}{}
	Let $\bdens{\rho}_{P}$ and $\bdens{\varphi}_{Q}$ denote multipartite density states.
	\begin{enumerate}
		\item For any fixed $w \in \mathbb{C}$, the index $\mathfrak{X}^{w}(\bdens{\rho}_{P})$ is an entire function in in the parameters $\alpha,\,q,$ and $r$.

		\item The state index is invariant under local invertible transformations of multipartite states. 

		\item $\mathfrak{X}^{w}(\bdens{\rho}_{P} \otimes \bdens{\varphi}_{Q}) = \mathfrak{X}^{w}(\bdens{\rho}_{P}) \mathfrak{X}^{w}(\bdens{\varphi}_{Q})$.
		
		\item $\mathfrak{X}^{w}(\bdens{\rho}_{P}) = -w \left[ \mathfrak{X}^{w}(\bdens{\rho}_{\partial_{i} P}) + \mathfrak{X}^{w}(\bdens{\rho}_{\partial_{j} P}) - \mathfrak{X}^{w}(\lambda_{ij} [\bdens{\rho}_{P}]) \right]$.
	\end{enumerate}
\end{theorem}

It is worth remarking on generalizations of the state index to the situation of infinite dimensional Hilbert spaces.
\begin{remark}[label=rmk:state_index_inf_dim]{}{}
	The appearance of dimensions of Hilbert spaces in the state index makes the $\alpha$ parameter only suitable for finite dimensional situations.  However, the $\alpha \rightarrow 0$ limit/specialization of the state index: 
	\begin{align*}
		w^{|P|} \sum_{T \subseteq P} (-1)^{|T|}  \left(\Tr \left[\left(\dens{\rho}_{T} \right)^{q} \right] \right)^{r},
	\end{align*}
	is sensible when working with infinite dimensional Hilbert spaces.\footnote{It might be possible to re-introduce a parameter like $\alpha$ in infinite dimensions by introducing a ``relative generalized" index defined for pairs ``unnormalized" density states (more precisely, ``normal weights" in the von Neumann algebra literature).  What we are calling the state index above should arise as the index of $\dens{\rho}$ relative to the unit trace ``identity" state $\dim(\hilb_{P})^{-1} \mathbbm{1}$, whereas the specialization to $\alpha = 0$ can be thought of as the index of $\dens{\rho}$ relative to $\mathbbm{1}$ (representing the unnormalized trace, which in infinite dimensions is no longer a state but a weight).}
	Unlike finite dimensions, the function $q \mapsto \Tr(\dens{\rho}^{q})$ does not necessarily define an entire function on $\mathbb{C}$.
	However it is guaranteed to be holomorphic on the half plane defined by $\mathrm{Re}(q) \geq 1$ (beyond this it might have an maximal analytic continuation to a branched cover of the plane $\mathbb{C}$).
\end{remark}

As a fun remark: the state index can be expressed as a weighted partition function of a family of auxiliarly many body theories (parameterized by $r$) of fermions on a lattice, with non-factorizability of the state being related to the presence of multi-body interaction terms.
We avoid the temptation to provide such a description in this paper, although the interested knowledgeable reader can likely work out the relevant details.\footnote{As a hint, the Hamiltonian should be given as the log of the partial traces of $\rho^{\otimes r}$, and the parameter $\alpha$ is related to shifts by ground state energies by an amount proportional to the log of the dimension of the Hilbert space at each lattice site.}
Such an approach might provide insight into using field theoretic techniques the state index for e.g.\ $|P|$-partite qubit systems with large $P$; moreover it might be an indicator that this index has a less contrived realization in terms of a natural quantum field theory.\footnote{It is also not unreasonable to speculate that the cochain complexes above might arise naturally as in the context of some twisted supersymmetric field theory or quantum mechanics.}

\subsection{Mutual Information in the Limit \texorpdfstring{$q \rightarrow 1$}{As q Goes to 1} and Some Deformations \label{sec:mutual_info_as_limit}}
As indicated in Rmk.~\ref{rmk:entropy_from_dimension}, one can extract mutual information from the state index by an application of partial derivatives with respect\footnote{Note that, using this identity, one can derive the fact that $I(\bdens{\rho}_{P} \otimes \bdens{\varphi}_{Q}) = 0$ by taking partial derivatives with respect to $q$ of the identity $\mathfrak{X}_{\alpha,q,r}(\bdens{\rho}_{P} \otimes \bdens{\varphi}_{Q}) = \mathfrak{X}_{\alpha,q,r}(\bdens{\rho}_{P}) \mathfrak{X}_{\alpha,q,r}(\bdens{\varphi}_{Q})$ and using the fact that $\mathfrak{X}_{0,1,r}(\bdens{\rho}_{P}) \equiv 0$ for any multipartite density state $\bdens{\rho}_{P}$.} to $q$, indeed: 
\begin{align*}
	\left. \left[\frac{\partial}{\partial q} \mathfrak{X}_{0,q,r} (\bdens{\rho}_{P}) \right] \right|_{q=1} &= w^{|P|} \sum_{\emptyset \subseteq T \subseteq P} (-1)^{|T|} \left. \left \{ \frac{\partial}{\partial q} \Tr \left[ \left(\dens{\rho}_{T}\right)^{q} \right]^{r} \right \} \right|_{q = 1} \\
																										   &= r w^{|P|} \sum_{\emptyset \subseteq T \subseteq P} (-1)^{|T|} \Tr\left[\dens{\rho}_{T} \log(\dens{\rho}_{T}) \right]\\ 
																										   &= r w^{|P|} I(\bdens{\rho}_{P}). 
\end{align*}
Note that evaluating $w = -1$, we arrive at the quantity $I^{-}(\bdens{\rho}_{P}) = (-1)^{|P|}I(\bdens{\rho}_{P})$ mentioned in Remark~\ref{rmk:multipartite_mutual_info_sign}.
Alternatively, because $\mathfrak{X}_{0,1,r}(\bdens{\rho}_{P}) \equiv 0$, the partial derivative above coincides with the limit:\footnote{The left hand side can be thought of as the ``$q$-derivative" of $\mathfrak{X}_{0,q,r}$ at the point $q=1$.}
\begin{align*}
	\lim_{q \rightarrow 1}\frac{\mathfrak{X}_{0,q,r}(\bdens{\rho}_{P})}{q-1} = -r w^{|P|} \sum_{T \subseteq P} (-1)^{|T|} S^{\mathrm{vN}}(\dens{\rho}) = r w^{|P|} I(\bdens{\rho}_{P}) 
\end{align*}
Note that evaluating $w = -1$, we arrive at the quantity $I^{-}(\bdens{\rho}_{P}) = (-1)^{|P|}I(\bdens{\rho}_{P})$ mentioned in Remark~\ref{rmk:multipartite_mutual_info_sign}.

In general the $\alpha \rightarrow 0$ limit of the state index can be thought of as a deformed version of mutual information (up to rescaling).
Indeed, the discussion above, suggests that we consider the following quantity: 
\begin{align}
	\newmath{I_{q,r}(\bdens{\rho}_{P})} &:= \frac{\mathfrak{X}_{0,q,r}^{w=1}(\bdens{\rho}_{P})}{r(q-1)},
	\label{eq:tsallis_mi_definition}
\end{align}
for each $q,r \in \mathbb{C}$.
When all Hilbert spaces are finite dimensional, the function $(q,r) \mapsto I_{q,r}(\bdens{\rho}_{P})$ is separately entire\footnote{When working with infinite dimensional Hilbert spaces, $I_{q,r}(\bdens{\rho}_{P})$ is holomorphic in the region $\{(q,r) \in \mathbb{C}^{2}: \mathrm{Re}(q) \geq 1 \}$.  C.f.\ Rmk.~\ref{rmk:state_index_inf_dim}.} in the parameters $q$ and $r$; deforms mutual information in the sense that
\begin{align*}
	\lim_{q \rightarrow 1} I_{q,r}(\bdens{\rho}_{P}) = I(\bdens{\rho}_{P}).
\end{align*}
Using the fact that the state index takes tensor products to products, it is easy to see that our two-parameter deformation of mutual information takes tensor products to products up to rescaling:
\begin{align}
	I_{q,r}(\bdens{\rho}_{P} \otimes \bdens{\varphi}_{Q}) = r(q-1)I_{q,r}(\bdens{\rho}_{P}) I_{q,r}(\bdens{\varphi}_{Q}),
	\label{eq:multiplicativity_tsallis_mi}
\end{align}
which is a generalization of the condition $I(\bdens{\rho}_{P} \otimes \bdens{\varphi}_{Q}) =0$.

This deformation of mutual information arises from a two-parameter deformation of von Neumann entropy.
To see this, we write $I_{q,r}(\bdens{\rho}_{P})$ as an inclusion-exclusion sum:
\begin{align*}
	I_{q,r}(\bdens{\rho}_{P}) & = \sum_{\emptyset \subseteq T \subseteq P} (-1)^{|T|-1} \mathrm{S}^{\mathrm{TR}}_{q,r}(\dens{\rho}_{T})
\end{align*}
where we are defining
\begin{align*}
	\newmath{S^{\mathrm{TR}}_{q,r}(\dens{\rho})}  &:= \frac{1}{r(q-1)} \left(1 - \Tr \left[\dens{\rho}^{q} \right]^{r} \right).
\end{align*}
When $\dens{\rho}$ is a density state on a finite dimensional Hilbert space,\footnote{If $\dens{\rho}$ is a trace 1 density state on an infinite dimensional Hilbert space, then $\dens{\rho}^{q}$ is trace class for $\mathrm{Re}(q) \geq 1$. Hence, $I_{q,r}$ is finite in the region $\{(q,r) \in \mathbb{C}^{2}: \mathrm{Re}(q) \geq 1 \}$ when considering multipartite density states on infinite dimensional Hilbert spaces.} this defines a two-parameter function that is separately entire in both $q$ and $r$ (defining the function at values $(q,r)$ such that $r(q-1)=0$, by taking limits).
Moreover, it deforms von Neumann entropy in the sense that 
\begin{align*}
	\lim_{q \rightarrow 1} \mathrm{S}^{\mathrm{TR}}_{q,r}(\dens{\rho}) = -\Tr \left[\dens{\rho} \log(\dens{\rho}) \right] = S^{\mathrm{vN}}(\dens{\rho}).
\end{align*}
Unlike the von Neumann entropy, for arbitrary values of $q$ and $r$, it is ``non-extensive"/non-additive, but in a mild way:\footnote{Moreover, like von Neumann entropy, if we restrict our attention to certain regions of $(q,r)$, it is non-negative and vanishes only for pure states.
Indeed, using the fact that $\dens{\rho}$ has unit trace, we have $S_{q,r}(\dens{\rho}) \geq 0$ in the region $\{(q,r): q \geq 1,\, r \geq 0\}$.
Moreover, $S_{q,r}(\dens{\rho}) = 0$ for some $(q,r) \in \mathbb{R}_{>1} \times \mathbb{R}_{\geq 0}$ if and only if $\dens{\rho}$ is pure.}
\begin{align*}
	S^{\mathrm{TR}}_{q,r}(\dens{\rho} \otimes \dens{\varphi}) &= S^{\mathrm{TR}}_{q,r}(\dens{\rho}) + S^{\mathrm{TR}}_{q,r}(\dens{\varphi}) + r(1-q)S^{\mathrm{TR}}_{q,r}(\dens{\rho}) S^{\mathrm{TR}}_{q,r}(\dens{\varphi}).
\end{align*}

The superscript $\mathrm{TR}$ stands for ``Tsallis-R\'{e}nyi", and is motivated by the observation that by specializing the parameter $r$ we can recover two famous deformations of the von Neumann entropy:
\begin{enumerate}
	\item When $r \rightarrow 0$ we recover the \textit{R\'{e}nyi} entropy \cite{renyi,quant_renyi}:
		\begin{align*}
			\lim_{r \rightarrow 0}  \mathrm{S}^{\mathrm{Ts}}_{q,r}(\dens{\rho}) &= \frac{1}{1-q} \log \left(\Tr \left[\dens{\rho}^{q} \right] \right) =: S^{\mathrm{Ry}}_{q}(\dens{\rho}).
		\end{align*}
		It is an extensive/additive deformation of von Neumann entropy:
		\begin{align*}
			S^{\mathrm{Ry}}_{q}(\dens{\rho} \otimes \dens{\varphi}) &= S^{\mathrm{Ry}}_{q}(\dens{\rho}) + S^{\mathrm{Ry}}_{q}(\dens{\varphi}).
		\end{align*}
	\item When $r \rightarrow 1$ we recover the \textit{Tsallis}/\textit{$q$-logarithmic} entropy: 
		\begin{align*}
			\mathrm{S}^{\mathrm{Ts}}_{q,1}(\dens{\rho}) &:= \frac{1}{q-1} \left(1 - \Tr \left[\dens{\rho}^{q} \right] \right). 
		\end{align*}
		which is non-extensive for general $q$:
		\begin{align*}
			S^{\mathrm{Ts}}_{q}(\dens{\rho} \otimes \dens{\varphi}) &= S^{\mathrm{Ts}}_{q}(\dens{\rho}) + S^{\mathrm{Ts}}_{q}(\dens{\varphi}) + (1-q)S^{\mathrm{Ts}}_{q}(\dens{\rho}) S^{\mathrm{Ts}}_{q}(\dens{\varphi}).
		\end{align*}
		The name ``Tsallis entropy" is derived from C. Tsallis's introduction of this quantity in \cite{tsallis} as a basis for generalizing Boltzmann-Gibbs statistics, although it was known before (c.f.\ Rmk.~3.2.ii of \cite{leinster:entropy} for a list of references).
\end{enumerate}

As a fun observation, we note that the R\'{e}nyi deformed mutual information
\begin{align*}
	I_{q}^{\mathrm{Ry}}(\bdens{\rho}_{P}) := I_{q,1}(\bdens{\rho}_{P}) = \sum_{T \subseteq P} (-1)^{|T|-1} S_{q}^{\mathrm{Ry}}(\dens{\rho}_{T})
\end{align*}
obeys
\begin{align*}
	0 \equiv I_{q}^{\mathrm{Ry}}(\bdens{\rho}_{P} \otimes \bdens{\varphi}_{Q})
\end{align*}
for any multipartite density states $\bdens{\rho}_{P}$ and $\bdens{\varphi}_{Q}$; this follows either by a specialization of \eqref{eq:multiplicativity_tsallis_mi} to $r=0$, or by directly using additivity of the R\'{e}nyi entropy as in the suggested proof of Prop.~\ref{prop:lambda_fact_then_mi_vanish}.
Thus, we we have a deformation the mutual information $I(\bdens{\rho}_{P})$ into an entire function $q \mapsto I^{\mathrm{Ry}}_{q}(\bdens{\rho}_{P})$ that is identically vanishing if $\bdens{\rho}_{P}$ is factorizable with respect to any partition of $P$. 
The following remark is a superficial demonstration of how one can introduce the additional deformation parameter $r$ into the one-parameter Tsallis and R\'{e}nyi entropies.
\begin{remark}{}{}
	To see how we can recover the two-parameter deformation from the one-parameter R\'{e}nyi and Tsallis deformations, we introduce the $q$-logarithm:
	\begin{align*}
		\log_{q}(x) := \frac{1}{1-q} \left(x^{1-q} - 1 \right).
	\end{align*}
	thinking of $q$ as a parameter in $\mathbb{C}$, this is a deformation of the logarithm in the sense that $\lim_{q \rightarrow 1} \log_{q}(x) = \log(x)$.
	We can use the $q$-logarithm to deform the function $x \mapsto x \log(x)$ used in the definition of Shannon/von Neumann entropy: for any $q \in \mathbb{C}$ define the function
	\begin{align*}
		\mathcal{L}_{q}: \mathbb{R}_{\geq 0} &\rightarrow \mathbb{C}\\
		x &\longmapsto x \log_{q}(x).
	\end{align*}
	Then the Tsallis/$q$-logarithmic entropy is defined as:
	\begin{align*}
		S^{\mathrm{Ts}}_{q}(\dens{\rho}) &= -\Tr \left[\mathcal{L}_{q}(\dens{\rho}) \right],
	\end{align*}
	where $\mathcal{L}_{q}(\dens{\rho})$ is defined via the continuous functional calculus.\footnote{Which is a fancy way of saying use the spectral decomposition of $\dens{\rho}$ and the fact that $\dens{\rho}$ has positive spectrum to define it.}
	Positive integer values of the parameter $r$ arise naturally when one considers independent copies of the density state/``physical system" (i.e.\ note that $\Tr[\dens{\rho}^{\otimes r}] = \Tr[\dens{\rho}]^{r}$): more precisely,
	\begin{align*}
		S^{\mathrm{TR}}_{q,r}(\dens{\rho}) = -\frac{1}{r}\Tr \left[\mathcal{L}_{q}(\dens{\rho}^{\otimes r}) \right] = \frac{1}{r} S^{\mathrm{Ts}}_{q}(\dens{\rho}^{\otimes r}),\, r \in \mathbb{Z}_{\geq 1}.
	\end{align*}

	On the other hand, we can can write: 
	\begin{align*}
		S^{\mathrm{TR}}_{q,r} = \frac{1}{1-q}\log_{1-r}\left( \Tr[\dens{\rho}^{q}] \right)
	\end{align*}
	So the parameter $r$ can be thought of as a deformation of the R\'{e}nyi entropy using the $(1-r)$-logarithm.

\end{remark}

\subsection{Euler Characteristics of Complexes of Vector Spaces and the Limit \texorpdfstring{$q \rightarrow 0$}{q Goes to 0}}
Note that at the $q = 0$ specialization, we have $\rho^{0} = \supp_{\dens{\rho}}$ (the support projection of $\dens{\rho}$), giving:
\begin{align*}
	\mathpzc{D}_{\alpha,0,r}[(\hilb, \dens{\rho})] = \dim(\hilb)^{\alpha} \rank(\dens{\rho})^{r}.
\end{align*}
So, for $q,\,r$ positive integers, $\mathpzc{D}_{0,q,r}[(\hilb, \dens{\rho})] \in \mathbb{Z}_{\geq 0}$, making it a likely candidate for the dimension of some vector space associated to $(\hilb, \dens{\rho})$.
As a result, the associated integer-valued $\mathfrak{X}^{w=1}_{\alpha,0,r}(\bdens{\rho}_{P})$ is a likely candidate for a chain complex of such vector spaces.  
Indeed, we can recover the Euler characteristics of the GNS and commutant complexes defined above as specializations of $\mathfrak{X}^{w=1}_{\alpha,0,r}$.
The specialization to $(\alpha,r) = (1,1)$ produces the Euler characteristic of the GNS complex: 
\begin{align}
	\chi \left[\cpx{G}(\bdens{\rho}_{P}) \right] &= \sum_{k=-1}^{|P|-1} (-1)^{k} \left[ \sum_{|T|=k+1} \dim(\hilb_{T}) \rank \left(\dens{\rho}_{T} \right)\right] = -\mathfrak{X}^{w=1}_{1,0,1}(\bdens{\rho}_{P}) 
	\label{eq:index_to_GNS_ec}
\end{align} 
while the specialization to $(\alpha,r) = (0,2)$ produces the Euler characteristic of the commutant complex:
\begin{align}
	\chi \left[\cpx{E}(\bdens{\rho}_{P}) \right] &= \sum_{k=-1}^{|P|-1} (-1)^{k} \left[ \sum_{|T|=k+1} \rank \left(\dens{\rho}_{T} \right)^2 \right] = -\mathfrak{X}^{w=1}_{0,0,2}(\bdens{\rho}_{P}).
	\label{eq:index_to_commutant_ec}
\end{align} 
The sign can be eliminated by looking at the shifted complexes defined in Def.~\ref{def:shifted_complexes}:
\begin{align*}
	\chi \left[\twid{\cpx{G}}(\bdens{\rho}_{P}) \right] &= \mathfrak{X}^{w=1}_{1,0,1}(\bdens{\rho}_{P})\\
	\chi \left[\twid{\cpx{E}}(\bdens{\rho}_{P}) \right] &= \mathfrak{X}^{w=1}_{0,0,2}(\bdens{\rho}_{P}).
\end{align*}

\begin{remark}{}{}
	The appearance of the sign in \eqref{eq:index_to_GNS_ec} and \eqref{eq:index_to_commutant_ec} is then in line with the fact that the Poincar\'{e} polynomials of the unshifted complexes obey (c.f.\ Cor.~\ref{cor:poincare_factorization})
	\begin{align*}
		P_{\cpx{G}}(\bdens{\rho}_{P} \otimes \bdens{\varphi}_{Q}) &= y P_{\cpx{G}}(\bdens{\rho}_{P}) P_{\cpx{G}}(\bdens{\varphi}_{Q}),\\
		P_{\cpx{E}}(\bdens{\rho}_{P} \otimes \bdens{\varphi}_{Q}) &= y P_{\cpx{E}}(\bdens{\rho}_{P}) P_{\cpx{E}}(\bdens{\varphi}_{Q}).
	\end{align*}

	so that, specializing to $y = -1$ to recover Euler characteristics, we have:
	\begin{align*}
		\chi \left[\cpx{G}(\bdens{\rho}_{P} \otimes \bdens{\varphi}_{Q}) \right] = -\chi \left[ \cpx{G}(\bdens{\rho}_{P}) \right] \chi\left[ \cpx{G}(\bdens{\varphi}_{Q}) \right];
	\end{align*}
	while the Poincar\'{e} polynomials of the shifted complexes obey
	\begin{align*}
		P_{\cpx{\twid{G}}}(\bdens{\rho}_{P} \otimes \bdens{\varphi}_{Q}) &=  P_{\cpx{\twid{G}}}(\bdens{\rho}_{P}) P_{\cpx{\twid{G}}}(\bdens{\varphi}_{Q}),\\
		P_{\cpx{\twid{E}}}(\bdens{\rho}_{P} \otimes \bdens{\varphi}_{Q}) &=  P_{\cpx{\twid{E}}}(\bdens{\rho}_{P}) P_{\cpx{\twid{E}}}(\bdens{\varphi}_{Q});
	\end{align*}
	giving
	\begin{align*}
		\chi \left[\twid{\cpx{G}}(\bdens{\rho}_{P} \otimes \bdens{\varphi}_{Q}) \right] = \chi \left[ \twid{\cpx{G}}(\bdens{\rho}_{P}) \right] \chi\left[ \twid{\cpx{G}}(\bdens{\varphi}_{Q}) \right].
	\end{align*}
\end{remark}

The following remark provides some speculation about how to recover the general index $\mathfrak{X}^{w=1}_{\alpha,q,r}$ as an Euler characteristic or Lefschetz index by a more sophisticated construction of the complexes above.
Such a sophisticated construction could be considered a complete realization of a categorification of mutual information, and might provide generalizations of the integer coefficient Poincar\'{e} polynomials above---which only depend on the support equivalence class of multipartite density states---to polynomials with coefficients in functions holomorphic in the parameters $\alpha,\,q,\,$ and $r$.
These latter polynomials should require more information than just the support equivalence class of the density state, and might offer a way in which one can recover invariants that vary continuously for, e.g.\ a family of pure multipartite states that pass from a factorizable state to a non-factorizable state (the integer valued polynomials we have constructed so far jump abruptly as the support projection suddenly changes in such a family).

\begin{remark}{}{Speculation About a Full Categorification of Mutual Information}
	Heuristically, $1/q$ should be thought of as the parameter that defines a non-commutative $L^{1/q}$ space: one can construct complexes by completing the GNS module with respect to $L^{1/q}$ (quasi-)norms constructed via partial trace reductions of the state $\dens{\rho}$.
	In the limit $q \rightarrow 0$, these norms approach the restriction of the operator norm to the GNS module, which happens to remain complete in the operator norm (even in infinite dimensions), and we recover our complexes above.
	The indices for $q \neq 0$ should be recoverable from a more sophisticated discussion along these lines.
	This interpretation is related to the vague intuition that the appearance of an expression like $\dens{\rho}^{q}$ is related to the appearance of (relative) modular flows ($t \mapsto \dens{\rho}^{it}$ is the relative modular flow with respect to the trace): indeed, the construction of the non-commutative $L^{p}$-spaces is intimately related to relative Tomita-Takesaki modular theory.
	An classical reference on modular flows is \cite{takesaki_vol2}.
	The author's favorite approach to the construction of non-commutative $L^{p}$-spaces is through the construction of the modular algebra \cite{yamagami, pavlov}.
\end{remark}

\section{W State vs. GHZ state \label{sec:W_vs_GHZ}}
As a testing ground for multipartite cohomologies, we consider the GHZ and W states. Traditionally these are defined as tripartite states, but we will also consider multipartite generalizations.

\subsection{Tripartite States \label{sec:tripartite_states}}
Let $\hilb_{\sX},\, \sX \in \{\sA, \sB, \sC \}$ be single-qubit Hilbert spaces (two-dimensional Hilbert spaces) equipped with (orthonormal) computational bases $\{\ket{0_{\sX}}, \ket{1_{\sX}} \} \subseteq \hilb_{\sX}$.
We define pure tripartite states 
\begin{align*}
	\bdens{\mathrm{GHZ}}_{3} := \left( (\sA, \sB, \sC), (\hilb_{\sA}, \hilb_{\sB}, \hilb_{\sC}), \mathrm{GHZ}_{3} \otimes \left( \mathrm{GHZ}_{3} \right)^{\vee} \right),\\
	\bdens{\mathrm{W}}_{3} := \left((\sA, \sB, \sC), (\hilb_{\sA}, \hilb_{\sB}, \hilb_{\sC}), \mathrm{W}_{3} \otimes \left(\mathrm{W}_{3} \right)^{\vee} \right).
\end{align*}
where $\mathrm{GHZ}_{3},\, \mathrm{W}_{3} \in \hilb_{\sA} \otimes \hilb_{\sB} \otimes \hilb_{\sC}$, are defined as 
\begin{align*}
	\mathrm{GHZ}_{3} &:= \frac{1}{\sqrt{2}} \left[ \ket{000} + \ket{111} \right],\\
	\mathrm{W}_{3} 	&:= \frac{1}{\sqrt{3}} \left[ \ket{100} + \ket{010} + \ket{001} \right],
\end{align*}
using traditional simplified notation where, e.g. $\ket{010}$ is to be read as $\ket{0_{\sA}} \otimes \ket{1_{\sB}} \otimes \ket{0_{\sC}}$.

\subsubsection{State Indices and Mutual Informations}
One can calculate 
\begin{align*}
	\mathfrak{X}^{w}_{\alpha,q,r}[\bdens{\mathrm{GHZ}}_{3}] &= w^{3}\left \{ 1 - 8^{\alpha} + \frac{3}{2^{r (q-1)}} \left[\left(4^{\alpha} - 2^{\alpha} \right) \right] \right \},\\
	\mathfrak{X}^{w}_{\alpha,q,r}[\bdens{\mathrm{W}}_{3}] &= w^{3} \left \{ 1 - 8^{\alpha} + 3  \left[\left(\frac{2}{3} \right)^{q} + \left(\frac{1}{3} \right)^{q} \right]^{r} \left(4^{\alpha} - 2^{\alpha} \right) \right \}.
\end{align*}
Setting $w=1$, the resulting functions are not the same;\footnote{Which can be verified by evaluating at (for instance) $(\alpha,q,r) = (1,2,1)$.} so, in particular, we can surmise that $\bdensl{W}_{3}$ and $\bdensl{GHZ}_{3}$ are not related by local unitary transformations or, more generally, local invertible transformations.

The $\alpha$ parameter is necessary in this distinction: the specializations to $\alpha = 0$ give
\begin{align*}
	\mathfrak{X}^{w}_{0,q,r}[\mathrm{GHZ}_{3}] = \mathfrak{X}^{w}_{0,q,r}[\mathrm{W}_{3}] \equiv 0.
\end{align*}
In particular, we must have that the $(q,r)$-Tsallis deformed and classical tripartite mutual informations must vanish:
\begin{align*}
	I_{q,r}[\mathrm{GHZ}_{3}] = I_{q,r}[\mathrm{W}_{3}] = 0.
\end{align*}
So tripartite mutual information and its two-parameter deformations do not distinguish these two states and cannot detect any shared information due to their tripartite entanglement. 

Before computing the $\mathrm{G}$ and $\mathrm{E}$ cohomologies, we can calculate their Euler characteristics by specializing the state index to $\alpha = r = 1$ and $q =  0$; the result is: 
\begin{align*}
	\chi_{1,1,0}[\mathrm{GHZ}_{3}] = \chi_{1,1,0}[\mathrm{W}_{3}] &= -5
\end{align*}
which are non-zero---indicating that such states must be entangled---but do not distinguish $\mathrm{GHZ}_{3}$ and $\mathrm{W}_{3}$. 

\subsubsection{Distinguishing Via Cohomology \label{sec:tripartite_cohomology_distinguish}}
As we have seen, one cannot distinguish $\bdensl{W}_{3}$ and $\bdensl{GHZ}_{3}$ by computing any $\alpha = 0$ specializations of the state index: in particular, the mutual information.  Moreover, we cannot distinguish these two tripartite density states via the Euler-characteristics of the (co)chain complexes we have defined in this paper.
However, there remains the possibility that they can be distinguished by the computation of GNS and commutant Poincar\'{e} polynomials. 
A brute force computation (e.g.\ using the provided software) of ranks of GNS cohomologies gives us 
\begin{align*}
	P_{\cpx{G}}(\bdensl{GHZ}_{3}) = P_{\cpx{G}}(\bdensl{W}_{3}) = 1 + 6 y.
\end{align*}
So GNS polynomials do not distinguish these two states.
However, the fact that these polynomials are non-vanishing tells us that there exist tuples of operators that exhibit non-local correlations in the sense described in \S\ref{sec:classes_and_correlations}; something the tripartite mutual information fails to detect.
In particular there is a single tuple of $1$-body operators exhibiting non-trivial correlations between individual tensor factors/primitive subsystems, 
and six linearly independent tuples of $2$-body operators that exhibit non-trivial correlations between pairs of tensor factors. 
Explicit representatives of a basis for the GNS cohomologies of $\bdensl{GHZ}_{3}$ and $\bdensl{W}_{3}$ are tabulated in Appendix~\ref{app:ghz_w_cohomology_generators}.

On the other hand, the two tripartite density states \textit{are} distinguished when we compute the cohomologies of the commutant complexes:
\begin{align*}
	P_{\cpx{E}}(\bdensl{GHZ_{3}}) &= 7 + 7 y,\\
	P_{\cpx{E}}(\bdens{W_{3}}) &= 3 + 3 y.
\end{align*} 
This computation is an alternative way (other than computing state indices) of demonstrating that $\bdensl{W}_{3}$ and $\bdensl{GHZ}_{3}$ are not related by local invertible transformations: i.e.\ they are in distinct SLOCC equivalence classes.\footnote{Recalling the statements made at the end of \S\ref{sec:interpretation}, it is also tempting to think of this as an example where the module structure of GNS modules---as opposed to just their underlying vector spaces---is needed to distinguish two states.}

\subsection{Generalized \texorpdfstring{$d$}{d}-partite GHZ and W states \label{sec:generalized_W_and_GHZ}}
We also consider generalizations of these results to $d$-partite states for arbitrary $d$.
For any integer $n \geq 2$ let $\hilb_{[n]}$ denote an $n$-dimensional Hilbert space and choose a pair of orthonormal vectors $\ket{0}$ and $\ket{1}$ spanning a two-dimensional subspace.
For $d \geq 1$ and $\boldsymbol{n} = (n_{1}, \cdots, n_{d})$, we define
\begin{align*}
	\mathrm{GHZ}_{d,\boldsymbol{n}} &:= \frac{1}{\sqrt{2}} \left(\ket{0}^{\otimes d} + \ket{1}^{\otimes d} \right) \in \hilb_{[n_{1}]} \otimes \hilb_{[n_{2}]} \cdots \otimes \hilb_{[n_{d}]}
\end{align*}
and the generalized W state as:
\begin{align*}
	\mathrm{W}_{d,\boldsymbol{n}} &:= \frac{1}{\sqrt{N}} \left(\sum_{i = 1}^{d} \ket{0}^{\otimes (i-1)} \otimes \ket{1} \otimes \ket{0}^{\otimes(d-i-1)} \right) \in \hilb_{[n_{1}]} \otimes \hilb_{[n_{2}]} \cdots \otimes \hilb_{[n_{d}]}.
\end{align*}
Taking into account the data of the Hilbert space factors, these define pure $d$-partite density states $\bdensl{GHZ}_{d, \boldsymbol{n}}$ and $\bdensl{W}_{d, \boldsymbol{n}}$.

\subsubsection{State Indices}
We can directly compute the state indices of these states as:
\begin{align*}
	\chi^{w}_{\alpha,q,r}[\bdensl{W}_{d,\boldsymbol{n}} ] &= w^{d} \sum_{k = 0}^{d} (-1)^{k} \left[\left(\frac{d-k}{d} \right)^{q} + \left(\frac{k}{d} \right)^{q} \right]^{r}D_{\alpha}(k;\boldsymbol{n}) 
\end{align*}
and
\begin{align}
	\chi_{\alpha,q,r}[\bdensl{GHZ}_{d,\boldsymbol{n}}] &= w^{d} \sum_{k = 0}^{d} (-1)^{k} F_{q,r}(k;d)D_{\alpha}(k;\boldsymbol{n}) 
	\label{eq:GHZ_index}
\end{align}
where
\begin{align*}
	D_{\alpha}(k;\boldsymbol{n}) &:= 
	\left \{
		\begin{array}{ll}
			\sum_{i_{0} < \cdots < i_{k}} (n_{i_{0}} \cdots n_{i_{k}})^{\alpha},  & \text{if $k > 0$}\\
			1, & \text{if $k = 0$}
		\end{array}
	\right.
\end{align*}
and
\begin{align*}
	F_{q,r}(k;d) &:= \left \{
		\begin{array}{ll}
			2^{(1-q)r} = \left[\left(\frac{1}{2} \right)^{q} + \left(\frac{1}{2} \right)^{q} \right]^{r}, & \text{if $k \neq 0, d$}\\
			1, & \text{if $k = 0,d$}
		\end{array}
	\right. .
\end{align*}
In the special case that $n_{1} = n_{2} = \cdots n_{d} = n$, these expressions simplify a bit more as we can write
\begin{align*}
	D_{\alpha}(k;(n,n,\cdots,n)) = \binom{d}{k} n^{k \alpha}.
\end{align*}
So define $\bdensl{GHZ}_{d,n} := \bdensl{GHZ}_{d,(n,n,\cdots,n)}$, and $\bdensl{W}_{d,n} := \bdensl{W}_{d,(n,n,\cdots,n)}$, then we can use \eqref{eq:GHZ_index} to deduce the compact expression: 
\begin{align*}
	\mathfrak{X}^{w}_{\alpha,q,r}[\bdensl{GHZ}_{d,n}] &= w^{d} \left\{ \frac{1}{2^{(q-1)r}} \left[(1 - n^{\alpha})^{d} - (1 + (-1)^{d} n^{\alpha d}) \right] + 1 + (-1)^{d} n^{\alpha d} \right\}.
\end{align*}
For general $r$, the expression for the index of $\bdensl{W}_{d,n}$ is not significantly simplified when $n_{1} = n_{2} = \cdots = n_{k} = n$.  However, when $r=1$, we can simplify the expression for the index of $\bdensl{W}_{d,n}$ to 
\begin{align*}
	\mathfrak{X}^{w}_{\alpha,q,1}[\bdensl{W}_{d,n}] &= \frac{-w^{d}}{d^{q-1}}\left( 1 + (-1)^{d}n^{-\alpha d} \right) \left[\sum_{l=0}^{d-1}(-1)^{l} \binom{d-1}{l} (n^{\alpha})^{l}(l+1)^{q-1} \right]
\end{align*}
The expression in the square brackets is a generalization of the binomial transform \cite{prodinger} of the sequence $\left((l+1)^{q-1}\right)_{l = 0}^{\infty}$, reducing to the usual binomial transform when $\alpha = 1$.  
Lovers of special functions might be pleased to note that we can rewrite this expression in terms of the order $(1-q)$-polylogarithm function.
\begin{align*}
	\mathfrak{X}^{w}_{\alpha,q,1}[\bdensl{W}_{d,n}] &= \frac{-w^{d}}{d^{q-1}}\left( 1 + (-1)^{d}n^{-\alpha d} \right) \left[ [z^{d}]  \left\{\mathrm{Li}_{1-q}\left(\frac{-n^{\alpha}z}{1-z} \right) \right\} \right ].
\end{align*}
Similarly, in the $\alpha \rightarrow 0$ limit (where the index is always independent of the dimensions of the ambient Hilbert space) we have the simplifications:
\begin{align*}
	\mathfrak{X}^{w}_{0,q,r}[\bdensl{GHZ}_{d,\boldsymbol{n}}] &= w^{d}\left(1+(-1)^{d} \right)\left(1-\frac{1}{2^{(q-1)r}} \right)\\
	\mathfrak{X}^{w}_{0,q,r}[\bdensl{W}_{d,\boldsymbol{n}}] &= w^{d} \sum_{k = 0}^{d} (-1)^{k} \binom{d}{k} \left[\left(\frac{d-k}{d} \right)^{q} + \left(\frac{k}{d} \right)^{q} \right]^{r}\\
															&= \frac{w^{d}}{2} \left(1 + (-1)^{d}\right) \left \{\sum_{k = 0}^{d} (-1)^{k} \binom{d}{k} \left[\left(\frac{d-k}{d} \right)^{q} + \left(\frac{k}{d} \right)^{q} \right]^{r} \right \}.
\end{align*}
both of these expressions vanish when $d$ is odd.

\subsubsection{\texorpdfstring{$d$}{d}-partite Mutual Informations}
To extract mutual informations from the state index we can use the identity:
\begin{align*}
	I(\bdens{\rho}_{P}) = \lim_{q \rightarrow 1} \left[ \frac{\mathfrak{X}^{w=1}_{0,q,1}}{q-1} \right]
\end{align*}
along with the fact that
\begin{align*}
	\lim_{q \rightarrow 1} \frac{1 - \left(a^{q} + (1-a)^{q} \right)^{r}}{1-q} &= r \left[ a \log(a) + (1-a) \log(a) \right],
\end{align*}
for any $a \in [0,1]$ (and using the convention that $0 \log 0 = 0$).
For odd $d$, the $\alpha = 0$ specializations of state indices vanish; so we have 
\begin{align*}
	0=I(\bdensl{GHZ}_{d,\boldsymbol{n}}) &= I(\bdensl{W}_{d,\boldsymbol{n}}), \: d \notin 2 \mathbb{Z}_{\geq 1}.
\end{align*}
While for even $d$:
\begin{align*}
	I(\bdensl{GHZ}_{d,\boldsymbol{n}}) &= 2 \log(2);
\end{align*}
and (again using the convention $0 \log 0 = 0$)
\begin{align*}
	I(\bdensl{W}_{d,\boldsymbol{n}}) &= \sum_{k = 0}^{d} (-1)^{k} \binom{d}{k} \left[\frac{k}{d} \log \left(\frac{k}{d}\right) + \left(1-\frac{k}{d} \right) \log\left(1 - \frac{k}{d} \right) \right]\\
									 &=\left(1 + (-1)^{d} \right) \sum_{k = 0}^{d} (-1)^{k} \binom{d}{k} \frac{k}{d} \log \left(\frac{k}{d}\right)\\
									 &=\frac{\left(1 + (-1)^{d} \right)}{d} \log \left[\prod_{k=1}^{d} k^{(-1)^{k} k \binom{d}{k}} \right]\\
									 &=-\left(1 + (-1)^{d} \right) \log \left[\prod_{j=0}^{d-1} (j+1)^{(-1)^{j} \binom{d-1}{j}} \right].
\end{align*}
The following remark is a fun side observation that gives an integral representation of $I(\bdensl{W}_{d,\boldsymbol{n}})$; such a representation allows us to deduce that $I(\bdensl{W}_{d,\boldsymbol{n}}) > 0$ for even $d$ (and moreover, $I(\bdensl{W}_{d,\boldsymbol{n}})$ approaches $0$ as $d \rightarrow \infty$).
In particular, one can use the mutual information to detect shared information among all subsystems for even $d$, but for odd $d$ it fails to detect any.

\begin{remark}{Integral Identities and Positivity of $I(\bdensl{W}_{d,\boldsymbol{n}})$}{}
	Using the identities
	\begin{align*}
		\int_{1}^{\infty} \frac{(d-1)!}{t(t+1) \cdots (t+d-1)} dt &= -\log \left[\prod_{j=0}^{d-1} (j+1)^{(-1)^{j} \binom{d-1}{j}} \right].
	\end{align*}
	and
	\begin{align*}
		\frac{(d-1)!}{t(t+1) \cdots (t+d-1)} = \int_{0}^{1} u^{t-1} (1-u)^{d-1} du 
	\end{align*}
	we have
	\begin{align*}
		-\log \left[\prod_{j=0}^{d-1} (j+1)^{(-1)^{j} \binom{d-1}{j}} \right] & = \int_{0}^{\infty} \left[\int_{0}^{1} u^{t-1} (1-u)^{d-1} du \right] dt\\
																			  & = -\int_{0}^{1} \frac{(1-u)^{d-1}}{\log(u)} du
	\end{align*}
	Hence,
	\begin{align*}
		I(\bdensl{W}_{d,\boldsymbol{n}}) &= (1 + (-1)^{d}) \int_{0}^{1} \left[-\frac{(1-u)^{d-1}}{\log(u)} \right] du
	\end{align*}
	The integrand is non-negative on $[0,1]$; so it immediately follows that $I(\bdens{W}_{d, \boldsymbol{n}}) \geq 0$. 
	Moreover, using the bounds 
	\begin{align*}
		(1-u) \leq -\log(u) \leq \frac{1}{u} - 1 
	\end{align*}
	for all $u \in \mathbb{R}_{>0}$,  it follows that
	\begin{align*}
		\frac{1 + (-1)^{d}}{d(d-1)} \leq I(\bdensl{W}_{d,\boldsymbol{n}}) \leq \frac{1 + (-1)^{d}}{d-1}.
	\end{align*}
	Hence $I(\bdensl{W}_{d,\boldsymbol{n}})>0$ for $d$ even; moreover $I(\bdensl{W}_{d,\boldsymbol{n}}) \rightarrow 0$ as $d \rightarrow \infty$.
\end{remark}

\subsubsection{Computational Observations of Cohomologies \label{sec:multipartite_cohomology_distinguish}}
For odd $d$, the $\alpha=0$ specializations of the state indices of the generalized $W$ and GHZ states vanish (hence so do their mutual informations).
On the other hand, computational observation suggests that the Poincar\'{e} polynomials associated to commutant complexes are non-trivial for all $d$.
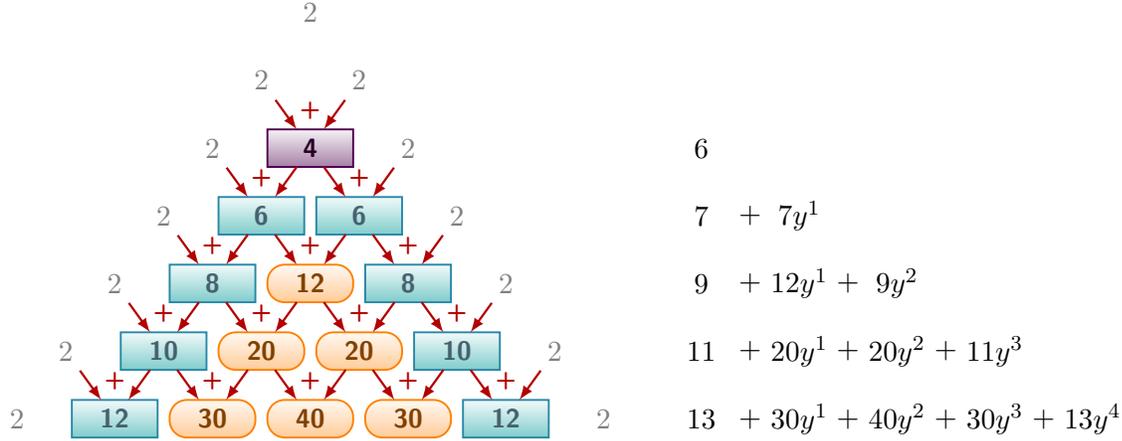
\begin{figure}
	\begin{tikzpicture}[x=13mm,y=9mm]
		\colorlet{outer}{cyan!60!black}
		\colorlet{top}{violet!60!black}
		\colorlet{inner}{orange!100!black}
		\colorlet{links}{red!70!black}

		\tikzset{
			box/.style={
				minimum height=5mm,
				inner sep=.7mm,
				outer sep=0mm,
				text width=10mm,
				text centered,
				font=\small\bfseries\sffamily,
				text=#1!50!black,
				draw=#1,
				line width=.25mm,
				top color=#1!5,
				bottom color=#1!40,
				shading angle=0,
				rounded corners=2.3mm,
				rotate=0,
			},
			bdrycolor/.style={text=gray},
			link/.style={-latex,links,line width=.3mm},
			plus/.style={text=links,font=\footnotesize\bfseries\sffamily},
			outerbox/.style={
				minimum height=5mm,
				inner sep=.7mm,
				outer sep=0mm,
				text width=10mm,
				text centered,
				font=\small\bfseries\sffamily,
				text=#1!50!black,
				draw=#1,
				line width=.25mm,
				top color=#1!5,
				bottom color=#1!50,
				shading angle=0,
				rotate=0,
			},
		}
		\pgfmathtruncatemacro{\bdrytrival}{2}
		\pgfmathtruncatemacro{\maxrow}{6}
		\node[bdrycolor] (t-0-0) at (0,0) {\bdrytrival};
		\foreach \row in {1,...,\maxrow}{
			\pgfmathtruncatemacro{\trival}{\bdrytrival};
			\node[bdrycolor] (t-\row-0) at (-\row/2,-\row) {\trival};
			\foreach \col in {1,...,\row}{
				\pgfmathtruncatemacro{\trival}{\trival*((\row-\col+1)/\col)+0.5};
				\global\let\trival=\trival
				\coordinate (tripos) at (-\row/2+\col,-\row);
				\coordinate (polpos) at (\maxrow/2+\col, -\row);
				\coordinate (pluspos) at (\maxrow/2+\col-1/2, -\row);

				\pgfmathtruncatemacro{\degree}{\col-1};
				\pgfmathtruncatemacro{\bdrycoeff}{\trival+1};
				\pgfmathtruncatemacro{\topcoeff}{\trival+2};

				\ifnum \col=\row
					\node[bdrycolor] (t-\row-\col) at (tripos) {\trival};
				\else
					\ifnum \col=1
						\ifnum \row=2
							\node (p-\row-\col) at (polpos) {$\topcoeff$};
							\node[outerbox=top] (t-\row-\col) at (tripos) {\trival};
						\else
							\node (p-\row-\col) at (polpos) {$\bdrycoeff$};
							\node[outerbox=outer] (t-\row-\col) at (tripos) {\trival};
						\fi
						\else
						\ifnum \col=\numexpr\row-1\relax
							\node[outerbox=outer] (t-\row-\col) at (tripos) {\trival};
							\node (p-\row-\col) at (polpos) {$\bdrycoeff y^{\degree}$};
							\node at (pluspos) {+};
						\else
							\node[box=inner] (t-\row-\col) at (tripos) {\trival};
							\node (p-\row-\col) at (polpos) {$\trival y^{\degree}$};
							\node at (pluspos) {+};
						\fi
					\fi
				\fi
				\ifnum \col<\row
					\node[plus,above=0mm of t-\row-\col]{+};
					\pgfmathtruncatemacro{\prow}{\row-1}
					\pgfmathtruncatemacro{\pcol}{\col-1}
					\draw[link] (t-\prow-\pcol) -- (t-\row-\col);
					\draw[link] (t-\prow-\col) -- (t-\row-\col);
				\fi
			}
		}
		\end{tikzpicture}
		\caption{Right: Commutant Poincar\'{e} polynomials associated to $\bdensl{GHZ}_{d,\boldsymbol{n}}$, beginning at $d = 2$ at the top and increasing to $d = 6$ at the bottom.
			Left: A Pascal's triangle from which we can recover the coefficients of commutant Poincar\'{e} polynomials with modifications at boundary values: values in the blue boxes with sharp edges are increased by 1; the value in the violet box at the apex is increased by 2.
			This pattern of recovering Poincar\'{e} polynomials was checked up to $d=11$ using the provided software.
			\label{fig:Pascal}
		}
\end{figure}

Using the provided software, a computational analysis of the coefficients of commutant Poincar\'{e} polynomials for the $d$-partite GHZ state $\bdensl{GHZ}_{d,\boldsymbol{n}}$ with $2 \leq d \leq 11$ suggests that such coefficients fit into a generalized Pascal's triangle, with modifications on the boundary: see Fig.~\ref{fig:Pascal}.
Explicitly, we observe: 
\begin{align*}
	P_{\cpx{E}}(\bdensl{GHZ}_{d,\boldsymbol{n}}) = \frac{1}{y} \left[T_{d} - 2(1 + y^{d}) \right] + (1+y^{d-2}), \; d \geq 2
\end{align*}
where the polynomials $T_{d}$ are determined by the recurrence relation
\begin{align*}
	T_{d} = y T_{d-1} + T_{d-1}, \; d \geq 1, 
\end{align*}
and the initial condition $T_{0} = 2$.
The closed form solution is $T_{d} = 2(1+y)^{d}$; so we have the following conjecture.
\begin{conjecture}{}{}
	$P_{\cpx{E}}(\bdensl{GHZ}_{d,\boldsymbol{n}}) = 1+ y^{d-2} + \frac{2}{y} \left[(1+y)^{d} - (1+y^{d}) \right]$
\end{conjecture}
This should be compared with both the state indices and $d$-partite mutual information, which vanish for odd $d$ and so cannot detect multipartite entanglement of this state.
On the other hand, the coefficients of the Poincar\'{e} polynomials---which are dimensions of spaces of operators---grow like multiples of binomial coefficients.
The following remark demonstrates that, if the conjecture is true, the ranks of cohomology groups can be extracted from a two-variable rational function.
\begin{remark}{}{}
	One can extract the Poincar\'{e} polynomials above from a rational two-variable generating function: an equivalent form of the conjecture is that
\begin{align*}
	\dim H^{k}[\cpx{E}(\bdensl{GHZ}_{d,\boldsymbol{n}})] = [y^{k} x^{d}] \left( -\frac{x^2 (x y+x-3) (x y+x-2)}{(x-1) (x y-1) (x y+x-1)} \right)
\end{align*}
where $[y^{k} x^{d}]Q(x,y)$ denotes the coefficient of $y^{k} x^{d}$ in the (Taylor) series expansion of $Q(x,y)$ about $x=y=0$. 
In other words, the Poincar\'{e} polynomial can be extracted as the coefficient of $x^{d}$.
\end{remark}

An analysis of commutant cohomologies of the generalized $W$-states leads to a much simpler conjecture of its general form: 
\begin{conjecture}{}{}
	$P_{\cpx{E}}(\bdensl{W}_{d,\boldsymbol{n}}) = 3 \left(1 + y^{d-2} \right)$
\end{conjecture}

On the other hand, a handful of computational observations seem to suggest that the GNS Poincar\'{e} polynomials of $\bdensl{W}_{d,\boldsymbol{n}}$ and $\bdensl{GHZ}_{d,\boldsymbol{n}}$ are identical for all $d \geq 2$ and choices of $\boldsymbol{n} = (n_{1}, \cdots, n_{d})$.
In particular, for the qubit embeddings ($\boldsymbol{n} = (2,2,\cdots,2)$) we observe:
\begin{align*}
	P_{\cpx{G}}(\bdensl{GHZ}_{d,2}) = 1+ (2^{d} - 2)y^{d-1}.
\end{align*}
The reader is invited to use the provided software to identify patterns in associated GNS Poincar\'{e} polynomials when the list of ambient dimensions $\boldsymbol{n} = (n_{1}, \cdots, n_{d})$ is varied.
For instance, one can observe patterns that seem to suggest that
\begin{align*}
	\dim H^{0}[\cpx{G}( \bdensl{W}_{d,(n_{1},\cdots,n_{d})} )] = \dim H^{0}[\cpx{G}( \bdensl{GHZ}_{d,(n_{1},\cdots,n_{d})} )] &= 1, \\ 
	\dim H^{N-1}[\cpx{G}( \bdensl{W}_{d,(n_{1},\cdots,n_{d})} )] = \dim H^{N-1}[\cpx{G}( \bdensl{GHZ}_{d,(n_{1},\cdots,n_{d})} )] &= \prod_{i = 1}^{d} (n_{i} - 2),\\
	\dim H^{1}[\cpx{G}( \bdensl{W}_{3,n} )] = \dim H^{1}[\cpx{G}( \bdensl{GHZ}_{3,n} )] &= 6(n-1).
\end{align*}

\appendix

\section{Proof of Lemma~\ref{lem:equivariant_endo} \label{app:equivariant_endo_lemma}}
It is easy to verify that the right multiplication map $\mathpzc{r}$ is a well-defined map of $k$-algebras.
To see that $\mathpzc{r}$ descends to an injection on $(I/L)^{\mathrm{op}}$, suppose that there exists an $a \in I$ such that $\mathpzc{r}_{a} \equiv 0$; then $0 + L = \mathpzc{r}_{a}(1 + L) = a + L$ so that $a \in L$; moreover, for any $a \in L$ we have $\mathpzc{r}_{a} \equiv 0$ as $L$ is a left ideal.
To see that $\mathpzc{r}$ is surjective, note that
\begin{enumerate}
	\item $I/L \leq A/L$ is precisely the set of elements of $A/L$ that are annihilated by the left action of $L$;

	\item Every $\theta \in \End_{A} \left[{}_{A}(A/L) \right]$  takes $0 + L$ to $0 + L$.
\end{enumerate}	
The first statement follows directly from the definition of $I$.
To see the second statement note that $\theta(0 + L) = \theta( 0 \cdot 0 + L) = 0 \cdot \theta(0 + L) = 0 + L$.
Now, for any $\theta \in \End_{A} \left[{}_{A}(A/L) \right]$, define $c := \theta(1+L) \in A/L$, then for any $l \in L$ we have 
$l \cdot c = l \cdot \theta(1 + L) = \theta(l + L) = \theta(0 + L) = 0 + L$; hence, $c \in I/L$; hence right action by $c$ (or any of its lifts) is sensible.
Now that we have this, let $\widehat{c}$ denote any lift of $c$ to $I$, then for any $x \in A$ we have
\begin{align*}
	\theta(x + L) &= x \cdot \theta(1 + L)\\
				  &= x \cdot \widehat{c} + L\\
				  &= (x + L) \cdot \widehat{c}\\
				  &= \mathpzc{r}_{c}(x + L)
\end{align*}
so $\theta \equiv \mathpzc{r}_{\theta(1 + L)}$.

\section{Proof of Thm.~\ref{thm:projective_measurement_indistinguishability} \label{app:projective_indistinguishability}}
We begin by recalling the Born-rule: a way of assigning a probability measure on $\mathbb{R}$ to the data of a self-adjoint operator and a density state.
The probability measure on $\mathbb{R}$ is interpreted as the probability measure associated to a projective measurement of $x$ in the state $\dens{\rho}$.
True to the spirit of this paper we will work with finite dimensional Hilbert spaces where the theory simplifies significantly (knowledgeable readers should see the generalization to infinite dimensions using projection valued measures and the Borel functional calculus).
Let $\hilb$ be finite dimensional and let $x \in \algebra{\hilb}$ be self-adjoint (i.e.\ $x = x^{*}$), then $x$ admits a spectral decomposition:
\begin{align*}
	x = \sum_{\lambda \in \sigma(x)} \lambda \mathbbm{P}^{x}_{\lambda},
\end{align*}	
where $\sigma(x) \subseteq \mathbb{R}$ is the spectrum of $x$ (its eigenvalues in finite dimensions) and $\newmath{\mathbbm{P}^{x}_{\lambda}}$ is the projection onto the eigenspace associated to $\lambda$.
For each  $\lambda \in \mathbb{R}$ can define a probability measure on the finite discrete set $\sigma(x)$ given pointwise by
\begin{align*}
	\mu_{x}:\sigma(x) &\longrightarrow [0,1] \\
	\lambda &\longmapsto \Tr[\dens{\rho} \mathbbm{P}^{x}_{\lambda}].
\end{align*}
The Born rule states that $\mu_{x}(\lambda)$ is the probability of measuring $\lambda$ in a projective measurement of $x$.

Now suppose that $(x,y)$ is a pair of commuting operators, so their spectral decompositions are compatible in the sense that $\mathbbm{P}^{x}_{\lambda}$ and $\mathbbm{P}^{y}_{\eta}$ commute for all $(\lambda, \eta) \in \sigma(x) \times \sigma(y)$.
It makes sense to speak of a simultaneous projection-valued measurement of the pair $(x,y)$; the result of such a measurement is a pair $(\lambda, \eta) \in \sigma(x) \times \sigma(y)$, and the associated probability measure is given by\footnote{One can derive this expression using the fact that the density state after a measurement of $x$ (and before observing the value of the measurement) is given by $\sum_{\lambda \in \sigma(x)} \mu(\lambda) \Tr[P_{\lambda}\dens{\rho} P_{\lambda}]$, where $\rho$ is the state before the measurement.
When $x$ and $y$ commute, the Born-rule probabilities and final state do not depend on the time-order that they are measured.}
\begin{align*}
	\mu_{(x,y)}: \sigma(x) \times \sigma(y) &\longmapsto [0,1]\\
	 (\lambda,\eta) &\longmapsto \Tr\left[\dens{\rho} \mathbbm{P}^{x}_{\lambda} \mathbbm{P}^{y}_{\eta} \right].
\end{align*}
We would like to say that two commuting operators are maximally correlated if the measurement of one completely determines the value of the other.
A straightforward way to give this meaning is to declare that two operators are maximally correlated if the support of the probability measure $\mu_{(x,y)}$ (the set of points where it is non-vanishing) in $\sigma(x) \times \sigma(y) \subseteq \mathbb{R}$ defines an invertible relation: i.e.\ it is the graph of a function $f: \sigma(x) \rightarrow \sigma(y)$.
This accounts for any possible non-linear correlations in the sense that $f$ might be the restriction of some non-linear function $\mathbb{R} \rightarrow \mathbb{R}$; however, defining the operator 
\begin{align*}
	f(x) &= \sum_{\lambda \in \sigma(x)} f(\lambda) \mathbbm{P}_{\lambda},
\end{align*}
the pair $(f(x),y)$ has support contained along the diagonal of $\mathbb{R} \times \mathbb{R}$.
Thus, any ``non-linearly" correlated pair $(x,y)$ can be taken to a pair with linear correlation, and it suffices to consider pairs of operators whose joint probability measure is supported along the diagonal.
This motivates our definition Def.~\ref{def:EPR_pair}of an EPR pair, which we state again below for convenience.
\begin{definition}{}{}
	Let $\bdens{\rho}_{\sAB} = (\hilb_{\sA}, \hilb_{\sB}, \dens{\rho}_{\sAB})$ be a bipartite density state.
	A pair of self-adjoint operators $(a,b) \in \algebra{\hilb_{\sA}} \times \algebra{\hilb_{\sB}}$ is an \newword{EPR Pair} if the result of any projective measurement of $(a,b)$ lies on the diagonal of $\mathbb{R} \times \mathbb{R}$.
	That is, the probability measure
	\begin{align*}
		\mu_{(x,y)}: \sigma(a) \times \sigma(b) &\longrightarrow [0,1]\\
		(\lambda, \eta) &\longmapsto \Tr[\dens{\rho}_{\sAB} \mathbbm{P}^{a}_{\lambda} \otimes \mathbbm{P}^{b}_{\eta}]
	\end{align*}
	has support contained in $\{(\lambda, \lambda) \in \mathbb{R} \times \mathbb{R} \}$.
\end{definition}
Now let us prove Thm.~\ref{thm:projective_measurement_indistinguishability} which we restate below:
\begin{theorem}{}{}
	Let $(a,b) \in \mathtt{GNS}(\dens{\rho}_{\sA}) \times \mathtt{GNS}(\dens{\rho}_{\sB})$ be self-adjoint operators.\footnote{In particular we must have $a \in \supp_{\sA} \algebra{\hilb_{\sA}} \supp_{\sA} = \mathtt{Com}(\dens{\rho}_{\sA})$ and $b \in \supp_{\sB} \algebra{\hilb} \supp_{\sB}=\mathtt{Com}(\dens{\rho}_{\sB})$.}
	Then $(a,b) \in \ker(d^{0}_{\cpx{G}})$ if and only if $(a,b)$ is an EPR pair.
\end{theorem} 
\begin{proof}
	We proceed with a string of if and only if statements: $(a,b) \in \ker(d^{0}_{\mathrm{G}})$ if and only if
	\begin{equation}
		\begin{aligned}
			0 &= \Tr[\dens{\rho}_{\sAB} (a \otimes 1_{\sB} - 1_{\sA} \otimes b)^{*} (a \otimes 1_{\sB} - 1_{\sA} \otimes b)]\\
			  &= \sum_{(\lambda, \eta) \in \sigma(a) \times \sigma(b)} (\lambda - \eta)^{*} (\lambda - \eta) \Tr[\dens{\rho}_{\sAB} \mathbbm{P}^{a}_{\lambda} \otimes \mathbbm{P}^{b}_{\eta}].
		\end{aligned}
		\label{eq:spectral_decomp_expansion}
	\end{equation}
	but $(\lambda - \eta)^{*} (\lambda - \eta) \geq 0$ and $\mathbbm{P}^{a}_{\lambda} \otimes \mathbbm{P}^{b}_{\eta}$ is a positive operator, so $\Tr[\dens{\rho}_{\sAB} \mathbbm{P}^{a}_{\lambda} \otimes \mathbbm{P}^{b}_{\eta}] \geq 0$.
	As a result we have that the second line of \eqref{eq:spectral_decomp_expansion} vanishes if and only if $\Tr[\dens{\rho}_{\sAB} \mathbbm{P}^{a}_{\lambda} \otimes \mathbbm{P}^{b}_{\eta}] = 0$ for all $\lambda \neq \eta$, and the theorem follows. 
\end{proof}

\section{Proof of Theorem~\ref{thm:pure_bipartite_cohomology} \label{app:pure_bipartite_cohomology}}
In this section $\sX$ will denote one of the subsystems $\sA$ or $\sB$.
As in the statement of the theorem, begin with a pure bipartite state $\bdens{\rho}_{\sAB}$ with $\dens{\rho}_{\sAB} = \psi \otimes \psi^{\vee}$ for $\psi \in \hilb_{\sA} \otimes \hilb_{\sB}$, and let
\begin{align*}
	\psi = \sum_{i = 1}^{S} \sqrt{p_{i}} \xi^{\sA}_{i} \otimes \xi^{\sB}_{i},
\end{align*} 
denote a Schmidt decomposition of $\psi$.
Note that $(\xi^{\sX}_{i})_{i = 1}^{S}$ forms an orthonormal basis for the subspace $\image(\dens{\rho}_{\sX}) \leq \hilb_{\sX}$; by choosing an orthonormal basis $(\kappa_{j})_{j=1}^{K}$ for $\ker(\dens{\rho}_{\sX}) = \image(\dens{\rho}_{\sX})^{\perp}$ we then have an orthonormal basis $(\chi_{\mu})_{\mu = 1}^{S+K}$ for $\hilb_{X} \cong \image(\dens{\rho}_{\sX}) \oplus \ker(\dens{\rho}_{\sX})$ defined by
\begingroup
\renewcommand{\arraystretch}{1.5}
\begin{align*}
	\chi_{\mu}^{\sX} &= 
	\left \{
		\begin{array}{ll}
			\xi^{\sX}_{\mu},    & \text{if $1 \leq \mu \leq S$}\\
			\kappa^{\sX}_{\mu}, & \text{if $S + 1 \leq \mu \leq S+K$}
		\end{array}
	\right. .
\end{align*}
\endgroup
Using this orthonormal basis, define
\begin{align*}
	\mathbbm{e}^{\sX}_{\mu \nu} &:= \chi_{\mu}^{\sX} \otimes \left(\chi^{\sX}_{\nu}\right)^{\vee}
\end{align*}
which can be thought of as an element of $\algebra{\hilb_{\sX}}$.
To prove the theorem it suffices to show that $\ker(d^{0}_{\mathrm{G}})$ is given by the subspace 
\begin{align*} 
	\mathrm{span}_{\mathbb{C}} \left\{(\mathbbm{e}_{ij}^{\sA}, \mathbbm{e}_{ij}^{\sB}): 1 \leq i \leq S, \, 1 \leq j \leq S \right\}.
\end{align*}
We will do this by utilizing Prop.~\ref{thm:kernel_covariance}, which states that:
\begin{align*}
	\ker(d^{0}_{\cpx{G}}) =
	\left \{(a,b) \in \mathtt{GNS}(\dens{\rho}_{\sA}) \times \mathtt{GNS}(\dens{\rho}_{\sA}):
	\text{$\Cov(a,b) = \Var(a) = \Var(b)$ and $\Tr[\dens{\rho}_{\sA}a] = \Tr[\dens{\rho}_{\sB}b]$} \right \}.
\end{align*}
Begin by noting that, with the choices and definitions above, we can identify $\algebra{\hilb_{\sA}}$ and $\algebra{\hilb_{\sB}}$ with the algebra of $(S+K) \times (S+K)$-matrices using the isomorphisms (of $C^{*}$-algebras) 
\begin{align*}
	\phi_{\sX}: \algebra{\hilb_{\sX}} & \longrightarrow \mathrm{Mat} \left[ \mathbb{C}^{S} \oplus \mathbb{C}^{K} \right] \\
                          a           & \longmapsto \sum_{\mu,\nu =1}^{S+K} \Tr[a \mathbbm{e}^{\sX}_{\mu \nu}] \mathbb{E}_{\mu \nu}.
\end{align*}
where, $\sX \in \{\sA, \sB\}$, and $\mathbb{E}_{\mu \nu}$ is the standard matrix with a 1 in position $(\mu, \nu)$ and zeros elsewhere.
Next observe that for any $r \in \algebra{\hilb_{\sX}}$, we have 
\begin{align*}
	\Tr\left[\dens{\rho}_{\sX} r \right] = \Tr[\mathbbm{D} \phi^{\sX}(r)]
\end{align*}
where $\mathbbm{D}$ is the $(S+K) \times (S+K)$ diagonal matrix 
\begin{align*}
	\mathbbm{D} = \sum_{i = 1}^{S} p_{i} \mathbbm{E}_{ii}.
\end{align*}
This allows us to easily rewrite the covariance and variances in terms of traces of products of matrices: for any $(a,b) \in \algebra{\hilb_{\sA}} \times \algebra{\hilb_{\sB}}$ and $r \in \algebra{\hilb_{\sX}}$: 
\begin{align*}
	\Cov(a,b)     & = \langle \phi^{\sA}(\twid{a}), \phi^{\sB}(\twid{b}) \rangle,\\
	\Var_{\sX}(r) & = \langle \phi^{\sX}(\twid{r}), \phi^{\sX}(\twid{r}) \rangle.
\end{align*}
where $\langle - , - \rangle$ is the sesquilinear form
\begin{align*}
	\langle - , - \rangle: \mathrm{Mat} \left[ \mathbb{C}^{S} \oplus \mathbb{C}^{K} \right] \times \mathrm{Mat} \left[ \mathbb{C}^{S} \oplus \mathbb{C}^{K} \right] &\longrightarrow \mathbb{C}\\
	(\mathbbm{X}, \mathbbm{Y}) &\longmapsto \Tr\left[ \mathbbm{D} \mathbbm{X}^{*} \mathbbm{Y} \right].
\end{align*}
and 
\begin{align*}
	\twid{r} := r - \Tr\left[\dens{\rho}_{\sX} r\right] 1_{\sX}.
\end{align*}
Suppose now that $(a,b) \in \algebra{\hilb_{\sA}} \times \algebra{\hilb_{\sB}}$ are such that
\begin{align}
	\Cov(a,b) = \Var_{\sA}(a) = \Var_{\sB}(b)
	\label{eq:cov_eq_var}
\end{align}
then, by the Cauchy-Schwarz inequality applied to the sesquilinear form $\langle - , - \rangle$, we have
\begin{align}
	\phi^{\sA}(\twid{a}) = \phi^{\sB}(\twid{b}) + \mathbbm{N}
	\label{eq:CS_result}
\end{align}
where $\mathbbm{N}$ is an element of the subspace $N$ of null vectors defined by 
\begin{align*}
	N &:= \{\mathbbm{X}: 0 = \langle \mathbbm{X}, \mathbbm{X} \rangle \}\\
	  &= \mathrm{span}_{\mathbb{C}} \{ \mathbbm{E}_{\kappa \mu}: S+1 \leq \kappa \leq S+K,\, 1 \leq \mu \leq S+K \}.
\end{align*}
Letting $\mathbbm{1}$ be the identity matrix, then \eqref{eq:CS_result} expands to
\begin{align*}
	\phi^{\sA}(a) - \Tr[\dens{\rho}_{\sA} a] \mathbbm{1} = \phi^{\sB}(b) - \Tr[\dens{\rho}_{\sB} b] \mathbbm{1} + \mathbbm{N}.
\end{align*}
Now, imposing the further condition that
\begin{align}
	\Tr[\dens{\rho}_{\sA} a] = 	\Tr[\dens{\rho}_{\sB} b]
	\label{eq:equal_expectation}
\end{align}
we must have
\begin{align}
	\phi^{\sA}(a)  = \phi^{\sB} (b) + \mathbbm{N}.
	\label{eq:CS_equal_expectation}
\end{align}
Condition \eqref{eq:CS_equal_expectation} follows for any pair $(a,b) \in \algebra{\hilb_{\sA}} \times \algebra{\hilb_{\sB}}$ satisfying  \eqref{eq:cov_eq_var} and \eqref{eq:equal_expectation}.
However, there is one further constraint to consider:
\begin{align*}
	(a,b) \in \mathtt{GNS}(\dens{\rho}_{\sA}) \times \mathtt{GNS}(\dens{\rho}_{\sB}) \leq \algebra{\hilb_{\sA}} \times \algebra{\hilb_{\sB}}.
\end{align*}
Under the isomorphisms $\phi_{\sX}$, the images of the subspaces $\mathtt{GNS}(\dens{\rho}_{\sX}) \leq \algebra{\hilb_{\sX}} ,\, \sX \in \{\sA, \sB \}$ are both equal to the subspace 
\begin{align*}
	G := \mathrm{span}_{\mathbb{C}} \{ \mathbbm{E}_{i\mu}: 1 \leq i  \leq S, \, 1 \leq \mu \leq S+K \} \leq \mathrm{Mat} \left[ \mathbb{C}^{S} \oplus \mathbb{C}^{K} \right] 
\end{align*}
Noting that $G \cap N = 0$ and $\phi^{\sA}(a),\, \phi^{\sB}(b) \in G$, it follows that the element $\mathbbm{N} \in N$ of \eqref{eq:CS_equal_expectation} must vanish.  Hence, 
\begin{align*}
	\phi^{\sA}(a) = \phi^{\sB}(b).
\end{align*}
It follows that we must have 
\begin{align*}
	a &= \sum_{i = 1}^{S} r_{ij} \mathbbm{e}^{\sA}_{ij}\\
	b &= \sum_{i = 1}^{S} r_{ij} \mathbbm{e}^{\sB}_{ji}.
\end{align*}
for some $(r_{ij})_{i,j = 1}^{S} \in \mathbb{C}$.
These are precisely elements of $\mathrm{span}_{\mathbb{C}} \{(\mathbbm{e}_{ij}^{\sA}, \mathbbm{e}_{ij}^{\sB}) \}$.
Summarizing, we have shown $\ker(d_{\cpx{G}}^{0}) = \mathrm{span}_{\mathbb{C}} \{(\mathbbm{e}_{ij}^{\sA}, \mathbbm{e}_{ij}^{\sB}) \}$. 

\section{Proof of Proposition~\ref{prop:triviality_means_coboundary} \label{app:proof_prop_triviality_means_coboundary}}

Let $\bdens{\rho}_{P}$ be an arbitrary multipartite density state and $\twid{\bdens{\rho}}_{P}$ be its fully factorizable form (c.f.\ \eqref{eq:fully_factorizable_form}). 
Recall that 
\begin{align*}
	\mathrm{Triv}^{k}(\bdens{\rho}_{P}) =  \left\{R \in \prod_{|T|=k+1} \mathtt{B}(\bdens{\rho}_{T}) \colon \text{$\exists Q \in \ker\left(\twid{d}^{k}_{\cpx{C}}\right)$ s.t. $R_{T} \overset{\mathtt{B}}{\sim}_{T} Q_{T},\, \forall T \subseteq P$ with $|T| = k+1$} \right\},
\end{align*}
for $0 \leq k \leq |P|-2$, and, to remove some notational ambiguity, we are letting $d_{\cpx{C}}$ denote the coboundary associated to the complex $\cpx{C}(\bdens{\rho}_{P})$ and $\twid{d}_{\cpx{C}}$ denote the coboundary associated to the complex $\cpx{C}(\twid{\bdens{\rho}}_{P})$.
We want to show
\begin{align*}
	\mathrm{Triv}^{k}(\bdens{\rho}_{P}) = \image \left(d_{\cpx{C}}^{k-1} \right).
\end{align*}
By Thm.~\ref{thm:support_fact_multipartite_cohomologies}, the complex $\cpx{C}^{k}(\twid{\bdens{\rho}}_{P})$ has vanishing cohomology in all degrees $< |P|-1$; hence,
\begin{align}	
	\ker(\twid{d}^{k}_{\cpx{C}}) = \image\left(\twid{d}_{\cpx{C}}^{k-1}\right),
	\label{eq:ker_is_image}
\end{align}
for $0 \leq k \leq |P|-2$.

Next, let us choose $K \in \mathrm{Triv}^{k}(\bdens{\rho}_{P})$, then for every $T \subseteq P$ with $|T|=k+1$ we have $K_{T} \overset{\mathtt{B}}{\sim}_{T} Q_{T}$ for some $Q \in \mathrm{EQ}^{k}(\bdens{\rho}_{P})$, but by \eqref{eq:ker_is_image}, we have $Q \in \image\left(\twid{d}_{\cpx{C}}^{k-1}\right)$; hence, there exists an $F \in \cpx{C}^{k-1}(\twid{\bdens{\rho}}_{P})$ such that 
\begin{align}
	K_{T} \overset{\mathtt{B}}{\sim}_{T} Q_{T} = \sum_{m = 0}^{k} (-1)^{m} \Sigma_{(T,m)} \left[ F_{\partial_{m} T} \otimes \left(\otimes_{s \in T(m)} \supp_{s} \right) \right] \otimes_{t \in T} \supp_{t}
	\label{eq:K_sim_coboundary}
\end{align}
for all $T \subseteq P$ with $|T| = k+1$.
Let us treat the case that $\cpx{C}(\bdens{\rho}_{P})$ is the GNS complex $\cpx{G}(\bdens{\rho}_{P})$ so that $\mathtt{B} = \mathtt{GNS}$, then \eqref{eq:K_sim_coboundary} is the statement that
\begin{align*}
	K_{T} &= \left[\sum_{m = 0}^{k} (-1)^{m} \Sigma_{(T,m)} \left[ F_{\partial_{m} T} \otimes \left(\otimes_{s \in T(m)} \supp_{s} \right) \right] \otimes_{t \in T} \supp_{t} \right]\supp_{T}
\end{align*}
using the compatibility of supports lemma, we can write
\begin{align*}
	K_{T} &= \sum_{m = 0}^{k} (-1)^{m} \Sigma_{(T,m)} \left[ F_{\partial_{m} T} \otimes \left(\otimes_{s \in T(m)} \supp_{s} \right) \right] \supp_{T}\\
		  &= \sum_{m = 0}^{k} (-1)^{m} \Sigma_{(T,m)} \left[ F_{\partial_{m} T} \left(\otimes_{w \in \partial_{m} T} \supp_{w} \right) \otimes \left(\otimes_{s \in T(m)} \supp_{s} \right) \right] \supp_{T}\\
		  &= \sum_{m = 0}^{k} (-1)^{m} \Sigma_{(T,m)} \left[ F_{\partial_{m} T} \left(\otimes_{w \in \partial_{m} T} \supp_{w} \right) \otimes \left(\otimes_{s \in T(m)} \supp_{s} \right) \right] \left[\supp_{\partial_{m} T} \otimes \supp_{T(m)} \right] \supp_{T}\\
		  &= \sum_{m = 0}^{k} (-1)^{m} \Sigma_{(T,m)} \left[ F_{\partial_{m} T}\supp_{\partial_{m} T} \otimes \supp_{T(m)} \right] \supp_{T}.
\end{align*}
This gives an element $S \in \cpx{C}^{k-1}(\bdens{\rho}_{P})$, defined componentwise by $S_{T} := F_{T} \supp_{T}$, that satisfies $K = d^{k-1}_{\cpx{G}} S$.
In the case that $\cpx{C}(\bdens{\rho}_{P}) = \cpx{G}(\bdens{\rho}_{P})$ we have thus shown: 
\begin{align*}
	\mathrm{Triv}_{k}(\bdens{\rho}_{P}) \leq \image \left(d_{\cpx{G}}^{k-1}: \cpx{G}^{k-1}(\bdens{\rho}_{P}) \rightarrow \cpx{G}^{k}(\bdens{\rho}_{P}) \right).
\end{align*}
The argument for the above inclusion when $\cpx{C}(\bdens{\rho}_{P}) = \cpx{E}(\bdens{\rho}_{P})$ is nearly identical.\footnote{One need only multiply all equations on the left by the appropriate support projections.} 

To show the opposite inclusion, let $K \in \image(d_{\cpx{C}}^{k-1})$.
If $\cpx{C}(\bdens{\rho}_{P})$ is the GNS complex $\cpx{G}(\bdens{\rho}_{P})$, then there exists an $F \in \cpx{G}^{k-1}(\bdens{\rho}_{P})$ such that
\begin{align*}
	K_{T} &= \sum_{m = 0}^{k-2} (-1)^{m} \Sigma_{(T,m)} \left[ F_{\partial_{m} T} \otimes \supp_{T(m)} \right] \supp_{T},
\end{align*}
however, $\cpx{G}^{k-1}(\bdens{\rho}_{P}) \leq \cpx{G}^{k-1}(\twid{\bdens{\rho}}_{P})$; so, with an implicit use of the compatibility of supports lemma, we can write this as: 
\begin{align*}
	K_{T} &= \left( \twid{d}^{k-1}F \right) \supp_{T}.
\end{align*}
This is just the statement that
\begin{align*}
	K_{T} \overset{\mathtt{GNS}}{\sim}_{T} Q_{T},\: \forall |T| = k+1,
\end{align*}
for some $Q_{T} \in \image(\twid{d}^{k-1}_{\cpx{G}})$; hence, $K \in \mathrm{Triv}^{k}(\bdens{\rho}_{P})$.
Once again, the argument for the case that $\cpx{C}(\bdens{\rho}_{P})$ is the commutant complex $\cpx{E}(\bdens{\rho}_{P})$ is nearly identical.

\section{Mutual Information and The Incidence Algebra of Posets \label{app:incidence_alg_mutual_info}}
Let $P$ be a finite set (we do not require a total order at this point).  Consider the set of functions that assign a real number to each inclusion $V \subseteq T$ of subsets $V, T \subseteq P$ (allowing for the case that $V=T$).  We can equip this set with the structure of an $\mathbb{R}$-algebra---denoted $A_{P}$ in the following way: scalar multiplication and addition are defined via: 
\begin{align*}
	\lambda \cdot f : V \subseteq T &\longmapsto \lambda f(V \subseteq T),\\
	f + g : V \subseteq T &\longmapsto f(V \subseteq T) + g(V \subseteq T),
\end{align*}
for $f,g \in A_{P},\, \lambda \in \mathbb{R}$, and $V, T$ subsets with $V \subseteq T$; the product is defined by
\begin{align*}
	f*g : V \subseteq T &\longmapsto \sum_{\{W: V \subseteq W \subseteq T \}} f(V \subseteq W) g(W \subseteq T).
\end{align*}
The corresponding algebra is sometimes referred to as the ``incidence algebra" associated to the poset of subsets of $P$ (it is also a special case of the ``category algebra" associated to the poset of subsets thought of as a category).
The incidence algebra has an identity element given by the ``delta function" $\delta$: defined to be valued $1$ on all self-inclusions $V \subseteq V$ and zero otherwise.  Functions that are non-vanishing on all self-inclusions $V \subseteq V$ are invertible in this algebra.  The constant function $\zeta$ (also called the ``zeta function") that takes the value $1$ on all inclusions has an inverse $\mu$ known as the M\"{o}bius-mu function: 
\begin{align*}
	\mu: V \subseteq T \longmapsto \sum_{V \subseteq W \subseteq T} (-1)^{|T| - |W|}.
\end{align*}
The $\mathbb{R}$-vector space $\mathrm{Fun}(P, \mathbb{R})$ of real valued functions on $P$ forms a right module for $A_{P}$ with action specified by: 
\begin{align*}
	m \cdot f : V &\longmapsto \sum_{\emptyset \subseteq W \subseteq V} m(W) f(W \subseteq V),
\end{align*}
for $m \in \mathrm{Fun}(P, \mathbb{R})$ and $f \in A_{P}$.
Now, suppose we have a multipartite density state $\bdens{\rho}_{P}$ associated to the set of tensor factors $P$ (equipped with some total order); define the function $s \in \mathrm{Fun}(P,\mathbb{R})$ as the function that assigns associated von Neumann entropies to each subset: $s(V) := S^{\mathrm{vN}}(\dens{\rho}_{V})$.
One can verify that the mutual information (with the additional global sign $(-1)^{|P|}$) arises by the right action of the M\"{o}bius-mu function on $s$: 
\begin{align*}
	s \cdot \mu: V \longmapsto I^{-}(\bdens{\rho}_{V}).
\end{align*}
One advantage of this presentation is that we can use the identity $\mu * \zeta = \delta$ to recover the function $s$: as a result we recover the identity: 
\begin{align*}
	S(\bdens{\rho}_{V}) = \sum_{\emptyset \subseteq T \subseteq V} I^{-}(\bdens{\rho}_{V}).
\end{align*}
This is an example of a M\"{o}bius inversion formula in the context of incidence algebras associated to posets.

\section{Generators for GNS cohomology of tripartite GHZ and W-states \label{app:ghz_w_cohomology_generators}}
In the following, we work with states in the Hilbert space defined by $\mathrm{span}_{\mathbb{C}} \left \{ \ket{0}, \ket{1} \right\} ^{\otimes 3}$ (where the span is taken to be orthonormal).
All operators are identified with matrices using the ordered basis $(\ket{0}, \ket{1})$ for Hilbert spaces associated to a single tensor factor; we use the lexicographical ordering to order the bases for Hilbert spaces associated to multiple tensor factors: e.g.\ we have the ordered basis $(\ket{00}, \ket{01}, \ket{10}, \ket{11})$ for Hilbert spaces associated to two tensor factors.

\subsection{The GHZ State}
The tripartite GHZ state is given by:
\begin{align*}
	\frac{1}{\sqrt{2}}\left( \ket{0}^{\otimes 3} + \ket{1}^{\otimes 3}  \right) \in  \mathrm{span}_{\mathbb{C}} \left \{ \ket{0}, \ket{1} \right\} ^{\otimes 3}. 
\end{align*}

\subsubsection{Generators for \texorpdfstring{$H^{0}$}{Degree 0 GNS Cohomology}}
A 1-cocycle representative for the one dimensional linear family spanning the first GNS cohomology component is   
\begin{align*}
	\left\{\sA\mapsto \left(
			\begin{array}{cc}
				1 & 0 \\
				0 & -1 \\
			\end{array}
			\right),\sB\mapsto \left(
			\begin{array}{cc}
				1 & 0 \\
				0 & -1 \\
			\end{array}
			\right),\sC\mapsto \left(
			\begin{array}{cc}
				1 & 0 \\
				0 & -1 \\
			\end{array}
	\right)\right\}
\end{align*}
The support projections $\supp_{\sA},\,\supp_{\sB},\,$ and $\supp_{\sC}$ are given by the $2 \times 2$ identity matrix.

\subsubsection{Generators for \texorpdfstring{$H^{1}$}{Degree 1 GNS cohomology}}
1-cocycle representatives for each of the six linearly independent cohomology classes of the first GNS cohomology component are given by:
\begin{align*}
	&\left\{\sAB \mapsto \left(
			\begin{array}{cccc}
				0 & 0 & 0 & 0 \\
				0 & 0 & 0 & 0 \\
				-\frac{1}{2} & 0 & 0 & 0 \\
				0 & 0 & 0 & 0 \\
			\end{array}
			\right),\sAC \mapsto \left(
			\begin{array}{cccc}
				0 & 0 & 0 & 0 \\
				0 & 0 & 0 & 0 \\
				\frac{1}{2} & 0 & 0 & 0 \\
				0 & 0 & 0 & 0 \\
			\end{array}
			\right),\sBC \mapsto \left(
			\begin{array}{cccc}
				0 & 0 & 0 & 1 \\
				0 & 0 & 0 & 0 \\
				0 & 0 & 0 & 0 \\
				0 & 0 & 0 & 0 \\
			\end{array}
	\right)\right\} \\
	&\left\{\sAB\mapsto \left(
			\begin{array}{cccc}
				0 & 0 & 0 & 0 \\
				0 & 0 & 0 & 0 \\
				0 & 0 & 0 & 1 \\
				0 & 0 & 0 & 0 \\
			\end{array}
			\right),\sAC\mapsto \left(
			\begin{array}{cccc}
				0 & 0 & 0 & 0 \\
				0 & 0 & 0 & 0 \\
				0 & 0 & 0 & 0 \\
				2 & 0 & 0 & 0 \\
			\end{array}
			\right),\sBC\mapsto \left(
			\begin{array}{cccc}
				0 & 0 & 0 & 0 \\
				0 & 0 & 0 & 1 \\
				0 & 0 & 0 & 0 \\
				0 & 0 & 0 & 0 \\
			\end{array}
	\right)\right\} \\
	&\left\{\sAB\mapsto \left(
			\begin{array}{cccc}
				0 & 0 & 0 & -2 \\
				0 & 0 & 0 & 0 \\
				0 & 0 & 0 & 0 \\
				0 & 0 & 0 & 0 \\
			\end{array}
			\right),\sAC\mapsto \left(
			\begin{array}{cccc}
				0 & 0 & 0 & 0 \\
				-1 & 0 & 0 & 0 \\
				0 & 0 & 0 & 0 \\
				0 & 0 & 0 & 0 \\
			\end{array}
			\right),\sBC\mapsto \left(
			\begin{array}{cccc}
				0 & 0 & 0 & 0 \\
				1 & 0 & 0 & 0 \\
				0 & 0 & 0 & 0 \\
				0 & 0 & 0 & 0 \\
			\end{array}
	\right)\right\} \\
	&\left\{\sAB\mapsto \left(
			\begin{array}{cccc}
				0 & 0 & 0 & 0 \\
				0 & 0 & 0 & 0 \\
				0 & 0 & 0 & 0 \\
				-2 & 0 & 0 & 0 \\
			\end{array}
			\right),\sAC\mapsto \left(
			\begin{array}{cccc}
				0 & 0 & 0 & 0 \\
				0 & 0 & 0 & 0 \\
				0 & 0 & 0 & -1 \\
				0 & 0 & 0 & 0 \\
			\end{array}
			\right),\sBC\mapsto \left(
			\begin{array}{cccc}
				0 & 0 & 0 & 0 \\
				0 & 0 & 0 & 0 \\
				0 & 0 & 0 & 1 \\
				0 & 0 & 0 & 0 \\
			\end{array}
	\right)\right\} \\
	&\left\{\sAB\mapsto \left(
			\begin{array}{cccc}
				0 & 0 & 0 & 0 \\
				1 & 0 & 0 & 0 \\
				0 & 0 & 0 & 0 \\
				0 & 0 & 0 & 0 \\
			\end{array}
			\right),\sAC\mapsto \left(
			\begin{array}{cccc}
				0 & 0 & 0 & 2 \\
				0 & 0 & 0 & 0 \\
				0 & 0 & 0 & 0 \\
				0 & 0 & 0 & 0 \\
			\end{array}
			\right),\sBC\mapsto \left(
			\begin{array}{cccc}
				0 & 0 & 0 & 0 \\
				0 & 0 & 0 & 0 \\
				1 & 0 & 0 & 0 \\
				0 & 0 & 0 & 0 \\
			\end{array}
	\right)\right\} \\
	&\left\{\sAB\mapsto \left(
			\begin{array}{cccc}
				0 & 0 & 0 & 0 \\
				0 & 0 & 0 & -\frac{1}{2} \\
				0 & 0 & 0 & 0 \\
				0 & 0 & 0 & 0 \\
			\end{array}
			\right),\sAC\mapsto \left(
			\begin{array}{cccc}
				0 & 0 & 0 & 0 \\
				0 & 0 & 0 & \frac{1}{2} \\
				0 & 0 & 0 & 0 \\
				0 & 0 & 0 & 0 \\
			\end{array}
			\right),\sBC\mapsto \left(
			\begin{array}{cccc}
				0 & 0 & 0 & 0 \\
				0 & 0 & 0 & 0 \\
				0 & 0 & 0 & 0 \\
				1 & 0 & 0 & 0 \\
			\end{array}
	\right)\right\}
\end{align*}
The support projection of relevance (for checking the cocycle condition) is
\begin{align*}
	\supp_{\sABC} &= \left(
\begin{array}{cccccccc}
 \frac{1}{2} & 0 & 0 & 0 & 0 & 0 & 0 & \frac{1}{2} \\
 0 & 0 & 0 & 0 & 0 & 0 & 0 & 0 \\
 0 & 0 & 0 & 0 & 0 & 0 & 0 & 0 \\
 0 & 0 & 0 & 0 & 0 & 0 & 0 & 0 \\
 0 & 0 & 0 & 0 & 0 & 0 & 0 & 0 \\
 0 & 0 & 0 & 0 & 0 & 0 & 0 & 0 \\
 0 & 0 & 0 & 0 & 0 & 0 & 0 & 0 \\
 \frac{1}{2} & 0 & 0 & 0 & 0 & 0 & 0 & \frac{1}{2} \\
\end{array}
\right).	
\end{align*}

\subsection{The W-state}
The tripartite W-state is given by:
\begin{align*}
	\frac{1}{\sqrt{3}} \left( \ket{100} + \ket{010} + \ket{001} \right).
\end{align*}

\subsubsection{Generators for \texorpdfstring{$H^{0}$}{Degree 0 GNS Cohomology}}
A 1-cocycle representative for the one dimensional linear family spanning the first GNS cohomology component is   
\begin{align*}
	\left\{\sA \mapsto \left(
			\begin{array}{cc}
				0 & 1 \\
				0 & 0 \\
			\end{array}
			\right),\sB \mapsto \left(
			\begin{array}{cc}
				0 & 1 \\
				0 & 0 \\
			\end{array}
			\right),\sC \mapsto \left(
			\begin{array}{cc}
				0 & 1 \\
				0 & 0 \\
			\end{array}
	\right)\right\}
\end{align*}
The support projections $\supp_{\sA},\,\supp_{\sB},\,$ and $\supp_{\sC}$ are given by the $2 \times 2$ identity matrix.

\subsubsection{Generators for \texorpdfstring{$H^{1}$}{Degree 1 GNS cohomology}}
1-cocycle representatives for each of the six linearly independent cohomology classes of the first GNS cohomology component are given by:
\begin{align*}
	&\left\{\sAB\mapsto \left(
			\begin{array}{cccc}
				1 & 0 & 0 & 0 \\
				0 & -\frac{1}{4} & -\frac{1}{4} & 0 \\
				0 & -\frac{1}{4} & -\frac{1}{4} & 0 \\
				0 & 0 & 0 & 0 \\
			\end{array}
			\right),\sAC\mapsto \left(
			\begin{array}{cccc}
				-1 & 0 & 0 & 0 \\
				0 & \frac{1}{4} & \frac{1}{4} & 0 \\
				0 & \frac{1}{4} & \frac{1}{4} & 0 \\
				0 & 0 & 0 & 0 \\
			\end{array}
			\right),\sBC\mapsto \left(
			\begin{array}{cccc}
				1 & 0 & 0 & 0 \\
				0 & -\frac{1}{4} & -\frac{1}{4} & 0 \\
				0 & -\frac{1}{4} & -\frac{1}{4} & 0 \\
				0 & 0 & 0 & 0 \\
			\end{array}
	\right)\right\},\\
	&\left\{\sAB\mapsto \left(
			\begin{array}{cccc}
				0 & 0 & 0 & 0 \\
				-\frac{5}{6} & 0 & 0 & 0 \\
				-\frac{13}{12} & 0 & 0 & 0 \\
				0 & \frac{5}{24} & \frac{5}{24} & 0 \\
			\end{array}
			\right),\sAC\mapsto \left(
			\begin{array}{cccc}
				0 & 0 & 0 & 0 \\
				-\frac{5}{6} & 0 & 0 & 0 \\
				\frac{5}{12} & 0 & 0 & 0 \\
				0 & -\frac{1}{24} & -\frac{1}{24} & 0 \\
			\end{array}
			\right),\sBC\mapsto \left(
			\begin{array}{cccc}
				0 & 0 & 0 & 0 \\
				1 & 0 & 0 & 0 \\
				0 & 0 & 0 & 0 \\
				0 & 0 & 0 & 0 \\
			\end{array}
	\right)\right\}, \\
	&\left\{\sAB\mapsto \left(
			\begin{array}{cccc}
				0 & 0 & 0 & 0 \\
				\frac{5}{6} & 0 & 0 & 0 \\
				-\frac{5}{12} & 0 & 0 & 0 \\
				0 & \frac{1}{24} & \frac{1}{24} & 0 \\
			\end{array}
			\right),\sAC\mapsto \left(
			\begin{array}{cccc}
				0 & 0 & 0 & 0 \\
				\frac{5}{6} & 0 & 0 & 0 \\
				\frac{13}{12} & 0 & 0 & 0 \\
				0 & -\frac{5}{24} & -\frac{5}{24} & 0 \\
			\end{array}
			\right),\sBC\mapsto \left(
			\begin{array}{cccc}
				0 & 0 & 0 & 0 \\
				0 & 0 & 0 & 0 \\
				1 & 0 & 0 & 0 \\
				0 & 0 & 0 & 0 \\
			\end{array}
	\right)\right\},\\
	&\left\{\sAB\mapsto \left(
			\begin{array}{cccc}
				0 & 0 & 0 & 0 \\
				-\frac{1}{2} & 0 & 0 & 0 \\
				-\frac{5}{4} & 0 & 0 & 0 \\
				0 & \frac{5}{8} & \frac{5}{8} & 0 \\
			\end{array}
			\right),\sAC\mapsto \left(
			\begin{array}{cccc}
				0 & 0 & 0 & 0 \\
				\frac{1}{2} & 0 & 0 & 0 \\
				\frac{5}{4} & 0 & 0 & 0 \\
				0 & -\frac{5}{8} & -\frac{5}{8} & 0 \\
			\end{array}
			\right),\sBC\mapsto \left(
			\begin{array}{cccc}
				0 & 0 & 0 & 0 \\
				0 & 0 & 0 & 0 \\
				0 & 0 & 0 & 0 \\
				0 & \frac{1}{2} & \frac{1}{2} & 0 \\
			\end{array}
	\right)\right\},\\
	&\left\{\sAB\mapsto \left(
			\begin{array}{cccc}
				0 & 0 & 0 & 0 \\
				0 & 0 & 0 & 0 \\
				0 & 0 & 0 & 0 \\
				-1 & 0 & 0 & 0 \\
			\end{array}
			\right),\sAC\mapsto \left(
			\begin{array}{cccc}
				0 & 0 & 0 & 0 \\
				0 & 0 & 0 & 0 \\
				0 & 0 & 0 & 0 \\
				0 & 0 & 0 & 0 \\
			\end{array}
			\right),\sBC\mapsto \left(
			\begin{array}{cccc}
				0 & 0 & 0 & 0 \\
				0 & 0 & 0 & 0 \\
				0 & 0 & 0 & 0 \\
				1 & 0 & 0 & 0 \\
			\end{array}
	\right)\right\},\\
	&\left\{\sAB\mapsto \left(
			\begin{array}{cccc}
				0 & 0 & 0 & 0 \\
				0 & 0 & 0 & 0 \\
				0 & 0 & 0 & 0 \\
				-1 & 0 & 0 & 0 \\
			\end{array}
			\right),\sAC\mapsto \left(
			\begin{array}{cccc}
				0 & 0 & 0 & 0 \\
				0 & 0 & 0 & 0 \\
				0 & 0 & 0 & 0 \\
				-1 & 0 & 0 & 0 \\
			\end{array}
			\right),\sBC\mapsto \left(
			\begin{array}{cccc}
				0 & 0 & 0 & 0 \\
				0 & 0 & 0 & 0 \\
				0 & 0 & 0 & 0 \\
				0 & 0 & 0 & 0 \\
			\end{array}
	\right)\right\}.
\end{align*}

The support projection of relevance is:
\begin{align*}
	\supp_{\sABC} &= \left(
\begin{array}{cccccccc}
 0 & 0 & 0 & 0 & 0 & 0 & 0 & 0 \\
 0 & \frac{1}{3} & \frac{1}{3} & 0 & \frac{1}{3} & 0 & 0 & 0 \\
 0 & \frac{1}{3} & \frac{1}{3} & 0 & \frac{1}{3} & 0 & 0 & 0 \\
 0 & 0 & 0 & 0 & 0 & 0 & 0 & 0 \\
 0 & \frac{1}{3} & \frac{1}{3} & 0 & \frac{1}{3} & 0 & 0 & 0 \\
 0 & 0 & 0 & 0 & 0 & 0 & 0 & 0 \\
 0 & 0 & 0 & 0 & 0 & 0 & 0 & 0 \\
 0 & 0 & 0 & 0 & 0 & 0 & 0 & 0 \\
\end{array}
\right).
\end{align*}	

\printbibliography

\end{document}